\documentclass{AIMS}
 \usepackage{paralist}
\usepackage{amsmath}
\usepackage{stmaryrd}
\usepackage{cite}
\usepackage{epsfig}
\usepackage{graphics}
\usepackage{amssymb}
\usepackage{verbatim,enumerate}
 \usepackage[colorlinks=true]{hyperref}
\hypersetup{urlcolor=blue, citecolor=red}

\def\currentvolume{X}
 \def\currentissue{X}
  \def\currentyear{200X}
   \def\currentmonth{XX}
    \def\ppages{X--XX}
\theoremstyle{plain}
\newtheorem{theorem}{Theorem}[section]
\newtheorem{lemma}{Lemma}[section]
\newtheorem{proposition}{Proposition}[section]

\theoremstyle{definition}
\newtheorem{definition}{Definition}[section]
\newtheorem{example}{Example}[section]
\newtheorem{remark}{Remark}[section]

\def\thesection{\arabic{section}}
\def\thesubsection{\thesection.\arabic{subsection}}
\def\thesubsubsection{\thesubsection.\arabic{subsubsection}}
\def\theequation {\thesection.\arabic{equation}}
\def\thetable{\thesection.\arabic{table}}
\def\thefigure{\thesection.\arabic{figure}}
\def\qed{\mbox{}\hfill {\small \fbox{}}\\}

\def\be{\begin{equation}\displaystyle}
\def\ee{\end{equation}}
\def\bel{\begin{equation} \displaystyle \begin{array}{l} }
\def\eel{\end{array} \end{equation} }
\def\bell{\begin{equation} \displaystyle \begin{array}{ll}  }
\def\eell{\end{array} \end{equation} }

\def\bea{\begin{eqnarray}}
\def\eea{\end{eqnarray} }
\def\beas{\begin{eqnarray*}}
\def\eeas{\end{eqnarray*} }

\def\NN{\mathbb{N}}

\def\RR{\mathbb{R}}

\def\bx{\mathbf{x}}
\def\by{\mathbf{y}}
\def\bz{\mathbf{z}}
\def\bH{\mathbf{H}}

\def\eps{\varepsilon}

\def\bar#1{{\overline #1}}
\def\R2+{\RR ^2_+}

\def\calB{{\mathcal B}}

\def\calF{{\mathcal F}}
\def\calE{{\mathcal E}}

\def\calT{{\mathcal T}}

\def\p{\partial}

\def\vphi{{\varphi}}

\newcommand{\pl}[2]{\frac{\partial#1}{\partial#2}}

\newcommand{\bn}{{\bf n}}
\newcommand{\bp}{{\bf p}}
\newcommand{\br}{{\bf r}}
\newcommand{\og}{\omega}
\newcommand{\Og}{\Omega}
\newcommand{\rd}{{\rm d}}
\newcommand{\fl}[2]{\frac{#1}{#2}}
\newcommand{\dt}{\delta}
\newcommand{\tm}{\times}

\newcommand{\nn}{\nonumber}
\newcommand{\ap}{\alpha}
\newcommand{\bt}{\beta}
\newcommand{\ld}{\lambda}
\newcommand{\Gm}{\Gamma}
\newcommand{\gm}{\gamma}

\newcommand{\tht}{\theta}
\newcommand{\ift}{\infty}
\newcommand{\vep}{\varepsilon}

\newcommand{\Dt}{\Delta}

\newcommand{\sg}{\sigma}

\newcommand{\btd}{\nabla}
\newcommand{\btu}{\Delta}

\newcommand{\ba}{\begin{array}}
\newcommand{\ea}{\end{array}}

\title[Mathematics and numerics for BEC]{Mathematical theory and
numerical methods for Bose-Einstein condensation}
\author[W. Bao and Y. Cai]{}


 \email{bao@math.nus.edu.sg}
 \email{yongyong.cai@gmail.com}

\date{}
\begin{document}
\maketitle

\centerline{\scshape Weizhu Bao }
\medskip
{\footnotesize\centerline{Department of Mathematics and Center  for Computational
Science and Engineering}
\centerline{National University of Singapore,
Singapore 119076}
}

\medskip

\centerline{\scshape Yongyong
Cai }
\medskip
{\footnotesize\centerline{Department of Mathematics, National University of Singapore,
Singapore 119076}
\centerline{and}
{\footnotesize\centerline{Beijing Computational Science Research Center, Beijing 100084, P. R. China}
}
}

\bigskip

\begin{abstract}
In this paper, we mainly review recent results on mathematical theory and
numerical methods for Bose-Einstein condensation (BEC),
based on the Gross-Pitaevskii equation (GPE). Starting from the  simplest case with one-component BEC of the weakly interacting bosons, we study the  reduction of GPE to lower dimensions,  the ground states of BEC including the existence and uniqueness  as well as nonexistence results, and the dynamics of GPE including dynamical laws, well-posedness of the Cauchy problem as well as the finite time blow-up. To compute the ground state, the gradient flow with discrete normalization (or imaginary time) method is reviewed and various full discretization methods are presented and compared. To simulate the dynamics, both finite difference methods and time splitting spectral methods are reviewed, and their error estimates are briefly outlined. When the GPE has symmetric properties, we  show how to  simplify the numerical methods. Then we compare two widely used scalings, i.e. physical scaling (commonly used) and semiclassical scaling, for BEC in strong repulsive interaction regime (
 Thomas-Fermi regime), and discuss semiclassical limits of the GPE. Extensions of these results for one-component BEC are then
carried out for rotating BEC by GPE with an angular momentum rotation, dipolar BEC by GPE with long range dipole-dipole interaction, and two-component BEC by coupled GPEs. Finally, as a perspective, we show briefly the mathematical models for
spin-1 BEC, Bogoliubov excitation and BEC at finite temperature.
\end{abstract}

\tableofcontents

\section{Introduction}
\setcounter{equation}{0}\setcounter{figure}{0}\setcounter{table}{0}
Quantum theory is one of the most important science discoveries in the last century. It asserts that all objects behave like waves in the micro length scale. However, quantum world remains a mystery as it is hard to observe  quantum phenomena due to the extremely small wavelength. Now, it is possible to explore quantum world in experiments due to the remarkable discovery of a new state of matter, Bose-Einstein condensate (BEC). In the state of BEC, the temperature is very cold (near absolute zero). In such case,  the wavelength of an object increases extremely, which leads to  the incredible and observable  BEC.

\subsection{Background}
\label{subsec:background}
The idea of BEC originated in 1924-1925, when A. Einstein generalized a work of S. N. Bose on the  quantum statistics for photons \cite{Bose} to a gas of non-interacting  bosons \cite{Einstein1,Einstein2}. Based on the quantum statistics,  Einstein predicted that, below a critical temperature, part of the bosons would occupy the same  quantum state to form a condensate. Although Einstein's work was carried out for non-interacting bosons, the idea can be applied to interacting system of bosons. When temperature $T$ is decreased,  the de-Broglie wavelength $\lambda_{dB}$ of the particle increases, where $\lambda_{dB}=\sqrt{2\pi\hbar^2/mk_BT}$, $m$ is the mass of the particle, $\hbar$ is the Planck constant and $k_B$ is the Boltzmann constant. At a critical temperature $T_c$, the wavelength $\lambda_{dB}$   becomes  comparable to the inter-particle average spacing, and the de-Broglie waves  overlap. In this situation, the particles behave coherently as a giant atom  and a BEC
 is formed.

Einstein's prediction did not receive much attention until F. London suggested the superfluid $^4$He as an evidence of BEC   in 1938 \cite{London}. London's idea had inspired extensive studies on the superfluid and interacting boson system.   In 1947, by developing the idea of London,  Bogliubov   established the first microscopic theory of superfluid in a system consisting of interacting bosons \cite{Bogoliubov}.    Later, it was found in experiment that  less then $10\%$ of the superfluid $^4$He is in the condensation due to the strong interaction between helium atoms. This fact motivated physicists to search for weakly interacting   system of Bose gases  with higher occupancy of BEC. The  difficulty is that almost all substances become solid or liquid at  temperature which the BEC phase transition occurs.  In 1959, Hecht \cite{Hecht} pointed out that spin-polarized hydrogen atoms would remain gaseous even at 0K. Hence, H atoms become an attractive candidate
  for BEC. In 1980, spin-polarized hydrogen  gases were realized by  Silvera and Walraven \cite{Silvera}. In the following decade, extensive efforts had been devoted to the experimental realization of hydrogen BEC, resulting in the developments of magnetically trapping and evaporative cooling techniques. However, those attempts to observe BEC failed.

In 1980s, due to the developments of laser trapping and cooling, alkali atoms became suitable candidates for BEC experiments as they are well-suited to laser cooling and trapping. By combining the advanced laser cooling and the evaporative cooling techniques together, the first BEC of dilute $^{87}$Rb gases was achieved in 1995, by  E. Cornell and C. Wieman's group in JILA \cite{Anderson}. In the same year, two successful experimental observations  of BEC, with $^{23}$Na   by Ketterle's group  \cite{Davis} and $^7$Li by Hulet's group \cite{Bradley}, were announced.  The experimental realization of  BEC for alkali vapors has two stages: the laser pre-cooling  and evaporative cooling. The alkali gas can be cooled down to several $\mu $K by laser cooling, and then be further cooled down to 50nK--100nK by evaporative cooling. As laser cooling can not be applied to hydrogen, it took atomic physicists much more time to achieve hydrogen BEC. In 1998, atomic condensate of hydrogen
 was finally realized \cite{Fried}. For better understanding of the long history towards the Bose-Einstein condensation, we refer to the Nobel lectures \cite{Cornell,Ketterle}.

The experimental advances \cite{Anderson,Davis,Bradley}  have spurred great excitement in the
atomic physics community and condensate physics community.   Since 1995, numerous efforts have been devoted to the studies of ultracold atomic gases and various kinds of condensates of dilute gases have been produced for both bosonic particles and fermionic particles \cite{Andersen,Dalfovo,Fetter,Leggett,Morsch,Ozeri,Posazhen}. In this rapidly growing research area,   numerical simulation has been playing an important role  in  understanding the theories and the experiments.   Our aim is to  review   the numerical methods and mathematical theories for BEC that have been developed over these years.

\subsection{Many body system and mean field approximation}
\label{subsec:manybody}
We are interested in  the ultracold dilute bosonic gases  confined in an external trap,  which is the case for most of the BEC experiments. In these cold dilute gases, only binary interaction is important. Hence, the many body Hamiltonian for $N$ identical bosons  held in a  trap can be written as \cite{LiebSeiringer,Leggett}
\begin{equation}
H_{N}=\sum\limits_{j=1}^N\left(-\frac{\hbar^2}{2m}\Delta_j+V(\bx_j)\right)+\sum\limits_{1\leq j<k\leq N}V_{\rm int}(\bx_j-\bx_k),
\end{equation}
where $\bx_j\in\Bbb R^3$ ($j=1,\ldots,N$) denote the positions of the particles,
 $m$ is the mass of a boson, $\Delta_j$ is the Laplace operator with respect to $\bx_j$, $V(\bx_j)$ is the external trapping potential, and $V_{\rm int}(\bx_j-\bx_k)$ denotes the inter-atomic two body interactions. The wave function $\Psi_N:=\Psi_N(\bx_1,\ldots,\bx_N,t)\in L^2(\Bbb R^{3N}\times \Bbb R)$ is symmetric, with respect to any permutation of the positions $\bx_j$. The evolution of the system is then described by the time-dependent  Schr\"odinger equation
\begin{equation}\label{eq:Nbdy}
i\hbar\partial_t\Psi_N(\bx_1,\ldots,\bx_N,t)=H_N\Psi_N(\bx_1,\ldots,\bx_N,t).
\end{equation}
Here $i$ denotes the imaginary unit. In the sequel, we may omit time $t$ when we write the $N$ body wave function $\Psi_N$.

In principle, the above many body system can be solved, but the cost increases quadratically  as $N$ goes large, due to the binary interaction term. To simplify the interaction, mean-field potential is introduced to approximate the two-body interactions. In the ultracold dilute regime, the binary interaction $V_{\rm int}$ is well approximated by the effective interacting potential:
\begin{equation}\label{eq:mf-app}
V_{\rm int}(\bx_j-\bx_k)=g\,\delta(\bx_j-\bx_k),
\end{equation}
where $\delta(\cdot)$ is the Dirac distribution and the constant $g=\frac{4\pi\hbar^2 a_s}{m}$. Here $a_s$ is the $s$-wave scattering length of the bosons (positive for repulsive interaction and negative for attractive interaction), and it is related to the potential $V_{\rm int}$ \cite{LiebSeiringer}. The above approximation (\ref{eq:mf-app}) is valid for the dilute regime case,  where the scattering length $a_s$ is much smaller than the average distance between the particles.

 For a BEC, all  particles are in the same quantum state and we can formally take the Hartree ansatz for the many body wave function as
\begin{equation}\label{eq:HF}
\Psi_N(\bx_1,\ldots,\bx_N,t)=\prod_{j=1}^N\psi_{H}(\bx_j,t),
\end{equation}
with the normalization condition for the single-particle wave function $\psi_{H}$ as
\begin{equation}
\int_{\Bbb R^3}|\psi_{H}(\bx,t)|^2\,d\bx=1.
\end{equation}
Then the energy of the state (\ref{eq:HF}) can be written as
\begin{equation}
E=N\int_{\Bbb R^3}\left[\frac{\hbar^2}{2m}|\nabla \psi_{H}(\bx,t)|^2+V(\bx)|\psi_{H}(\bx,t)|^2
+\frac{N-1}{2}g|\psi_{H}(\bx,t)|^4\right]\,d\bx.
\end{equation}
Let us introduce the wave function for the whole condensate
\begin{equation}
\psi(\bx,t)=\sqrt{N}\psi_{H}(\bx,t).
\end{equation}
Neglecting terms of order $1/N$, we obtain the energy of the $N$ body system as
\be\label{eq:GPenergy}
E(\psi)=\int_{\Bbb R^3}\left[\frac{\hbar^2}{2m}|\nabla \psi(\bx,t)|^2+V(\bx)|\psi(\bx,t)|^2
+\frac{1}{2}g|\psi(\bx,t)|^4\right]\,d\bx,
\ee
where the wave function is normalized according to the total number of the particles,
\be\label{eq:GPnorm}
\int_{\Bbb R^3}|\psi(\bx,t)|^2\,d\bx=N.
\ee
Eq. (\ref{eq:GPenergy}) is the well-known Gross-Pitaevskii energy functional. The equation governing the motion of the condensate can be derived by \cite{PitaevskiiStringari}
\be\label{eq:GPderive}
i\hbar \p_t\psi(\bx,t)=\frac{\delta E(\psi)}{\delta \overline{\psi}}=\left[-\frac{\hbar^2}{2m}\nabla^2+V(\bx)+g|\psi|^2\right]\psi,
\ee
where $\overline{\psi}$ denotes the complex conjugate of $\psi:=\psi(\bx,t)$. Eq. (\ref{eq:GPderive}) is a
 nonlinear Schr\"odinger equation (NLSE) with cubic nonlinearity, known as the Gross-Pitaevskii equation (GPE).

In the derivation, we have used both the dilute property of the gases and the Hartree ansatz (\ref{eq:HF}). Eq. (\ref{eq:HF}) requires that the BEC system is  at extremely low temperature such that almost all  particles are in the same states. Thus, mean field approximation (\ref{eq:GPenergy}) and (\ref{eq:GPderive}) are only valid for dilute boson gases (or usually called weakly interacting boson gases) at temperature $T$ much smaller than the critical temperature $T_c$.

The Gross-Pitaevskii (GP) theory (\ref{eq:GPderive}) was developed by Pitaevskii \cite{Pitaevskii}  and  Gross \cite{Gross} independently in 1960s.   For a long time,  the validity of this mean field approximation lacks of rigorous mathematical justification. Since the first experimental observation of BEC in 1995, much attention has been paid to the GP theory. In 2000, Lieb et al.  proved that the energy (\ref{eq:GPenergy})  describes the ground state energy of the many body system correctly in the mean field regime \cite{LiebSeiringer,LiebSeiringerPra2000}. Later H. T. Yau and his collaborators  studied the validity of    GPE  (\ref{eq:GPderive}) as an approximation for  (\ref{eq:Nbdy}) to describe the dynamics of BEC \cite{Erdos}, without the trapping potential $V(\bx)$.

GP theory, or mean field theory, has been proven to predict many properties of BEC quite well. It has become the fundamental mathematical model to understand BEC. In this review article, we will concentrate on the GP theory.

\subsection{The Gross-Pitaevskii equation}
\label{subsec:gpe}
As shown in section \ref{subsec:manybody}, at temperature $T\ll T_c$, the dynamics of a BEC is well described by the Gross-Pitaevskii equation (GPE) in three dimensions (3D)
\be\label{eq:GPE}
i\hbar\p_t\psi(\bx,t)=\left[-\frac{\hbar^2}{2m}\nabla^2+V(\bx)+Ng|\psi(\bx,t)|^2\right]\psi(\bx,t),\quad \bx\in\Bbb R^3,\, t>0,
\ee
where $\bx=(x,y,z)^T\in {\Bbb R}^3$ is the Cartesian coordinates, $\nabla$ is the gradient operator and $\nabla^2:=\nabla\cdot\nabla=\Delta$ is the Laplace operator. In fact, the above GPE (\ref{eq:GPE}) is obtained
 from the GPE (\ref{eq:GPderive}) by a rescaling  $\psi\to \sqrt{N}\psi$, noticing (\ref{eq:GPnorm}),  the  wave function $\psi$ in (\ref{eq:GPE}) is normalized by
\be\label{eq:norm}
\|\psi(\cdot,t)\|_2^2=\int_{\Bbb R^3}|\psi(\bx,t)|^2\,dx=1.
\ee

\subsubsection{Different external trapping potentials}
In the early BEC experiments, a single harmonic oscillator well  was
used to trap the atoms in the condensate  \cite{Dalfovo,BradleySackett}.
Recently  more advanced and complicated traps are applied in studying
BEC in  laboratory \cite{PitaevskiiStringari,Milburn,Bronski,Carr}.  Here we
present several  typical  trapping potentials  which are widely used in
current  experiments.

\bigskip

\noindent I. Three-dimensional (3D) harmonic oscillator potential
\cite{PitaevskiiStringari}:
\be
\label{eq:hp}
V_{\rm ho}({\bx}) = V_{\rm ho}(x) + V_{\rm  ho}(y)+ V_{\rm ho}(z),
\quad  V_{\rm  ho}(\alpha)= \fl{m}{2}\og_\alpha^2
\alpha^2, \ \alpha =x,y,z,
\ee
where $\og_x$, $\og_y$ and $\og_z$ are the trap frequencies in $x$-,
$y$- and $z$-direction, respectively. Without loss of generality, we assume that $\omega_x\leq\omega_y\leq\omega_z$ throughout the paper.

\medskip

\noindent II. 2D harmonic oscillator + 1D double-well potential  (Type I)
\cite{Milburn}:
\be
\label{eq:dwp1}
V_{\rm dw}^{(1)}({\bx}) =  V_{\rm dw}^{(1)}(x) + V_{\rm ho}(y)+
V_{\rm ho}(z), \quad  V_{\rm dw}^{(1)}(x)
=  \fl{m}{2}\nu_x^4\left(x^2-\hat{a}^2\right)^2,
\ee
where $\pm \hat{a}$ are the double-well centers in $x$-axis, $\nu_x$
is a given constant with physical dimension  1/[s$\;$m]$^{1/2}$.

\medskip

\noindent III. 2D harmonic oscillator + 1D double-well potential
(Type II) \cite{Holthaus,Capuzzi}:
\be
\label{eq:dwp2} V_{\rm dw}^{(2)}({\bx}) =
V_{\rm dw}^{(2)}(x) + V_{\rm ho}(y)+ V_{\rm ho}(z),  \quad  V_{\rm dw}^{(2)}(x)=
 \fl{m}{2}\og_x^2\left(|x|-\hat{a}\right)^2.
\ee

\medskip

\noindent IV. 3D harmonic oscillator + optical lattice potential
\cite{Choi,PitaevskiiStringari,Adhikari}:
\be
\label{eq:olp2}
V_{\rm hop}({\bx}) = V_{\rm  ho}(\bx) + V_{\rm opt}(x) +V_{\rm opt}
(y)+V_{\rm opt}(z), \quad
V_{\rm opt}(\alpha)=I_\alpha\;E_\alpha  \sin^2(\hat{q}_\alpha \alpha),
\ee
where  $\hat{q}_\alpha = 2\pi/\lambda_\alpha$ is
fixed by the wavelength  $\lambda_\alpha$ of the laser
light creating the stationary 1D  lattice wave,
$E_\alpha = \hbar^2\hat{q}_\alpha^2/2m$ is the so-called  recoil energy, and
$I_\alpha$ is a dimensionless parameter providing  the intensity of the
laser beam. The optical lattice potential has  periodicity $T_\alpha =
\pi/\hat{q}_\alpha = \lambda_\alpha/2$ along  $\alpha$-axis ($\alpha =x, y,z$).

\medskip

\noindent V. 3D box potential \cite{PitaevskiiStringari}:
\be
\label{eq:box3d}
V_{\rm box}({\bx}) =  \left\{\begin{array}{ll}
0, & \quad 0 < x, y, z<L,\\
\infty,& \quad\hbox{otherwise}.
\end{array}\right.
\ee
where $L$ is the length of the box in the $x$-, $y$-, $z$-direction.   \\

For more types of external trapping potential, we refer to  \cite{PitaevskiiStringari,Pethick}.
When a harmonic potential is considered, a typical set of parameters used in  experiments
with ${}^{87}$Rb is given
by
\begin{equation*}
m=1.44\tm 10^{-25}[kg], \
\og_x=\og_y=\og_z=20\pi[rad/s], \
a=5.1\tm 10^{-9}[m], \ N:\ 10^2\sim 10^7
\end{equation*}
and the Planck constant has the value
\begin{equation*} \hbar =1.05\tm 10^{-34}\ [Js].\end{equation*}

\subsubsection{Nondimensionlization}
In order to nondimensionalize Eq. (\ref{eq:GPE}) under the normalization (\ref{eq:norm}),
 we introduce
\be
\label{eq:scale} \tilde{t} =  \fl{t}{t_s}, \quad \tilde{\bx}=\fl{\bx} {x_s},
\quad  \tilde{\psi}\left(\tilde{\bx},\tilde{t}\right) =  x_s^{3/2}\psi
\left({\bx}, t\right), \quad  \tilde{E}(\tilde{\psi}) = \fl{E(\psi)}{E_s},
\ee
where $t_s$, $x_s$ and $E_s$ are the scaling parameters of dimensionless
time, length and energy units,
respectively.  Plugging (\ref{eq:scale}) into (\ref{eq:GPE}), multiplying by
$t_s^2/mx_s^{1/2}$, and  then removing all $\tilde{}$,  we obtain  the
following dimensionless GPE under the normalization  (\ref{eq:norm}) in 3D:
\bea
\label{eq:ngpe}
i\p_t\psi({\bx},t)
&=&-\fl{1}{2}\nabla^2\psi({\bx},t)+V({\bx})\psi({\bx},t)
+\kappa|\psi({\bx},t)|^2\psi({\bx},t),
\eea
where the dimensionless energy functional $E(\psi)$ is defined as
\be
\label{eq:denergy}
E(\psi)=\int_{{\Bbb R}^3} \left[ \fl{1}{2}|\btd \psi|^2 + V(\bx)
 |\psi|^2 +\fl{\kappa}{2} |\psi|^4\right]\; d\bx,
\ee
and the choices for the scaling parameters $t_s$ and  $x_s$,
the dimensionless potential $V({\bx})$ with $\gm_y=t_s \og_y$ and
$\gm_z=t_s \og_z$, the energy unit $E_s=\hbar/t_s = \hbar^2/m x_s^2$,
and the interaction parameter $\kappa=4\pi a_s N/ x_s$
for different external trapping potentials are given below
\cite{Lim}:

\bigskip

\noindent I. 3D harmonic oscillator potential:
\[
t_s = \fl{1}{\og_x}, \quad x_s = \sqrt{\fl{\hbar}
{m\og_x}}, \quad
V({\bx}) = \fl{1}{2}\left(x^2+\gm_y^2y^2+\gm_z^2z^2\right).
\]

\medskip

\noindent II. 2D harmonic oscillator + 1D double-well potential
(type I):
\[
t_s = \left(\fl{m}{\hbar\nu_x^4}\right)^{1/3}, \
x_s = \left(\fl{\hbar}{m\nu_x^2}\right)^{1/3}, \
 a  =\fl{\hat{a}}{x_s}, \
V({\bx}) =  \fl{1}{2}\left[\left(x^2-a^2\right)^2+\gm_y^2y^2+
\gm_z^2z^2  \right].
\]

\medskip

\noindent III. 2D harmonic oscillator + 1D double-well potential
(type II):
\[
t_s = \fl{1}{\og_x}, \quad x_s = \sqrt{\fl{\hbar}
{m\og_x}}, \quad  a  =\fl{\hat{a}}{x_s}, \quad
V({\bx}) =  \fl{1}{2}\left[(|x|-a)^2+\gm_y^2y^2+\gm_z^2z^2\right].
\]

\medskip

\noindent IV. 3D harmonic oscillator + optical lattice potentials:
\beas
&&t_s = \fl{1}{\og_x}, \quad x_s = \sqrt{\fl{\hbar}
{m\og_x}},\quad k_\tau = \fl{2\pi^2x_s^2 I_\tau}{\ld_\tau^2},
\quad q_\tau =  \fl{2\pi x_s}{\ld_\tau}, \quad \tau=x,y,z,\\
&&V({\bx}) = \fl{1}{2}(x^2+\gm_y^2y^2+\gm_z^2z^2)+k_x\sin^2(q_xx)+k_y
\sin^2(q_y y) +k_z\sin^2(q_z z).
\eeas

\medskip

\noindent V. 3D Box potential:
\[
t_s = \fl{mL^2}{\hbar},\quad x_s = L,  \quad
V({\bx}) = \left\{\begin{array}{ll}
0, & 0<x, y, z<1,\\
\infty, & \textrm{otherwise}.
\end{array}\right.
\]

\subsubsection{Dimension reduction}
\label{subsubsec:dred}
 Under the external potentials I--IV, when $\og_y \approx 1/t_s=\omega_x$
and $\og_z\gg 1/t_s=\omega_x$ ($\Leftrightarrow$ $\gm_y \approx 1$ and
$\gm_z\gg 1$), i.e. a disk-shape condensate,  the 3D GPE can be reduced to a two dimensional (2D) GPE.
In the following discussion, we take potential I, i.e. the harmonic potential as an example.

 For a disk-shaped condensate with small height in $z$-direction, i.e.
 \be
\label{eq:r2d}
\og_x\approx \og_y, \quad \og_z\gg \og_x, \qquad \Longleftrightarrow \qquad
\gm_y\approx1, \quad \gm_z\gg 1,
 \ee
the 3D GPE (\ref{eq:ngpe}) can be reduced to a 2D  GPE by
assuming that the time evolution does not cause excitations along the
$z$-axis since these excitations have  larger energies at the order of
$\hbar \og_z$ compared to
excitations along the $x$ and $y$-axis with energies at the order of $\hbar \og_x$.

To understand this \cite{BaoJakschP}, consider the total condensate energy $E\left(\psi(t)\right)$ with $\psi(t):=\psi(\bx,t)$:
 \bea \label{eq:energy}
E\left(\psi(t)\right)&=&\fl{1}{2}\int_{{\Bbb R}^3}|\btd\psi(t)|^2d\bx
+\fl{1}{2}\int_{{\Bbb R}^3} \left(x^2+\gm_y^2
y^2\right)|\psi(t)|^2d\bx \nn\\
&&+\fl{\gm_z^2}{2}\int_{{\Bbb R}^3}
z^2|\psi(t)|^2d\bx+\frac{\kappa}{2} \int_{{\Bbb R}^3} |\psi(t)|^4
d\bx.
 \eea
Multiplying (\ref{eq:ngpe}) by $\overline{\psi_t}$ and integrating by
parts show the energy conservation
 \be
\label{eq:energyc}
E\left(\psi(t)\right)=E\left(\psi_I\right), \qquad  t\ge0,
 \ee
where $\psi_I=\psi(t=0)$ is the initial function which may depend
on all parameters  $\gm_y$, $\gm_z$ and $\kappa$. Now assume
that $\psi_I$ satisfies
 \be \label{eq:energyl}
\fl{E(\psi_I)}{\gm_z^2}\to 0, \qquad \hbox{as} \quad \gm_z\to
\ift.
 \ee
Take a sequence $\gm_z\to\ift$ (and keep all other parameters
fixed). Since $\int_{{\Bbb R}^3}|\psi(t)|^2\;d\bx=1$, we
conclude from weak compactness that there is a positive measure
$n^0(t)$ such that
\[|\psi(t)|^2 \rightharpoonup n^0(t)
\quad \hbox{weakly as}\quad \gm_z\to\ift.\]
Energy conservation implies
\[\int_{{\Bbb R}^3}\; z^2|\psi(t)|^2\;d\bx \to 0, \quad \hbox{as}
\quad \gm_z\to\ift,
\]
and thus we conclude concentration of the condensate in the plane
$z=0$:
\[n^0(x,y,z,t)=n_2^0(x,y,t)\dt(z),\]
where $n_2^0(t):=n_2^0(x,y,t)$ is a positive measure on ${\Bbb R}^2$.

  Now let $\psi_3=\psi_3(z)$ be a wave function with
\[\int_{{\Bbb R}} \; |\psi_3(z)|^2\;dz =1,\]
depending on $\gm_z$ such that
 \be
 |\psi_3(z)|^2 \rightharpoonup
\dt(z), \qquad \hbox{as} \quad \gm_z\to \ift.
 \ee
Denote by $S_{\rm fac}$ the subspace
\be S_{\rm fac}=\{\psi=\psi_2(x,y)\psi_3(z)\; |\; \psi_2\in L^2({\Bbb R}^2)\}\ee
and let
\be\Pi:\; L^2({\Bbb R}^3)\to S_{\rm fac}\subseteq L^2({\Bbb R}^3)\ee
be the projection on $S_{\rm fac}$:
\be (\Pi\psi)(x,y,z)=\psi_3(z)\int_{{\Bbb R}} \; \overline{\psi_3}(z^\prime)
\; \psi(x,y,z^\prime)\; dz^\prime.\ee
Now write the equation (\ref{eq:ngpe}) in the form
\be i \p_t\psi={\mathcal A}\psi +{\mathcal F}(\psi),\ee
where ${\mathcal A}\psi$ stands for the linear part and ${\mathcal
F}(\psi)$ for the nonlinearity. Applying $\Pi$ to the GPE
gives \begin{align} \label{eq:proj}
i\p_t(\Pi \psi)=&\,\Pi {\mathcal A} \psi +\Pi {\mathcal F}(\psi)\nn\\
=&\,\Pi {\mathcal A} (\Pi\psi) +\Pi {\mathcal F}(\Pi\psi) +\Pi\left((\Pi
{\mathcal A}-{\mathcal A}\Pi)\psi+(\Pi{\mathcal F}(\psi)- {\mathcal F}(\Pi
\psi))\right).
 \end{align}
 The projection approximation of (\ref{eq:ngpe})
is now obtained by dropping the commutator terms and  it reads \bea
\label{eq:proj2j}
&&i\p_t (\Pi \sg)=\Pi {\mathcal A} (\Pi\sg) +\Pi {\mathcal F}(\Pi\sg),\\
\label{eq:proj3j} &&(\Pi\sg)(t=0)=\Pi \psi_I, \eea or explicitly,
with
 \be \label{eq:proj4j}
 (\Pi \sg)(x,y,z,t)=:\psi_2(x,y,t)\psi_3(z),
 \ee
we find
 \be \label{eq:gpe2d}
 i
\p_t{\psi_2}=-\fl{1}{2}\btd^2 \psi_2+
\fl{1}{2}\left(x^2+\gm_y^2 y^2 + C \right) \psi_2 + \left(\kappa\int_{-\ift}^\ift \psi_3^4(z)\,dz\right)
|\psi_2|^2\psi_2, \ee where
\[C=\gm_z^2\int_{-\ift}^\ift \; z^2|\psi_3(z)|^2\;dz+
\int_{-\ift}^\ift\; \left|\fl{d\psi_3}{dz}\right|^2\;dz.\]
Since this GPE is time-transverse invariant, we can replace
$\psi_2\to \psi\;e^{-iC/2}$ and drop the constant $C$ in
the trap potential. The observables are not affected by this. For the same reason,  we will always assume that $V(\bx)\ge0$ in (\ref{eq:GPE}).

The `effective' GPE (\ref{eq:gpe2d}) is well known in the
physical literature, where the projection method
is often referred to as `integrating out the $z$-coordinate'.
However, an analysis of the limit process $\gm_z\to\ift$ has to be
based on the derivation as presented above, in particular on
studying the commutators $\Pi {\mathcal A}-{\mathcal A}\Pi$, $\Pi {\mathcal F}
-{\mathcal F}\Pi $ .   In the case of small interaction $\beta=o(1)$ \cite{bamsw}, a good choice for $\psi_3(z)$ is the
ground state of the harmonic oscillator in $z$-dimension: \be
\label{eq:psi3} \psi_3(z)=\left(\fl{\gm_z}{\pi}\right)^{1/4}\;
e^{-\gm_z z^2/2}.
 \ee

For condensates with interaction other than small interaction the choice of
$\psi_3$ is much less obvious. Often one assumes  that the
condensate density along the $z$-axis is well described by the
$(x,y)$-trace of the ground state position density $|\phi_g|^2$
  \be
\label{eq:pps}
|\psi(x,y,z,t)|^2 \approx |\psi_2(x,y,t)|^2 \int_{{\Bbb R}^2}
\; |\phi_g(x_1,y_1,z)|^2\;dx_1dy_1
\ee
and (taking a pure-state-approximation)
\be
\label{eq:pph}
\psi_3(z)=\left(\int_{{\Bbb R}^2}
\; |\phi_g(x,y,z)|^2\;dxdy\right)^{1/2}.
\ee

Similarly, when $\og_y\gg 1/t_s=\omega_x$ and $\og_z\gg 1/t_s=\omega_x$
($\Leftrightarrow$  $\gm_y\gg 1$ and $\gm_z\gg 1$),
i.e. a cigar-shaped condensate, the 3D GPE can be reduced to a
1D GPE. For a cigar-shaped condensate \cite{BaoJakschP,Pethick,PitaevskiiStringari}
\be
\label{eq:r2dd}
\og_y\gg \og_x, \quad \og_z\gg \og_x, \qquad \Longleftrightarrow\qquad
\gm_y\gg1, \quad \gm_z\gg 1,
 \ee
the 3D GPE (\ref{eq:GPE}) can be reduced to a 1D GPE by proceeding
analogously.

Then the 3D GPE (\ref{eq:GPE}), 2D and 1D GPEs can  be
written in a unified way
 \be
\label{eq:gpeg}
i\p_t{\psi(\bx,t)}=-\fl{1}{2}\btd^2 \psi(\bx,t)+
V(\bx)\psi(\bx,t)
+ \beta\; |\psi(\bx,t)|^2\psi(\bx,t), \qquad \bx\in {\Bbb R}^d,
\ee
where
\be
\label{eq:dhp:sec1}
\beta=\kappa\;\begin{cases}
\int_{{\Bbb R}^2} \psi_{23}^4(y,z)\;dy dz, \\
\int_{{\Bbb R}} \psi_3^4(z)\;dz, \\
1,
\end{cases}
\,
V(\bx)=\begin{cases}
 \fl{1}{2}\gm_x^2x^2,  & d=1, \\
 \fl{1}{2}\left(\gm_x^2x^2+\gm_y^2 y^2\right), & d=2, \\
 \fl{1}{2}\left(\gm_x^2x^2+\gm_y^2 y^2+\gm_z^2 z^2\right), &d=3;
\end{cases}
\ee
where $\gm_x\ge1$ is a constant and $\psi_{23}(y,z)\in L^2({\Bbb R^2})$
is often chosen to be the $x$-trace of the ground state $\phi_g(x,y,z)$ in 3D
as $\psi_{23}(y,z)=\left(\int_{\Bbb R}|\phi_g(x,y,z)|^2\,dx\right)^{1/2}$
which is usually approximated by the ground state of the corresponding
2D harmonic oscillator \cite{BaoJakschP,Pethick,PitaevskiiStringari}.
The normalization condition for (\ref{eq:gpeg}) is
\be
\label{eq:normg}
\int_{{\Bbb R}^d} \; |\psi(\bx,t)|^2\;d\bx=1,
\ee
and the energy of (\ref{eq:gpeg}) is given by
\be\label{eq:energyuniform}
E(\psi(\cdot,t)):=\int_{\Bbb R^d}\;\left[\frac12|\nabla\psi(\bx,t)|^2+V(\bx)|\psi(\bx,t)|^2
+\frac{\beta}{2}|\psi(\bx,t)|^4\right]\,d\bx.
\ee
For a weakly interacting condensate, choosing $\psi_{23}$ and $\psi_3$ as the ground states of the corresponding 2D and 1D harmonic oscillator \cite{BaoJakschP,Pethick,PitaevskiiStringari}, respectively,
we derive,
\be
\label{eq:ufw}
\beta:=\kappa \left\{\ba{ll}
 \fl{(\gm_y \gm_z)^{1/2}}{2\pi}, &\qquad d=1, \\
\sqrt{\fl{\gm_z}{2\pi}}, &\qquad d=2, \\
1, &\qquad d=3.\\
\ea\right.
\ee

\subsubsection{BEC on a ring} BEC on a ring has been realized by choosing Toroidal potential (3D harmonic oscillator +2D Gaussian potential) \cite{Ryu}:
\be\label{eq:ring}
V_{\rm tor}(\bx)= V_{\rm  ho}(\bx)+V_{\rm gau}(x,y),\quad V_{\rm gau} (x,y)=V_0e^{-2\frac{x^2+y^2}{w_0^2}},
\ee
where $V_{\rm gau}$ is produced by a laser beam, $w_0$ is the beam waist, and $V_0$ is related to the power of the plug-beam.

In the quasi-1D regime \cite{Ryu},  $\og_x=\og_y=\og_r$, the toroidal potential can be written in cylindrical coordinate $(r,\theta,z)$ as
\be
V_{\rm tor}(r,\theta,z)=\frac{m}{2}\og_r^2r^2+\frac{m}{2}\og_z^2z^2+V_0e^{-2\frac{r^2}{w_0^2}}.
\ee
When $\og_r,\og_z\gg1$, the dynamics of BEC in the ring trap (\ref{eq:ring}) would be confined in $r=R$ and $z=0$, where  $\frac{m}{2}\og_r^2r^2+V_0e^{-2\frac{r^2}{w_0^2}}$ attains the minimum at $R$. Then similar to the above  dimension reduction process and nondimensionlization, we can obtain the dimensionless 1D GPE for BEC on a ring as \cite{Halk}:
\be
i\p_t\psi(\theta,t)=-\frac12\p_{\theta\theta}\psi(\theta,t)+\beta|\psi|^2\psi(\theta,t),\quad \theta\in[0,2\pi],\quad t>0,
\ee
with periodic boundary condition, where $\psi:=\psi(\theta,t)$ is the wave function and $\beta$ is a dimensionless parameter.

\subsection{Outline of the review}
Concerning the GPE (\ref{eq:gpeg}), there are two basic issues,
the ground state and the dynamics.
 Mathematically speaking, the dynamics include the time dependent
 behavior of GPE, such as the well-posedness of the Cauchy problem,
 finite time blow-up,  stability of traveling waves, etc.
 The ground state  is usually defined as the minimizer of the
 energy functional (\ref{eq:energyuniform}) under the normalization
 constraint (\ref{eq:normg}).  In the remaining part of the paper,
 we will review the mathematical theories and numerical methods for
ground states and  dynamics of BECs.

In section \ref{sec:mathgpe}, we review the  theories of GPE for
single-component BEC. Existence and uniqueness, as well as other properties
for the ground states are presented. Well-posedness of the Cauchy
problem for GPE is also reviewed. The rigorous analysis
on the convergence rates for the  dimension reduction is introduced in
section \ref{subsec:dredrate}. After an overview on the mathematical
results for GPE, we list the numerical methods to find the ground
states and compute the dynamics for GPE in sections \ref{sec:numgs} and
\ref{sec:numdym}, respectively. The most popular way  for computing the ground states of BEC is the gradient flow with discrete normalization (or  imaginary time) method. Section \ref{sec:numgs} provides a solid mathematical background on the method and details on the full discretizations. For computing the dynamics of GPE, the traditional finite difference methods and the popular time splitting methods are taken into consideration in section \ref{sec:numdym}, with rigorous error analysis.

 In section \ref{sec:rotat}, we investigate the rotating BEC with quantized vortices. There exist critical rotating speeds for the vortex configuration. In order to compute the ground states and dynamics of  rotating BEC in the presence of the multi-scale  vortex structure, we report the efficient and accurate numerical methods in section \ref{sec:numrotat}. For fast rotating BEC, the semiclassical scaling is usually adopted other than the physical scaling used in the introduction. We demonstrate these two different scalings in section \ref{sec:semiclass}, for the whole space case (harmonic trap) and the bounded domain case (box potential). In fact, the semiclassical scaling is very useful in the case of Thomas-Fermi regime.

 Section \ref{sec:dipole} is devoted to the mathematical theory and numerical methods for dipolar BEC. There are both isotropic contact interactions (short range) and anisotropic dipole-dipole interactions (long range) in a dipolar BEC, and the dipolar GPE involves a highly singular kernel representing the dipole-dipole interaction. We overcome the difficulty caused by the singular kernel via a reformulation of the dipolar GPE, and carry out accurate and efficient numerical methods for dipolar BECs.  In section \ref{sec:2bec}, we consider a two component BEC, which is the simplest multi component BEC system. Ground state properties as well as dynamical properties are described. Efficient numerical methods are proposed by generalizing the existing methods for single component BEC. Finally, we briefly introduce some other important topics that are not covered in the current review in section \ref{sec:chall}, such as spinor BEC, Bogoliubov excitations and BEC at finite temperatu
 re.

Throughout the paper, we adopt the standard Sobolev spaces and write the $\|\cdot\|_p$ for standard $L^p(\Bbb R^d)$ norm when there is no confusion on the spatial variables. The notations are consistent in each section, and the meaning of notation remains the same if not specified.

\section{Mathematical theory for the Gross-Pitaevskii equation}\label{sec:mathgpe}
\setcounter{equation}{0}\setcounter{figure}{0}\setcounter{table}{0}
In this section, we consider the dimensionless GPE in $d$ ($d=1,2,3$) dimensions (\ref{eq:gpeg}),
\be\label{eq:gpe:sec2}
i\p_t\psi(\bx,t)=-\fl{1}{2}\btd^2 \psi(\bx,t)+
V(\bx)\psi(\bx,t)
+ \beta\; |\psi(\bx,t)|^2\psi(\bx,t),\quad \bx\in\Bbb R^d,
\ee
where $V(\bx)\ge0$ is a real-valued potential and $\beta\in\Bbb R$ is
treated as an arbitrary dimensionless parameter. The GPE (\ref{eq:gpe:sec2}) can be
generalized to any dimensions and many results presented here are valid in
higher dimensions, but we focus on the most relevant cases $d=1,2,3$ for BEC.

There are two important invariants, i.e., the  normalization (mass),
\be
\label{eq:norm:sec2}
N(\psi(\cdot,t)) =  \int_{\Bbb R^d}|\psi({\bx},t)|^2\; d{\bx} \equiv N(\psi_0)=
\int_{\Bbb R^d}|\psi({\bx},0)|^2\;d{\bx}=1, \quad t\geq0,
\ee
and the  energy per particle
\be
\label{eq:energy:sec2}
E(\psi(\cdot,t)) =  \int_{\Bbb R^d}\left[\fl{1}{2}|\nabla  \psi|^2+V({\bx})|\psi|^2
+\fl{\bt}{2}|\psi|^4\right]d{\bx}\equiv  E(\psi(\cdot,0)),\quad t\geq 0.
\ee
In fact, the energy functional  $E(\psi)$ can be split into three
parts, i.e. kinetic energy  $E_{\rm kin}(\psi)$, potential energy
$E_{\rm pot}(\psi)$ and  interaction energy $E_{\rm int}(\psi)$,
which are defined as
\bea
\label{eq:kinpot:sec2}
&&E_{\rm int}(\psi) =  \int_{\Bbb R^d}\fl{\bt}{2}|\psi({\bx},t)|^4d{\bx},
\qquad E_{\rm pot}(\psi) = \int_{\Bbb R^d}V({\bx})|\psi({\bx},t)|^2d{\bx},
\\
\label{eq:inten:sec2}
&&E_{\rm kin}(\psi) =  \int_{\Bbb R^d}\fl{1}{2}|\nabla\psi({\bx},t)|^2\;
d{\bx}, \quad E(\psi)  = E_{\rm kin}(\psi)+E_{\rm pot}(\psi)
+E_{\rm int}(\psi). \qquad
\eea
For convenience, we introduce the following function spaces:
\be\label{eq:funcspace:sec2}
L_V(\Bbb R^d)=\left\{\phi|\int_{\Bbb R^d}V(\bx)|\phi(\bx)|^2d\bx<\infty\right\},\quad
X:=X(\Bbb R^d)=H^1(\Bbb R^d)\cap L_V(\Bbb R^d).
\ee
\subsection{Ground states}
To find the stationary solution of (\ref{eq:gpe:sec2}), we write
\be
\label{eq:state-az:sec2}
\psi({\bx},t)=\phi({\bx})\; e^{-i\mu t},
\ee
where $\mu$ is the chemical potential of the condensate and
$\phi({\bx})$ is a function independent of time. Substituting
(\ref{eq:state-az:sec2}) into (\ref{eq:gpe:sec2}) gives the following equation
for $(\mu, \phi({\bx}))$:
\be
\label{eq:charactereq:sec2}
\mu\;\phi({\bx}) =  -\fl{1}{2}\Dt\phi({\bx})+V({\bx})\phi({\bx})
+\bt|  \phi({\bx})|^2\phi({\bx}), \qquad {\bx}\in \Bbb R^d,
\ee
under  the normalization condition
\be
\label{eq:gpenorm:sec2}
\|\phi\|_2^2:=\int_{\Bbb R^d}|\phi({\bx})|^2d{\bx}=1.
\ee
This is a  nonlinear eigenvalue problem with a constraint and any
eigenvalue  $\mu$ can be computed from its corresponding eigenfunction
$\phi({\bx})$ by
\bea
\label{eq:mu-energy:sec2}
\mu & = & \mu(\phi)
=\int_{\Bbb R^d}\left[\fl{1}{2}|\nabla\phi({\bx})|^2+
V({\bx})|\phi({\bx})|^2+\bt|\phi({\bx})|^4\right]d{\bx} \nn \\
&=&  E(\phi)+\int_{\Bbb R^d}\fl{\bt}{2}|\phi({\bx})|^4d{\bx}=E(\phi)
+E_{\rm  int}(\phi).
\eea

The ground state of a BEC is usually defined as the minimizer of  the
following minimization problem:

Find $ \phi_g\in S$ such that
\be
\label{eq:minp:sec2}
E_g:=E(\phi_g) = \min_{\phi\in S} E(\phi),
\ee
where $S=\{\phi \ |\ \|\phi\|_2=1, \ E(\phi)<\ift\}$ is the  unit sphere.

 It is easy to show that the ground state $\phi_g$ is an  eigenfunction
of the nonlinear eigenvalue problem.  Any eigenfunction of
(\ref{eq:charactereq:sec2}) whose energy is larger than that  of the ground
state is usually called as excited states in the  physics literatures.
\subsubsection{Existence}
In this section, we discuss the  existence and uniqueness of the ground state (\ref{eq:minp:sec2}).
Denote the best Sobolev constant $C_b$ in 2D as
\be\label{eq:bestcons:2d}
C_{b}=\inf_{0\neq f\in H^1(\Bbb R^2)}\frac{\|\nabla f\|_{L^2(\Bbb R^2)}^2\,\|f\|_{L^2(\Bbb R^2)}^2}{\|f\|_{L^4(\Bbb R^2)}^4}.
\ee
The best constant $C_b$ can be attained at
some $H^1$ function \cite{Weinstein} and it is crucial in
considering the existence of ground states in 2D.

For existence and uniqueness of the ground state, we have the following results.
\begin{theorem}\label{thm:gs}(Existence and uniqueness) Suppose  $V(\bx)\ge 0$ ($\bx\in\Bbb R^d$) satisfies the confining condition
\be\label{eq:confine:sec2}
\lim\limits_{|\bx|\to\infty}V(\bx)=\infty,
 \ee
 there exists a ground state $\phi_g\in S$ for (\ref{eq:minp:sec2}) if one of the following holds

(i)  $d=3$,  $\beta\ge0$;

(ii)  $d=2$,  $\beta>-C_{b}$;

(iii)  $d=1$, for all $\beta\in \Bbb R$.

\noindent Moreover, the ground state  can be chosen as nonnegative $|\phi_g|$ , and
$\phi_g=e^{i\theta}|\phi_g|$ for some constant  $\theta\in\Bbb R$. For $\beta\ge0$, the nonnegative ground state $|\phi_g|$ is unique. If potential $V(\bx)\in L_{\rm loc}^2$, the nonnegative ground state is strictly positive.

In contrast, there exists no ground state,  if one of the following holds:

(i$^\prime$)  $d=3$,  $\beta<0$;

(ii$^\prime$)  $d=2$,  $\beta\leq-C_{b}$.
\end{theorem}
To prove the theorem, we present the following lemmas.
\begin{lemma}\label{lem:comp:sec2}Suppose that  $V(\bx)\ge 0$ ($\bx\in\Bbb R^d$) satisfies
$\lim\limits_{|\bx|\to\infty}V(\bx)=\infty$, the embedding $X\hookrightarrow L^p(\Bbb R^d)$ is compact,  where $p\in[2,\infty]$ for $d=1$,  $p\in[2,\infty)$ for $d=2$, and $p\in [2,6)$ for  $d=3$.
\end{lemma}
\begin{proof} It suffices to prove the case for $p=2$ and the other cases can be obtained by interpolation in view of the Sobolev inequalities. Since $X$ is a Hilbert space, we need show  that any weakly convergent sequence in $X$ has a strong convergent subsequence in $L^2(\Bbb R^d)$. Taking a bounded sequence  $\{\phi^n\}_{n=1}^\infty\subset X$ such that
\be
\phi^n\rightharpoonup \phi \text{ in } X,
\ee
 in order to prove the strong $L^2(\Bbb R^d)$ convergence of the sequence,  we need only prove that \be
\|\phi^n\|_{L^2(\Bbb R^d)}\to \|\phi\|_{L^2(\Bbb R^d)}.
\ee
Using the weak convergence, there exists $C>0$ such that $\int_{\Bbb R^d}V(\bx)|\phi^n|^2\,d\bx\leq C$.
For any $\vep>0$, from $\lim\limits_{|\bx|\to\infty}V(\bx)=\infty$, there exists $R>0$ such that
$V(\bx)\ge \frac{C}{\vep}$ for $|\bx|\ge R$,
which implies that
\be\label{eq:part1:sec2}
\int_{|\bx|\ge R}|\phi^n|^2\leq \vep.
\ee
For $|\bx|\ge R$, applying Sobolev embedding theorem, we obtain
\be\label{eq:part2:sec2}
\int_{|\bx|\leq R}|\phi|^2\,d\bx=\lim\limits_{n\to\infty}\int_{|\bx|\leq R}|\phi^n|^2\,d\bx.
\ee
Combining (\ref{eq:part1:sec2}) and (\ref{eq:part2:sec2}) together as well as the lower semi-continuity  of the $L^2(\Bbb R^d)$ norm, we have
\be
\limsup\limits_{n\to\infty}\|\phi^n\|_{L^2(\Bbb R^d)}^2-\vep\leq \|\phi\|_{L^2(\Bbb R^d)}^2\leq \liminf\limits_{n\to\infty}\|\phi^n\|_{L^2(\Bbb R^d)}^2.
\ee
Hence we get $\|\phi^n\|_{L^2(\Bbb R^d)}\to \|\phi\|_{L^2(\Bbb R^d)}$
and the strong convergence  in $L^2(\Bbb R^d)$ holds true. The conclusion then follows.
\end{proof}
The following lemma ensures that the ground state must be nonnegative.
\begin{lemma}\label{lem:pos:sec2} For any $\phi\in X(\Bbb R^d)$ and energy $E(\cdot)$ (\ref{eq:energy:sec2}), we have
\be
E(\phi)\ge E(|\phi|),
\ee
and the equality holds iff $\phi=e^{i\theta}|\phi|$ for some constant $\theta\in\Bbb R$.
\end{lemma}
\begin{proof} Noticing the inequality for $\phi\in H^1(\Bbb R^d)$ ($d\in\Bbb N$)\cite{LiebLoss},
\be
\|\nabla |\phi|\|_{L^2(\Bbb R^d)}\leq \|\nabla \phi\|_{L^2(\Bbb R^d)},
\ee
where the equality holds iff $\phi=e^{i\theta}|\phi|$ for some constant $\theta\in\Bbb R$, a direct application implies the conclusion. \end{proof}
The minimization problem (\ref{eq:minp:sec2}) is nonconvex, but it can be transformed to a convex minimization problem through the following lemma when $\beta\ge0$.
\begin{lemma}(\cite{LiebSeiringerPra2000})\label{lem:convex:sec2} Considering the  density $\rho(\bx)=|\phi(\bx)|^2\ge0$, for $\sqrt{\rho}\in S$,  the energy $E(\sqrt{\rho})$ (\ref{eq:energy:sec2}) is strictly convex in $\rho$ if $\beta\ge0$.
\end{lemma}
\begin{proof}The potential energy (\ref{eq:kinpot:sec2}) is linear in $\rho$ and  the interaction energy  (\ref{eq:kinpot:sec2}) is quadratic in $\rho$. Hence, $E_{\rm pot}+E_{\rm int}$ is convex in $\rho$.
For $\phi_1(\bx)=\sqrt{\rho_1(\bx)},\phi_2(\bx)=\sqrt{\rho_2(\bx)}\in S$ ($\rho_1,\rho_2\ge0$), we have $\phi_\theta(\bx)=\sqrt{\theta\rho_1(\bx)+(1-\theta)\rho_2(\bx)}\in S$ for any $\theta\in(0,1)$. Using Cauchy inequality, we get
\begin{align*}
\left|\nabla\phi_\theta(\bx)\right|^2=&
\frac{\left|\sqrt{\theta\rho_1(\bx)}\sqrt{\theta}\nabla\phi_1(\bx)
+\sqrt{(1-\theta)\rho_2(\bx)}\sqrt{1-\theta}\nabla\phi_2(\bx)\right|^2}
{\theta\rho_1(\bx)+(1-\theta)\rho_2(\bx)}\\
\leq&\frac{\left(\theta\rho_1(\bx)+(1-\theta)\rho_2(\bx)\right)
\left(\theta|\nabla\phi_1(\bx)|^2
+(1-\theta)|\nabla\phi_2(\bx)|^2\right)}
{\theta\rho_1(\bx)+(1-\theta)\rho_2(\bx)}\\
=&\theta|\nabla\phi_1(\bx)|^2
+(1-\theta)|\nabla\phi_2(\bx)|^2,
\end{align*}
which implies the convexity of the kinetic energy $E_{\rm kin}$ (\ref{eq:inten:sec2})  (with possible approximation procedure). The conclusion then follows.
\end{proof}

\noindent{\it{Proof of Theorem \ref{thm:gs}}:} We separate the proof into the existence and nonexistence parts.

(1) Existence. First, we claim that the energy $E$ (\ref{eq:energy:sec2}) is bounded below under the
 assumptions. Case (i) is clear. For case (ii), using the constraint $\|\phi\|_2^2=1$ and
 Gagliardo-Nirenberg inequality, we have
\begin{equation*}
\beta\|\phi\|_4^4
\ge -\|\phi\|_2^2\cdot\|\nabla\phi\|_2^2=-\|\nabla\phi\|_2^2.
\end{equation*}
For case (iii), using Cauchy inequality
and Sobolev inequality, for any $\vep>0$, there exists $C_\vep>0$ such that
\begin{equation*}
\|\phi\|_4^4\leq \|\phi\|_{\infty}^2\|\phi\|_2^2\leq \|\phi\|_{\infty}^2
\leq \|\nabla\phi\|_2\|\phi\|_2\leq \vep \|\nabla\phi\|_2^2+C_\vep,
\end{equation*}
which yields the claim.  Hence, in all cases, we can take a sequence $\{\phi^n\}_{n=1}^\infty$ minimizing the energy $E$ in $S$, and the sequence is uniformly bounded in $X$. Taking a weakly convergent subsequence (denoted as the original sequence for simplicity)  in $X$, we have
 \be
 \phi^n\rightharpoonup \phi^\infty,\quad \text{weakly in } X.
 \ee
 Lemma \ref{lem:comp:sec2} ensures that $\left\{\phi^n\right\}_{n=1}^\infty$ converges  to $\phi^\infty$ in $L^p$ where $p$ is given in Lemma \ref{lem:comp:sec2}. Combining the lower-semi-continuity of the $H^1$ and $L_V$ norms, we  conclude that $\phi^\infty\in S$ is a ground state \cite{LiebSeiringerPra2000}. Lemma \ref{lem:pos:sec2} ensures that the ground state can be chosen as the nonnegative one. Actually, the nonnegative ground state is strictly positive \cite{LiebSeiringerPra2000}.  The uniqueness comes from the strict convexity of the energy in Lemma \ref{lem:convex:sec2}.

 (2) Nonexistence. Firstly, we consider the case $d=3$, i.e. case (i$^\prime$). If $\beta<0$, let $\phi(\bx)=\pi^{-\frac34}e^{-|\bx|^2/2}\in S$ and denote
 \be
 \phi^{\vep}(\bx)=\vep^{-3/2}\phi(\bx/\vep)\in S,\quad \vep>0,
 \ee
we find
 \be
 E(\phi^\vep)=\frac{C_1}{\vep^2}+\frac{\beta C_2}{\vep^3}+C_3+O(1), \quad C_1,C_2>0.
 \ee
 Hence $E(\phi^\vep)\to -\infty$ as $\vep\to0^+$ which shows that there exists no ground state.

 Secondly, we consider the case $d=2$. Let $\phi_b(\bx)$ ($\bx\in\Bbb R^2$) be  the smooth, radial symmetric (decreasing) function such that the best constant $C_{b}$ is attained in (\ref{eq:bestcons:2d}). If $\beta<-C_b$, let $\phi_b^\vep(\bx)=\vep^{-1}\phi_b(\bx/\vep)$ ($\vep>0$), and we have
 \be
 E(\phi_b^\vep)=\frac{\beta+C_{b}}{2\vep^2}+C_4+O(1)\text{ as } \vep\to0^+.
 \ee
  As $\vep\to0^+$, $ E(\phi_b^\vep)\to -\infty$, which shows that there exists no ground state.  For $\beta=-C_b$, as $\vep\to 0^+$, $\phi_b^\vep$ will converge to the Dirac distribution and the infimum of the energy $E$ will be the minimal of $V(\bx)$ (suppose $V(\bx)$ take minimal at origin), given by the sequence $\phi_b^\vep$. Thus, there exists no ground state for $\beta=-C_b$. The proof is complete.
$\hfill$ $\Box$

\begin{remark}The conclusions in Theorem \ref{thm:gs} hold for potentials satisfying the confining condition, including the box potential as in (\ref{eq:box3d}). Since box potentials are not in $L^2_{\rm loc}$, there exists zeros  in the ground state at the points where $V(\bx)=+\infty$. Results for the 3D case   were first obtained by Lieb et al. \cite{LiebSeiringerPra2000}.
\end{remark}

\subsubsection{Properties of ground states}
In this section, when we refer to the ground state, the conditions guaranteeing the existence in Theorem \ref{thm:gs} are always assumed and potentials are locally bounded.

For the ground state $\phi_g\in S$, we have the following Virial theorem when $V(\bx)$ is homogenous.
 \begin{theorem}(Virial identity) Suppose $V(\bx)$ ($\bx\in\Bbb R^d$, $d=1,2,3$) is homogenous of order $s>0$, i.e. $V(\lambda\bx)=\lambda^sV(\bx)$ for all $\lambda\in\Bbb R$, then the ground state solution $\phi_g\in S$ for (\ref{eq:minp:sec2}) satisfies
 \be\label{eq:virial:sec2}
 2E_{\rm kin}(\phi_g)-s\;E_{\rm pot}(\phi_g)+d\;E_{\rm int}(\phi_g)=0.
 \ee
 \end{theorem}
 \begin{proof} Consider $\phi^\vep(\bx)=\vep^{-d/2}\phi_g(\bx/\vep)\in S$ ($\vep>0$), and use the stationary condition of the energy $E(\phi^\vep)$ at $\vep=1$, then we get $\frac{d{E(\phi^\vep)}}{d\vep}\big|_{\vep=1}=0$, which yields the Virial identity (\ref{eq:virial:sec2}).
 \end{proof}
Many properties of the ground state are determined by the potential $V(\bx)$.
\begin{theorem}\cite{LiebSeiringerPra2000}(Symmetry) Suppose $V(\bx)$ is spherically symmetry and monotone increasing, then the positive ground state solution $\phi_g\in S$ for (\ref{eq:minp:sec2}) must be spherically symmetric and monotonically decreasing.
 \end{theorem}
\begin{proof}This fact comes from the symmetric rearrangements.
\end{proof}
To learn more on the ground state, we study the Euler-Lagrange equation (\ref{eq:charactereq:sec2}).
\begin{theorem} The ground state of (\ref{eq:minp:sec2})  satisfies the Euler-Lagrange equation (\ref{eq:charactereq:sec2}). Suppose $V(\bx)\in L_{\rm loc}^\infty$, the  ground state $\phi_g\in S$ of (\ref{eq:minp:sec2}) is $H^2_{\rm loc}$. In addition, if $V\in C^\infty$, the ground state is also $C^\infty$.
\end{theorem}
\begin{proof}It is easy to show the ground state satisfies the nonlinear eigenvalue problem (\ref{eq:charactereq:sec2}). The regularity follows from the elliptic theory.
\end{proof}
For confining potentials, we can show that ground states decay exponentially fast when $|\bx|\to\infty$.
\begin{theorem}\label{thm:exponential}
Suppose that $0\leq V(\bx)\in L_{\rm loc}^2$ satisfies (\ref{eq:confine:sec2}) and $\phi_g\in S$ is a ground state  of (\ref{eq:minp:sec2}). When $\beta\ge0$, for any $\nu>0$,  there exists a constant $C_{\nu}>0$ such that
\be\label{eq:expon:sec2}
|\phi_g(\bx)|\leq C_{\nu}e^{-\nu|\bx|},\quad \bx\in\Bbb R^d,\,\, d=1,2,3.
\ee
\end{theorem}
\begin{proof}  The proof for  $d=3$ is given in \cite{LiebSeiringerPra2000} and the cases for $d=1,2$ are the same. For any $\nu>0$, rewrite the Euler-Lagrange equation (\ref{eq:charactereq:sec2}) for $\phi_g$ as
\be
\left(-\frac{1}{2}\nabla^2+\frac{\nu^2}{2}\right)\phi_g=\left(\mu+\frac{\nu^2}{2}-V-\beta|\phi_g|^2\right)\phi_g.
\ee
Making use of the $d$-dimensional Yukawa potential $Y_d^\nu(\bx)$ ($d=1,2,3$) \cite{LiebLoss} associated with $-\frac{1}{2}\nabla^2+\frac{\nu^2}{2}$, $\phi_g$ can be expressed as
\be
\phi_g(\bx)=\int_{\Bbb R^d}Y_d^\nu(\bx-{\bf y})\left[\mu+\frac{\nu^2}{2}-V(\by)-\beta|\phi_g(\by)|^2\right]\,d\,\by.
\ee
Noticing that $\phi_g$ and the Yukawa potential are positive and $V$ is confining potential, we see that for sufficiently large $R>0$,
$\mu+\frac{\nu^2}{2}-V(\bx)-\beta|\phi_g(\bx)|^2\leq0$ for $|\bx|\ge R$. Thus, we get
\be
\phi_g(\bx)\leq \int_{|\by|<R}Y_d^\nu(\bx-{\bf y})\left[\mu+\frac{\nu^2}{2}-V(\by)-\beta|\phi_g(\by)|^2\right]\,d\,{\by}.
\ee
Noticing that $Y_d^\nu\in L_{\rm loc}^2$ ($d=1,2,3$) and $|Y_d^\nu(\bx)|\leq C e^{-\nu|\bx|}$ for sufficiently large $|\bx|$, we find
\be
C_\nu=\sup_{\bx}\int_{|\by|<R}e^{\nu|\bx|}Y_d^\nu(\bx-{\bf y})\left[\mu+\frac{\nu^2}{2}-V(\by)-\beta|\phi_g(\by)|^2\right]\,d\,\by<\infty,
\ee
and the conclusion (\ref{eq:expon:sec2}) holds.
\end{proof}
\begin{remark} Results (\ref{eq:expon:sec2}) can be generalized to 1D case for arbitrary $\beta$, where $\|\phi_g\|_{\infty}$ is bounded by Sobolev inequality. The proof is the same.
\end{remark}
For convex potentials, the ground states are shown to be log concave.
\begin{theorem}Suppose $V(\bx)$ ($\bx\in\Bbb R^d$, $d=1,2,3$) is convex, then the positive ground state $\phi_g$ of (\ref{eq:minp:sec2}) is log concave, i.e. $\ln(\phi_g(\bx))$ is concave,
\begin{equation*}
\ln(\phi_g(\lambda\bx+(1-\lambda)\by))\ge \lambda \ln(\phi_g(\bx))+(1-\lambda)\ln(\phi_g(\by)),\quad
\bx,\by\in \Bbb R^d,\,\,\lambda\in[0,1].
\end{equation*}
\end{theorem}
\begin{proof}See \cite{LiebSeiringerPra2000}.
\end{proof}
When $\beta>0$, we can actually estimate the $L^\infty$ bound for the ground state.
\begin{theorem}Suppose that $0\leq V(\bx)\in C^\alpha_{\rm loc}$ ($\alpha>0$) satisfies (\ref{eq:confine:sec2}) and $\beta>0$.  Let $\phi_g$ be the unique positive ground state of (\ref{eq:minp:sec2}), we have
\be
\|\phi_g\|_{\infty}\leq \sqrt{\frac{\mu_g}{\beta}},\quad \mu_g=E(\phi_g)+\frac{\beta}{2}\|\phi_g\|_{4}^4.
\ee
The chemical potential $\mu_g\leq 2 E(\phi_g)$ and hence can be bounded by choosing arbitrary testing function.
\end{theorem}
\begin{proof} Applying elliptic theory to the Euler-Lagrange equation,
\be\label{eq:el:sec2}
\mu_g\phi_g=\left(-\frac12\nabla^2+V+\beta|\phi_g|^2\right)\phi_g,
\ee
we get $\phi_g\in C^{2,\alpha}_{\rm loc}$. From Theorem \ref{thm:exponential}, $\phi_g$ is bounded in $L^\infty$. Consider the point $\bx_0$ where $\phi_g$ takes its maximal, we can obtain
\begin{align*}
\mu_g\phi_g(\bx_0)&=\left(-\frac{1}{2} \nabla^2\phi_g+V+
\beta|\phi_g|^2\right)\bigg|_{\bx_0}\phi_g(\bx_0)\nonumber\\
&\ge \left[V(\bx_0)+\beta|\phi_g(\bx_0)|^{2}\right]\phi_g(\bx_0)\ge \beta|\phi_g(\bx_0)|^{2}\phi_g(\bx_0),
\end{align*}
and so
\be
\|\phi_g\|_{\infty}^2=|\phi_g(\bx_0)|^{2}\leq \frac{\mu_g}{\beta}.
\ee
\end{proof}

\begin{remark} In 2D and 3D, for small $\beta>0$ or $\beta<0$, the $L^\infty$ estimate above can be improved by employing the $W^{2,p}$ estimates for (\ref{eq:el:sec2}) and the embedding $H^2(\Bbb R^d)\hookrightarrow L^\infty(\Bbb R^d)$ ($d=2,3$). In 1D, $L^\infty$ bound can be simply obtained by $H^1(\Bbb R)\hookrightarrow L^\infty(\Bbb R)$, while the $H^1$ norm can be estimated by the energy.
\end{remark}

\subsubsection{Approximations of ground states}
\label{subsubsec:gsapp}

For a few external potentials, we can find approximations of ground states in the weakly interaction regime, i.e.
$|\beta|=o(1)$, and strongly repulsive interaction regime, i.e.  $\beta\gg 1$ \cite{BaoJinP,BaoLimZhang}. These approximations show the leading order behavior of the ground states and they can be used as initial data for computing ground states numerically.

{\sl Under a box potential}, i.e. we take
\be\label{eq:dbox:sec7}
V(\bx)=\begin{cases}
0, & \bx=(x_1,\ldots,x_d)^T\in U=(0,1)^d,\\
+\infty,& \text{otherwise},
\end{cases}
\ee
in (\ref{eq:charactereq:sec2}). When $\beta=0$, i.e. linear case, (\ref{eq:charactereq:sec2}) collapses to
\be\label{eq:linear:sec7}
\mu\phi=-\frac12\btd^2\phi,\quad \phi|_{\partial U}=0,\quad \|\phi\|_2^2=\int_U|\phi(\bx)|^2d\bx=1.
\ee
For this linear eigenvalue problem, it is easy to find an orthonormal set of eigenfunctions as \cite{BaoJakschP,Pethick,PitaevskiiStringari}
\be
\label{eq:sollin:sec7}
\phi_{\bf J}({\bx}) = \prod_{m=1}^d  \phi_{j_m}(x_m), \ \phi_l(x) =
\sqrt{2} \sin (l\pi x), \  l\in {\Bbb N}, \ {\bf J}=(j_1, \cdots, j_d)
\in {\Bbb N}^d,
\ee
with the  corresponding eigenvalues as
\be
\label{eq:eingbox:sec7}
\mu_{\bf J} = \sum_{m=1}^d \mu_{j_m}, \qquad
\mu_{l} = \fl{1}{2}l^2\pi^2, \qquad l\in {\Bbb N}.
\ee
Thus, for linear case, we can find the exact ground state as $\phi_g(\bx)=\phi_{(1,\cdots,1)}(\bx)$.
In addition, when $|\beta|=o(1)$, we can approximate the ground state as
$\phi_g(\bx)\approx\phi_{(1,\cdots,1)}(\bx)$. The corresponding energy  and chemical
potential can be found as
\begin{eqnarray*}
&&E_g=E(\phi_g)\approx E(\phi_{(1,\cdots,1)}(\bx))=d \pi^2/2+O(\beta),\\
&&\mu_g =\mu(\phi_g)\approx \mu(\phi_{(1,\cdots,1)}(\bx))= d \pi^2/2+O(\beta).
\end{eqnarray*}
On the other hand, when $\beta\gg1$, by dropping the diffusion term (i.e. the first term
on the right hand side of (\ref{eq:charactereq:sec2})) -- Thomas-Fermi (TF)
approximation -- \cite{LiebLoss,Andersen}, we obtain
 \be \label{eq:gpegs:sec2}
\mu_g^{\rm TF} \phi_g^{\rm TF}(\bx)= \beta |\phi_g^{\rm TF}(\bx)|^2\phi_g^{\rm TF}(\bx), \qquad
\bx\in U.
 \ee
From (\ref{eq:gpegs:sec2}), we obtain
\be
\label{eq:TFsolution:sec7}
\phi_g^{\rm TF}({\bx}) =  \sqrt{\fl{\mu_g^{\rm TF}}{\bt}}, \qquad  \bx
\in U.
\ee
Plugging (\ref{eq:TFsolution:sec7}) into the normalization condition, we obtain
\be
\label{eq:normbox:sec7}
1 = \int_U  |\phi_g^{\rm TF}({\bx})|^2\; d{\bx} = \int_U
\fl{\mu_g^{\rm  TF}}{\bt} \; d{\bx} = \fl{\mu_g^{\rm TF}} {\bt}
\quad\Rightarrow \quad \mu_g^{\rm TF} = \bt.
\ee
The TF energy $E_g^{\rm TF}$ is obtained via (\ref{eq:mu-energy:sec2}),
\be
\label{eq:engasp3:sec7}
E_g^{\rm TF} =  \mu_g^{\rm TF} - \fl{\bt}{2}\int_U
|\phi_g^{\rm TF}|^4\;  d{\bx} = \fl{\mu_g^{\rm TF}}{2} = \fl{\bt}{2}.
\ee
Therefore, we  get the TF approximation for the ground state, the energy
and the chemical potential when $\bt \gg1$:
\bea
\label{eq:gs1d2:sec7}
&&\phi_g(\bx) \approx \phi_g^{\rm TF}(\bx)=1, \qquad \bx \in U, \\
\label{eq:eng5:sec7}
&&E_g \approx E_g^{\rm TF} = \fl{\bt}{2}, \qquad \mu_g \approx
\mu_g^{\rm TF} =\bt.
\eea
 It is easy to see that the  TF approximation for the ground state
does not  satisfy the boundary condition  $\phi|_{\p U}=0$.  This is  due to
removing the diffusion term in (\ref{eq:charactereq:sec2}) and  it suggests that a
boundary layer will appear in the ground  state when $\bt\gg1$.
Due to the existence of the boundary  layer, the kinetic energy does not
go to zero when $\bt\to \ift$  and thus it cannot be neglected. Better approximation with matched asymptotic expansion can be found in \cite{BaoLimZhang}.

{\sl Under a harmonic potential}, i.e. we take $V(\bx)$ as (\ref{eq:dhp:sec1}).
When $\beta=0$, the exact ground state can be found as \cite{BaoJakschP,Pethick,PitaevskiiStringari}
\[\mu_g^0=\left\{\ba{l}
\fl{\gm_x}{2}, \\
\fl{\gm_x+\gm_y}{2}, \\
\fl{\gm_x+\gm_y+\gm_z}{2},\\
\ea \right.
\quad
\phi_g^0(\bx)=\left\{\ba{ll}
\fl{(\gm_x )^{1/4}}{(\pi)^{1/4}} \;
e^{-(\gamma_xx^2)/2}, &d=1,\\
\fl{(\gm_x\gm_y )^{1/4}}{(\pi)^{1/2}} \;
e^{-(\gamma_xx^2+\gm_y y^2)/2}, &d=2,\\
\fl{(\gm_x\gm_y \gm_z)^{1/4}}{(\pi)^{3/4}} \;
e^{-(\gamma_xx^2+\gm_y y^2+\gm_z z^2)/2}, &d=3.\\
\ea\right.
\]
Thus when $|\beta|=o(1)$, the ground state $\phi_g$ can be approximated by $\phi_g^0$, i.e.
\[\phi_g(\bx)\approx \phi_g^0(\bx), \qquad \bx\in {\Bbb R}^d.\]
Again, when $\beta\gg 1$, by dropping the diffusion term (i.e. the first term
on the right hand side of (\ref{eq:charactereq:sec2})) -- Thomas-Fermi (TF)
approximation -- \cite{LiebLoss,Andersen}, we obtain
 \be \label{eq:gpegs:sec7}
\mu_g^{\rm TF} \phi_g^{\rm TF}(\bx)= V(\bx)\phi_g^{\rm TF}(\bx)+
\beta |\phi_g^{\rm TF}(\bx)|^2\phi_g^{\rm TF}(\bx), \qquad
\bx\in {\Bbb R}^d.
 \ee
Solving the above equation, we get
 \be \label{eq:gss:sec7}
  \phi_g(\bx)\approx \phi_g^{\rm TF}(\bx)=\left\{\ba{ll}
\sqrt{\left(\mu_g^{\rm TF} -V(\bx)\right)/\bt},
&\ V(\bx)< \mu_g^{\rm TF},\\
0, & \hbox{otherwise},
\ea\right.
 \ee
where $\mu_g^{\rm TF}$ is chosen to satisfy the normalization $\|\phi_g^{\rm TF}\|_2=1$.
After some tedious computations \cite{BaoJinP,BaoLimZhang}, we get
\begin{eqnarray*}
\mu_g^{\rm TF}=\left\{\ba{l}
 \fl{1}{2}\left(\fl{3\bt\gm_x}
{2}\right)^{2/3},\\
\left(\frac{\beta \gamma_x\gamma_y}{\pi}\right)^{1/2},\\
\frac{1}{2}\left(\frac{15\beta \gamma_x\gamma_y\gamma_z}{4\pi}\right)^{2/5},\\
\ea\right.
\qquad
E_g^{\rm TF}=\left\{\ba{ll}
\fl{3}{10}\left(\fl{3\bt\gm_x}{2}
\right)^{2/3}, &d=1,\\
\frac{2}{3}\left(\frac{\beta \gamma_x\gamma_y}{\pi}\right)^{1/2}, &d=2,\\
\frac{5}{14}\left(\frac{15\beta \gamma_x\gamma_y\gamma_z}{4\pi}\right)^{2/5}, &d=3.\\
\ea\right.
 \end{eqnarray*}
It is easy to verify that the Thomas-Fermi approximation (\ref{eq:gss:sec7}) does not have limit as $\beta\to\infty$.

\begin{remark}\label{rmk:tflimit}
For the harmonic potential (\ref{eq:dhp:sec1}), the energy of the Thomas-Fermi approximation is unbounded, i.e.
\begin{equation}
E(\phi_g^{\rm TF}) = + \infty.
\end{equation}
This is due to the low regularity of $\phi_g^{\rm TF}$ at the free boundary
$V(\bx)=\mu_g^{\rm TF}$. More precisely, $\phi_g^{\rm TF}$ is locally
$C^{1/2}$ at the interface. This is a
typical behavior for solutions of free boundary value problems,
which indicates that an interface layer correction has to be constructed
in order to improve the approximation quality.
\end{remark}


\subsection{Dynamics}
Many properties of dynamics for BEC can be reported   by solving  GPE (\ref{eq:gpe:sec2}). In this section, we will consider the well-posedness for Cauchy problem of GPE (\ref{eq:gpe:sec2}). For BEC, energy (\ref{eq:energy:sec2}) is an important physical quantity and thus it is natural to study the well-posedness in the energy space $X(\Bbb R^d)$ ($d=1,2,3$) (\ref{eq:funcspace:sec2}).
\subsubsection{Well-posedness}
To investigate the Cauchy problem of (\ref{eq:energy:sec2}), dispersive estimates (Strichartz estimates) have played  very important roles. For smooth potentials $V(\bx)$ with at most quadratic growth in far field, i.e.,
\be\label{eq:potcon:sec2}
V(\bx)\in C^\infty(\Bbb R^d) \hbox{ and
}D^{{\bf k}} V(\bx)\in L^\infty(\Bbb R^d),\quad \hbox{for all }
{\bf k}\in{\Bbb N}_0^d\  \hbox{with}\  |{\bf k}|\ge 2,
\ee
where $\NN_0=\{0\}\cup {\Bbb N}$,
Strichartz estimates are well established \cite{Cazenave,Strichartz}.

\begin{definition} In $d$ dimensions ($d=1,2,3$), let $q^\prime$ and $r^\prime$ be the conjugate
 index of $q$ and $r$ ($1\leq q,r\leq \infty$), respectively, i.e. $1=1/q^\prime+1/q
 =1/r^\prime+1/r$, we call the pair $(q,r)$  admissible and $(q^\prime,r^\prime)$ conjugate admissible if
\be \frac{2}{q}=d\left(\frac{1}{2}-\frac{1}{r}\right), \ee
and
\be
2\leq
r<\frac{2d}{d-2},\quad (2\leq r\leq \infty \quad\text{if}\quad d=1;\;2\leq r< \infty \quad\text{if}\quad d=2).
\ee
\end{definition}

 Consider the unitary group  $e^{it{\bf H}_{\bx}^V}$ generated by ${\bf H}_{\bx}^V=-\frac{1}{2}\nabla^2+V(\bx)$, for $V(\bx)$ satisfying (\ref{eq:potcon:sec2}), then the following estimates are available.

\begin{lemma}\label{lem:stri}(Strichartz's estimates) Let $(q,r)$ be an admissible
pair and $(\gamma,\varrho)$ be a conjugate admissible pair, $I\subset
{\Bbb R}$ be a bounded interval satisfying $0\in I$, then we have

(i) There exists a constant  $C$ depending on $I$ and $q$ such that
\be \left\|e^{-it{\bf H}_{\bx}^V}\varphi\right\|_{L^q(I,L^r(\Bbb R^d))}\leq
C(I,q) \|\varphi\|_{L^2(\Bbb R^d)}. \ee

(ii)  If $f\in L^{\gamma}(I,L^{\varrho}(\Bbb R^d))$, there exists a
constant  $C$ depending on $I$, $q$ and $\varrho$, such that \be
\left\|\int_{I\bigcap s\leq
t}e^{-i(t-s){\bf H}_{\bx}^V}f(s)\,ds\right\|_{L^q(I,L^r(\Bbb R^d))}\leq
C(I,q,\varrho) \|f\|_{L^{\gamma}(I,L^{\varrho}(\Bbb R^d))}. \ee
\end{lemma}

Using the above lemma, we can get the following results \cite{Cazenave,Sulem}.
\begin{theorem}\label{thm:dy}
(Well-posedness of Cauchy problem) Suppose the real-valued trap
potential satisfies $V(\bx)\ge0$ ($\bx\in\Bbb R^d$, $d=1,2,3$) and the condition (\ref{eq:potcon:sec2}),  then we have

(i) For any initial data $\psi(\bx,t=0)=\psi_0(\bx)\in X(\Bbb R^d)$,
 there exists a
$T_{{\rm max}}\in(0,+\infty]$ such that the Cauchy problem of
(\ref{eq:gpe:sec2})
 has a unique maximal solution
$\psi\in C\left([0,T_{{\rm max}}),X\right)$. It is maximal in
the sense that if $T_{{\rm max}}<\infty$, then
$\|\psi(\cdot,t)\|_{X}\to\infty$ when  $t\to T^-_{{\rm
max}}$.

(ii) As long as the solution $\psi(\bx,t)$ remains in the energy
space $X$, the {\sl $L^2$-norm} $\|\psi(\cdot,t)\|_2$ and {\sl
energy} $E(\psi(\cdot,t))$ in (\ref{eq:energy:sec2}) are conserved for
$t\in[0,T_{\rm max})$.

(iii) The solution of the Cauchy problem for (\ref{eq:gpe:sec2}) is global in time, i.e.,
  $T_{\rm max}=\infty$, if $d=1$ or $d=2$ with $\beta>C_b/\|\psi_0\|_2^2$ or $d=3$ with $\beta\ge0$.
\end{theorem}
\subsubsection{Dynamical properties}
From Theorem \ref{thm:dy}, the GPE (\ref{eq:gpe:sec2}) conserves the energy (\ref{eq:energy:sec2}) and the mass ($L^2$-norm) (\ref{eq:norm:sec2}). There are other important quantities that measure the dynamical properties of BEC.  Consider the momentum defined as
\be\label{eq:mom:sec2}
{\bf P}(t)=\int_{\Bbb R^d}{\rm Im}(\bar{\psi}(\bx,t)\nabla \psi(\bx,t))\,d\bx,\quad t\ge0,
\ee
 where ${\rm Im}(c)$ denotes the imaginary part of $c$. Then we can get the following result.
 \begin{lemma}\label{lem:monode} Suppose $\psi({\bx},t)$ is the solution of the problem
(\ref{eq:gpe:sec2}) and $|\nabla V(\bx)|\leq C(V(\bx)+1)$ ($V(\bx)\ge0$) for some constant $C$, then we have
\be
\label{eq:momode:sec2}
\qquad \dot{{\bf P}}(t)=-\int_{\Bbb R^d}|\psi(\bx,t)|^2\nabla V(\bx)\,d\bx.
\ee
In particular, for $V(\bx)\equiv0$, the momentum is conserved.
 \end{lemma}
 \begin{proof}
 Differentiating (\ref{eq:mom:sec2}) with respect to $t$, noticing (\ref{eq:gpe:sec2}),
integrating by parts and taking into account that $\psi$ decreases to 0
exponentially when $|{\bx}|\to\infty$ (see also \cite{Cazenave}), we have
\begin{align*}
\dot{{\bf P}}(t)
=& -i\int_{\mathbb{R}^d}\left[\;\bar{\psi}_t\nabla\psi
-\nabla\bar{\psi}\psi_t\right]d{\bx}
= \int_{\mathbb{R}^d}\left[\left(-i\bar{\psi}_t\right)\nabla\psi
+i\psi_t\nabla\bar{\psi}\;\right]d{\bx}\\
=&\int_{\mathbb{R}^d}\left[\left(-\fl{1}{2}\nabla^2\bar{\psi}+V
\bar{\psi}+\bt|\psi|^2\bar{\psi}\right)\nabla\psi+\left(-\fl{1}{2}\nabla^2\psi+V\psi+\bt|\psi|^2\psi\right)
\nabla\bar{\psi}\right]d{\bx}\\
=&\int_{\mathbb{R}^d}\bigg[\fl{1}{2}\nabla |\nabla\psi|^2+V(\bx)\nabla|\psi|^2+\frac{\beta}{2}\nabla|\psi|^4\bigg]\,d\bx\\
=& -\int_{\mathbb{R}^d}|\psi|^2\nabla V(\bx)\,d\bx,\quad t\geq 0.
\end{align*}
The proof is complete.
 \end{proof}

 Another quantity characterizing the dynamics of  BEC is the
condensate width  defined as
\bea
\label{eq:def_sigma:sec2}
\sigma_{\ap}(t) = \sqrt{\dt_{\ap}(t)},\quad\mbox{where}\;\;
\dt_\ap(t)  =
\int_{\mathbb{R}^d}\ap^2|\psi({\bx},t)|^2d{\bx},
\eea
for $t\geq 0$ and $\alpha$ being either $x,y$ or $z$, with $\bx=x$ in 1D, $\bx=(x,y)^T$ in 2D and $\bx=(x,y,z)^T$ in 3D.
For the dynamics of condensate widths, we have the following lemmas:
\begin{lemma}\label{lem:variance}
Suppose $\psi({\bx},t)$ is the solution of
(\ref{eq:gpe:sec2}) in $\Bbb R^d$ ($d=1,2,3$) with initial data $\psi(\bx,0)=\psi_0(\bx)$, then we have
\bea
\label{eq:sigma_ODE2:sec2}
&&\ddot{\dt}_{\ap}(t) = \int_{{\mathbb  R}^d}\left[2|\p_{\alpha}\psi|^2+\beta
|\psi|^4-2\alpha|\psi|^2\p_{\alpha}V(\bx)\right]\,d\bx,\quad t\geq 0,\\
\label{eq:sigma_init0:sec2}
&&\dt_{\ap}(0) = \dt_{\ap}^{(0)} = \int_{\mathbb{R}^d}\ap^2|\psi_0({\bx})|^2
d{\bx},\qquad \ap = x, y, z,\\
\label{eq:sigma_init1:sec2}
&&\dot{\dt}_{\ap}(0) = \dt_{\ap}^{(1)} = 2\int_{\mathbb{R}^d}\ap\, {{\rm Im}}
\left(\bar{\psi}_0\p_\ap\psi_0\right)\; d{\bx}.
\eea
\end{lemma}
\begin{proof}
 Differentiating (\ref{eq:def_sigma:sec2}) with respect to  $t$, applying (\ref{eq:gpe:sec2}),
and integrating by parts, we obtain \be
\dot{\delta}_\alpha(t)=-i\int_{\Bbb
R^d}\left[\alpha\bar{\psi}(\bx,t)\p_{\alpha}\psi(\bx,t)-
\alpha\psi(\bx,t)\p_{\alpha}\bar{\psi}(\bx,t)\right]\,d\bx, \qquad
t\ge0.\ee Similarly, we have \be \label{eq:d2ap22:sec2}
\ddot{\delta}_\alpha(t)=\int_{\Bbb
R^d}\left[2|\p_{\alpha}\psi|^2+\beta
|\psi|^4-2\alpha|\psi|^2\p_{\alpha}V(\bx)\right]\,d\bx, \ee
and the conclusion follows.
\end{proof}

Based on the above Lemma, when $V(\bx)$ is taken as the harmonic potential (\ref{eq:dhp:sec1}),
it is easy to show that the condensate width is a periodic function
whose frequency is doubling the trapping frequency in a few special cases \cite{BaoZhang}.

\begin{lemma}
(i) In 1D without interaction, i.e. $d=1$ and $\beta=0$ in (\ref{eq:gpe:sec2}), for any initial data
$\psi(x,0)=\psi_0=\psi_0(x)$, we have
\be
\label{eq:solution_dt_r:sec27}
\dt_x(t) = \fl{E(\psi_0)}{\gm_x^2}+\left(\dt_x^{(0)}-\fl{E(\psi_0)}{\gm_x^2}\right)
\cos(2\gm_xt)+\fl{\dt_x^{(1)}}{2\gm_x}\sin(2\gm_x t),\qquad t\ge0.
\ee
(ii) In 2D with a radially symmetric trap, i.e. $d = 2$ and
$\gm_x=\gm_y:=\gm_r$ in (\ref{eq:dhp:sec1}) and (\ref{eq:gpe:sec2}), for any initial data
$\psi(x,y,0)=\psi_0=\psi_0(x,y)$, we have
\be
\label{eq:solution_dt_r:sec28}
\dt_r(t) = \fl{E(\psi_0)}{\gm_r^2}+\left(\dt_r^{(0)}-\fl{E(\psi_0)}{\gm_r^2}\right)
\cos(2\gm_rt)+\fl{\dt_r^{(1)}}{2\gm_r}\sin(2\gm_r t), \qquad t\ge0,
\ee
where $\dt_r(t)= \dt_x(t)+\dt_y(t)$,
$\dt_r^{(0)}:=\dt_x(0)+\dt_y(0)$, and
$\dt_r^{(1)}:=\dot{\dt}_x(0)+\dot{\dt}_y(0)$.
Furthermore, when the initial condition $\psi_0(x,y)$  satisfies
\be
\label{eq:vortex_initial:sec5}
\psi_0(x,y)=f(r)e^{im\theta}\quad{\rm with}\quad m\in{\mathbb Z}\quad
{\rm and} \quad f(0) = 0\quad{\rm when}\quad m\neq0,
\ee
we have, for any $t\geq 0$,
\begin{align}
\label{eq:solution_dt_xy:sec29}
\dt_x(t)=&\dt_y(t) = \fl{1}{2}\dt_r(t)\nonumber\\
=&\fl{E(\psi_0)}{2\gm_x^2}+\left(\dt_x^{(0)}-\fl{E(\psi_0)}{2\gm_x^2}\right)
\cos(2\gm_xt)+\fl{\dt_x^{(1)}}{2\gm_x}\sin(2\gm_x t), \qquad t\ge0.
\end{align}
\end{lemma}

For the dynamics of BEC, the center of mass is also important, which is given by
\be\label{eq:centerofm:sec2}
{\bx}_c(t)=\int_{\Bbb R^d}\bx |\psi(\bx,t)|^2\,d\bx,\quad t\ge0.
\ee
Following the proofs for Lemmas \ref{lem:monode} and \ref{lem:variance}, we can get the equation governing the motion of ${\bx}_c$.
\begin{lemma}\label{lem:centerofmass}
Suppose $\psi({\bx},t)$ is the solution of
(\ref{eq:gpe:sec2}) in $\Bbb R^d$ ($d=1,2,3$) with initial data $\psi(\bx,0)=\psi_0(\bx)$, then we have
\bea
\label{eq:mc_ODE2:sec2}
&&\dot{\bx}_c(t)=
{\bf P}(t),\qquad\ddot{\bx}_{c}(t) = -\int_{\Bbb R^d}|\psi(\bx,t)|^2\nabla V(\bx)\,d\bx,\quad t\geq 0,\\
\label{eq:mc_init0:sec2}
&&{\bx}_{c}(0) = {\bx}_{c}^{(0)} = \int_{\mathbb{R}^d}\bx|\psi_0({\bx})|^2
d{\bx},\\
\label{eq:mc_init1:sec2}
&&\dot{\bx}_c(0) = {\bx}_{c}^{(1)} = {\bf P}(0)=\int_{\Bbb R^d}{\rm Im}(\bar{\psi}_0\nabla \psi_0)\,d\bx.
\eea
\end{lemma}
\begin{proof} Analogous  calculation to Lemma \ref{lem:monode} shows that
\be
\dot{\bx}_c(t)=\frac{i}{2}\int_{\Bbb R^d}(\psi\nabla\bar{\psi}-\bar{\psi}\nabla\psi)\,d\bx=
{\bf P}(t),\quad t\ge0.
\ee
Hence, Lemma \ref{lem:monode} leads to the expression for $\ddot{\bx}_c(t)$.
\end{proof}
\begin{remark}\label{lem:harmmc} When $V(\bx)$ is the harmonic potential (\ref{eq:dhp:sec1}),  Eq. (\ref{eq:mc_ODE2:sec2}) can be rewritten as
\be\label{eq:mc_ODE:sec2}
\ddot{\bx}_{c}(t)+A{\bx}_c(t)=0,\quad t\ge0,
\ee
where $A$ is a $d\times d$ diagonal matrix as $A=(\gamma_x^2)$ when $d=1$, $A={\rm diag}(\gamma_x^2,\gamma_y^2)$ when $d=2$ and $A={\rm diag}(\gamma_x^2,\gamma_y^2,\gamma_z^2)$ when $d=3$. This immediately implies that each
component of $\bx_c$ is a periodic function whose frequency is the same as the trapping frequency in that component.
\end{remark}

For the harmonic potential (\ref{eq:dhp:sec1}), Remark \ref{lem:harmmc} provides a way to construct the exact solution of the GPE (\ref{eq:gpe:sec2}) with a stationary state as initial data. Let $\phi_e(\bx)$ be a stationary state of the GPE (\ref{eq:gpe:sec2})
with a chemical potential $\mu_e$ \cite{BaoTang,BaoWangP}, i.e. $(\mu_e,\phi_e)$
satisfying
\be
\label{eq:nep:sec2}
\mu_e\phi_e({\bx}) = -\fl{1}{2}\nabla^2\phi_e + V({\bx}) \phi_e +
\bt|\phi_e|^2\phi_e , \qquad \|\phi_e\|_2^2 = 1.
\ee
If the initial data $\psi_0(\bx)$ for the Cauchy problem of (\ref{eq:gpe:sec2}) is chosen as
a stationary state with a shift in its center, one can construct an exact
solution of the GPE (\ref{eq:gpe:sec2}) with a harmonic oscillator potential
(\ref{eq:dhp:sec1}) \cite{Bialy,Garcia}. This kind of analytical construction can be used, in
particular, in the benchmark and validation of numerical algorithms for GPE.

\begin{lemma}\label{lem:shift:sec2} Suppose $V(\bx)$ is given by (\ref{eq:dhp:sec1}),
if the initial data $\psi_0(\bx)$  for the Cauchy problem of (\ref{eq:gpe:sec2}) is chosen as
\be
\label{eq:init5:sec2}
\psi_0(\bx)=\phi_e(\bx-\bx_0), \qquad
\bx \in {\mathbb R}^d,
\ee
where $\bx_0$ is a given point in ${\mathbb R}^d$, then the exact solution
of (\ref{eq:gpe:sec2}) satisfies:
\be
\label{eq:exacts1:sec2}
\psi(\bx,t)=\phi_e(\bx-\bx_c(t))\;e^{-i\mu_e t}\; e^{iw(\bx,t)},
\qquad \bx\in{\mathbb R}^d, \quad t\ge 0,
\ee
where for any time $t\ge0$, $w(\bx,t)$ is linear for $\bx$, i.e.
\be
\label{eq:exact3:sec2}
w(\bx,t) = {\bf c}(t) \cdot \bx + g(t), \qquad {\bf c}(t)=(c_1(t), \cdots,
c_d(t))^T, \qquad \bx\in {\mathbb R}^d, \quad t\ge0,
\ee
and $\bx_c(t)$ satisfies the second-order ODE system (\ref{eq:mc_ODE:sec2}) with initial condition
\be
\bx_c(0)=\bx_0,\quad \dot{\bx}_c(0)={\bf 0}.
\ee
\end{lemma}
\begin{proof} See detailed proof in \cite{BaoDuZhang}.
\end{proof}

\subsubsection{Finite time blow-up and damping}\label{subsubsec:blowup}
According to Theorem \ref{thm:dy}, there is a maximal time $T_{\rm max}$ for the existence of the  solution in energy space. If $T_{\rm max}<\infty$, there exists finite time blow up.

\begin{theorem}\label{thm:blowup}(Finite time blow-up) In 2D and 3D, assume $V(\bx)$ satisfies (\ref{eq:potcon:sec2}) and $d\,V(\bx)+ \bx\cdot \nabla V(\bx)\ge0$ for
$\bx\in{\Bbb R}^d$ ($d=2,3$). When $\beta<0$, for any initial data
$\psi(\bx,t=0)=\psi_0(\bx)\in X$ with finite variance $\int_{\Bbb R^d}|\bx|^2|\psi_0|^2\,d\bx<\infty$  to the Cauchy problem of
(\ref{eq:gpe:sec2}), there exists finite time blow-up, i.e.,
$T_{\rm max}<\infty$, if one of the following holds:

(i) $E(\psi_0)<0$;

(ii) $E(\psi_0)=0$ and ${\rm Im}\left(\int_{\Bbb
R^d}\bar{\psi}_0(\bx)\ (\bx\cdot\nabla\psi_0(\bx))\,d\bx\right)<0$;

(iii) $E(\psi_0)>0$ and ${\rm Im}\left(\int_{\Bbb R^d}
\bar{\psi}_0(\bx)\ (\bx\cdot\nabla\psi_0(\bx))\,d\bx\right)
<-\sqrt{d\;E(\psi_0)}\|\bx\psi_0\|_{L^2}$.
\end{theorem}
\begin{proof} Define the variance
\be
\delta_{_V}(t)=\int_{\Bbb R^d}|\bx|^2|\psi(\bx,t)|^2\,d\bx.
\ee
Lemma \ref{lem:variance} indicates that $\delta^\prime_{_V}(t)=2\,{\rm Im}\left(\int_{\Bbb R^d}\bar{\psi}(\bx,t)(\bx\cdot\nabla\psi(\bx,t))\,d\bx\right)$ and
\begin{align*}
\ddot\delta_{_V}(t)=&2d\int_{\Bbb R^d}\left(
\frac{1}{d}|\nabla\psi|^2+\frac{\beta}{2}|\psi|^4-\frac{1}{d}|\psi|^2\bx\cdot\nabla V(\bx)\right)\,d\bx\\
=&2d\,E(\psi)-(d-2)\int_{\Bbb R^d}|\nabla\psi|^2\,d\bx-2\int_{\Bbb R^d}|\psi(\bx,t)|^2(d\,V(\bx)+\bx\cdot V(\bx))\,d\bx\\
\leq & 2d\,E(\psi)=2d\,E(\psi_0),\qquad d=2,3.
\end{align*}
Thus,
\be
\delta_{_V}(t)\leq d\,E(\psi_0)t^2+\delta^\prime_{_V}(0)t+\delta_{_V}(0).
\ee
When one of the conditions (i), (ii) and (iii) holds, there exists a finite time $t^*>0$ such that $\delta_{_V}(t^*)\leq0$, which means that there is a singularity at or before $t=t^*$.
\end{proof}

Theorem \ref{thm:blowup} shows that the solution of the GPE (\ref{eq:gpe:sec2}) may blow up for negative $\beta$ (attractive interaction)
in 2D and 3D. However, the physical quantities modeled by $\psi$ do not become
infinite which implies that the validity of (\ref{eq:gpe:sec2}) breaks
down near the singularity. Additional physical mechanisms, which
were initially small, become important near the singular point and
prevent the formation of the singularity. In BEC, the particle
density $|\psi|^2$ becomes large close to the critical point and
inelastic collisions between particles which are negligible for
small densities become important. Therefore a small damping
(absorption) term is introduced into the NLSE (\ref{eq:gpe:sec2}) which
describes inelastic processes. We are interested in the cases
where these damping mechanisms are important and, therefore,
restrict ourselves to the case of focusing nonlinearity, i.e. $\bt<0$,
where $\bt$ may also be time dependent. We consider the following
damped nonlinear Schr\"{o}dinger equation:
\begin{eqnarray} \label{eq:sdged:sec2}
&&i\; \p_t\psi=-\fl{1}{2}\;\nabla^2 \psi+ V(\bx)\; \psi +\bt
|\psi|^{2\sg}\psi-i\; f(|\psi|^2)\psi,
\quad t>0, \; \bx\in {\Bbb R}^d, \qquad\\
\label{eq:sdgid:sec2}
&&\psi(\bx,t=0)=\psi_0(\bx), \qquad \bx\in {\Bbb
R}^d,
\end{eqnarray}
where $f(\rho)\ge 0$ for $\rho=|\psi|^2\ge 0$ is a real-valued
monotonically increasing function.

The general form of (\ref{eq:sdged:sec2}) covers many damped NLSE arising
in various different applications. In BEC, for example, when
$f(\rho)\equiv 0$, (\ref{eq:sdged:sec2}) reduces to the usual GPE
(\ref{eq:gpe:sec2}); a linear damping term $f(\rho)\equiv \dt$ with
$\dt>0$ describes inelastic collisions with the background gas;
cubic damping $f(\rho)=\dt_1 |\bt| \rho$ with $\dt_1>0$ corresponds
to two-body loss \cite{Saito,Roberts}; and a quintic damping term
of the form $f(\rho)=\dt_2 \bt^2 \rho^2$ with $\dt_2>0$ adds
three-body loss to the GPE (\ref{eq:gpe:sec2}) \cite{Saito,Roberts}. It is
easy to see that the decay of the normalization according to
(\ref{eq:sdged:sec2}) due to damping is given by
\be \label{eq:normNt:sec2}
\dot{N}(t)=\fl{\rd}{\rd t} \int_{{\Bbb R}^d}\; |\psi({\bf x}, t)|^2\;
d{\bf x} = -2\int_{{\Bbb R}^d}\; f(|\psi({\bf x}, t)|^2)|\psi({\bf
x}, t)|^2\; d{\bf x} \le 0, \quad t >0.
\ee
Particularly, if
$f(\rho)\equiv \dt$ with $\dt>0$, the normalization is given by
\be \label{eq:drnt:sec2}
N(t)=\int_{{\Bbb R}^d}\; |\psi({\bf x},
t)|^2\; d{\bf x}= e^{-2\dt\; t} N(0)=e^{-2\dt\; t}\; \int_{{\Bbb
R}^d}\; |\psi_0({\bf x})|^2\; d{\bf x}, \quad t\ge 0.
\ee
For more discussions, we refer to \cite{BaoJaksch}.

\subsection{Convergence of dimension reduction}
\label{subsec:dredrate}
In an experimental setup with harmonic potential (\ref{eq:dhp:sec1}), the trapping frequencies in
different directions can be very different. Especially, disk-shaped
 and cigar-shaped condensate were observed in
experiments. In section \ref{subsubsec:dred},
the 3D GPE  is formally reduced to 2D GPE in
disk-shaped condensate and to 1D GPE in cigar-shaped
condensate. Mathematical
and numerical justification for the dimension reduction of 3D GPE
is only available in the weakly interaction regime, i.e.
$\beta=o(1)$ \cite{BaoP,bamsw,bacm}. Unfortunately, in the intermediate ($\beta=O(1)$)
or strong repulsive interaction regime ($\beta\gg1$), no mathematical results are available and numerical
studies can be found in \cite{BaoGeJakschPW}.

For weak interaction regime,  the dimension reduction  is verified by energy type method with projection discussed in section \ref{subsubsec:dred}  \cite{BaoP,bamsw}. Later,  Ben Abdallah et al. developed an averaging technique and proved the more general forms of the lower dimensional GPE  \cite{bacm} without using the projection method. A more refined model in lower dimensions is developed in \cite{BenCai}. Here, we introduce this general approach. We refer to \cite{bacm} and references therein for more discussions.

Consider the 3D GPE for $(\bx,\bz)\in\Bbb R^{d}\times \Bbb R^n$ with $d+n=3$ ($d=1$ or $d=2$)
\begin{equation}\label{eq:gpedred:sec2}
\begin{split}
&i\p_t\psi(\bx,\bz,t)=\left[-\frac12\left(\Delta_{\bx}+\Delta_{\bz}\right)+
V^\vep(\bx,\bz)+\beta|\psi|^2\right]\psi(\bx,\bz,t),\\
&\psi(\bx,\bz,0)=\Psi^{\rm init}(\bx,\bz),\quad V^\vep(\bx,\bz)=\frac{|\bx|^2}{2}+\frac{|\bz|^2}{2\vep^2},\quad \bx\in\Bbb R^d,\ \bz\in\Bbb R^n,
\end{split}
\end{equation}
where $\Delta_{\bx}$ and $\Delta_{\bz}$ are the Laplace operators in $\bx\in\Bbb R^d$ and $\bz\in\Bbb R^n$, respectively. The wave function is normalized as $\|\Psi^{\rm init}\|_{L^2(\Bbb R^3)}^2=1$. Compared with the situation in section \ref{subsubsec:dred} (\ref{eq:gpeg}), we have $0<\vep=1/\gamma_z\ll1$ ($d=2$) with $\gamma_y=1$ in disk-shaped BEC (2D) and $0<\vep=1/\gamma\ll1$ ($d=1$) with $\gamma_y=\gamma_z=\gamma$ in cigar-shaped BEC (1D). Our purpose is to describe the limiting dynamics of (\ref{eq:gpedred:sec2}) for $0<\vep\ll1$.

First, we introduce the rescaling $\bz\to\vep^{1/2}\bz$ and rescale $\psi\to e^{-itn/2\vep}\vep^{-n/4}\psi^\vep(\bx,\bz,t)$ to keep the normalization. Then Eq. (\ref{eq:gpedred:sec2}) becomes
\be \label{eq:GPE:sec2}
i\p_t \psi^\vep(\bx,\bz,t) =
\bH_{\bx}\psi^\vep + {1\over \vep} \bH_{\bz} \psi^\vep +\frac{\beta}{\vep^{n/2}}
|\psi^\vep|^2\psi^\vep,\quad
(\bx,\bz)\in\mathbb{R}^d\times\mathbb{R}^n,\ee with initial data \be
 \psi^\eps(t=0)=\Psi^{\rm init}\in L^2(\Bbb R^d\times\Bbb R^n),
\ee where
\be\label{eq:weak:sec2}
\bH_{\bx}=\frac12\left(-\Delta_{\bx} +|\bx|^2\right),\quad \bH_{\bz}=\frac12\left(-\Delta_{\bz}+|\bz|^2-n\right),
\quad\frac{\beta}{\vep^{n/2}}:=\delta \in\mathbb{R}.
\ee
Here $\beta=\delta\vep^{n/2}$ with a constant $\delta\in\Bbb R$ means that we are working in the weak interaction regime, i.e., $\beta=O(\vep^{1/2})$ in 2D disk-shaped BEC and $\beta=O(\vep)$ in 1D cigar-shaped BEC. Notice that the singularly perturbed Hamiltonian $\bH_{\bz}$ is a harmonic oscillator (conveniently shifted here such that it admits integer eigenvalues).

By introducing the filtered unknown
\be\Psi^\vep=e^{it\bH_{\bz}/\vep}\,\psi^\vep,\ee
we get the equation
\be \label{eq:filtre1:sec2}
i\p_t \Psi^\vep(\bx,\bz,t) =
\bH_{\bx}\Psi^\vep(\bx,\bz,t)
+F\left(\frac{t}{\vep},\Psi^\vep\right),\qquad
\Psi^\vep(t=0)=\Psi^{\rm init}, \ee where $F$ is equal to
\be\label{eq:Fdef:sec2}
F(s,\Psi)=\delta \,e^{is\bH_{\bz}}\left(\left|e^{-is\bH_{\bz}}\Psi\right|^2e^{-is\bH_{\bz}}\,\Psi\right).\ee
When $\vep$ is small, (\ref{eq:GPE:sec2}) (or, equivalently, (\ref{eq:filtre1:sec2})) couples the high oscillations in time generated by the strong confinement operator with a nonlinear dynamics in the $\bx$ plane, which is the only phenomenon that we want to describe.

In \cite{bacm}, Ben Abdallah et al. have developed an averaging technique and proved that, for general confining potentials in the $z$ direction, the limiting model as $\eps$ goes to zero is
\begin{equation} \label{eq:1storder:sec2}i\p_t \Psi = \bH_{\bx}\Psi +F_{\rm av}\left(\Psi\right),\qquad
\Psi(t=0)=\Psi^{\rm init}, \end{equation} where the long time average of $F$ is
defined by
\be
F_{\rm av}(\Psi)=\lim_{T\to +\infty}\frac{1}{T}\int_0^T F(s,\Psi)\,ds.\ee
For general confining operator $\bH_{\bz}$, the convergence is proved  using the fact that $F(s,\Psi)$ is almost periodic \cite{bacm}, but the convergence rates are generally unclear.
In the specific case of a harmonic confinement operator, like here, this convergence result can be quantified. The important point is that $\bH_{\bz}$ admits only integer eigenvalues and the function $F$ is $2\pi$-periodic with respect to the $s$ variable. Therefore, the expression of $F_{av}$ is not a limit but a simple integral, and we have in fact
\be\label{eq:Fav:sec2}F_{\rm av}(\Psi)=\frac{1}{2\pi}\int_{0}^{2\pi} F(s,\Psi)\,ds.\ee
On top of that, one can characterize the rate of convergence and prove that $\Psi$ is a first order approximation of $\Psi^\vep$ in $\vep$.

 Rigourously,
in order to state the convergence, we introduce the convenient scale of functional spaces. For all $\ell \in \RR^+$, we set
$$\calB_\ell:=\left\{\psi\in H^\ell(\RR^3) \bigg |
(|\bx|^2+|\bz|^2)^{\ell/2}\psi\in L^2(\RR^3)\right\}$$
endowed with one of the two following equivalent norms:
\be
\|u\|_{\calB_\ell}^2:=\|u\|_{L^2(\mathbb{R}^3)}^2+\|\bH_{\bx}^{\ell/2}u\|^2_{L^2(\mathbb{R}^3)}+
\|\bH_{\bz}^{\ell/2}u\|^2_{L^2(\mathbb{R}^3)}
\label{eq:norm1:sec2}\ee
or
\be
\|u\|_{\calB_\ell}^2:=\|u\|_{H^\ell(\mathbb{R}^3)}^2+\|(|\bx|^2+|\bz|^2)^{\ell/2}u\|^2_{L^2(\mathbb{R}^3)}.
\label{eq:norm2:sec2}\ee
For the  equivalence, see e.g. Theorem 2.1 in \cite{bacm}.

We have the convergence as the following.
\begin{theorem}\label{thm:dred} For some real number $m>3/2$, assume that the initial
datum $\Psi^{\rm init}$ belongs to $\calB_{m+4}$. Let $\Psi^\vep(\bx,\bz,t)=e^{it\bH_{\bz}/\vep}\psi^\vep$ be the solution of the filtered equation
\be\label{eq:filtre:sec2} i\p_t \Psi^\vep(\bx,\bz,t) =
\bH_{\bx}\Psi^\vep(t,x,z) +F\left(\frac{t}{\eps},\Psi^\vep\right),\quad
\Psi^\vep(t=0)=\Psi^{\rm init}, \ee where
\be
F(s,\Psi)=\delta
\,e^{is\bH_{\bz}}\,\left|e^{-is\bH_z}\Psi\right|^2e^{-is\bH_{\bz}}\,\Psi\,.\ee
Define also $\widetilde{\Psi}$ as the solution of the averaged problem
\begin{equation} \label{eq:averaged:sec2}
i\p_t \widetilde{\Psi} = \bH_{\bx}\widetilde{\Psi} +F_{av}\left(\widetilde{\Psi}\right),\qquad
\widetilde{\Psi}(t=0)=\Psi^{\rm init},
\end{equation}
where $F_{av}$ is defined by (\ref{eq:Fav:sec2}).
 Then, we have the following
conclusions.
\begin{enumerate}\renewcommand{\labelenumi}{(\roman{enumi})}
      \item There exists $T_0>0$, depending only on $\|\Psi^{\rm init}\|_{\calB_{m+4}}$,
      such that $\Psi^\vep$ and $\widetilde{\Psi}$
      are uniquely defined and are uniformly bounded in the space $C([0,T_0];\calB_{m+4})$, independently of $\vep\in (0,1]$.
       \item The function $\widetilde \Psi$ is a first order approximation of the solution $\Psi^\eps$ in
$C([0,T_0];\calB_m)$, i.e., for some $C>0$, we have
\be
\|\Psi^\eps(t)-\widetilde \Psi(t)\|_{\calB_m}\leq C \eps,\quad\, \forall t\in [0, T_0].
\ee
\end{enumerate}
\end{theorem}
The readers are referred to \cite{bacm,BenCai} for a detailed proof of Theorem  \ref{thm:dred}.
\begin{remark} The key property here is the periodicity of $F(s,\Psi)$, and the result can be generalized to other dimensions
$\Bbb R^p=\Bbb R^{d}\times\Bbb R^{p-d}$, more general nonlinearities
$f(|\psi|^2)\psi$ in (\ref{eq:GPE:sec2}) and other operators $\bH_{\bz}$ such that $F(s,\Psi)$ defined by (\ref{eq:Fdef:sec2}) is periodic.
\end{remark}

Theorem \ref{thm:dred} implies results of lower dimensional GPE (\ref{eq:gpeg}). Let us take disk-shaped BEC as an example, i.e., $n=1$ and $d=2$. Thus,  the eigenvalues of $\bH_{\bz}$ are the nonnegative integers. Let $\chi_p(\bz)$ be the normalized eigenfunction associated to the eigenvalue $p\in \NN_0$:
\be \bH_{\bz} \chi_p = p \chi_p\,, \qquad \int \chi_p^2 \, d\bz = 1.\ee
In particular,
\be\label{eq:hermit:sec2}
\chi_0(\bz)=e^{-\bz^2/2}/\pi^{1/4},\quad \bz\in\Bbb R.
\ee
Consider a function $\Psi\in \calB_m$ expanded on this basis as
\be\Psi(\bx,z) = \sum_{p=0}^{+\infty} \vphi_p(\bx) \chi_p (\bz)\,.\ee
Then we have
\be\label{eq:Fexpa:sec2} F(s,\Psi)=\delta\sum_{p_1,p_2,p_3,p_4}a_{\mbox{\tiny $p_1p_2p_3p_4$}}\,e^{is\Omega_{\mbox{\tiny $p_1p_2p_3p_4$}}}\,\vphi_{p_2}(\bx)\vphi_{p_3}(\bx)\overline{\vphi_{p_4}}(\bx)\,\chi_{p_1}(\bz),\ee
where we define the coefficients
$$\Omega_{pqrs}=p+s-q-r,\,\quad a_{pqrs} = \langle \chi_p\chi_q\chi_r\chi_s\rangle.$$
Here and in the sequel, $\langle \cdot\rangle$ denotes the integration over the $\bz$ variable.
 We write the expansion (\ref{eq:Fexpa:sec2}) shortly as
\be\label{eq:Fexpas:sec2} F(s,\Psi)=\delta\sum_{p_1,p_2,p_3,p_4}a_{\mbox{\tiny 1234}}\,e^{is\Omega_{\mbox{\tiny 1234}}}\,\vphi_{p_2}\vphi_{p_3}\overline{\vphi_{p_4}}\otimes\chi_{p_1}\,.\ee
In the above sums, and in the sequel, $a_{\mbox{\tiny 1234}}$ and $\Omega_{\mbox{\tiny 1234}}$ stand for $a_{p_1p_2p_3p_4}$ and $\Omega_{p_1p_2p_3p_4}$, respectively.

The expansion of $F_{\rm av}$ (\ref{eq:Fav:sec2}) is obtained by averaging $F(s,\Psi)$ over time $s$. Noticing that the average of  $e^{is\Omega_{\mbox{\tiny 1234}}}$ vanishes if $\Omega_{\mbox{\tiny 1234}}\neq 0$,
 let us  define  the following index set, whose information is preserved after averaging $F(s,\Psi)$ given by (\ref{eq:Fexpas:sec2}), \ for any $p \in \NN$,
\be\label{eq:Lambdap:sec2} \Lambda(p) = \{(q,r,s),  \,\, \mathrm{such\,\, that}\,\, p + s = q+r\}. \ee
Then we have
\begin{equation}
\label{eq:Favexp:sec2}
F_{\rm av}(\Phi)=\delta \sum_{p_1=0}^\infty\,\,\sum_{\scriptsize(p_2,p_3,p_4) \in \Lambda(p_1)}a_{\mbox{\tiny 1234}}\,\vphi_{p_2}\vphi_{p_3}\overline{\vphi_{p_4}}\otimes\chi_{p_1}.
\end{equation}

The solution of (\ref{eq:filtre:sec2}) is written as $\Psi^\vep(\bx,\bz,t)=\sum\limits_{p=0}^\infty \vphi^\vep_p(\bx,t)\chi_p(\bz)$ and the solution of (\ref{eq:averaged:sec2}) is written as $\Psi(\bx,\bz,t)=\sum\limits_{p=0}^\infty \vphi_p(\bx,t)\chi_p(\bz)$. If the initial data is polarized on the ground  mode of the confinement Hamiltonian, i.e., we have
$$
\forall p\in\NN\setminus\{0\},\qquad \vphi^\vep_p(t=0)=0\quad\mbox{and}\quad \vphi^\vep_{0}(t=0)=\vphi^{\rm init}.$$
In this case, the averaged system (\ref{eq:averaged:sec2}) reads
\begin{eqnarray*}
&&i\p_t\vphi_{p_1}=\bH_{\bx} \vphi_{p_1}+\delta\sum_{\scriptsize (p_2,p_3,p_4)\in \Lambda(p_1)}a_{\mbox{\tiny 1234}}\,\vphi_{p_2}\vphi_{p_3}\overline{\vphi_{p_4}},\\
&&\vphi_{p_1}(t=0)=\delta_{0p_1}\,\vphi^{\rm init},
\end{eqnarray*}
where $\delta_{0p_1}$ is the Kronecker delta.
It is readily seen from this expression that $\vphi_{p}(t)=0$ for all $t$ as soon as $p\neq 0$. Hence the averaged system  (\ref{eq:averaged:sec2}) reduces to the single equation for $\vphi_0$ as
\be i\p_t\vphi_{0}=\bH_{\bx} \vphi_{0}+\delta a_{0000}|\vphi_0|^2\vphi_0,\ee
where $a_{0000}=\frac{1}{\sqrt{2\pi}}$. This is exactly the 2D GPE (\ref{eq:gpeg}) with the choice (\ref{eq:ufw}) (notice that  we adopt a rescaling here).

Similarly, for a cigar-shaped BEC, i.e. $n=2$, when initial data is polarized on the ground  mode of the confinement Hamiltonian, we recover the 1D GPE (\ref{eq:gpeg}).

This averaging technique has solved the dimension reduction of 3D GPE in the weak interaction regime $\beta=o(1)$ (\ref{eq:gpedred:sec2}). However, there seems no progress for the intermediate interaction regime $\beta=O(1)$ (\ref{eq:gpedred:sec2}) yet.

\section{Numerical methods for computing  ground states}
\label{sec:numgs}
\setcounter{equation}{0}\setcounter{figure}{0}\setcounter{table}{0}
In this section, we  review different numerical methods for computing the ground states of BEC (\ref{eq:minp:sec2}). Due to the presence of the confining potential, the ground state decays exponentially fast when $|\bx|\to\infty$ and thus it is natural to truncate the whole space problem (\ref{eq:gpe:sec2}) to a bounded domain $U\subset\Bbb R^d$ with homogeneous Dirichlet boundary conditions. Thus, we consider the  GPE (\ref{eq:gpe:sec2}) in $U$ as
\bea
\label{eq:sdge:sec3}
&&i\; \psi_t=-\fl{1}{2}\;\nabla^2 \psi+ V(\bx)\; \psi +\bt
|\psi|^{2}\psi, \qquad t>0, \qquad \bx\in U\subset{\mathbb R}^d, \\
\label{eq:sdge1:sec3}
&&\psi(\bx,t)= 0, \qquad \bx \in \Gm=\p U, \quad t\ge0.
\eea
The normalization (\ref{eq:norm:sec2}) and energy (\ref{eq:energy:sec2}) then become
\begin{equation} \label{eq:mass:sec3}
N(\psi)=\|\psi(\cdot,t)\|_2^2:=\int_{U}\; |\psi({\bf x}, t)|^2\; d{\bf x}=1,
\qquad t\ge 0,
\end{equation}
and
\begin{equation} \label{eq:energy:sec3}
E(\psi) = \int_{U}\; \left[\fl{1}{2}|\btd \psi({\bf x},
t)|^2
+V(\bx)|\psi(\bx,t)|^2+\fl{\bt}{2}|\psi(\bx,t)|^{4}\right]\;
d\bx, \quad t\ge 0.
\end{equation}
Replacing $\Bbb R^d$ with $U$, many results presented in section \ref{sec:mathgpe} can be directly  generalized to bounded domain case. Similarly, finding the ground state $\phi_g$ of (\ref{eq:sdge:sec3}), i.e. minimizing energy $E(\phi)$ (\ref{eq:energy:sec3}) under normalization  constraint $N(\phi)=1$ (\ref{eq:mass:sec3}), is equivalent to solving the nonlinear eigenvalue problem (\ref{eq:charactereq:sec2}) with boundary condition (\ref{eq:sdge1:sec3}). According to Theorem \ref{thm:gs}, ground state $\phi_g$ can be chosen as nonnegative, and we will restrict ourselves in real-valued wave function $\phi$ throughout this section.
\subsection{Gradient flow with discrete normalization}
\label{subsubsec:gfdn}
One of the most popular
techniques for dealing with the normalization constraint (\ref{eq:mass:sec3})
is through the following construction:
choose a time sequence $0=t_0<t_1<t_2<\cdots<t_n<\cdots$ with
$\tau_n:=\btu t_{n}=t_{n+1}-t_{n}>0$ and $\tau=\max_{n\ge 0} \; \tau_n$.
To adapt an algorithm for the solution of the
usual gradient flow to the minimization problem under a constraint,
it is natural to consider the following gradient flow with discrete normalization (GFDN)
which is  widely used in physical literatures
 for computing the ground state solution of BEC \cite{Bao,BaoDu}:
\begin{align}
\label{eq:ngf1:sec3}
&\phi_t = -\fl{1}{2}\fl{\dt E(\phi)}{\dt \phi}=
\fl{1}{2}\nabla^2 \phi - V(\bx) \phi -\bt\; |\phi|^2\phi, \
\bx\in U,\ t_n<t<t_{n+1}, \ n\ge0, \quad \\
\label{eq:ngf2:sec3}
&\phi(\bx,t_{n+1})\stackrel{\triangle}{=}
\phi(\bx,t_{n+1}^+)=\fl{\phi(\bx,t_{n+1}^-)}{\|\phi(\cdot,t_{n+1}^-)\|_2},
\qquad \bx\in U, \quad n\ge 0,\\
\label{eq:ngf3:sec3}
&\phi(\bx,t)= 0, \qquad \bx \in \Gm,\qquad \quad
\phi(\bx,0)=\phi_0(\bx), \qquad \bx \in U,
\end{align}
where $\phi(\bx, t_n^\pm)=\lim_{t\to t_n^\pm} \phi(\bx,t)$ and
$\|\phi_0\|_2=1$.
In fact, the gradient flow  (\ref{eq:ngf1:sec3}) can be viewed as
applying the steepest decent method to the energy functional
$E(\phi)$ without constraint  and (\ref{eq:ngf2:sec3}) then projecting
the solution back to the unit sphere in order to satisfy the
constraint (\ref{eq:mass:sec3}). From the numerical point of view, the
gradient flow  (\ref{eq:ngf1:sec3}) can  be solved via
traditional techniques and the normalization of the
gradient flow is simply achieved by a projection at
the end of each time step.
In fact, Eq. (\ref{eq:ngf1:sec3}) can be obtained from the GPE (\ref{eq:sdge:sec3}) by $t\to it$.
Thus GFDN is also known as the imaginary time method in physics literatures.

The GFDN (\ref{eq:ngf1:sec3})-(\ref{eq:ngf3:sec3}) possesses the following properties \cite{BaoDu}.
\begin{lemma}\label{lem:energyd:sec3}
Suppose $V(\bx)\ge0$ for all $\bx \in U$, $\bt\ge 0$ and
$\|\phi_0\|_2=1$, then

(i) $\|\phi(\cdot, t)\|_2 \le \|\phi(\cdot, t_n)\|_2=1$ for
$t_n\le t\le t_{n+1}$,
 $n\ge0$.

(ii) For any $\bt\geq 0$,
\be
\label{eq:energydg:sec3}
E(\phi(\cdot,t))\le  E(\phi(\cdot,t')), \qquad t_n\le t'<t\le t_{n+1},
\qquad n\ge0.
\ee

(iii) For $\bt=0$,
\be
\label{eq:energydl:sec3}
E\left(\frac{\phi(\cdot,t)}{\|\phi(\cdot,t)\|_2}\right)\le E\left(\frac{\phi(\cdot,t_n)}{\|\phi(\cdot,t_n)\|_2}\right),
\quad t_n\le t\le t_{n+1}, \quad n\ge0.
\ee
\end{lemma}
The property
(\ref{eq:energydg:sec3}) is often referred as the energy diminishing
property of the gradient flow. It is interesting to note
that (\ref{eq:energydl:sec3}) implies that
the energy diminishing property is preserved even with
the normalization of the solution of the
 gradient flow for $\bt=0$, that is, for
linear evolutionary equations.
\begin{theorem}\label{thm:edh:sec3}
Suppose $V(\bx)\ge0$ for all $\bx \in U$ and
$\|\phi_0\|_2=1$.  For $\bt=0$,
the GFDN (\ref{eq:ngf1:sec3})-(\ref{eq:ngf3:sec3})
is energy diminishing for any time step $\tau$ and
initial data $\phi_0$, i.e.
\be
\label{eq:dphi:sec3}
E(\phi(\cdot,t_{n+1}))\le E(\phi(\cdot, t_n))\le \cdots
\le E(\phi(\cdot,0))=E(\phi_0), \;\; n=0,1,2,\cdots.
\ee
\end{theorem}
For $\beta>0$, the GFDN (\ref{eq:ngf1:sec3})-(\ref{eq:ngf3:sec3}) does not
preserve the diminishing property for the normalization of the solution (\ref{eq:energydl:sec3}) in general.

In fact, the normalized step (\ref{eq:ngf2:sec3}) is equivalent to solving
the following ODE {\sl exactly}
\bea
\label{eq:Ode1:sec3}
&&\phi_t(\bx,t) = \mu_\phi(t,\tau) \phi(\bx,t), \qquad \bx\in U,
\quad t_n < t<t_{n+1}, \quad n\ge0,\\
\label{eq:Ode2:sec3}
&&\phi(\bx,t_n^+)= \phi(\bx,t_{n+1}^-), \qquad   \bx\in U;
\eea
where
\be
\label{eq:sgtk:sec3}
\mu_\phi(t,\tau)\equiv
\mu_\phi(t_{n+1},\tau_n) = -\fl{1}{2\; \tau_n}
\ln \|\phi(\cdot,t_{n+1}^-)\|_2^2,
\qquad t_n\le t\le t_{n+1}.
\ee
Thus the GFDN
(\ref{eq:ngf1:sec3})-(\ref{eq:ngf3:sec3})
can be viewed as a first-order
splitting method for the gradient flow with discontinuous coefficients \cite{BaoDu}:
\bea
\label{eq:nngf1:sec3}
&&\phi_t = \fl{1}{2}\nabla^2 \phi - V(\bx) \phi -\bt\; |\phi|^2\phi
+\mu_\phi(t,\tau)\phi, \qquad
\bx\in U, \quad t> 0,\\
\label{eq:nngf2:sec3}
&&\phi(\bx,t)= 0, \qquad \bx \in \Gm,\qquad \quad
\phi(\bx,0)=\phi_0(\bx), \qquad \bx \in U.
\eea
Let $\tau\to 0$,
we see that
\begin{equation*}
\lim_{\tau\to0^+}\mu_\phi(t,\tau)= \mu_\phi(t)
=\fl{1}{\|\phi(\cdot,t)\|_2^2}\int_{U}
\left[\fl{1}{2}|\btd \phi(\bx,t)|^2+V(\bx)\phi^2(\bx,t)+
\bt\phi^4(\bx,t)\right]d\bx.
\end{equation*}
This suggests us to consider the  following
continuous normalized gradient flow (CNGF) \cite{BaoDu}:
\bea
\label{eq:nkngf1:sec3}
&&\phi_t = \fl{1}{2}\nabla^2 \phi - V(\bx) \phi -\bt\; |\phi|^2\phi
+\mu_\phi(t)\phi, \qquad
\bx\in U, \quad t\ge 0,\\
\label{eq:nkngf2:sec3}
&&\phi(\bx,t)= 0, \qquad \bx \in \Gm, \qquad \quad
\phi(\bx,0)=\phi_0(\bx), \qquad \bx \in U.
\eea
In fact, the right hand side of  (\ref{eq:nkngf1:sec3}) is the same
as (\ref{eq:mu-energy:sec2}) if we view $\mu_\phi(t)$
as a Lagrange multiplier for the
constraint (\ref{eq:mass:sec3}).

Furthermore, for the above CNGF,
as observed in \cite{AftalionDu,BaoDu,Du},
the solution of (\ref{eq:nkngf1:sec3}) also satisfies the following theorem:
\begin{theorem}\label{thm:edhh:sec3}
Suppose $V(\bx)\ge0$ for all $\bx \in U$, $\bt\ge0$ and
$\|\phi_0\|_2=1$. Then the CNGF
(\ref{eq:nkngf1:sec3})-(\ref{eq:nkngf2:sec3}) is normalization conservative and
energy diminishing, i.e.
\bea
\label{eq:ncphi:sec3}
&&\|\phi(\cdot,t)\|_2^2=\int_U \phi^2(\bx,t)\; d\bx =
\|\phi_0\|_2^2=1, \qquad t\ge0,\\
\label{eq:edcngf:sec3}
&&
\fl{\rd}{\rd t}E(\phi)=- 2\left\|\phi_t(\cdot,t)\right\|_2^2\le 0\;,
\qquad t\ge0,
\eea
which in turn implies
$$
E(\phi(\cdot, t_1))\ge E(\phi(\cdot,t_2)), \qquad 0\le t_1\le t_2<\ift.
$$
\end{theorem}
\subsection{Backward Euler finite difference discretization}
\label{subsec:BEFD}
In this section, we will present a backward Euler finite difference  method
 to discretize
the GFDN (\ref{eq:ngf1:sec3})-(\ref{eq:ngf3:sec3}) (or a full discretization
of the CNGF (\ref{eq:nkngf1:sec3})-(\ref{eq:nkngf2:sec3})).
 For simplicity of
notation, we introduce the method for the case of one spatial
dimension $d=1$ with
homogeneous Dirichlet boundary conditions.
Generalizations to higher dimension with a rectangle $U=[a,b]\times[c,d]\subset \Bbb R^2$ and a box
$U=[a,b]\times[c,d]\times[e,f]\subset\Bbb R^3$ are straightforward
for tensor product grids and the results remain valid without
modifications. For $d=1$, we have \cite{BaoDu}
\begin{eqnarray} \label{eq:sdge1d:sec3}
&&\phi_t = \fl{1}{2}\phi_{xx} - V(x) \phi -\bt\; |\phi|^2\phi, \quad
x\in U=(a,b),\ t_n<t<t_{n+1}, \ n\ge0,\qquad \\
\label{eq:sdge1d2:sec3}
&&\phi(x,t_{n+1})\stackrel{\triangle}{=}
\phi(x,t_{n+1}^+)=\fl{\phi(x,t_{n+1}^-)}{\|\phi(\cdot,t_{n+1}^-)\|_2},
\qquad a\le x\le b, \quad n\ge 0,\\
\label{eq:sdge1d3:sec3}
&&\phi(x,0)=\phi_0(x), \quad a\le x \le b,\qquad
\phi(a,t)=\phi(b,t)=0, \qquad t\ge0;
\eea
with $\|\phi_0\|_2^2=\int_a^b \phi_0^2(x)\; dx=1$.

We choose the spatial mesh size $h=\btu x>0$ with $\Delta x=(b-a)/M$,  choose the time step  $\tau=\btu
t>0$ and define the index sets
 \begin{equation}\label{eq:index:sec3} {\mathcal T}_{M}=\{j\ |\ j=1,2,\ldots,M-1\},
 \qquad{\mathcal T}_{M}^0=\{j\ |\ j=0,1,2,\ldots,M\}.
 \end{equation}
 We denote grid points and time steps by
\be\label{eq:mesh1d:sec3}
x_j:=a+j\;h, \qquad j\in\calT_M^0;\qquad t_n := n\; \tau,  \qquad
n=0,1,2,\cdots.
\ee

Let $\phi^{n}_j$ be the numerical approximation of $\phi(x_j,t_n)$
and $\phi^{n}$ the solution vector at time $t=t_n=n\tau$ with
components $\phi_j^{n}$. Introduce the following finite difference
operators:
\be\label{eq:fdnotation:sec3}\begin{split}
&\delta_x^+\phi^n_{j}=\frac{1}{h}(\phi_{j+1}^{n}-\phi_{j}^n),\quad
\delta_x^-\phi_{j}^n=\frac{1}{h}(\phi_{j}^{n}-\phi_{j-1}^n), \quad \delta_x\phi_{j}^n=\frac{
\phi_{j+1}^{n}-\phi_{j-1}^n}{2\,h}, \\ &
\delta_t^+\phi_{j}^n=\frac{1}{\tau}(\phi_{j}^{n+1}-\phi_{j}^n),\quad
\delta_t^-\phi_j^n=\frac{1}{\tau}(\phi_j^n-\phi_j^{n-1}),\quad
\delta_t\phi_{j}^n=\frac{\phi_{j}^{n+1}-\phi_{j}^{n-1}}{2\tau},\\
& \delta_x^2\phi_{j}^n=\frac{\phi_{j+1}^n-2\phi_{j}^n+\phi_{j-1}^n}{h^2},\quad \delta_t^2\phi_{j}^n=\frac{\phi_{j}^{n+1}-2\phi_{j}^n+\phi_{j}^{n-1}}{\tau^2}.
\end{split}
\ee

We denote
 \be
  X_{M}=\left\{v=\left(v_{j}\right)_{j\in{\mathcal T}_M^0}\ |\
v_{0}=v_{M}=0\right\}\subset {\Bbb C}^{M+1},\ee
and define the discrete $l^p$, semi-$H^1$ and $l^\infty$ norms  over $X_M$ as
\begin{equation}\label{eq:fdnorm:sec3}\begin{split} &\|v\|_p^p=h\sum\limits_{j=0}^{M-1}|v_{j}|^p,\qquad
\|\delta_x^+ v\|_2^2=h\sum\limits_{j=0}^{M-1}
\left|\delta_x^+ v_{j}\right|^2, \qquad
\|v\|_{\infty}=\sup\limits_{j\in\mathcal{T}_M^0}|v_{j}|,\\
&\left(u,v\right)=\sum\limits_{j=0}^{M-1}u_j\bar{v}_j,\qquad \langle u,v\rangle=\sum\limits_{j=1}^{M-1}u_j\bar{v}_j,\quad \forall u,v\in X_M.\end{split}
\end{equation}

The {\bf Backward Euler finite difference (BEFD)} method is to use backward Euler
for time discretization and second-order centered finite difference
for spatial derivatives. The detailed scheme is \cite{BaoDu}:
\bea\label{eq:befd3:sec3}
&&\fl{\phi_j^{(1)}-\phi_j^n}{\tau}=\frac{1}{2}\delta_x^2\phi^{(1)}_j-V(x_j)\phi_j^{(1)}-\bt \left(\phi_j^{n}\right)^2
\phi_j^{(1)},  \quad j\in \calT_M,\\
&&\phi_0^{(1)}=\phi_M^{(1)}=0,\quad
\phi_j^0= \phi_0(x_j), \quad \phi_j^{n+1}=\fl{\phi_j^{(1)}}{\|\phi^{(1)}\|_2}, \quad j\in \calT_M^0,
\quad n=0,1,\cdots.\nn
\eea
The above BEFD method is  implicit  and unconditionally stable.
The discretized system can be solved by Thomas' algorithm.
The memory cost is $O(M)$ and computational cost per time step is $O(M)$.  In higher
dimensions (such as 2D or 3D), the associated discretized system can
be solved by iterative methods, for example the Gauss-Seidel
or conjugate gradient (CG) or multigrid (MG) iterative
method  \cite{Bao,BaoDu,BaoWang2}. With the approximation
$\phi^n$ of $\phi$ by BEFD, the energy and chemical potential can be computed as
\beas
&&E(\phi(\cdot,t_n))\approx E^n=h\sum_{j=0}^{M-1}\left[\frac{1}{2}\left|\delta_x^+\phi_j^n\right|^2+V(x_j)|\phi_j^n|^2
+\frac{\beta}{2}|\phi_j^n|^4\right],\\
&&\mu(\phi(\cdot,t_n))\approx \mu^n=E^n+h\sum_{j=0}^{M-1}\frac{\beta}{2}|\phi_j^n|^4,\qquad n\ge0.
\eeas

For $\beta=0$, i.e., linear case, the BEFD discretization (\ref{eq:befd3:sec3})
is energy diminishing and monotone for any $\tau>0$ (see \cite{BaoDu}).

\subsection{Backward Euler pseudospectral method}
\label{subsec:besp}
Spectral method enjoys high accuracy for smooth problems such as the ground state problems in BEC. Thus, it is favorable to use spectral method in numerical computation of ground states.  For simplicity, we shall introduce the method in
1D (\ref{eq:sdge1d:sec3})-(\ref{eq:sdge1d3:sec3}), i.e. $d=1$.
Generalization to $d>1$ is straightforward for tensor product
grids and the results remain valid without modifications. We adopt the same mesh strategy and notations as those in section \ref{subsec:BEFD}.

For any function $\psi(x)\in L^2(U)$ ($U=(a,b)$), $\phi(x)\in C_0(\bar{U})$,
 and vector $\phi=(\phi_0,\phi_1,\ldots,\phi_{M})^T\in X_M$ with
$M$ an even positive integer, denote finite dimensional spaces
 \be
 Y_M={\rm span}\left\{\Phi_l(x)=\sin\left(\mu_l(x-a)\right),
\quad \mu_l=\frac{\pi l}{b-a},\quad x\in U,\,l\in{\calT}_M\right\}.\ee
Let $P_M:L^2\left(U\right)\to Y_M$ be the standard $L^2$ projection
 onto $Y_M$ and $I_M:C_0(\overline{U})\to Y_M$  and $I_M: X_M\to Y_M$
 be the standard sine interpolation operator as
\be\label{eq:interpolation:sec3}
\left(P_M\psi\right)(x)=\sum\limits_{l=1}^{M-1}\hat{\psi}_l
\sin\left(\mu_l(x-a)\right),\quad \left(I_M\phi\right)(x)
=\sum\limits_{l=1}^{M-1}\tilde{\phi}_l
\sin\left(\mu_l(x-a)\right),
\ee
and the coefficients are given by
\be\label{eq:sinetran:sec3}
\hat{\psi}_l=\frac{2}{b-a}\int_a^b\psi(x)\sin\left(\mu_l(x-a)\right)\,dx,\; \tilde{\phi}_l=\frac{2}{M}\sum\limits_{j=1}^{M-1}\phi_j\sin\left(\frac{jl\pi}{M}\right),\; l\in{\calT}_M,
\ee
where $\phi_j=\phi(x_j)$ when $\phi$ is a function instead of a vector.

The  backward Euler sine spectral  discretization  for (\ref{eq:ngf1:sec3})-(\ref{eq:ngf3:sec3}) reads \cite{BaoChernLim}:\\
Find
$\phi^{n+1}(x)\in Y_{M}$ (i.e. $\phi^{+}(x)\in Y_{M}$)  such that
\begin{align}
\label{eq:besp11:sec3}
&\frac{\phi^+(x)-\phi^n(x)}{\tau}=\frac{1}{2} \nabla^2
\phi^{+}(x)-P_M\left[\left(V(x)+\beta |\phi^n(x)|^2\right)
\phi^+(x)\right],  \quad x\in U, \\
\label{eq:besp21:sec3}
&\phi^{n+1}(x)=\frac{\phi^+(x)}{\|\phi^+(x)\|_2}, \quad x\in U,\quad
n=0,1,\cdots;\quad \phi^0(x)=P_{M}\left(\phi_0(x)\right).
\end{align}

The above discretization can be solved in phase space and it is not
suitable in practice due to the difficulty of computing the
integrals in (\ref{eq:sinetran:sec3}). We now present an efficient implementation
by choosing $\phi^0(x)$ as the interpolation of $\phi_0(x)$ on
the grid points $\{x_j, \ j\in{\calT}_{M}^0\}$,
i.e $\phi^0(x_j) =\phi_0(x_j)$ for $j\in{\calT}_{M}^0$, and approximating the integrals in (\ref{eq:sinetran:sec3}) by a
quadrature rule on the grid points.
Let $\phi_{j}^n$  be the approximations of $\phi(x_j,t_n)$, which is the solution
of (\ref{eq:ngf1:sec3})-(\ref{eq:ngf3:sec3}).  Backward Euler sine pseudospectral  (BESP) method   for discretizing (\ref{eq:ngf1:sec3})-(\ref{eq:ngf3:sec3}) reads \cite{BaoChernLim}:
\begin{align}
\label{eq:besp1:sec3}
&\frac{\phi_j^{(1)}-\phi_j^n}{\tau}=\frac{1}{2} \left.
D_{xx}^s\phi^{(1)}\right|_{x=x_j}-V(x_j)\phi_j^{(1)} -\beta |\phi_j^{n}|^2
\phi_j^{(1)},  \quad j\in\calT_M, \\
\label{eq:besp2:sec3}
&\phi_0^{(1)}=\phi_M^{(1)}=0, \; \phi_j^0=\phi_0(x_j),
\;\phi_j^{n+1}=\frac{\phi_j^{(1)}}{\|\phi^{(1)}\|_2}, \;j\in\calT_M^0, \;
n=0,1,\cdots.
\end{align}
 Here $D_{xx}^s$, a pseudospectral differential operator approximation
of $\partial_{xx}$, is defined as
\be\label{eq:sineps:sec3}\left. D_{xx}^s u\right|_{x=x_j}= -\sum_{l=1}^{M-1}\;
\mu_l^2 \tilde{u}_l\; \sin(\mu_l(x_j-a)). \qquad j\in\calT_M. \ee

  In the discretization (\ref{eq:besp1:sec3}), at every time
step, a nonlinear system has to be solved. Here we present
an efficient way to solve it iteratively by introducing
a stabilization term with constant coefficient and using
discrete sine transform (DST):
\begin{equation}
\label{eq:iter1:sec3}
\frac{\phi_j^{(1),m+1}-\phi_j^n}{\tau}=\frac{1}{2} \left.
D_{xx}^s\phi^{(1),m+1}\right|_{x=x_j}-\alpha\phi_j^{(1),m+1}
+\left(\alpha-V(x_j)-\beta |\phi_j^{n}|^2\right) \phi_j^{(1),m},
\end{equation}
 where $m\ge0$, $\phi_j^{(1),0}=\phi_j^n$ and $j\in\calT_M^0$.
Here $\alpha\ge0 $ is called as a stabilization parameter
 to be determined.
Taking discrete sine transform at both sides of (\ref{eq:iter1:sec3}),
we obtain
\begin{equation}
\label{eq:iter3:sec3}
\frac{(\widetilde{\phi^{(1),m+1}})_l-(\widetilde{\phi^n})_l}{\tau}
=-\left(\alpha+\frac{1}{2}\mu_l^2\right)(\widetilde{\phi^{(1),m+1}})_l
+(\widetilde{G^m})_l, \qquad l\in\calT_M,
\end{equation}
where $(\widetilde{G^m})_l$ are the sine transform
coefficients of the vector $G^m=(G_0^m,\cdots,G_M^m)^T$ defined as
\be G_j^m = \left(\alpha-V(x_j)-\beta |\phi_j^{n}|^2\right) \phi_j^{(1),m},
\qquad j\in\calT_M^0.\ee
Solving (\ref{eq:iter3:sec3}), we get
\begin{equation}
\label{eq:iter4:sec3}
 (\widetilde{\phi^{(1),m+1}})_l = \frac{2}{2+\tau (2\alpha +\mu_l^2)}
\left[(\widetilde{\phi^n})_l +\tau \; (\widetilde{G^m})_l\right],
\qquad l\in\calT_M.
\end{equation}
Taking inverse discrete sine transform for (\ref{eq:iter4:sec3}), we get
the solution for (\ref{eq:iter1:sec3}) immediately.

In order to make the iterative method (\ref{eq:iter1:sec3}) for
solving (\ref{eq:besp1:sec3}) converges as fast as possible,  the `optimal' stabilization parameter $\alpha$ in (\ref{eq:iter1:sec3}) is suggested as \cite{BaoLim}:
\begin{equation}
\label{eq:para:sec3}
\alpha_{\rm opt}= \frac{1}{2}\left(b_{\rm max}+b_{\rm min}\right),
\end{equation}
where
\be
b_{\rm max}= \max_{1 \le j \le M-1} \left(V(x_j)+\beta |\phi_j^n|^2\right),
\qquad b_{\rm min}= \min_{1 \le j \le M-1}
\left(V(x_j)+\beta |\phi_j^n|^2\right).
\ee

Similarly, with the approximation
$\phi^n$ of $\phi$ by BESP, the energy and chemical potential can be computed as
\beas
&&E(\phi(\cdot,t_n))\approx E^n=\frac{b-a}{4}\sum_{l=1}^{M-1}\mu_l^2 |\widetilde{(\phi^n)}_l|^2+ h\sum_{j=0}^{M-1}\left[V(x_j)|\phi_j^n|^2
+\frac{\beta}{2}|\phi_j^n|^4\right],\\
&&\mu(\phi(\cdot,t_n))\approx \mu^n=E^n+h\sum_{j=0}^{M-1}\frac{\beta}{2}|\phi_j^n|^4,\qquad n\ge0.
\eeas

\begin{remark}In practice, Fourier pseudospectral method or cosine pseudospectral method can also be applied to spatial discretization  for discretizing (\ref{eq:sdge1d:sec3})-(\ref{eq:sdge1d3:sec3}) when the homogeneous
Dirichlet boundary condition in (\ref{eq:sdge1d3:sec3}) is replaced by periodic boundary condition
or homogeneous Neumann boundary condition, respectively.
\end{remark}

\subsection{Simplified methods under symmetric potentials}
\label{subsec:sympotgs}
The ground state $\phi_g$ of (\ref{eq:minp:sec2}) shares the same symmetric properties with $V(\bx)$ ($\bx\in\Bbb R^d$) ($d=1,2,3$). In such cases, simplified numerical  methods, especially with less memory requirement,
 for computing the ground states are available.

{\it Radial symmetry in 1D, 2D and 3D}. When  the potential $V(\bx)$ is radially symmetric in $d=1,2$ and spherically symmetric in $d=3$, the problem is reduced to 1D.  Due to the symmetry, the GPE (\ref{eq:gpe:sec2}) essentially
collapses to a 1D problem with $r=|\bx|\in[0,+\infty)$ for $\psi:=\psi(r,t)$ ($d=1, 2, 3$):
 \begin{align}
\label{eq:gperadial:sec3}
&i\p_t\psi(r,t) = \frac{-1}{2r^{d-1}}\frac{\p}{\p r}\left(r^{d-1}\frac{\p}{\p r}\psi\right)+
\left(V(r)+
\bt|\psi|^2\right)\psi,\; r\in(0,+\infty),\\
&\frac{\p\psi(0,t)}{\p r}=0,\qquad \qquad \psi(r,t)\to 0,\quad {\rm as}\quad r\to\infty.
\end{align}
The normalization condition (\ref{eq:norm:sec2}) becomes
\be\label{eq:normradial:sec3}
N_r(\psi)=\omega(d)\int_{0}^\infty|\psi(r,t)|^2r^{d-1}\,dr=1.
\ee
Here $\omega(d)$ is the area  of unit sphere in $d$ dimensions, where $\omega(1)=2$, $\omega(2)=2\pi$ and $\omega(3)=4\pi$. The energy (\ref{eq:energy:sec2})  can be rewritten for radial wave function as
\be\label{eq:energyradial:sec3}
E_r(\psi)=\omega(d)\int_0^\infty \left(\frac12|\partial_r\psi(r,t)|^2+V(r)|\psi(r,t)|^2+\frac{\beta}{2}|\psi(r,t)|^4\right)r^{d-1}\,dr.
\ee
Then, the minimization problem (\ref{eq:minp:sec2}) collapses to :

Find $ \vphi_g\in S_r$ such that
\be
\label{eq:minpradial:sec3}
E_g:=E_r(\vphi_g) = \min_{\vphi\in S_r} E_r(\vphi),
\ee
where $S_r=\{\vphi \ |\omega(d)\int_{0}^\infty|\varphi(r)|^2r^{d-1}\,dr=1 , \ E_r(\vphi)<\ift\}$.

 The nonlinear eigenvalue problem (\ref{eq:charactereq:sec2}) collapses to
\be
\label{eq:charactereqradial:sec3}
\mu\vphi(r) = -\frac{1}{2r^{d-1}}\frac{\rd}{\rd r}\left(r^{d-1}\frac{\rd}{\rd r}\vphi(r)\right)+V(r)\vphi(r)
+\bt|\vphi(r)|^2\vphi(r),\; r\in(0,+\infty),
\ee
with boundary conditions
\be\label{eq:boundaryradial:sec3}
\varphi^\prime(0)=0,\qquad \vphi(r)\to0,\quad{\rm when}\quad r\to\infty,
\ee
under the normalization constraint (\ref{eq:normradial:sec3}) with $\psi=\varphi$.

The eigenvalue problem (\ref{eq:charactereqradial:sec3})-(\ref{eq:boundaryradial:sec3}) is defined in a semi-infinite  interval $(0,+\infty)$. In practical computation, this is approximated by a problem defined on a finite
interval. Since the full wave function  vanishes exponentially fast as $r\to\infty$, choosing $R>0$ sufficiently
large, then the eigenvalue problem (\ref{eq:charactereqradial:sec3})-(\ref{eq:boundaryradial:sec3}) can be approximated by
\be
\label{eq:charactereqradial2:sec3}
\mu\;\vphi(r) =  -\frac{1}{2r^{d-1}}\frac{\rd}{\rd r}\left(r^{d-1}\frac{\rd}{\rd r}\vphi(r)\right)+[V(r)
+\bt|\vphi|^2]\vphi(r),\quad 0<r<R,
\ee
with boundary conditions
\be\label{eq:boundaryradial2:sec3}
\varphi^\prime(0)=0,\qquad \vphi(R)=0,
\ee
under the normalization
\be\label{eq:normradial2:sec3}
\omega(d)\int_{0}^R|\varphi(r)|^2r^{d-1}\,dr=1.
\ee
 To compute the ground state $\varphi_g$, a method based on GFDN  (\ref{eq:ngf1:sec3})-(\ref{eq:ngf3:sec3}) can be simplified. Here, we only present full discretization using a simplified BEFD method.

Choose time steps as (\ref{eq:mesh1d:sec3}), mesh size $\Delta r=2R/(2M+1)$ with positive integer $M$ and grid points as
\be\label{eq:meshradial:sec3}
r_j=j\Delta r,\quad r_{j+\frac12}=\left(j+\frac12\right)\Delta r,\quad j\in\calT_M^0.
 \ee
We adopt the same notation for finite difference operator as (\ref{eq:fdnotation:sec3}). Let $\vphi_{j+\frac12}^n$ be the numerical approximation of $\varphi(r_{j+\frac12},t_n)$ and $\vphi^{n}$ be the solution vector at time $t=t_n$ with
components $\vphi_{j+\frac12}^{n}$. Then a simplified BEFD method for computing the ground state of (\ref{eq:minpradial:sec3}) by GFDN with an initial guess $\vphi_0(r)$ is given as \cite{BaoDu}
\begin{align}\label{eq:simbefd31:sec3}
&\fl{\vphi_{j+\frac12}^{(1)}-\vphi_{j+\frac12}^n}{\tau}=\left[
\frac{1}{2}\delta_{r,d}^2-V(r_{j+\frac12})-
\bt \left(\vphi_{j+\frac12}^n\right)^2\right]
\vphi_{j+\frac12}^{(1)}, \quad j\in\calT_M\cup\{0\}, \\
&\vphi_{-\frac12}^{(1)}=\vphi_{\frac12}^{(1)},\qquad \vphi_{M+\frac12}^{(1)}=0,\quad
\vphi_{j+\frac12}^0= \vphi_0(r_{j+\frac12}),\quad \quad j\in \calT_M^0, \nn\\ \label{eq:simbefd33:sec3}
&\quad \vphi_{j+\frac12}^{n+1}=\fl{\vphi_{j+\frac12}^{(1)}}{\|\vphi^{(1)}\|_{r}}, \quad j\in \calT_M^0,
\quad n=0,1,\cdots,
\end{align}
where
\[\delta_{r,d}^2\;\vphi_{j+\frac{1}{2}}^{(1)}=\frac{1}{(\Delta r)^2r_{j+\frac12}^{d-1}}\left[r_{j+1}^{d-1}\vphi^{(1)}_{j+\frac32}-(r_{j+1}^{d-1}+r_j^{d-1})
\vphi^{(1)}_{j+\frac12}+r_j^{d-1}\vphi^{(1)}_{j-\frac12}\right],\]
and the norm is defined as
\be\label{eq:fdnormradial:sec3}
\|\vphi^{(1)}\|_{r}^2=\omega(d)\Delta r\sum\limits_{j=0}^{M-1}|\vphi^{(1)}_{j+\frac12}|^2(r_{j+\frac12})^{d-1}.
\ee
Here, we have introduced a ghost point $r_{-\frac12}$ so that the Neumann boundary condition $\vphi^\prime(0)=0$ is approximated  with second order accuracy. The linear system (\ref{eq:simbefd31:sec3}) can be solved very efficiently by  the Thomas' algorithm, where the computational cost is $O(M)$ per time step, for all dimensions $d=1,2,3$. The memory cost is $O(M)$. This tremendously reduces memory and computation complexity in higher dimensions ($d=2,3$) from
$O(M^d)$ to $O(M)$ compared with the proposed BEFD (\ref{eq:befd3:sec3}) with Cartesian coordinates.

{\it Cylindrical symmetry in 3D}. For $\bx=(x,y,z)^T\in\Bbb R^3$, when $V$ is cylindrically symmetric, i.e., $V$ is of the form $V(r,z)$ ($r=\sqrt{x^2+y^2}$), the problem is reduced to 2D. Due to the symmetry, the GPE (\ref{eq:gpe:sec2}) essentially
collapses to a 2D problem with $r\in(0,+\infty)$ and $z\in\Bbb R$ for $\psi:=\psi(r,z,t)$ :
\begin{align}
\label{eq:gpecylin:sec3}
&i\p_t\psi(r,z,t) =-\frac12 \left[\frac{1}{r}\frac{\p}{\p r}\left(r\frac{\p\psi}{\p r}\right)+\frac{\p^2\psi}{\p z^2}\right]+
\left(V(r,z)+
\bt|\psi|^2\right)\psi,\\
&\frac{\p\psi(0,z,t)}{\p r}=0,\quad z\in\Bbb R,\quad
\psi(r,z,t)\to 0, \quad{\rm when}\quad r+|z|\to\infty.
\end{align}
The normalization condition (\ref{eq:norm:sec2}) becomes
\be\label{eq:normcylin:sec3}
N_c(\psi)=2\pi\int_{\Bbb R^+}\int_{\Bbb R}|\psi(r,z,t)|^2r\,dzdr=1.
\ee

Then, the minimization problem (\ref{eq:minp:sec2}) collapses to :

Find $ \vphi_g\in S_c$ such that
\be
\label{eq:minpcylin:sec3}
E_g:=E_c(\vphi_g) = \min_{\vphi\in S_c} E_c(\vphi),
\ee
where
\be\label{eq:energycylin:sec3}
E_c(\vphi)=\pi\int_{\Bbb R^+}\int_{\Bbb R} \left(|\vphi_r(r,z)|^2+|\vphi_z(r,z)|^2+2V(r,z)|\vphi|^2+\beta|\vphi|^4\right)r\,dzdr,
\ee
and $S_c=\{\vphi \ |2\pi\int_{\Bbb R^+}\int_{\Bbb R}|\varphi(r,z)|^2r\,dzdr=1, \ E_c(\vphi)<\ift\}$.

 The nonlinear eigenvalue problem (\ref{eq:charactereq:sec2}) collapses to
\be
\label{eq:charactereqcylin:sec3}
\mu\vphi(r,z) = -\frac12
\left[\frac{1}{r}\frac{\p}{\p r}\left(r\frac{\p\vphi}{\p r}\right)+\frac{\p^2\vphi}{\p z^2}\right]+\left(V(r,z)
+\bt|\vphi|^2\right)\vphi,\ \ r>0, \ {z\in \Bbb R},
\ee
with boundary conditions
\be\label{eq:boundarycylin:sec3}\begin{split}
&\frac{\p\vphi(0,z,t)}{\p r}=0,\quad z\in\Bbb R,\quad
\vphi(r,z,t)\to 0, \quad{\rm when}\quad r+|z|\to\infty,
\end{split}\ee
under the normalization constraint (\ref{eq:normcylin:sec3}) with $\psi=\varphi$.

The eigenvalue problem (\ref{eq:charactereqradial:sec3})-(\ref{eq:boundaryradial:sec3}) is defined in the $r$-$z$ plane. In practical computation, this is approximated by a problem defined on a bounded
domain. Since the full wave function  vanishes exponentially fast as $r+|z|\to\infty$, choosing $R>0$ and $Z_1<Z_2$ with $|Z_1|$, $|Z_2|$ and $R$ sufficiently
large, then the eigenvalue problem (\ref{eq:charactereqcylin:sec3})-(\ref{eq:boundarycylin:sec3}) can be approximated for $(r,z)\in(0,R)\times(Z_1,Z_2)$,
\be
\label{eq:charactereqcylin2:sec3}
\mu\;\vphi(r,z) =  -\frac12
\left[\frac{1}{r}\frac{\p}{\p r}\left(r\frac{\p\vphi}{\p r}\right)+\frac{\p^2\vphi}{\p z^2}\right]+\left(V(r,z)
+\bt|\vphi|^2\right)\vphi,
\ee
with boundary conditions
\be\label{eq:boundarycylin2:sec3}
\frac{\p\vphi(0,z)}{\p r}=0,\; \vphi(R,z)=\vphi(r,Z_1)=\vphi(r,Z_2)=0,
\quad z\in[Z_1,Z_2],\; r\in[0,R],
\ee
under the normalization
\be\label{eq:normcylinl2:sec3}
2\pi\int_{0}^R\int_{Z_1}^{Z_2}|\varphi(r,z)|^2r\,dzdr=1.
\ee
 To compute the ground state, the GFDN  (\ref{eq:ngf1:sec3})-(\ref{eq:ngf3:sec3}) collapses to a 2D problem. We present a full finite difference discretization. Choose time steps as (\ref{eq:mesh1d:sec3}) and $r$- grid points (\ref{eq:meshradial:sec3}) for positive integer $M>0$. For integer $N>0$, choose mesh size $\Delta z=(Z_2-Z_1)/N$ and define $z$- grid points $z_k=Z_1+k\Delta z$ for  $k\in\calT_N^0=\{k\ |\ k=0,1,\ldots,N\}$.

Let $\vphi_{j+\frac12\,k}^n$ be the numerical approximation of $\varphi(r_{j+\frac12},z_{k},t_n)$ and $\vphi^{n}$ be the solution vector at time $t=t_n$ with
components $\vphi_{j+\frac12\,k}^{n}$. Then a simplified BEFD method for computing the ground state of (\ref{eq:energycylin:sec3}) by GFDN with an initial guess $\vphi_0(r,z)$ is given below \cite{BaoDu}:
\begin{align} 
&\fl{\vphi_{j+\frac12\,k}^{(1)}-\vphi_{j+\frac12k}^n}{\tau}=\left[
\frac{1}{2}(\delta_r^2+\delta_z^2)-V(r_{j+\frac12},z_{k})-\bt
\left(\vphi_{j+\frac12k}^n\right)^2\right]\vphi_{j+\frac12k}^{(1)},\ (j,k)\in\calT_{MN}^*, \nn\\
&\vphi_{-\frac12\,k}^{(1)}=\vphi_{\frac12\,k}^{(1)},
\quad\vphi_{M+\frac12\,k}^{(1)}=\vphi_{j+\frac12\,0}^{(1)}=\vphi_{j+\frac12\,N}^{(1)}=0,\qquad (j,k)\in\calT_{MN}^0, \nn \\ \label{eq:cysimbefd33:sec3}
& \vphi_{j+\frac12\,k}^0= \vphi_0(r_{j+\frac12},z_{k}),\quad\vphi_{j+\frac12\,k}^{n+1}
=\fl{\vphi_{j+\frac12\,k}^{(1)}}{\|\vphi^{(1)}\|_{c}}, \quad (j,k)\in \calT_{MN}^0,
\, n\ge0,
\end{align}
 where $\calT_{MN}^*=\{(j,k)\ |\ 0\le j\le M-1, \ 1\le k\le N-1\}$, $\calT_{MN}^0=\{(j,k)\ |\ 0\le j\le M, \ 0\le k\le N\}$ and
 \begin{eqnarray*}
&&\delta_r^2\vphi_{j+\frac12\,k}^{(1)}= \frac{1}{(\Delta r)^2r_{j+\frac{1}{2}}}\left[r_{j+1}\vphi_{j+\frac32\,k}^{(1)}-(r_{j+1}+
r_{j})\vphi_{j+\frac12\,k}^{(1)}+r_{j}\vphi_{j-\frac12\,k}^{(1)}\right],\\
&&\delta_z^2\vphi_{j+\frac12k}^{(1)}=\frac{1}{(\Delta z)^2}[\vphi_{j+\frac12\,k+1}^{(1)}-2\vphi_{j+\frac12\,k}^{(1)}+\vphi_{j+\frac12\,k-1}^{(1)}], \qquad (j,k)\in \calT_{MN}^*,
\end{eqnarray*}
and the norm is defined by
\be\label{eq:fdnormcylin:sec3}
\|\vphi^{(1)}\|_{c}^2=2\pi\Delta r\,\Delta z\sum\limits_{j=0}^{M-1}\sum\limits_{k=0}^{N-1}|\vphi^{(1)}_{j+\frac12\,k}|^2
r_{j+\frac12}.
\ee
Here, we use ghost points to approximate the Neumann boundary conditions, which is the same as  the radially symmetric potential case.

\begin{remark}\label{rmk:excite} When the potential $V(\bx)$ is an even function, BEFD (\ref{eq:befd3:sec3}) and BESP (\ref{eq:besp1:sec3})-(\ref{eq:besp2:sec3}) can be used to compute the first excited states by choosing proper initial guess (see \cite{BaoDu,BaoLim}).
\end{remark}

\subsection{Numerical results}
In this section, we report numerical results on the ground state by the proposed BEFD and BESP methods.
\begin{example}\label{exm:1:sec3}
Ground and first excited states (Remark \ref{rmk:excite}) in 1D, i.e., we take
$d=1$ in (\ref{eq:gpe:sec2}) and study
two kinds of trapping potentials

Case I. A harmonic oscillator potential $V(x) =\frac{x^2}{2}$ and
$\beta=400$;

Case II. An optical lattice potential $V(x)= \frac{x^2}{2}+25
\sin^2\left(\frac{\pi x}{4}\right)$ and $\beta=250$.

The initial data (\ref{eq:ngf3:sec3}) is chosen
 as $\phi_0(x) = e^{-x^2/2}/\pi^{1/4}$ for computing the ground state,
and resp., $\phi_0(x) = \frac{\sqrt{2} x}{\pi^{1/4}} e^{-x^2/2}$
for computing the first excited state. We solve the problem with BESP (\ref{eq:besp1:sec3})-(\ref{eq:besp2:sec3})
on $[-16,16]$, i.e. $a=-16$ and $b=16$, and take time step
$\tau =0.05$ for computing the ground state, and resp., $\tau=0.001$
for computing the first excited state.
The steady state solution in our computation is reached when $\max_{1 \le j \le M-1}\;
|\phi_j^{n+1}-\phi_j^n|<10^{-12}$. Let $\phi_g$ and $\phi_1$
be the `exact' ground state and first excited state, respectively,
which are obtained numerically by
using BESP with a very fine mesh $h=\frac{1}{32}$
and $h=\frac{1}{128}$, respectively. We denote their
energy and chemical potential as
$E_g:=E(\phi_g)$, $E_1:=E(\phi_1)$, and
$\mu_g:=\mu(\phi_g)$, $\mu_1:=\mu(\phi_1)$.
Let $\phi_{g,h}^{\rm SP}$ and $\phi_{1,h}^{\rm SP}$
be the numerical ground state and first excited state
obtained by using BESP with mesh size $h$, respectively.
Similarly, $\phi_{g,h}^{\rm FD}$ and $\phi_{1,h}^{\rm FD}$
are obtained by using BEFD in a similar way.
Tabs.~\ref{tbl:1a:sec3} and \ref{tbl:1b:sec3} list the errors for Case I, and Tabs.~\ref{tbl:2a:sec3} and
\ref{tbl:2b:sec3}
show the errors for Case II.  Furthermore, we  compute the energy and chemical potential
for the ground state and first excited state based on our
`exact' solution $\phi_g$ and $\phi_1$.
For Case I, we have $E_g:=E(\phi_g)= 21.3601$ and
$\mu_g:=\mu(\phi_g) = 35.5775$ for ground state,
and $E_1:=E(\phi_1)= 22.0777$ and
$\mu_1:=\mu(\phi_1) = 36.2881$ for the first excited
 state. Similarly, for Case II,  we have
$E_g=26.0838$, $\mu_g= 38.0692$, $E_1=27.3408$ and
$\mu_1=38.9195$. Fig.~\ref{fig:1:sec3} plots $\phi_g$ and $\phi_1$ as well
as their corresponding trapping potentials
for Cases I\&II. Fig.~3.2a shows the excited states $\phi_1$ for potential in Case I with
different $\beta$.
\end{example}
\begin{table}[htbp]
\begin{center}
\begin{tabular}{ccccc}  \hline
 mesh size  &$h=1$  &$h=1/2$  &$h=1/4$   &$h=1/8$ \\  \hline
$\max|\phi_g - \phi_{g,h}^{\rm SP}|$  &1.310E-3  &7.037E-5
    &1.954E-8  &$<$E-12 \\
$\|\phi_g - \phi_{g,h}^{\rm SP}\|$  &1.975E-3  &7.425E-5
    &2.325E-8   &$<$E-12 \\
$|E_g - E(\phi_{g,h}^{\rm SP})|$  &5.688E-5  &2.642E-6
    &9E-12   &$<$E-12 \\
$|\mu_g - \mu(\phi_{g,h}^{\rm SP})|$  &1.661E-2 &8.705E-5
   &9.44E-10  &4E-12 \\    \hline
$\max|\phi_g - \phi_{g,h}^{\rm FD}|$   &2.063E-3  &1.241E-3
     &2.890E-4  &7.542E-5  \\
$\|\phi_g - \phi_{g,h}^{\rm FD}\|$  &3.825E-3  &1.439E-3
  &3.130E-4  &7.705E-5 \\
$|E_g - E(\phi_{g,h}^{\rm FD})|$  &2.726E-3  &9.650E-4
    &2.540E-4  &6.439E-5 \\
$|\mu_g - \mu(\phi_{g,h}^{\rm FD})|$  &2.395E-2  &6.040E-4
&2.240E-4  &5.694E-5 \\    \hline
\end{tabular}
\caption{Spatial resolution of BESP and BEFD for ground state of
Case I in Example \ref{exm:1:sec3}.}\label{tbl:1a:sec3}
\end{center}
\end{table}

\begin{table}[htbp]
\begin{center}
\begin{tabular}{ccccc}  \hline
 Mesh size  &$h=1/4$   &$h=1/8$  &$h=1/16$  &$h=1/32$ \\   \hline
$\max|\phi_1 - \phi_{1,h}^{\rm SP}|$ &2.064E-1  &6.190E-4	
     &2.099E-7  &$<$E-12 \\
$\|\phi_1 - \phi_{1,h}^{\rm SP}\|$  &1.093E-1 &3.200E-4
    &1.403E-7  &$<$E-12 \\
$|E_1 - E(\phi_{1,h}^{\rm SP})|$  &5.259E-2  &3.510E-4
   &5.550E-9   &$<$E-12 \\
$|\mu_1 - \mu(\phi_{1,h}^{\rm SP})|$  &1.216E-1	&1.509E-3
    &4.762E-8  &$<$E-12 \\     \hline
$\max|\phi_1 - \phi_{1,h}^{\rm FD}|$  &2.348E-1  &8.432E-3
    &2.267E-3  &6.040E-4 \\
$\|\phi_1 - \phi_{1,h}^{\rm FD}\|$  &1.197E-1  &4.298E-3
    &1.215E-3  &2.950E-4 \\
$|E_1 - E(\phi_{1,h}^{\rm FD})|$  &3.154E-1  &5.212E-2  &1.382E-2
    &3.449E-3 \\
$|\mu_1 - \mu(\phi_{1,h}^{\rm FD})|$  &4.216E-1  &5.884E-2
    &1.609E-2  &3.999E-3 \\   \hline
  \end{tabular}
 \caption{Spatial resolution of BESP and BEFD for the first excited
 state of Case I in Example \ref{exm:1:sec3}.} \label{tbl:1b:sec3}
\end{center}
\end{table}

\begin{table}[htbp]
\begin{center}
\begin{tabular}{ccccc}    \hline
Mesh size   &$h=1$  &$h=1/2$ &$h=1/4$  &$h=1/8$ \\     \hline
$\max|\phi_g -\phi_{g,h}^{\rm SP}|$  &7.982E-3 &1.212E-3
    &2.219E-6    &1.9E-11 \\
$\|\phi_g - \phi_{g,h}^{\rm SP}\|$  &1.304E-2  &1.313E-3
   &2.431E-6    &2.8E-11 \\
$|E_g - E(\phi_{g,h}^{\rm SP})|$  &4.222E-4  &1.957E-4
    &4.994E-8    &$<$E-12 \\
$|\mu_g - \mu(\phi_{g,h}^{\rm SP})|$  &9.761E-2  &4.114E-3
    &5.605E-7   &$<$E-12 \\     \hline
$\max|\phi_g - \phi_{g,h}^{\rm FD}|$  &1.019E-2  &5.815E-3
  &1.001E-3  &2.541E-4 \\
$\|\phi_g - \phi_{g,h}^{\rm FD}\|$  &1.967E-2  &7.051E-3
   &1.390E-3  &3.387E-4 \\
$|E_g - E(\phi_{g,h}^{\rm FD})|$  &7.852E-2  &2.961E-2
    &7.940E-3  &2.027E-3 \\
$|\mu_g - \mu(\phi_{g,h}^{\rm FD})|$   &1.786E-1  &1.716E-2
  &6.730E-3  &1.728E-3 \\     \hline
  \end{tabular}
  \caption{Spatial resolution of BESP and BEFD for ground state
of Case II in Example \ref{exm:1:sec3}.}\label{tbl:2a:sec3}
\end{center}
\end{table}

\begin{table}[htbp]
\begin{center}
\begin{tabular}{ccccc}    \hline
 Mesh size  &$h=1/4$  &$h=1/8$  &$h=1/16$  &$h=1/32$ \\   \hline
$\max|\phi_1 - \phi_{1,h}^{\rm SP}|$  &2.793E-1 &1.010E-3
    &4.240E-7   &2E-12 \\
$\|\phi_1 - \phi_{1,h}^{\rm SP}\|$  &1.477E-1  &5.241E-4
   &2.784E-7  &2E-12 \\
$|E_1 - E(\phi_{1,h}^{\rm SP})|$  &1.145E-1  &8.337E-4
   &1.943E-8  &$<$E-12  \\
$|\mu_1 - \mu(\phi_{1,h}^{\rm SP})|$   &1.593E-1  &2.357E-3
   &1.097E-7   &5E-12 \\      \hline
$\max|\phi_1 - \phi_{1,h}^{\rm FD}|$   &3.134E-1  &1.124E-2
   &3.231E-3  &8.450E-4 \\
$\|\phi_1 - \phi_{1,h}^{\rm FD}\|$  &1.599E-1  &5.779E-3
  &1.701E-3  &4.122E-4 \\
$|E_1 - E(\phi_{1,h}^{\rm FD})|$  &6.011E-1  &1.002E-1  &2.688E-2
   &6.707E-3 \\
$|\mu_1 - \mu(\phi_{1,h}^{\rm FD})|$  &6.315E-1  &9.887E-2
   &2.742E-2  &6.827E-3 \\     \hline
  \end{tabular}
  \caption{Spatial resolution of BESP and BEFD for the first excited
 state of Case II in Example \ref{exm:1:sec3}.} \label{tbl:2b:sec3}
\end{center}
\end{table}

\begin{figure}[htbp]
\centerline{a)\psfig{figure=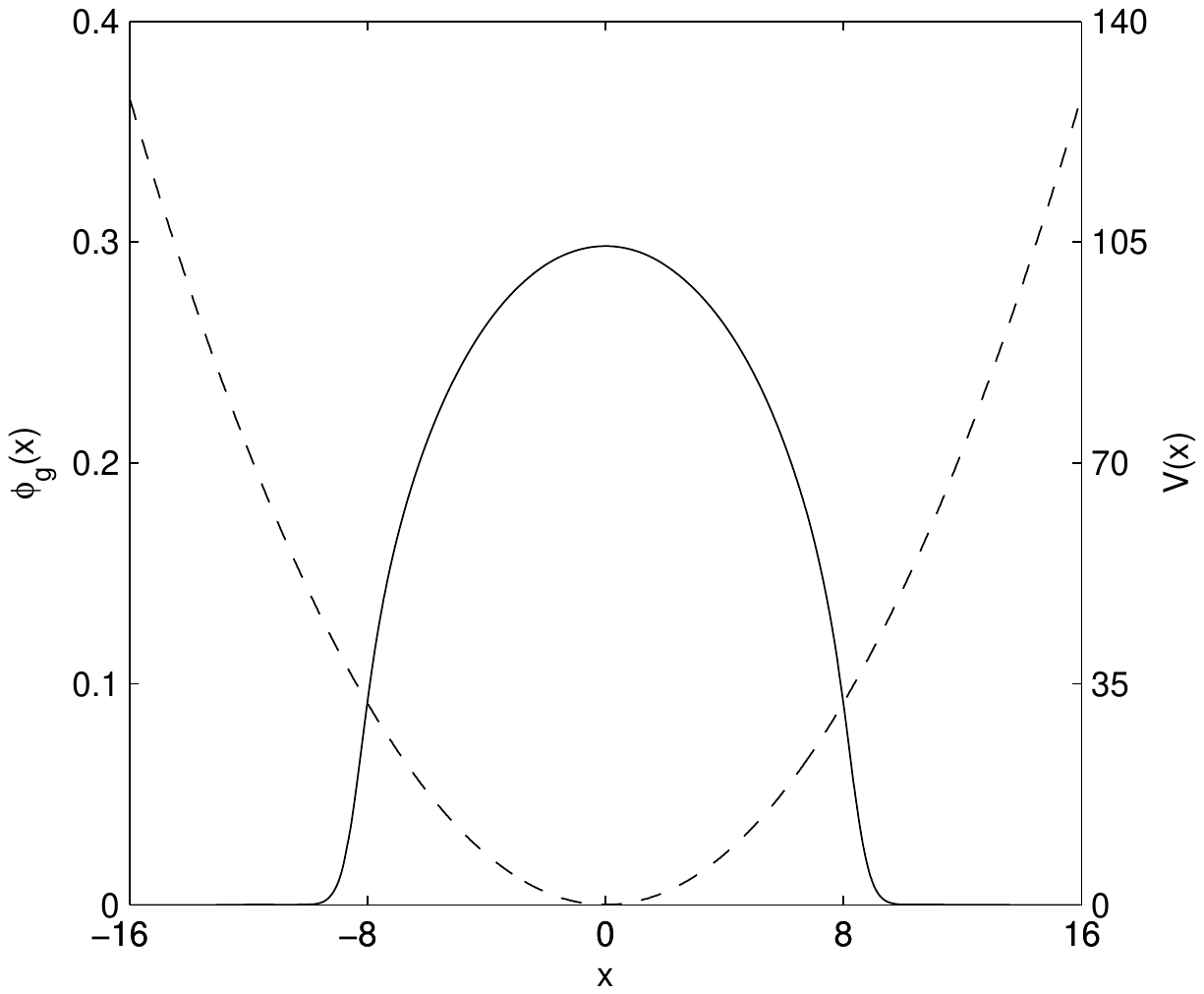,height=4.5cm,width=5cm,angle=0}
\quad \psfig{figure=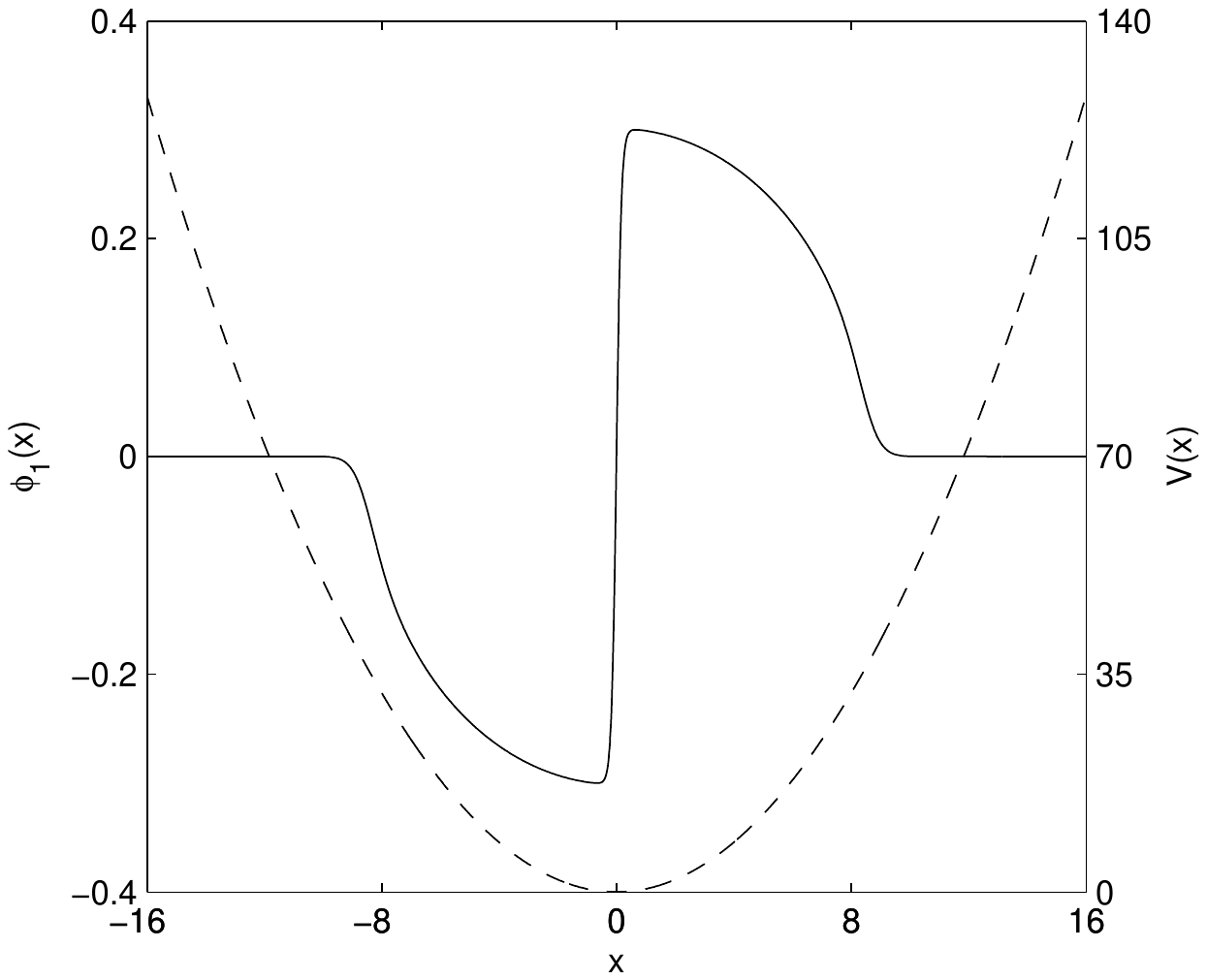,height=4.5cm,width=5cm,angle=0} }
\centerline{b)\psfig{figure=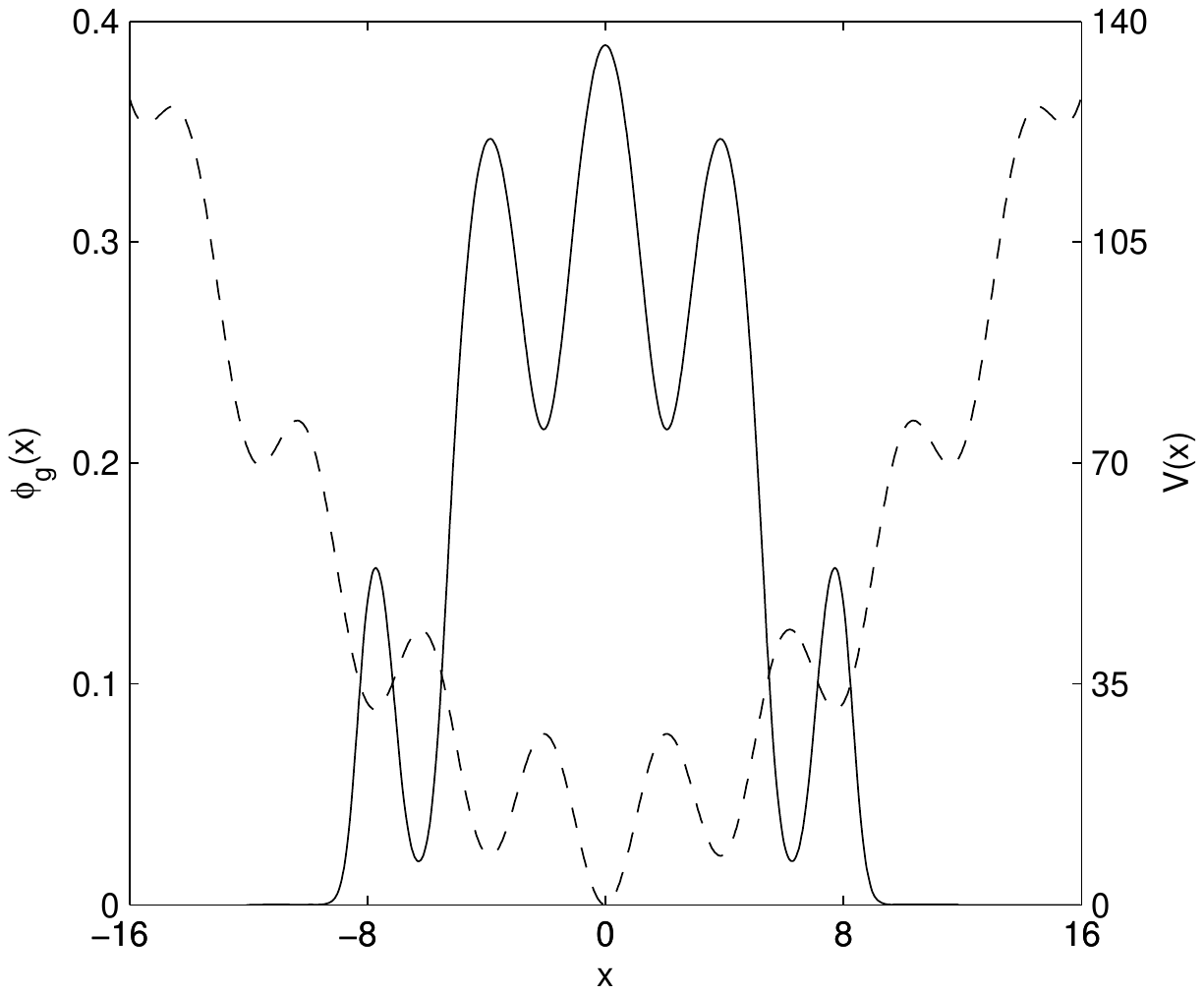,height=4.5cm,width=5cm,angle=0}
\quad \psfig{figure=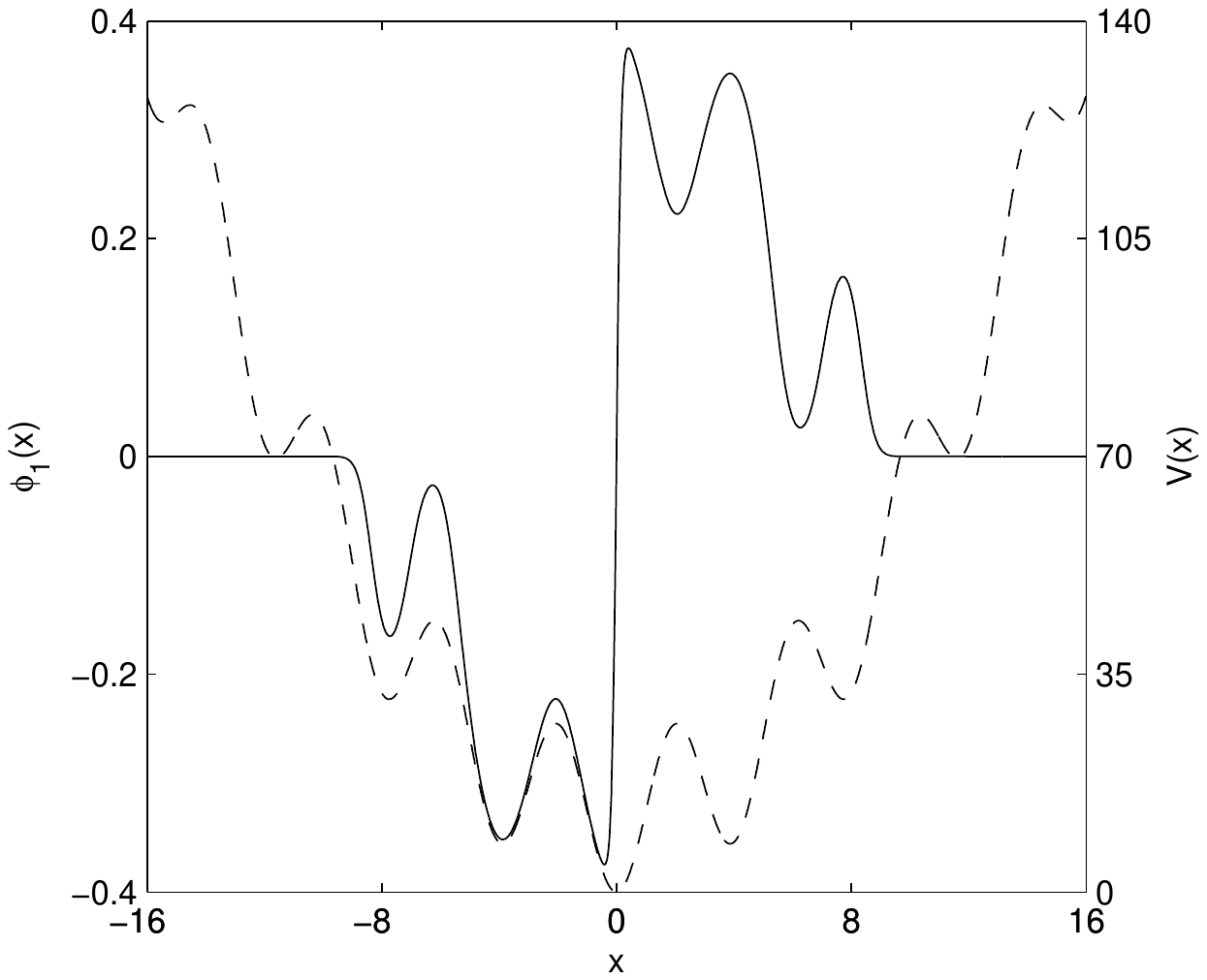,height=4.5cm,width=5cm,angle=0} }
\caption{Ground state $\phi_g$ (left column, solid lines)
and first excited state $\phi_1$ (right column, solid lines)
as well as trapping potentials (dashed lines) in Example \ref{exm:1:sec3}.
a): For Case I; b): For Case II.}\label{fig:1:sec3}
\end{figure}

From Tabs.~\ref{tbl:1a:sec3}-\ref{tbl:2b:sec3}, Figs.~\ref{fig:1:sec3} and 3.2a, we can draw the following conclusions. For BESP, it is spectrally accurate in spatial discretization; where for BEFD,
it is only second-order accurate. The error in the ground and first excited
states is only due to the spatial discretization.
Thus when high accuracy is required
or the solution has multiscale structure \cite{BaoLimZhang}, BESP is
 much better than BEFD in terms that it needs much less grid points.
Therefore BESP can save a lot of memory and computational time,
especially  in 2D \& 3D.

\begin{example}\label{exm:2:sec3}
Ground states in 2D with radial symmetric trap, i.e. we take $d=2$ in (\ref{eq:gpe:sec2}) and
\be V(x,y)=V(r)=\fl{1}{2}r^2, \qquad (x,y)\in {\Bbb R}^2,
\qquad  r=\sqrt{x^2+y^2}\ge0.\ee
The GFDN (\ref{eq:ngf1:sec3})-(\ref{eq:ngf3:sec3}) is solved in polar coordinate
with  $R=8+1/128$ under mesh size $\Delta r=\fl{1}{64}$ and
time step $\tau=0.1$ by using the simplified BEFD (\ref{eq:simbefd31:sec3})-(\ref{eq:simbefd33:sec3})
with initial data $\phi_0(x,y)=\phi_0(r)=\fl{1}{\sqrt{\pi}}\,e^{-r^2/2}$.
Fig.~3.2b
shows the ground state solution $\phi_g(r)$ with different
$\bt$.
Tab.~\ref{tab:tab2t:sec3} displays the values of
 $\phi_g(0)$, radius mean square $r_{\rm rms}=\sqrt{2\pi\int_0^\infty r^2|\phi_g(r)|^2rdr}$, energy $E(\phi_g)$ and
chemical potential $\mu_g$ for different $\bt$.
\end{example}

\begin{table}[htbp]\label{tab:tab2t:sec3}
\begin{center}
\begin{tabular}{ccccc}\hline
 $\beta$  &$\phi_g(0)$ &$r_{\rm rms}$  &$E(\phi_g)$ &$\mu_g=\mu(\phi_g)$
   \\ \hline
  0      &0.5642   &1.0000    &1.0000     &1.0000 \\
  10      &0.4104   &1.2619     &1.5923    &2.0637 \\
 50      &0.2832     &1.7018    &2.8960     &4.1430 \\
100    &0.2381    &1.9864    &3.9459    &5.7597 \\
250     &0.1892     &2.4655     &6.0789    &9.0031 \\
500      &0.1590    &2.9175     &8.5118    &12.6783\\
    \hline
\end{tabular}
\caption{Numerical results for radial symmetric ground  states in Example \ref{exm:2:sec3}.}
\end{center}
\end{table}

\begin{figure}[htbp] \label{fig:2:sec3}
\centerline{(a).\psfig{figure=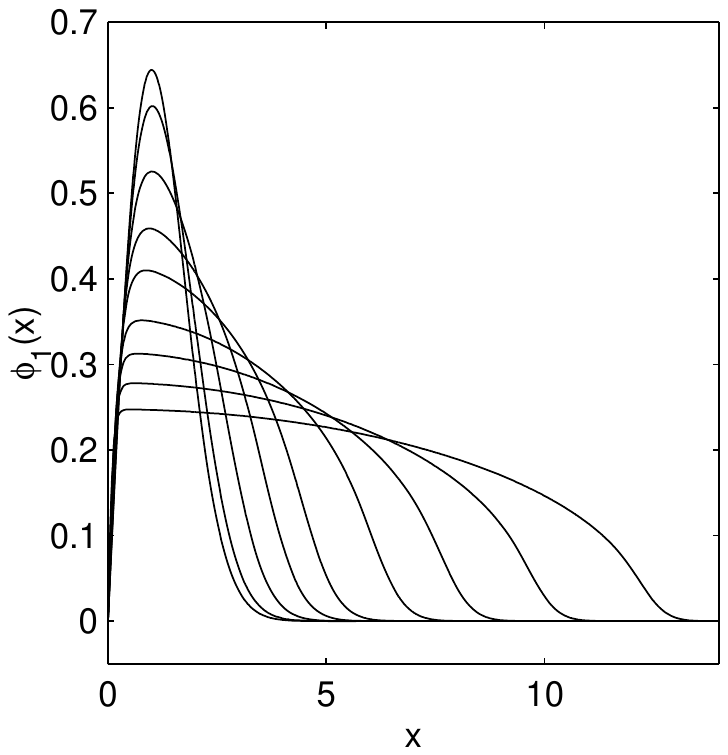,height=4.5cm,width=5cm,angle=0} \quad
(b).\psfig{figure=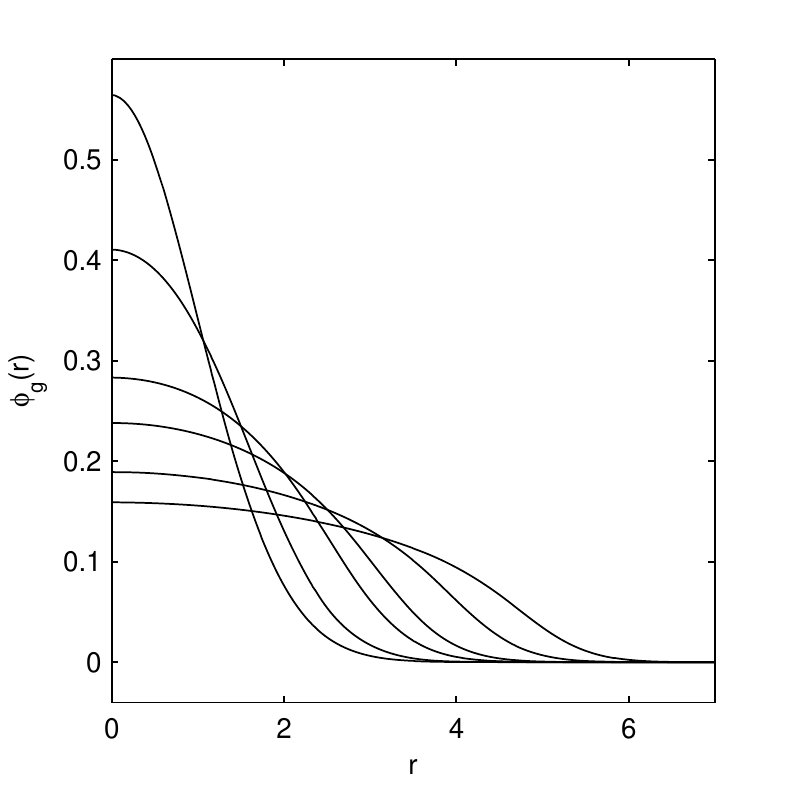,height=4.5cm,width=5cm,angle=0}
\qquad }
\caption{(a). First excited state solution $\phi_1(x)$ (an odd function)
in Example \ref{exm:1:sec3} with potential in Case I for different
$\bt=0,\; 3.1371,\; 12.5484,\; 31.371,\; 62.742,\;
156.855,\; 313.71,\; 627.42,\; 1254.8$ (with decreasing peak);
and (b). 2D ground states $\phi_g(r)$
in Example \ref{exm:2:sec3} for  $\bt=0, 10, 50,100,250,500$ (with decreasing peak).}
\end{figure}

\subsection{Comments of different methods}
In literatures, different numerical methods have been introduced to compute ground states of BEC.
In \cite{Ruprecht}, Ruprecht et al. used Crank-Nicolson finite difference method to compute the ground states of BEC based on the Euler-Lagrange equation (\ref{eq:charactereq:sec2}). Later, Edwards et al. proposed a Runge-Kutta method  to find the ground states in 1D and 3D with spherical symmetry. Dodd \cite{Dodd} gave an analytical expansion of the energy
$E(\phi)$ using the Hermite polynomials
 when the trap $V$ is harmonic. By minimizing the energy in terms of
 the expansion, approximated
 ground state results were reported in \cite{Dodd}. In \cite{Tosi},
 Succi et al. used an imaginary time
 method (equivalent to  GFDN) to compute the ground states with centered
 finite-difference discretization in space and explicit
 forward discretization in time. Lin et al. designed an iterative method
 in \cite{Chang1}. After discretization
 in space, they transformed the problem to a minimization problem in finite
 dimensional vector space.
  Gauss-Seidel iteration methods were proposed to solve the corresponding problem. Bao and Tang proposed a
  finite element method to find the ground state by directly minimizing the energy functional in \cite{BaoTang}.
 In \cite{BaoDu,BaoLim}, Bao et al. developed the GFDN method to calculate the ground state, which contains a gradient flow and a projection at each step.
Different discretizations have been discussed, including the finite difference discretization or spectral
discretization in space, explicit (forward Euler) discretization or implicit (backward Euler, Crank-Nicolson)
discretization in time.

 In the current studies of BEC, the most popular method for computing
the ground state  is the GFDN method. In fact, imaginary time method \cite{Tosi} is the same as the GFDN method,
while imaginary time is preferred in the physics community. There are many different discretizations for GFDN (\ref{eq:ngf1:sec3})-(\ref{eq:ngf3:sec3})  \cite{BaoDu,BaoLim}, including backward Euler, Crank-Nicolson, forward Euler, backward Euler for linear part and forward Euler for nonlinear part, time-splitting for the time discretization and centered finite difference or spectral method for spatial discretization. From our  experience, among these discretizations,  BEFD and BESP  are very easy to use, robust, very efficient and accurate in practical computation. Furthermore, energy diminishing is observed in linear case under
any time step $\tau>0$ and nonlinear case when time step $\tau$ is not too big \cite{BaoDu}. If high accuracy is crucial in computing ground states in BEC, e.g. under an optical lattice potential or
in a rotational frame, BESP is recommended.

\section{Numerical methods for computing  dynamics of GPE}
\label{sec:numdym}
\setcounter{equation}{0}\setcounter{figure}{0}\setcounter{table}{0}
In this section, we review different numerical methods to discretize the Cauchy problem of the GPE (\ref{eq:gpe:sec2})
for computing the dynamics of BEC.
In fact, many efficient and accurate numerical methods have been proposed for discretizing the above GPE, or
the nonlinear Schr\"{o}dinger equation (NLSE) in general,
such as time-splitting sine pseudospectral method \cite{BaoJakschP,BaoJinP,BaoJinP2,BaoZhang}, time-splitting finite difference
 method  \cite{Wang0,BaoTang0}, time-splitting Laguerre-Hermite pseudospectral method \cite{BaoShen}, conservative Crank-Nicolson finite difference method \cite{Chang,BaoCai2,BaoCai},
 semi-implicit finite difference method \cite{BaoCai2,BaoCai}, etc. Each method has its own advantages and disadvantages. Here we present the detailed algorithms for some of these methods.

In practice, we truncate the problem on a bounded domain $U\subset\Bbb R^d$ ($d=1,2,3$) as Eq. (\ref{eq:sdge:sec3}), with either homogeneous Dirichlet boundary conditions (\ref{eq:sdge1:sec3}) or periodic boundary conditions when $U$ is an interval (1D), a rectangle (2D) or a box (3D).
For simplicity of notation,  we shall introduce the method
in one space dimension ($d=1$). Generalizations to $d>1$ are
straightforward for tensor product grids and the results remain
valid without modifications. In 1D, the equation (\ref{eq:sdge:sec3})
 becomes \bea \label{eq:sdge1d:sec4}
&&i \pl{\psi(x,t)}{t}=-\fl{1}{2}\psi_{xx}(x,t)+
V(x)\psi(x,t)
+ \beta |\psi|^2\psi(x,t), \ a<x<b,\;t>0,\qquad \\
\label{eq:sdgi1d:sec4}
&&\psi(x,t=0)=\psi_0(x), \qquad  a\le x\le b,\eea
with the homogenous Dirichlet boundary condition
\be\label{eq:sdgd01d:sec4}
\psi(a,t)=\psi(b,t)=0,\qquad t>0.
\ee
 The following periodic boundary condition
\be
\label{eq:sdgb1d:sec4}
\psi(a,t)=\psi(b,t),\quad \p_x\psi(a,t)=\p_x\psi(b,t),
\qquad t>0,
\ee
or homogeneous Neumann boundary condition
\be\label{eq:sdgd01np:sec4}
\p_x\psi(a,t)=\p_x\psi(b,t)=0,\qquad t>0,
\ee
is also widely used in the literature.

To discretize Eq. (\ref{eq:sdge1d:sec4}), we use uniform grid points. Choose the spatial mesh size $h=\btu x>0$ with $\Delta x=(b-a)/M$  ($M$ an
even positive integer), the time step $\tau$, the grid points and the time step  as
in (\ref{eq:mesh1d:sec3}).
Let $\psi^n_j$ be the approximation of $\psi(x_j,t_n)$ and
$\psi^n$ be the solution vector with components $\psi_j^n$.

In the following, we will introduce two widely used schemes: the time-splitting methods and the finite difference time domain (FDTD) methods.

\subsection{Time splitting pseudospectral/finite difference method}
\label{subsec:ts}
The time splitting procedure was
presented for differential equations in \cite{Strang} and applied
to Schr\"{o}dinger equations in \cite{Hardin,Taha1}. For the simplest two-step case, consider an abstract initial value problem for $u:[0,T]\to \calB$ ($\calB$ Banach space),
\be\label{eq:generalinv:sec4}
\frac{\rd}{\rd t}u(t)=(A +B)u(t),\quad u(0)\in\calB,
\ee
where $A$ and $B$ are two operators, the solution can be written in the abstract form as $u(t)=e^{t(A+B)}u(0)$. For a given time step $\tau>0$, let $t_n = n \tau$, $n=0,1,\ldots$, and $u^n$ be the approximation of $u(t_n)$. The time-splitting approximation, or operator splitting (split-step) is usually  given as  \cite{Strang,Yoshida}
\be\label{eq:lietrotter:sec4}
u^{n+1}= e^{\tau A}e^{\tau B}u^n, \qquad\text{Lie-Trotter splitting},
\ee
or
\be\label{eq:strang:sec4}
u^{n+1}\approx e^{\tau A/2}e^{\tau B}e^{\tau A/2}u^n,\qquad\text{Strang splitting}.
\ee
Formally, from Taylor expansion, it is easy to see that the approximation error of Lie-Trotter splitting is of first order $O(\tau)$, and the error of Strang splitting is of second order $O(\tau^2)$. In principle, splitting approximations of higher order accuracy can be constructed as \cite{Yoshida}
\be
u^{n+1}\approx e^{a_1\tau A}e^{b_1\tau B}e^{a_2\tau A}e^{b_2\tau B}\cdots e^{a_m \tau A}e^{b_m\tau B}u^n,\quad m\ge1,
\ee
where coefficients $a_j$ and $b_j$ ($j=1,\cdots,m$) are chosen properly. One of the most frequently used higher order splitting scheme is the  fourth-order symplectic time integrator (cf. \cite{Yoshida,BaoShen}) for (\ref{eq:generalinv:sec4})  as:
\be\label{eq:4thsplit:sec4}
\begin{split}
&u^{(1)} = e^{\theta \tau A}\;u^n,\quad  u^{(2)}= e^{2\theta \tau B}\;u^{(1)},\quad
u^{(3)} = e^{(0.5-\theta)\tau A}\;u^{(2)},\\
&u^{(4)}=e^{(1-4\theta)\tau B}\;u^{(3)},\quad u^{(5)}=e^{(0.5-\theta)\tau A}\;u^{(4)},
\quad u^{(6)}=e^{2\theta \tau B}\;u^{(5)},\\& u^{n+1}=e^{\theta \tau A}\;u^{(6)},
\end{split}
\ee
where the coefficient $\theta$ can be calculated as
\be
\theta =\frac{1}{6}\left(2+2^{1/3}+2^{-1/3}\right)\approx 0.67560 35959 79828 81702.
\ee

\subsubsection{Time splitting sine pseudospectral method}
\label{subsubsec:TSSP}
In this section we present a time-splitting sine pseudospectral
method, to numerically solve the GPE (\ref{eq:sdge1d:sec4}) with homogenous Dirichlet boundary condition (\ref{eq:sdgd01d:sec4}).
The merit of this method is that it is
unconditionally stable, time reversible, time-transverse
invariant, and that it conserves the total particle number.

\medskip

{\bf Time-splitting sine pseudospectral (TSSP) method}. From time $t=t_n$ to
$t=t_{n+1}$,  the GPE (\ref{eq:sdge1d:sec4})  is solved in two splitting
steps. One solves first \be \label{eq:fstep:sec4}
i \psi_t=-
\fl{1}{2} \psi_{xx}, \ee for the time step of length $\tau$,
followed by solving \be \label{eq:sstep:sec4} i
\pl{\psi(x,t)}{t}=V(x)\psi(x,t) + \beta
|\psi(x,t)|^2\psi(x,t), \ee for the same time step. Eq.
(\ref{eq:fstep:sec4}) will be discretized in space by the sine spectral
method and integrated in time {\it exactly}. For
$t\in[t_n,t_{n+1}]$, the ODE (\ref{eq:sstep:sec4}) leaves $|\psi|$
invariant in $t$ \cite{BaoJakschP,BaoJinP} and  therefore  becomes \be
\label{eq:sstepp:sec4} i\pl{\psi(x,t)}{t}=V(x)\psi(x,t) +
\beta |\psi(x,t_n)|^2\psi(x,t) \ee and thus can be integrated {\it
exactly}. This is equivalent to choosing operators $A$, $B$ in (\ref{eq:generalinv:sec4}) as
\be
A\psi=\frac{i}{2}\p_{xx}\psi,\qquad B\psi=-i(V(x)+\beta |\psi|^2)\psi.
\ee
 From time $t=t_n$ to $t=t_{n+1}$, we combine the
splitting steps via the standard Strang splitting \cite{BaoJakschP,BaoJinP,BaoZhang}:
\be\label{eq:tssp:sec4}\begin{split}
&\psi_j^{(1)}=\fl{2}{M}\sum_{l=1}^{M-1}
  e^{-i\tau\mu_l^2/4}\;\widetilde{(\psi^n)}_l\;\sin(\mu_l(x_j-a)),
    \\
&\psi^{(2)}_j=e^{-i(V(x_j)+\beta |\psi_j^{(1)}|^2)\tau}\;\psi_j^{(1)}, \qquad j\in\calT_M,\\
&\psi^{n+1}_j=\fl{2}{M}\sum_{l=1}^{M-1}
  e^{-i\tau\mu_l^2/4}\;\widetilde{(\psi^{(2)})}_l\;\sin(\mu_l(x_j-a)), \qquad\psi_0^{n+1}=\psi_M^{n+1}=0,
\end{split}
\ee
where $\mu_l=l\pi/(b-a)$ for $l\in\calT_M$ and  $\widetilde{(\psi^n)}_l$ and $\widetilde{(\psi^{(2)})}_l$ are the discrete sine transform
coefficients of $\psi^n$  and $\psi^{(2)}$, respectively, which are defined  in (\ref{eq:sinetran:sec3}). One can also exchange the order of the steps (\ref{eq:sstep:sec4}) and (\ref{eq:sstepp:sec4}) in TSSP (\ref{eq:tssp:sec4}), and the numerical results are almost the same.

\begin{remark}For the  GPE (\ref{eq:sdge1d:sec4})  truncated in a bounded domain with periodic boundary condition (\ref{eq:sdgb1d:sec4}) or
homogenous Neumann boundary condition (\ref{eq:sdgd01np:sec4}), a time splitting Fourier or cosine pseudospectral  method similar to TSSP (\ref{eq:tssp:sec4}) is straightforward \cite{BaoWang,BaoJakschP,BaoJinP}, i.e., solve the linear part (\ref{eq:fstep:sec4}) by Fourier or cosine spectral discretization instead of sine spectral discretization.
\end{remark}

\subsubsection{Time splitting finite difference method}
\label{subsubsec:TSFD}
In section \ref{subsubsec:TSSP}, a sine pseudospectral method (or a Fourier pseudospectral method) is used to solve Eq. (\ref{eq:fstep:sec4}). Spectral method is favorable in view of its high accuracy if the solution of  continuous problem (\ref{eq:gpe:sec2}) is smooth. For non-smooth potentials $V(x)$ (random potential), the regularity of the solution of the GPE (\ref{eq:gpe:sec2}) would be low. In such case, finite difference method is suggested instead of spectral method.

{\bf Time-splitting finite difference (TSFD) method} From time $t=t_n$ to
$t=t_{n+1}$,  the GPE (\ref{eq:sdge1d:sec4}) with homogenous Dirichlet boundary condition (\ref{eq:sdgd01d:sec4}) is solved in two splitting
steps, (\ref{eq:fstep:sec4}) and (\ref{eq:sstepp:sec4}). As indicated in section \ref{subsubsec:TSSP}, Eq. (\ref{eq:sstepp:sec4}) can be integrated exactly. For the linear part (\ref{eq:fstep:sec4}), we use Crank-Nicolson finite difference method.  From time $t=t_n$ to $t=t_{n+1}$, we combine the
splitting steps via the standard Strang splitting \cite{BaoTang0}:
\be\label{eq:tsfd:sec4}\begin{split}
&\psi^{(1)}_j=e^{-i\tau (V(x_j)+\beta |\psi_j^n|^2)/2}\;\psi_j^n, \qquad j\in\calT_M^0,\\
&i\frac{\psi_j^{(2)}-\psi^{(1)}_j}{\tau}=-\frac14 (\delta_x^2\psi_j^{(2)}+\delta_x^2\psi_j^{(1)}),
  \quad j\in\calT_M, \qquad \psi_0^{(2)}=\psi_{M}^{(2)}=0,  \\
&\psi^{n+1}_j=e^{-i\tau (V(x_j)+\beta |\psi_j^{(2)}|^2)/2}\;\psi_j^{(2)},
 \qquad j\in\calT_M^0.
\end{split}
\ee

This method is also unconditionally stable, time reversible, time-transverse
invariant, and it conserves the total particle number. The linear part of TSFD
(\ref{eq:tsfd:sec4}) can be solved by  the Thomas' algorithm in 1D and discrete sine transform (DST) in 2D and 3D.  Thus, the computational costs of TSFD and TSSP are the same in 2D and 3D, while the cost of TSFD is cheaper in 1D due to the Thomas' algorithm.

\subsection{Finite difference time domain method}\label{subsec:fdtd}
Finite difference time domain (FDTD) methods for NLSE  have been extensively studied in the literature \cite{Akrivis1,Chang} and are  widely used. In this section, we present the most popular finite difference discretizations for the GPE (\ref{eq:sdge1d:sec4}) with homogenous Dirichlet boundary condition (\ref{eq:sdgd01d:sec4}).
\subsubsection{Crank-Nicolson finite difference method}
The {\sl conservative Crank-Nicolson finite difference} (CNFD)
discretization of  GPE (\ref{eq:sdge1d:sec4}) reads \cite{Chang,Glassey,BaoJinP2,BaoCai2,BaoCai}
\begin{equation}\label{eq:cnfd:sec4}
i\delta^+_t\psi_{j}^n=\left[-\frac{1}{2}\delta_x^2+V_{j}+
\frac{\beta}{2}(|\psi_{j}^{n+1}|^2+|\psi_{j}^n|^2)\right]
\psi_{j}^{n+1/2}, \quad j\in\calT_M,\ n\ge0,
 \end{equation}
where
\begin{equation*}
V_{j}=V(x_j),\qquad \psi_{j}^{n+1/2}=\frac
12\left(\psi_{j}^{n+1}+\psi_{j}^n\right), \qquad j\in\calT_M^0, \quad n=0,1,2,\ldots.
\end{equation*}
The boundary condition (\ref{eq:sdgd01d:sec4}) is discretized as
 \be \label{eq:bound1:sec4}
\psi^{n+1}_{0}=\psi^{n+1}_{M}=0,  \qquad
  n=0,1,\ldots\;,
 \ee
and the initial condition (\ref{eq:sdgi1d:sec4}) is discretized as
 \be \label{eq:init1:sec4}
\psi^0_{j}=\psi_0(x_j), \qquad j\in \calT_M^0.
 \ee
The above CNFD method conserves the mass
and energy in the discretized level. However, it is a fully implicit
method, i.e., at each time step, a fully nonlinear system must be
solved, which may be very expensive, especially in 2D and 3D. In
fact, if the fully nonlinear system is not solved numerically to
extremely high accuracy, e.g., at machine accuracy, then the mass and
energy of the numerical solution obtained in practical computation
are no longer conserved \cite{BaoCai2}. This motivates us also consider the
following semi-implicit discretization for the GPE.
\subsubsection{Semi-implicit finite difference method}
The {\sl semi-implicit finite difference} (SIFD) discretization  of
the GPE (\ref{eq:sdge1d:sec4}), is to use Crank-Nicolson/leap-frog schemes for
discretizing linear/nonlinear terms, respectively, as \cite{BaoCai2,BaoCai}
\begin{equation}\label{eq:sifd:sec4}
i\delta_t
\psi_{j}^n=\left[-\frac{1}{2}\delta_x^2+V_{j}\right]\frac{\psi_{j}^{n+1}+\psi_{j}^{n-1}}{2}
+\beta|\psi_{j}^{n}|^2\psi_{j}^{n},\qquad j\in
\calT_M,\quad n\ge1.
\end{equation}
Again, the boundary condition (\ref{eq:sdgd01d:sec4}) and initial condition
(\ref{eq:sdgi1d:sec4}) are discretized in (\ref{eq:bound1:sec4}) and (\ref{eq:init1:sec4}),
respectively. In addition, the first step can be computed by any
explicit second or higher order time integrator, e.g., the
second-order modified Euler method, as
\begin{eqnarray}\label{eq:sifd1:sec4}
&&\psi_{j}^1=\psi_{j}^0-i\tau\left[\left(-\frac{1}{2}\delta_
x^2+V_{j}\right)\psi_{j}^{(1)} +\beta
|\psi_{j}^{(1)}|^2\psi_{j}^{(1)}\right],\qquad j\in \calT_M,\qquad \\
&&\qquad\qquad
\psi_{j}^{(1)}=\psi_{j}^0-i\frac{\tau}{2}\left[\left(-\frac{1}{2}\delta_x^2+V_{j}\right)\psi_{j}^0
+\beta|\psi_{j}^{0}|^2\psi_{j}^{0}\right]. \nonumber
\end{eqnarray}

For this SIFD method, at each time step, only a linear system is to
be solved, which is much more cheaper than that of the CNFD method
in practical computation.

\subsection{Simplified methods for symmetric potential and initial data}
Similar to  section \ref{subsec:sympotgs}, when the potential $V(\bx)$ ($\bx\in\Bbb R^d$, $d=1,2,3$) and the initial data $\psi_0$ are symmetric, then the solution of the GPE (\ref{eq:gpe:sec2}) is symmetric.
In such cases, simplified numerical methods, especially with less memory requirement, are
available for computing the dynamics of GPE.

\subsubsection{Radial symmetry in 1D, 2D and 3D} When  potential $V(\bx)$ and initial data $\psi(\bx,0)=\psi_0(\bx)$ are radially symmetric for $d=1,2$ and spherically symmetric for $d=3$, the problem is reduced to 1D. In such case, the GPE (\ref{eq:gpe:sec2}) collapse to the 1D equation (\ref{eq:gperadial:sec3}) for $\psi:=\psi(r,t)$ with $r=|\bx|\ge0$.

In practical computation, the equation (\ref{eq:gperadial:sec3}) is approximated  on a finite
interval. Since the  wave function  vanishes  as $r\to\infty$, choosing $R>0$ sufficiently
large,  equation (\ref{eq:gperadial:sec3}) can be approximated by
\be
\label{eq:gperadial:sec4}
i\p_t\psi(r,t) =  -\frac{1}{2r^{d-1}}\frac{\p}{\p r}\left(r^{d-1}\frac{\p}{\p r}\psi\right)+V(r)\psi
+\bt|\psi|^2\psi,\quad 0<r<R,
\ee
with boundary conditions
\be\label{eq:boundaryradial:sec4}
\p_r\psi(0,t)=0,\qquad \psi(R,t)=0, \qquad t\ge0,
\ee
and initial condition
\be\label{eq:iniradial:sec4}
\psi(r,0)=\psi_0(r), \qquad 0\le r\le R.
\ee
For this 1D problem, an efficient and accurate method is the time-splitting finite difference (TSFD) method.
Choose time steps as (\ref{eq:mesh1d:sec3}) and mesh size  $\Delta r=2R/(2M+1)$ for some integer $M>0$. We use the same notations for the grid points as (\ref{eq:meshradial:sec3}), and let $\psi_{j+\frac12}^n$ be the numerical approximation of $\psi(r_{j+\frac12},t_n)$ and $\psi^{n}$ be the solution vector at time $t=t_n$ with
components $\psi_{j+\frac12}^{n}$. The time-splitting finite difference (TSFD)  discretization for (\ref{eq:gperadial:sec4}) reads
\begin{align}\label{eq:tsfdsim:sec4}
&\psi^{(1)}_{j+\frac{1}{2}}=e^{-i\tau \beta |\psi_{j+\frac{1}{2}}^n|^2/2}\;\psi_{j+\frac{1}{2}}^n, \qquad  j\in\calT_M^0, \nn\\
&i\frac{\psi_{j+\frac{1}{2}}^{(2)}-\psi^{(1)}_{j+\frac{1}{2}}}{\tau}=
\left[-\frac{1}{2}\delta_{r,d}^2 +V(r_{j+\frac{1}{2}})\right]\frac{\psi_{j+\frac{1}{2}}^{(2)}+\psi^{(1)}_{j+\frac{1}{2}}}{2},
\quad 0\le j\le M-1, \nn\\
&\psi_{-\frac{1}{2}}^{(2)}=\psi_{\frac{1}{2}}^{(2)}, \qquad \psi_{M+\frac{1}{2}}^{(2)}=0,
\qquad \psi_{j+\frac{1}{2}}^0=\psi_0(r_{j+\frac{1}{2}}), \qquad j\in\calT_M^0, \nn\\
&\psi^{n+1}_{j+\frac{1}{2}}=e^{-i\tau \beta |\psi_{j+\frac{1}{2}}^{(2)}|^2/2}\;\psi_{j+\frac{1}{2}}^{(2)},
 \qquad j\in\calT_M^0, \quad n\ge0.
\end{align}
This method is second-order accurate in space and time, unconditionally stable and it conserves the normalization in the discretized level. The memory cost is $O(M)$ and computational cost is $O(M)$ per time step,
which save significantly from $O(M^d)$ and $O(M^d\ln M)$, respectively,
in 2D and 3D when the Cartesian coordinates is used by TSSP.

\subsubsection{Cylindrical symmetry in 3D} For $\bx=(x,y,z)^T\in\Bbb R^3$, when $V$  and $\psi_0$ are cylindrically symmetric, i.e., $V$ and $\psi_0$ are of the form $V(r,z)$ and $\psi_0(r,z)$ ($r=\sqrt{x^2+y^2}$), respectively, the Cauchy problem of the GPE (\ref{eq:gpe:sec2}) is reduced to 2D.  Due to the symmetry, the GPE (\ref{eq:gpe:sec2}) essentially
collapses to the 2D problem  (\ref{eq:gpecylin:sec3}) with $r\in(0,+\infty)$ and $z\in\Bbb R$ for $\psi:=\psi(r,z,t)$.

In practice, the GPE (\ref{eq:gpecylin:sec3}) is truncated on a bounded domain. Choosing $R>0$ and $Z_1<Z_2$ with $|Z_1|$, $|Z_2|$ and $R$ sufficiently
large, then Eq. (\ref{eq:gpecylin:sec3}) can be approximated for $(r,z)\in(0,R)\times(Z_1,Z_2)$ as,
\be
\label{eq:gpecylin:sec4}
i\p_t\psi(r,z,t) =  -\frac12
\left[\frac{1}{r}\frac{\p}{\p r}\left(r\frac{\p\psi}{\p r}\right)+\frac{\p^2\psi}{\p z^2}\right]+\left(V(r,z)
+\bt|\psi|^2\right)\psi,
\ee
with boundary condition
\be\label{eq:boundarycylin:sec4}
\frac{\p\psi(0,z,t)}{\p r}=0,\; \psi(R,z,t)=\psi(r,Z_1,t)=\psi(r,Z_2,t)=0,
\quad z\in[Z_1,Z_2],\; r\in[0,R],
\ee
and initial condition
\be\label{eq:inicylin:sec4}
\psi(r,z,0)=\psi_0(r,z),\qquad z\in[Z_1,Z_2],\quad r\in[0,R].
\ee
Similar to section \ref{subsec:sympotgs}, choose time steps as (\ref{eq:mesh1d:sec3}) and $r$- grid points (\ref{eq:meshradial:sec3}) for positive integer $M>0$. For integer $N>0$, choose mesh size $\Delta z=(Z_2-Z_1)/N$ and define $z$- grid points $z_k=Z_1+k\Delta z$ for $k\in\calT_N^0$. Let $\psi_{j+\frac12\,k}^n$ be the numerical approximation of $\psi(r_{j+\frac12},z_{k},t_n)$ and $\psi^{n}$ be the solution vector at time $t=t_n$ with
components $\psi_{j+\frac12\,k}^{n}$.
Then the time-splitting finite difference (TSFD)  discretization for (\ref{eq:gpecylin:sec4}) reads
\begin{align}\label{eq:tsfdcyl:sec4}
&\psi^{(1)}_{j+\frac{1}{2}\,k}=e^{-i\tau[V(r_{j+\frac{1}{2}},z_k)+ \beta |\psi_{j+\frac{1}{2}\,k}^n|^2]/2}\;\psi_{j+\frac{1}{2}\,k}^n, \qquad (j,k)\in \calT_{MN}^0, \nn \\
&i\frac{\psi_{j+\frac{1}{2}\,k}^{(2)}-\psi^{(1)}_{j+\frac{1}{2}\,k}}{\tau}=
-\frac{1}{4}(\delta_{r}^2 +\delta_z^2)\left(\psi_{j+\frac{1}{2}\,k}^{(2)}+\psi^{(1)}_{j+\frac{1}{2}\,k}\right),
\qquad (j,k)\in \calT_{MN}^*,\nn \\
&\psi_{-\frac{1}{2}\,k}^{(2)}=\psi_{\frac{1}{2}\,k}^{(2)}, \ \psi_{M+\frac{1}{2}\,k}^{(2)}=\psi_{j+\frac{1}{2}\,0}^{(2)}=\psi_{j+\frac{1}{2}\,N}^{(2)}=0,
 \ (j,k)\in \calT_{MN}^0,\nn \\
&\psi^{n+1}_{j+\frac{1}{2}\,k}=e^{-i\tau [V(r_{j+\frac{1}{2}},z_k)+\beta |\psi_{j+\frac{1}{2}\,k}^{(2)}|^2]/2}\;\psi_{j+\frac{1}{2}\,k}^{(2)}, \quad (j,k)\in \calT_{MN}^0,
\end{align}
with $\ \psi_{j+\frac{1}{2}\,k}^0=\psi_0(r_{j+\frac{1}{2}},z_k)$ for $(j,k)\in \calT_{MN}^0$.
This method is  second-order accurate in space and time, unconditionally stable and it conserves the normalization in the discretized level. The memory cost is $O(MN)$ and computational cost is $O(MN\ln M)$ per time step,
which save significantly from $O(M^2N)$ and $O(M^2N\ln M)$, respectively,
 when the Cartesian coordinates is used by TSSP.

\subsection{Error estimates   for SIFD and CNFD}
In this section, we present the error estimates for CNFD (\ref{eq:cnfd:sec4}) and SIFD (\ref{eq:sifd:sec4}). The notations for finite difference operators, related finite dimensional vector space and norms are given in section \ref{subsec:BEFD}. Specially, for real valued nonnegative potential $V(x)$, we define the corresponding discrete weighted $l^2$ norm for $u=(u_0,u_1,\cdots,u_M)^T\in X_M$ as:
\be
\|u\|_V^2=h\sum\limits_{j=0}^{M-1}V_j|u_j|^2.
\ee

In the remaining part of this section, we use the notation $p
\lesssim q$ to represent that there exists a generic constant $C$
which is independent of time step $\tau$ and mesh size $h$ such that
$|p|\le C\,q$.

To state the error bounds, we define the `error' function $e^n\in X_M$ as
\be e_{j}^n =
\psi(x_j,t_n)-\psi_{j}^n, \qquad j\in \calT_M^0, \quad
n\ge0,\ee
where $\psi:=\psi(x,t)$ is the exact solution of the GPE (\ref{eq:sdge1d:sec4}) and $\psi_j^n$ is the numerical solution.
We make the following assumption
on the exact solution $\psi$, i.e., let $0<T<T_{\rm max}$ with
$T_{\rm max}$ the maximal existing time of the solution \cite{Cazenave,Sulem}:
\be \label{assum:fd:sec4}
\psi\in
C^3([0,T];W^{1,\infty})\cap C^2([0,T]; W^{3,\infty})
\cap C^0([0,T];W^{5,\infty}\cap
H_0^1),
\ee
where the spatial norms are taken in the interval $U=(a,b)$.

For the SIFD method, we have \cite{BaoCai2,BaoCai}
\begin{theorem}\label{thm:sifd:sec4} Assuming that
(\ref{assum:fd:sec4}) holds and that $V(x)\in C^1(\bar{U})$, there exist $h_0>0$ and $0<\tau_0\leq\frac{1}{4}$
sufficiently small, when $0<h\le h_0$ and $0<\tau\le \tau_0\leq\frac14$,  we
have the following  error estimates for the SIFD method
(\ref{eq:sifd:sec4}) with (\ref{eq:bound1:sec4}), (\ref{eq:init1:sec4}) and
(\ref{eq:sifd1:sec4})
 \be\label{eq:sifdes:sec4} \|e^n\|_2\lesssim h^2+\tau^2,\quad\|\delta_x^+
e^n\|_2\lesssim h^{3/2}+\tau^{3/2}, \quad \|\psi^n\|_\infty\leq 1+M_1,
\quad 0\le n\le \frac{T}{\tau}, \ee where
 $M_1=\max_{0\le t\le T}\|\psi(\cdot,t)\|_{L^\infty(U)}$.
In addition, if either  $\p_{n}V(x)|_{\Gamma}=0$ ($\Gamma=\p U$, $\p_{n}$ is the outer normal derivative) or $\psi\in C^0([0,T];H^2_0(U))$, we have the optimal
error estimates
\be\label{eq:sifdo:sec4}
\|e^n\|_2+ \|\delta_x^+e^n\|_2\lesssim h^{2}+\tau^{2},\quad \|\psi^n\|_\infty\leq 1+M_1,\quad 0\le n\le \frac{T}{\tau}.
\ee
\end{theorem}

Similarly, for the CNFD method, we have \cite{BaoCai2,BaoCai}

\begin{theorem}\label{thm:cnfd:sec4}
Assuming that (\ref{assum:fd:sec4}) holds and that $V(x)\in C^1(\bar{U})$, there exists $h_0>0$ and $\tau_0>0$ sufficiently small, when
$0<h\le h_0$ and $0<\tau\le \tau_0$,  the CNFD discretization (\ref{eq:cnfd:sec4}) with (\ref{eq:bound1:sec4})
and (\ref{eq:init1:sec4}) admits a unique solution $\psi^n$ ($0\le n\le
\frac{T}{\tau}$) such that the following  error estimates hold,
 \be \label{eq:cnfdes:sec4}\|e^n\|_2\lesssim h^2+\tau^2,\quad\|\delta_x^+ e^n\|_2\lesssim
h^{3/2}+\tau^{3/2},\quad \|\psi^n\|_\infty\leq 1+M_1, \quad 0\le n\le \frac{T}{\tau}.\ee
In addition, if either  $\p_{n}V(x)|_{\Gamma}=0$ or $\psi\in C^0([0,T];H^2_0(U))$, we have the optimal
error estimates
\be\label{eq:cnfdo:sec4}
\|e^n\|_2+ \|\delta_x^+e^n\|_2\lesssim h^{2}+\tau^{2},\quad \|\psi^n\|_\infty\leq 1+M_1,\quad 0\le n\le \frac{T}{\tau}.
\ee
\end{theorem}

Error bounds of conservative CNFD method for NLSE in 1D (without potential $V(x)$) was established in
\cite{Chang,Glassey}. In fact, their proofs for CNFD rely strongly
on the conservative property of the method and the discrete version
of the Sobolev inequality in 1D
\be\|f\|_{L^\infty}^2\leq \|\nabla f\|_{L^2}\cdot \|f\|_{L^2}, \qquad
\forall f\in H^1(U) \ \hbox{with}\ U\subset {\Bbb R}, \ee which
immediately implies {\it a priori} uniform bound for
$\|f\|_{L^\infty}$. However, the extension of the discrete version
of the above Sobolev inequality is no longer valid in 2D and 3D.
Thus the techniques used in \cite{Chang,Glassey} for obtaining error
bounds of CNFD for NLSE only work for conservative schemes in 1D and
they cannot be extended to either high dimensions or
non-conservative finite difference scheme like SIFD.

Here, we are going  to use different techniques to establish
optimal error bounds of CNFD and SIFD for the GPE (\ref{eq:gpe:sec2}) in 1D
which can be directly generalized to 2D and 3D. In the
analysis, besides the standard techniques of the energy method, for
SIFD, we adopt the mathematical induction; for CNFD, we cut off the nonlinearity.
\subsubsection{Convergence rate for SIFD}
Firstly,  SIFD (\ref{eq:sifd:sec4}) is uniquely solvable \cite{BaoCai}.
\begin{lemma}\label{lem:sifdex:sec4} (Solvability of the difference equations)
 Under the assumption (\ref{assum:fd:sec4}),
for any given initial data $\psi^0\in X_M$, there exists a unique
solution $\psi^n\in X_M$ of (\ref{eq:sifd1:sec4}) for $n=1$ and
(\ref{eq:sifd:sec4}) for $n>1$.
\end{lemma}
Define the local truncation error $\eta^n\in X_M$ of the SIFD method
(\ref{eq:sifd:sec4}) with (\ref{eq:bound1:sec4}), (\ref{eq:init1:sec4}) and
(\ref{eq:sifd1:sec4}) for $n\ge1$ and $j\in\calT_M$ as
\be\label{eq:etadef:sec4}
\eta_{j}^n:=\left(i\delta_t-\beta|\psi(x_j,t_{n})|^2\right)
\psi(x_j,t_n)+\left[\frac{\delta_x^2}{2}-V_{j}\right]
\frac{\psi(x_j,t_{n-1})+\psi(x_j,t_{n+1})}{2},
\ee
and by noticing (\ref{eq:init1:sec4}) for $n=0$ as
 \bea\label{eq:eta0:sec4}
&&\eta_{j}^0:=i\delta_t^+\psi(x_j,0) -\left(-\frac{1}{2}\delta_x^2
+V_{j}\right)\psi_{j}^{(1)} -\beta
|\psi_{j}^{(1)}|^2\psi_{j}^{(1)},\quad j\in {\calT}_M,\qquad \\
&&\quad
\psi_{j}^{(1)}=\psi_0(x_j)-i\frac{\tau}{2}\left[\left(-\frac{1}{2}\delta_x^2+V_{j}
\right)\psi_0(x_j)
+\beta|\psi_0(x_j)|^2\psi_0(x_j)\right]. \nn\eea
Then we have
\begin{lemma}\label{lem:localerrsifd:sec4} (Local truncation error)
 Assuming $V(\bx)\in C(
\overline{U})$,
under the assumption (\ref{assum:fd:sec4}), we have
 \be\label{eq:tau1:sec4}
\|\eta^n\|_\infty\lesssim \tau^2+h^2, \qquad  0\leq n \leq
\frac{T}{\tau}-1, \quad\mbox{and}\qquad
\|\delta_x^+\eta^0\|_\infty\lesssim \tau+h. \ee
In addition, assuming
 $V(\bx)\in C^1(\overline{U})$,  we have  for  $1\le n \leq \frac{T}{\tau}-1$
\be\label{eq:tauder1:sec4} |\delta_x^+\eta_{j}^n|\lesssim\begin{cases}
\tau^2+h^2,& 1\leq j\leq M-2,\\
  \tau+h, &j=0,M-1.
 \end{cases}
 \ee
Furthermore,  assuming either  $\p_nV(x)|_{\Gamma}=0$ or $u\in C([0,T];H^2_0(U))$, we have
\be\label{eq:tauder:sec4} \|\delta_x^+\eta^n\|_\infty\lesssim
\tau^2+h^2,\qquad 1\le n \leq \frac{T}{\tau}-1. \ee
 \end{lemma}
\begin{proof}
 First, we prove (\ref{eq:tau1:sec4}) and
(\ref{eq:tauder1:sec4}) when $n=0$. Rewriting $\psi_{j}^{(1)}$ and then
using Taylor's expansion at $(x_j,0)$, noticing (\ref{eq:sdge1d:sec4}) and
(\ref{eq:sdgi1d:sec4}), we get \begin{align*}
\psi_{j}^{(1)}=&\;\psi\left(x_j,\frac{\tau}{2}\right)+i\frac{\tau}{2}
\left[\left(\frac{\delta_x^2}{2}-V_{j}-\beta|\psi_0(x_j)|^2\right)\psi_0(x_j)+
i\frac{\psi\left(x_j,\frac{\tau}{2}\right)-\psi_0(x_j)}{\tau/2}\right]\nn\\
=&\;\psi\left(x_j,\frac{\tau}{2}\right)+i\frac{\tau}{2}\biggl[
\frac{h^2}{2}\int_0^1\int_{0}^{s_1}\int_0^{s_2}\int_{-s_3}^{s_3}
\p_{xxxx}\psi_0(x_j+s h)\,ds\,ds_3\,ds_2\,ds_1\nn\\
&+i\frac{\tau}{2}\int_0^1\int_0^\theta\p_{tt}\psi(x_j,s\tau/2)\,ds\,d\theta\biggr]
=\psi\left(x_j,\frac{\tau}{2}\right)+O\left(\tau^2+\tau
h^2\right),\quad
 j\in\mathcal{T}_M.\end{align*}
 Then, using  Taylor expansion at
$(x_j,\tau/2)$ in (\ref{eq:eta0:sec4}), noticing (\ref{eq:sdge1d:sec4}), in view of triangle inequality and the assumption (\ref{assum:fd:sec4}), we have
\begin{align*}
|\eta^0_{j}|\lesssim&\;\tau^2\|\psi_{ttt}\|_{L^\infty}+
h\tau\|\psi_{xxxxx}\|_{L^\infty}+\tau^2\|\psi_{ttxx}\|_{L^\infty}+\tau(h^2\|\psi_{xxxx}\|_{L^\infty}
+\tau\|\psi_{tt}\|_{L^\infty})\nn\\
&\cdot(\|\psi\|_{L^\infty}+\tau h^2\|\psi_{xxxx}\|_{L^\infty}+\tau^2\|\psi_{tt}\|_{L^\infty})^2\nn\\
\lesssim&\;\tau^2+h^2, \qquad
 j\in\calT_M,
\end{align*}
where the $L^\infty$-norm means $\|\psi\|_{L^\infty}:=\sup_{0\le t\le
T}\sup_{x\in U}|\psi(x,t)|$. This immediately implies (\ref{eq:tau1:sec4})
when $n=0$. Similarly, we can get
\begin{align*}
|\delta_x^+\eta^0_{j}|\lesssim &\;\tau\|\psi_{xxxxx}\|_{L^\infty}
+\tau^2\|\psi_{tttx}\|_{L^\infty}+\tau h\|\psi_{xxxx}\|_{L^\infty}(\|\psi\|_{W^{4,\infty}}+\|\psi_{tt}\|_{L^\infty})^2\\
&+\tau^2\|\psi_{ttx}\|_{L^\infty}(\|\psi\|_{W^{4,\infty}}+\|\psi_{tt}\|_{L^\infty})^2\\
\lesssim&\; \tau+h,\qquad  j\in \calT_M,
\end{align*}
where for $j=0$, we use equation (\ref{eq:sdgd01d:sec4}) to deduce that $\psi_{tt}(a,t)=\psi_{ttt}(a,t)=0$.

Now we prove (\ref{eq:tau1:sec4}), (\ref{eq:tauder1:sec4}) and (\ref{eq:tauder:sec4}) when $n\ge1$. Using
Taylor's expansion at $(x_j,t_n)$ in (\ref{eq:etadef:sec4}), noticing
(\ref{eq:sdge1d:sec4}), using triangle inequality and  assumption (\ref{assum:fd:sec4}), we
have for $1\le n\le
\frac{T}{\tau}-1$,
\begin{equation*}
|\eta_{j}^n|\lesssim h^2\|\psi_{xxxx}\|_{L^\infty}
+\tau^2\left(\|\psi_{ttt}\|_{L^\infty}+\|\psi_{ttxx}\|_{L^\infty}\right)
\lesssim\tau^2+h^2, \qquad j\in\mathcal{T}_M.
\end{equation*}
Thus, (\ref{eq:tau1:sec4}) is true. For $j=1,\ldots, M-2$, we know
\be\label{eq:lte11:sec4}
|\delta_x^+\eta_j^n|\lesssim h^2\|\psi_{xxxxx}\|_{L^\infty}
+\tau^2\left(\|\psi_{tttx}\|_{L^\infty}+\|\psi_{ttxxx}\|_{L^\infty}\right)
\lesssim\tau^2+h^2,\quad 1\le n\le
\frac{T}{\tau}-1.
\ee
However, for $j=0,M-1$, we derive that $\psi_{tt}$, $\psi_{xx}$ and $\psi_{ttxx}$ are all zero on the boundary $\Gamma$, and it follows that for $1\le n\le
\frac{T}{\tau}-1$ and $j=0,M-1$,
\be\label{eq:lte12:sec4}
|\delta_x^+\eta_j^n|\lesssim h\|\psi_{xxxx}\|_{L^\infty} +\tau^2\left(\|\psi_{tttx}\|_{L^\infty}+\|\psi_{ttxxx}\|_{L^\infty}\right)
\lesssim\tau+h.
\ee
(\ref{eq:lte11:sec4}) and (\ref{eq:lte12:sec4}) prove (\ref{eq:tauder1:sec4}).

In the case of $\p_nV(x)|_{\Gamma}=0$ or $u\in C([0,T];H^2_0(U))$, by differentiating (\ref{eq:sdge1d:sec4}), we can show that $\psi_{xx}|_{\Gamma}=\psi_{xxxx}|_{\Gamma}=\psi_{ttxx}|_{\Gamma}=\psi_{tt}|_{\Gamma}=0$. Then (\ref{eq:lte11:sec4}) holds for boundary case and (\ref{eq:tauder:sec4}) is correct. The proof is complete.
\end{proof}
Now, we are going to establish the estimates in Theorem \ref{thm:sifd:sec4} by mathematical induction \cite{BaoCai}.
\begin{proof}[Proof of Theorem \ref{thm:sifd:sec4}] We first prove the optimal
discrete semi-$H^1$ norm convergence rate in the case of either $\p_n V(x)|_{\Gamma}=0$ or $\psi\in C^0([0,T];H^2_0(U))$. Since $e^0={\bf 0}$, (\ref{eq:sifdo:sec4}) is true. For $n=1$, using Lemma \ref{lem:localerrsifd:sec4} and noticing $e_j^1=\psi(x_j,\tau)-\psi_j^1=-i\tau \eta_j^0$ ($j\in\calT_M^0$), we find that
\be
\|e^1\|_2+\|\delta_x^+e^1\|_2\lesssim h^2+\tau^2.
\ee
Recalling discrete Sobolev inequality which implies that
$\|e^1\|_\infty\leq C_1\|\delta_x^+e^1\|_2$, for sufficiently small $h$ and $\tau$, we derive
\be
\|\psi^1\|_\infty\leq \|e^1\|_\infty+\|\psi\|_{L^\infty}\leq M_1+1.
\ee
Now we assume that
(\ref{eq:sifdo:sec4}) is valid for all $0\le n\leq m-1\leq \frac{T}{\tau}-1$,
then we need to show that it is still valid when $n=m$.  In order to
do so, subtracting (\ref{eq:etadef:sec4}) from (\ref{eq:sifd:sec4}),
noticing (\ref{eq:sdgb1d:sec4}) and (\ref{eq:bound1:sec4}), we obtain the following
equation for the ``error" function $e^n\in X_M$:
\begin{equation}\label{eq:semi-error:sec4}
i\delta_te_{j}^{n}=\left[-\frac12\delta_x^2+V_{j}\right]\frac{e_{j}^{n+1}+e_{j}^{n-1}}{2}
+\xi_{j}^n+\eta^n_{j},\qquad j\in\calT_{M},\quad n\ge1,
\end{equation}
where $\xi^n\in X_M$ ($n\ge1$) is defined as
\be
\xi_{j}^n=\beta|\psi(x_j,t_n)|^2\psi(x_j,t_n)-\beta|\psi_{j}^n|^2
\psi_{j}^n, \qquad
j\in\calT_{M}. \label{eq:eta11:sec4}
\ee
By the assumption of mathematical induction, we have \cite{BaoCai}
\be\label{eq:eexi:sec4} \|\xi^n\|_2^2 \le
C_{2}\|e^n\|_2^2,\qquad
\|\delta_x^+\xi^n\|_2^2\leq C_{3}
\|\delta_x^+e^n\|_2^2+\|e^n\|_2^2,\qquad 1\le n\le m-1, \ee
where $C_2$ and $C_3$ are constants only depending on $M_1$ and $\beta$.

Multiplying both sides of (\ref{eq:semi-error:sec4}) by
$\overline{e_{j}^{n+1}+e_{j}^{n-1}}$ and summing all together for
$j\in\calT_M$,  taking imaginary parts, using the triangular
and Cauchy inequalities, noticing (\ref{eq:tau1:sec4}) and (\ref{eq:eexi:sec4}) , we
have for $1\le n\le m-1$
\begin{align*}
\|e^{n+1}\|_2^2-\|e^{n-1}\|_2^2&=2\tau\, {\rm
Im}\left(\xi^n+\eta^n,e^{n+1}
+e^{n-1}\right)\\
&\leq2\tau\left[\|e^{n+1}\|_2^2+\|e^{n-1}\|_2^2
+\|\eta^n\|_2^2+\|\xi^n\|_2^2\right]\nn\\
&\leq C_4\tau(h^2+\tau^2)^2+2\tau\left(\|e^{n+1}\|_2^2+\|e^{n-1}\|_2^2\right)
+2\tau C_2\|e^n\|_2^2.
\end{align*}
Summing  above inequality for $n=1,2,\ldots,m-1$, for $\tau\leq\frac14$, we get
\begin{equation}
\|e^{m}\|_2^2+\|e^{m-1}\|_2^2\leq C_4T(h^2+\tau^2)^2+C_5\tau
\sum\limits_{l=1}^{m-1}\|e^{l}\|_2^2,
\ 1\le m\le \frac{T}{\tau},
\end{equation}
with positive constants $C_4>0$ and $C_5>0$ independent of $\tau$ and $h$.
In view of the discrete Gronwall inequality \cite{Chang,Glassey,BaoCai}
 and noticing $\|e^0\|_2=0$ and $\|e^{1}\|_2\lesssim h^2+\tau^2$, we immediately
 obtain $\|e^n\|_2\lesssim h^2+\tau^2$ for $n=m$.

 For the semi-$H^1$ norm, define the discrete linear energy  for $e^n\in X_M$ as
  \be
  {\calE}(e^n)=\frac12\|\delta_x^+e^{n}\|_2+\|e^{n}\|_V^2.
  \ee
  Multiplying both sides of (\ref{eq:semi-error:sec4}) by
$\overline{e_{j}^{n+1}-e_{j}^{n-1}}$, summing over index
$j\in{\calT}_M$ and performing summation by parts, taking real part,
 we have
\be \label{eq:energy2:sec4}{\calE}(e^{n+1})-
{\calE}(e^{n-1})=-2\;{\rm
Re}\left \langle \xi^n+\eta^n,e^{n+1}-e^{n-1}\right\rangle,\qquad
1\leq n\leq m-1. \ee
Rewriting (\ref{eq:semi-error:sec4}) as \be \label{eq:en1nm1:sec4}
e_{j}^{n+1}-e_{j}^{n-1}=-2i\tau\left[\xi_{j}^n+\eta_{j}^n+\chi_{j}^n\right],
\qquad j\in {\calT}_M,\ee where $\chi^n\in X_M$ is defined as
\be\label{eq:chi11:sec4}\chi_{j}^n=\left[-\frac12\delta_x^2+V_{j}\right]\frac{e_{j}^{n+1}+e_{j}^{n-1}}{2}, \qquad j\in
{\calT}_M,\ee
then plugging (\ref{eq:en1nm1:sec4}) into (\ref{eq:energy2:sec4}), we
obtain
\begin{eqnarray} \label{eq:eng367:sec4}
{\calE}(e^{n+1})-
{\calE}(e^{n-1})&=&-4\tau \,{\rm
Im}\left\langle \xi^n+\eta^n, \xi^n
+\eta^n+\chi^n\right\rangle\nn\\
&=&-4 \tau\,{\rm Im}\left\langle \xi^n+\eta^n,\chi^n\right\rangle,
\quad 1\leq n\leq m-1.
\end{eqnarray}
From (\ref{eq:chi11:sec4}), (\ref{eq:eta11:sec4}) and summation by parts, we have
\bea
\label{eq:ett167:sec4}\left|\left\langle
\xi^n,\chi^n\right\rangle\right|&=&\frac{1}{2}\left|\left\langle
\xi^n,\left(-\frac12\delta_x^2+V\right)\left(e^{n+1}+e^{n-1}\right)\right\rangle\right|\nn\\
&\lesssim&\left|\left\langle
\delta_x^+\xi^n,\delta_x^+\left(e^{n+1}+e^{n-1}\right)\right\rangle\right|
+\left|\left\langle
\xi^n,V\left(e^{n+1}+e^{n-1}\right)\right\rangle\right|\nn\\
&\lesssim&\|\delta^+_x e^{n+1}\|_{2}^2+\|\delta^+_x
e^n\|^2_{2}+\|\delta^+_x
e^{n-1}\|_2^2+\|e^{n+1}\|_2^2+\|e^n\|_2^2+\|e^{n-1}\|_2^2\nn\\
&&+\|\delta^+_x \xi^{n}\|_{2}^2+\|\xi^n\|_2^2\nn\\
&\lesssim&\|\delta^+_x e^{n+1}\|_{2}^2+\|\delta^+_x
e^n\|^2_{2}+\|\delta^+_x e^{n-1}\|_2^2,\qquad 1\le n\le
m.\qquad
 \eea
Similarly, noticing (\ref{eq:eexi:sec4}), (\ref{eq:tau1:sec4}) and (\ref{eq:tauder:sec4}),
we have \begin{align}\label{eq:ett168:sec4} \left|\left\langle
\eta^n,\chi^n\right\rangle\right|&=\frac{1}{2}\left|\left\langle
\eta^n,\left(-\frac12\delta_x^2+V\right)\left(e^{n+1}+e^{n-1}\right)\right\rangle\right|\\
&\lesssim\left|\left\langle
\delta_x^+\eta^n,\delta_x^+\left(e^{n+1}+e^{n-1}\right)\right\rangle\right|
+\left|\left\langle
\eta^n,V\left(e^{n+1}+e^{n-1}\right)\right\rangle\right|\nn\\
&\lesssim\|\delta^+_x e^{n+1}\|_{2}^2+\|\delta^+_x
e^n\|^2_{2}+\|\delta^+_x
e^{n-1}\|_2^2+\|e^{n+1}\|_2^2+\|e^n\|_2^2+\|e^{n-1}\|_2^2\nn\\
&\quad+\|\delta^+_x\eta^{n}\|_{2}^2+\|\eta^n\|_2^2\nn\\
 &\lesssim\|\delta^+_x
e^{n+1}\|_{2}^2+\|\delta^+_x e^n\|^2_{2}+\|\delta^+_x
e^{n-1}\|_2^2+(\tau^2+h^2)^2, \quad 1\le n\le
m.\nn
 \end{align}
Plugging (\ref{eq:ett167:sec4}) and (\ref{eq:ett168:sec4}) into (\ref{eq:eng367:sec4}),
 we get \begin{align*}
\mathcal{E}(e^{n+1})-\mathcal{E}(e^{n-1})
&\lesssim\tau(\tau^2+h^2)^2+\tau\left[\|\delta^+_x
e^{n+1}\|_{2}^2+\|\delta^+_x e^n\|^2_{2}+\|\delta^+_x
e^{n-1}\|_2^2\right]
\\
&\lesssim\tau(\tau^2+h^2)^2+\tau\left[\mathcal{E}(e^{n+1})
+\mathcal{E}(e^{n})+\mathcal{E}(e^{n-1})\right], \ 1\le n\le
m.
 \end{align*}
 Summing  above inequality for $1\le n\le m-1$, we get \begin{equation*}
\mathcal{E}(e^{n+1})+\mathcal{E}(e^{n})\lesssim T(\tau^2+h^2)^2 +
\mathcal{E}(e^{1}) +\mathcal{E}(e^{0}) +\tau \sum_{l=1}^{n+1}
\mathcal{E}(e^{l}), \qquad 1\le n\le m-1.\end{equation*}  Using the
discrete Gronwall inequality \cite{BaoCai,Chang}, we have \be
\mathcal{E}(e^{n+1})+\mathcal{E}(e^{n})\lesssim
(\tau^2+h^2)^2+\mathcal{E}(e^{1})+\mathcal{E}(e^{0})
\lesssim(\tau^2+h^2)^2, \qquad 1\le n\le m-1. \ee
Thus $\calE(e^{m})\lesssim (\tau^2+h^2)^2$ and
\be
\|\delta_x^+e^m\|_2^2\leq \mathcal{E}(e^{m})+\|V\|_{L^\infty(U)}\|e^m\|_2^2\lesssim
(\tau^2+h^2)^2.
\ee
 In view of the discrete Sobolev inequality, we get
 \be
 \|e^m\|_\infty\lesssim \|\delta_x^+e^m\|_2\lesssim \tau^2+h^2.
 \ee
 Noticing that in all the above inequalities, the appearing  constants are independent of $h$ and $\tau$. Hence, for sufficiently small $\tau$ and $h$, we conclude that
 \be
 \|\psi^m\|_\infty\leq \|\psi(\cdot,t_m)\|_{L^\infty(U)}+\|e^m\|_{\infty}\leq M_1+1.
 \ee
 This completes the proof of (\ref{eq:sifdo:sec4}) at $n=m$. Therefore the  result is proved by  mathematical induction.

 For the case of assumption (\ref{assum:fd:sec4}) and $V\in C^1$ without further assumptions, we will lose half
order convergence rate in the semi-$H^1$ norm because of the boundary (\ref{eq:tauder1:sec4}).
Notice that the reminder term is $O(h^2+\tau^2)^{3/2}$ instead of
$O(h^2+\tau^2)$ in (\ref{eq:ett168:sec4}), and that the  remaining proof
is the same. Hence,   we will have the $3/2$ order convergence
rate for  discrete semi-$H^1$ norm.  The proof is complete.
 \end{proof}

 \begin{remark}\label{rmk:extension:sec4}
 Here we emphasis that the above approach can be extended to the
 higher dimensions, e.g. 2D and 3D, directly.
 The key point is the discrete Sobolev inequality in 2D and 3D as
 \be
 \|u_h\|_\infty\leq C|\ln h|\,\|u_h\|_{H^1_s},
 \qquad \|v_h\|_\infty\leq Ch^{-1/2}\|v_h\|_{H^1_s},
 \ee
 where $u_h$ and $v_h$  are 2D and 3D mesh functions with zero at the boundary,
  respectively,  and the discrete semi-$H^1$ norm $\|\cdot\|_{H_s^1}$ and $l^\infty$ norm
  $\|\cdot\|_\infty$ can be defined similarly as the discrete semi-$H^1$ norm
  and the $l^\infty$ norm  in (\ref{eq:fdnorm:sec3}). The same proof  works in 2D
  and 3D, with the above Sobolev
  inequalities and the  additional technical assumption
 $\tau=o(1/|\ln h|)$ in 2D and $\tau =o(h^{1/3})$ in 3D.
 \end{remark}

\subsubsection{Convergence rate for CNFD}
Let $\psi^n\in X_M$ be the
numerical solution of the CNFD (\ref{eq:cnfd:sec4})  and $e^n\in X_M$ be the error
function.

\begin{lemma} (Conservation of mass and energy)
\label{eq:concnfd:sec4} For the CNFD scheme (\ref{eq:cnfd:sec4})  with
(\ref{eq:bound1:sec4}) and (\ref{eq:init1:sec4}), for any mesh size $h>0$, time step
$\tau>0$ and initial data $\psi_0$, it conserves the mass and energy
in the discretized level, i.e. \be\label{eq:mass} \|\psi^n\|_2^2\equiv
\|\psi^0\|_2^2, \qquad E_h(\psi^n)\equiv E_h(\psi^0), \qquad
n=0,1,2,\ldots\;,
\end{equation}
where the discrete energy is given by
\be\label{eq:denergy:sec4}
E_h(\psi^n)=\frac12\|\delta_x^+\psi^n\|_2^2+\|\psi^n\|_V^2+\frac{\beta}{2}\|\psi^n\|_4^4.
\ee
\end{lemma}
\begin{proof}  Follow the analogous arguments of the  CNFD
method for the NLSE \cite{Chang1,Glassey} and we omit the details here
for brevity. \end{proof}
For the solvability, we have
\begin{lemma} (Solvability of the difference equations)
\label{eq:solve:sec4} For any given $\psi^n$, there exists a unique solution
$\psi^{n+1}$ of the CNFD discretization (\ref{eq:cnfd:sec4})  with
(\ref{eq:bound1:sec4}).
\end{lemma}
\begin{proof}The proof is standard \cite{Akrivis,BaoCai}.  In higher dimensions (2D and 3D),
additional assumption is needed for uniqueness, i.e., time step
$\tau$ is sufficiently small compared with mesh size.
\end{proof}
Denote the local truncation error $\widetilde{\eta}^n\in X_M$
($n\ge0$) of the CNFD scheme (\ref{eq:cnfd:sec4})  with (\ref{eq:bound1:sec4}) and
(\ref{eq:init1:sec4}) as
\begin{eqnarray}\label{eq:taudef:sec4}
\widetilde{\eta}_{j}^n:&=&i\delta_t^+\psi(x_j,t_n)-
\left[-\frac12\delta_x^2+V_{j}+\frac{\beta}{2}
\left(|\psi(x_j,t_{n+1})|^2+|\psi(x_j,t_n)|^2\right)\right]\nn \\
&&
\times\frac{\psi(x_j,t_n)+\psi(x_j,t_{n+1})}{2}, \qquad
 j\in\calT_M.
\end{eqnarray}
Similar to Lemma \ref{lem:localerrsifd:sec4}, we have
\begin{lemma} (Local truncation error)
\label{eq:localerr1:sec4} Assume $V(\bx)\in C(\bar{U})$ and under
assumption (\ref{assum:fd:sec4}), we have
 \be\label{eq:tauhat1:sec4}
\|\widetilde{\eta}^n\|_\infty\lesssim \tau^2+h^2, \qquad  0\le n
\leq \frac{T}{\tau}-1. \ee In addition, assuming $V(\bx)\in C^1(\bar{U})$,
we have for $1\le n \leq \frac{T}{\tau}-1$
\be\label{eq:tauhatder1:sec4}
 |\delta_x^+\widetilde{\eta}_{j}^n|\lesssim\begin{cases}
\tau^2+h^2,& 1\leq j\leq M-2,\\
  \tau+h, &j=0,M-1.
 \end{cases}
 \ee
 In addition,  if either  $\p_n V(x)|_{\Gamma}=0$ or $\psi\in C^0([0,T];H^2_0(U))$, we have
\be\label{eq:tauhatder:sec4} \|\delta_x^+\widetilde{\eta}^n\|_\infty\lesssim
\tau^2+h^2,\qquad 1\le n \leq \frac{T}{\tau}-1. \ee
\end{lemma}

One main difficulty in deriving error bounds for CNFD  in high dimensions is the
$l^\infty$ bounds for the finite difference  solutions. In \cite{BaoCai,Thomee,Akrivis},
this difficulty was overcome  by truncating the nonlinearity  to a
global Lipschitz function with compact support in $d$-dimensions ($d=1,2,3$).
This cutoff would not change  $\psi(x,t)$ and $\psi^n$  if the continuous solution  $\psi(x,t)$ is bounded and the numerical solution   $\psi^n$ is
close to the
continuous solution, i.e., if (\ref{eq:cnfdes:sec4}) or (\ref{eq:cnfdo:sec4}) holds.
\begin{proof}[Proof of Theorem \ref{thm:cnfd:sec4}] As in the proof of Theorem \ref{thm:sifd:sec4}, we only prove the optimal convergence under  assumptions (\ref{assum:fd:sec4}) with either   $\p_nV(x)|_{\Gamma}=0$ or $\psi\in C^0([0,T];H^2_0(U))$. Choose  a smooth function $\rho(s)\in C^\infty(\Bbb R)$  such that
\be\label{eq:rho:sec4}
\rho(s)=\left\{\begin{array}{ll} 1,
&0\le |s|\leq 1\,,\\
\in[0,1], &1\le |s|\le 2\,,\\
0, &|s|\ge2\,.\\
\end{array}\right.
\ee
Let us denote $B=(M_1+1)^2$ and denote
\be\label{eq:fbdef:sec4}
 f_{_B}(s)=\rho(s/B)s,\qquad s\in\Bbb R,
\ee
then $f_{_B}\in C_0^\infty(\Bbb R)$.
Choose $\phi^0=\psi^0\in X_M$ and define
$\phi^n\in X_M$ ($n\ge1$) as
 \begin{equation}\label{eq:appro:sec4}
i\delta_t^+\phi_{j}^{n}=\left[-\frac12\delta_x^2+V_{j}
+\frac{\beta}{2}\left(f_{_B}(|\phi_{j}^{n+1}|^2)
+f_{_B}\left(|\phi_{j}^{n}|^2\right)\right)
\right]\phi_{j}^{n+1/2},\quad j\in{\calT}_M,
 \end{equation}
where
\be
\phi^{n+1/2}_{j}=\frac{1}{2}(\phi^{n+1}_{j}+\phi^n_{j}),\qquad
j\in{\calT}_M^0, \quad n\ge0.\ee In fact, $\phi^n$ can be viewed
as another approximation of $\psi(x,t_n)$. Define the `error'
function $\hat{e}^n\in X_M$ ($n\ge0$) with components $\hat{e}_j^n=\psi(x_j,t_n)-\phi_j^n$ for $j\in {\calT}_M^0$ and
denote the local truncation error $\hat{\eta}^n\in X_M$ as
\begin{eqnarray}\label{eq:tautil:sec4}
\hat{\eta}_{j}^n&:=&i\delta_t^+\psi(x_j,t_n)-\left[-\frac12\delta_x^2+V_{j}
+\frac{\beta}{2}\left(f_{_B}(|\psi(x_j,t_{n+1})|^2)+f_{_B}(|\psi(x_j,t_n)|^2)\right)\right]
\nn\\&&\times\frac{\psi(x_j,t_n)+\psi(x_j,t_{n+1})}{2}, \quad
 j\in\calT_M,\ n\ge0.
\end{eqnarray}
Similar as Lemma \ref{eq:localerr1:sec4}, we can prove
\be \|\hat{\eta}^n\|_\infty+\|\delta_x^+\hat{\eta}^n\|_\infty\lesssim \tau^2+h^2,
\quad 0\le n\le \frac{T}{\tau}-1.\ee
Subtracting (\ref{eq:tautil:sec4}) from (\ref{eq:appro:sec4}), we obtain
\begin{eqnarray}\label{eq:erreq1:sec4}
i\delta^+_t\hat{e}_{j}^n&=&\left[-\frac12\delta_x^2+V_{j}
\right]\hat{e}_{j}^{n+1/2}+\hat{\xi}_{j}^n+\hat{\eta}_{j}^n,\quad
j\in\calT_M,\ n\ge0,
\end{eqnarray}
where $\hat{\xi}^n\in X_M$ defined as
\begin{align}
\hat{\xi}_{j}^n=&\frac{\beta}{2}\left(|\psi(x_j,t_{n+1})|^2
+|\psi(x_j,t_n)|^2\right)\frac{\psi(x_j,t_{n+1})
+\psi(x_j,t_n)}{2}\\&-\frac{\beta}{2}\left(f_{_B}(|\phi_{j}^{n+1}|^2)
+f_{_B}(|\phi_{j}^n|^2)\right)\phi_j^{n+1/2}, \ j\in\calT_M^0.
\end{align}
Recalling that $f_{_B}\in C_0^\infty$, for $0\leq n\leq\frac{T}{\tau}-1$,
we can show that \cite{BaoCai}
\be\label{eq:lipcn:sec4}
\|\hat{\xi}^n\|_2\lesssim \sum\limits_{k=n,n+1}
\|\hat{e}^k\|_2,\quad
\|\delta_x^+\hat{\xi}^n\|_2\lesssim \sum\limits_{k=n,n+1}\left(\|\hat{e}^k\|_2+
\|\delta_x^+\hat{e}^k\|_2\right).
\ee
This property is similar to that of the SIFD case (\ref{eq:eexi:sec4}). The same proof for Theorem \ref{thm:sifd:sec4} works. However, there is no need to use mathematical induction here, as $C_0^\infty$ property of $f_{_B}$ guarantees (\ref{eq:lipcn:sec4}). For simplicity, we omit the details here. Finally, we  derive the following for sufficiently small $h$ and $\tau$,
\be\label{eq:escnfd:sec4}
\|\hat{e}^n\|_2+\|\delta_x^+\hat{e}^n\|_2\lesssim h^2+\tau^2,\quad \|\phi^n\|_\infty\leq M_1+1,\quad 0\leq n\leq \frac{T}{\tau}.
\ee
From (\ref{eq:escnfd:sec4}), we know that (\ref{eq:appro:sec4}) collapses to CNFD (\ref{eq:cnfd:sec4}), i.e., $\psi^n=\phi^n$. Thus, we prove error estimates (\ref{eq:cnfdo:sec4}) for CNFD (\ref{eq:cnfd:sec4}).

Again, for the case of assumption (\ref{assum:fd:sec4}) and $V\in C^1$ without further assumptions, we will lose half order convergence rate in the semi-$H^1$ norm.
\end{proof}
\begin{remark} If the cubic nonlinear term
$\beta|\psi|^2\psi$ in (\ref{eq:sdge1d:sec4}) is replaced by a general
nonlinearity $f(|\psi|^2)\psi$, the numerical discretization CNFD
and its error estimates in $l^2$-norm,  $l^\infty$-norm and discrete
$H^1$-norm are still valid provided that the nonlinear real-valued
function $f(\rho)\in C^3([0,\infty))$. The higher dimensional case (2D or 3D) is the same as Remark \ref{rmk:extension:sec4}.
\end{remark}

\subsection{Error estimates for TSSP}
From now on, we investigate the error bounds for time-splitting method. In the last decade, there have been many studies  on the analysis of the splitting error  for Schr\"odinger equations \cite{Besse,Lubich,Thalhammer,Neuhauser,Debussche}. For NLSE, Besse et al. obtained order estimates for the Strang splitting error \cite{Besse}. Later, Lubich introduced  formal Lie derivatives to estimate the Strang splitting error \cite{Lubich}. The formal Lie calculus enables a systematical approach for studying splitting schemes.

In our consideration of the TSSP (\ref{eq:tssp:sec4}) for solving the GPE (\ref{eq:sdge1d:sec4}), we will restrict ourselves to certain subspaces of $H_0^1(U)$. Let $\phi(x)\in H^m(U)\cap H_0^1(U)$ be represented in sine series as
\be
\phi(x)=\sum\limits_{l=1}^{+\infty}\widehat{\phi}_l\sin(\mu_l(x-a)),\qquad x\in U=(a,b),
\ee
with $\widehat{\phi}_l$ given in (\ref{eq:sineps:sec3}), define the subspace $H_{\sin}^m(U)\subset H^m\cap H_0^1$  equipped with the norm
\be
\|\phi\|_{H_{\sin}^m(U)}=\left(\sum\limits_{l=1}^\infty \mu_l^{2m}|\widehat{\phi}_l|^2\right)^{\frac12},
\ee
which is equivalent to the $H^m$ norm in this subspace. We notice that $H_{\sin}^m(U)=\{\phi\in H^m(U)| \p_x^{2k}\phi(a)=\p_x^{2k}\phi(b)=0,\quad 0\leq 2k<m,\,k\in\Bbb Z\}$.   It is easy to see that $e^{it\Delta}$ will preserve the $H_{\sin}^m$ norm.

The TSSP (\ref{eq:tssp:sec4}) can be thought as the full discretization of the following semi-discretization scheme. Let $\psi^{[n]}(x)$ be the numerical approximation of $\psi(x,t_n)$. From time $t=t_n$ to $t=t_{n+1}$, we use the standard Strang splitting:
\be\label{eq:tsspsemi:sec4}\begin{split}
&\psi^{\{1\}}(x)=
  e^{i\tau\Delta/4}\psi^{[n]}(x),\quad
  \psi^{\{2\}}(x)=e^{-i(V(x)+\beta |\psi^{\{1\}}(x)|^2)\tau}\;\psi^{\{1\}}(x), \\
&\psi^{[n+1]}(x)=e^{i\tau\Delta/4}\psi^{\{2\}}(x),
 \qquad x\in U.
\end{split}
\ee
In order to guarantee that $e^{-i\tau(V+\beta|\phi|^2)}\phi$ is a flow in $H_{\sin}^m$, we make the following assumptions on the potential $V(x)$
\be\label{assum:pot:sec4}
V(x)\in H^m(U) \quad\text{and}\quad\p_xV(x)\in H_{\sin}^{m-1}(U),\quad m\ge1.
\ee
To derive the optimal error bounds, we make the following assumption
on the exact solution $\psi$, i.e., let $0<T<T_{\rm max}$ with
$T_{\rm max}$ the maximal existing time of the solution \cite{Cazenave,Sulem}:
\be \label{assum:sp:sec4}
\psi\in C\left([0,T];H^{m}(U)\cap
H_0^1(U)\right).
\ee
It is easy to get that $\psi$ is actually in $H_{\sin}^m$ under assumptions (\ref{assum:pot:sec4}) and (\ref{assum:sp:sec4}), by showing that $\p_x^{2k}\psi|_{\Gamma}=0$ ($0\leq 2k<m$) from the GPE (\ref{eq:sdge1d:sec4}).

Now, we could state the error estimates for the TSSP (\ref{eq:tssp:sec4}).
\begin{theorem}\label{thm:tssp:sec4} Let $\psi^n\in X_M$ be the numerical approximation obtained by the TSSP (\ref{eq:tssp:sec4}). Under assumptions (\ref{assum:pot:sec4}) and (\ref{assum:sp:sec4}),
there exist constants $0<\tau_0,h_0\leq1$, such that if $0<h\leq h_0$, $0<\tau\leq\tau_0$ and $m\ge5$, we have
\begin{equation}\label{eq:tsspes:sec4}\begin{split}
&\|\psi(x,t_n)-I_M(\psi^n)(x)\|_{L^2(U)}\lesssim h^{m}+\tau^2,\qquad \|\psi^n\|_\infty\leq M_1+1,
\\ &\|\nabla (\psi(x,t_n)
-I_M(\psi^n)(x))\|_{L^2(U)}\lesssim h^{m-1}+\tau^2,\qquad 0\leq n\leq\frac{T}{\tau},
\end{split}\end{equation}
where the interpolation operator $I_M$ is given in (\ref{eq:interpolation:sec3}) and $M_1=\max_{t\in[0,T]}\|\psi(\cdot,t)\|_{L^\infty}$.
\end{theorem}
By Parseval's identity, we can easily get the following.
\begin{lemma}(Conservation of mass) The TSSP (\ref{eq:tssp:sec4}) conserves the total mass, i.e.,
\be
\|\psi^n\|_2^2=\|\psi^0\|_2^2,\qquad n\ge1.
\ee
\end{lemma}
The proof of Theorem \ref{thm:tssp:sec4} is separated into two steps. The first step is to establish the error estimates for semi-discretization (\ref{eq:tsspsemi:sec4}). Then we analyze the error between the semi-discretization (\ref{eq:tsspsemi:sec4}) and the full discretization TSSP (\ref{eq:tssp:sec4}).

The $H^1$ error bound for the semi-discretization (\ref{eq:tsspsemi:sec4}) is the following.
\begin{theorem} \label{thm:tsspsemi:sec4} Let $\psi^{[n]}(x)$ be the solution given by the splitting scheme (\ref{eq:tsspsemi:sec4}). Under assumption (\ref{assum:pot:sec4}) and (\ref{assum:sp:sec4}) with $m\ge5$, we have $\psi^{[n]}(x)\in H_{\sin}^m$ and
\be\label{eq:tsspsemi:sec47}
\|\psi^{[n]}(x)-\psi(x,t_n)\|_{L^2}+\|\nabla(\psi^{[n]}(x)-\psi(x,t_n))\|_{L^2}\lesssim \tau^2,\quad 0\leq n\leq\frac{T}{\tau}.
\ee
 \end{theorem}
The semi-discretization (\ref{eq:tsspsemi:sec4}) can be simplified as
\be
\psi^{[n+1]}(x)=S_\tau(\psi^{[n]}(x)).
\ee
Using the fact $L^\infty(U)\subset H^1(U)$ and following \cite{Lubich}, we would easily obtain the $H^1_0$ stability of the splitting (\ref{eq:tsspsemi:sec4}).
\begin{lemma}\label{lem:sta:sec4} ($H_0^1$-conditional $L^2$- and $H_0^1$ stability). If $\psi,\phi\in H_{0}^1(U)$ with
\be
\|\psi\|_{H_0^1}\leq M_2,\quad \|\phi\|_{H^0_1} \leq M_2,
\ee
then we have
\be\begin{split}
\|S_\tau(\psi)-S_\tau(\phi)\|_{L^2}\leq e^{c_0\tau}\|\psi-\phi\|_{L^2},\\
\|S_\tau(\psi)-S_\tau(\phi)\|_{H^1_0}\leq e^{c_1\tau}\|\psi-\phi\|_{H^1_0},
\end{split}
\ee
where $c_0$ and $c_1$ depend on $M_2$, $V(x)$ and $\beta$.
\end{lemma}

To prove Theorem \ref{thm:tsspsemi:sec4}, we need  the local error, which is the key point in analyzing time-splitting methods.
\begin{lemma}\label{lem:tssploc:sec4} If $\psi_0\in H_{\sin}^5$, then the error after one step of (\ref{eq:tsspsemi:sec4}) in $H^1_0$ norm is given by
\be
\|\psi^{[1]}(x)-\psi(x,\tau)\|_{L^2}+\|\nabla(\psi^{[1]}(x)-\psi(x,\tau))\|_{L^2}\leq C \tau^3,
\ee
where $C$ only depends on $\|\psi_0\|_{H_{\sin}^5}$, $V(x)$ and $\beta$.
\end{lemma}
We use formal Lie derivative calculus to study the local error in Lemma \ref{lem:tssploc:sec4}. For a general differential equation $\phi_t=F(\phi)$ ($\phi\in H_0^1$), denote the evolution operator $\varphi_F^t(v)$ as the solution at time $t$ with initial value $\phi(0)=v$. The Lie derivative $D_F$ is given by \cite{Lubich}
\be
(D_FG)(v)=\frac{\rd}{\rd t}G(\varphi_F^t(v))\big|_{t=0}=G^\prime(v) F(v),\quad v\in H_0^1(U),
\ee
where $G$ is a vector field on $H^1_0$. Let $\hat{T}(\psi)=\frac{i}{2}\Delta \psi$, $\hat{V}(\psi)=-i\left(V(x)+\beta|\psi|^2\right)\psi$ and $\hat{H}=\hat{T}+\hat{V}$, and denote $D_T$, $D_V$ and $D_H$ as the corresponding Lie derivatives (cf. \cite{Lubich}) for $\hat{T}$, $\hat{V}$ and $\hat{H}$, respectively. Similar to \cite{Lubich}, one can compute the commutator
\begin{align*}
[\hat{T},\hat{V}](\psi)=&\;\hat{T}^\prime(\psi)\hat{V}(\psi)-\hat{V}^\prime(\psi)\hat{T}(\psi)\\
=&\;\frac{i}{2}(-i)\Delta(V(x)\psi+\beta|\psi|^2\psi)-\left[-iV\Delta\psi\frac{i}{2}
-i\beta\left(2|\psi|^2\Delta\psi\frac{i}{2}-\psi^2\Delta\bar{\psi}\frac{i}{2}\right)\right]\\
=&\; \frac12\psi\Delta V+\nabla V\cdot\nabla\psi+\beta \psi^2\Delta\bar{\psi}+\frac32\beta|\nabla\psi|^2\psi
+\frac{\beta}{2}\bar{\psi}\nabla\psi\cdot\nabla\psi.
\end{align*}
Under assumption (\ref{assum:pot:sec4}) on potential $V(x)$ with $m\ge5$, following analogous arguments in \cite{Lubich}, we have
\begin{align}
&\|[\hat{T},\hat{V}](\psi)\|_{H^1_0}\leq C\|\psi\|_{H_{\sin}^3}(1+\|\psi\|_{H_{\sin}^3}^2),\\
&\|[\hat{T},[\hat{T},\hat{V}]](\psi)\|_{H^1_0}\leq C\|\psi\|_{H_{\sin}^5}(1+\|\psi\|_{H_{\sin}^5}^2).
\end{align}
Now, we can prove the local error.
\begin{proof}[Proof of Lemma \ref{lem:tssploc:sec4}] The proof is analogous to that in section 5 of \cite{Lubich}. Here, we outline the main part. First of all, by using variation of constant formula (or Duhamel's principle), one can write the error as
\be
\psi^{[1]}(x)-\psi(x,\tau)=\tau f\left(\frac{\tau}{2}\right)-\int_0^\tau f(s)ds+r_2-r_1,
\ee
where $f(s)=\exp((\tau-s)D_{T})D_V \exp(sD_T){\rm Id}(\psi_0)$ ($\rm{Id}$ the identity operator), and the remainder terms,
\begin{align*}
&r_1=\int_0^\tau\int_0^{\tau-s}\exp((\tau-s-\sigma)D_H)D_V
\exp(\sigma D_T)D_V \exp(s D_T){\rm Id}(\psi_0)\,d\sigma ds,\\
&r_2=\tau^2\int_0^1(1-\theta)\exp(\frac{\tau}{2}D_T)\exp(\theta\tau D_T)D_V^2\exp(\frac{\tau}2D_T){\rm Id}(\psi_0)\,d\theta.
\end{align*}
For the principal part, we notice that the midpoint rule quadrature leads to
\be
\tau f\left(\frac{\tau}{2}\right)-\int_0^\tau f(s)ds=\tau^3\int_0^1{\rm ker}(\theta)f^{\prime\prime}(\theta\tau)\,d\theta,
\ee
where ${\rm ker}(\theta)$ is the Peano kernel for midpoint rule and \cite{Lubich}
\be
f^{\prime\prime}(s)=e^{i\frac{\tau}{2}\Delta}
[\hat{T},[\hat{T},\hat{V}]](e^{i\frac{\tau-s}{2}\Delta}\psi_0).
\ee
Hence the principal part is of order $O(\tau^3)$. For $r_2-r_1$, denote
\be
g(s,\sigma)=\exp((\tau-s-\sigma)D_T)D_V
\exp(\sigma D_T)D_V \exp(s D_T){\rm Id}(\psi_0),
\ee
then we have
\be
r_2-r_1=\frac{\tau^2}{2}g\left(\frac{\tau}{2},0\right)-\int_0^\tau\int_0^{\tau-s}g(s,\sigma)\,d\sigma ds
+\tilde{r}_2-\tilde{r}_1,
\ee
with
\be
\tilde{r}_2=r_2-\frac{\tau^2}{2}g\left(\frac{\tau}{2},0\right),\quad
\tilde{r}_1=r_1-\int_0^\tau\int_0^{\tau-s}g(s,\sigma)\,d\sigma ds.
\ee
For $\tilde{r}_1$, noticing that \begin{equation*}
\exp(\tau D_H){\rm Id} (\psi_0)=\exp(\tau D_T){\rm Id}(\psi_0)+\int_0^\tau \exp((\tau-s)D_H)D_V\exp(sD_T){\rm Id} (\psi_0),\end{equation*}
we can derive from the definition of $g(s,\sigma)$ and the form of $r_1$ that
\be
\|\tilde{r}_1\|_{H_0^1}\leq \tilde{C}_1 \tau^3.
\ee
Similarly, we get $\|\tilde{r}_2\|_{H_0^1}\leq \tilde{C_2}\tau^3$. Here, $\tilde{C}_1$ and $\tilde{C}_2$ only depend on $\|\psi_0\|_{H_{\sin}^3}$, $V$ and $\beta$. The  remainder term is also a quadrature rule and it follows that
\be
\left\|\frac{\tau^2}{2}g\left(\frac{\tau}{2},0\right)-\int_0^\tau\int_0^{\tau-s}g(s,\sigma)\,d\sigma ds\right\|_{H_0^1}
\leq C_r\tau^3,
\ee
where $C_r$ depends on $\|\psi_0\|_{H_{\sin}^3}$, $V$ and $\beta$. Combining the above results together, we can get Lemma \ref{lem:tssploc:sec4}.
\end{proof}
Theorem \ref{thm:tsspsemi:sec4} can be proved by a combination of Lemmas \ref{lem:sta:sec4} and \ref{lem:tssploc:sec4}, using induction \cite{Lubich}, and we omit this part here.

\begin{proof}[Proof of Theorem \ref{thm:tssp:sec4}] Having Theorem \ref{thm:tsspsemi:sec4}, we only need to  compare the full discretization solution $I_M(\psi^n)(x)$ and the semi-discretization solution $\psi^{[n]}(x)$ ($0\leq n\leq \frac{T}{\tau}$),
\be
I_M(\psi^n)(x)-\psi^{[n]}(x)=I_M(\psi^n)(x)-P_M(\psi^{[n]}(x))+P_M(\psi^{[n]}(x))-\psi^{[n]}(x).
\ee
From Theorem \ref{thm:tsspsemi:sec4},  there exists some $M_2>0$ such that $\|\psi^{[n]}\|_{H_{\sin}^m}\leq M_2$. Hence
\be
\|P_M(\psi^{[n]}(x))-\psi^{[n]}(x)\|_{L^2}\lesssim h^m,\quad \|\nabla[P_M(\psi^{[n]}(x))-\psi^{[n]}(x)]\|_{L^2}\lesssim h^{m-1}.
\ee
Denote $e^n_I(x)=I_M(\psi^n)(x)-P_M(\psi^{[n]}(x))$, then $e^0_I=I_M(\psi^0)(x)-P_M(\psi^{[0]})(x)$ and
\be
\|e^0_I(x)\|_{L^2}\lesssim h^m,\qquad \|\nabla e^0_I(x)\|_{L^2}\lesssim h^{m-1},\quad \| e^0_I(x)\|_{H^2}\lesssim h^{m-2}.
\ee
From the semi-discretization (\ref{eq:tsspsemi:sec4}), there holds
\begin{equation*}\begin{split}
&P_M(\psi^{\{1\}})(x)=
  e^{i\tau\Delta/4}P_M(\psi^{[n]}),\quad
  P_M(\psi^{\{2\}})(x)=P_M(e^{-i(V(x)+\beta |\psi^{\{1\}}|^2)\tau}\;\psi^{\{1\}})(x),\\
&P_M(\psi^{[n+1]})(x)=P_M(e^{i\tau\Delta/4}\psi^{\{2\}})(x),
 \qquad x\in U.
 \end{split}
\end{equation*}
Similarly, for the TSSP (\ref{eq:tssp:sec4}), there holds
\begin{equation*}\begin{split}
&I_M(\psi^{(1)})(x)=
  e^{i\tau\Delta/4}I_M(\psi^n)(x),\quad
  I_M(\psi^{(2)})(x)=I_M(e^{-i(V(x)+\beta |\psi^{(1)}(x)|^2)\tau}\;\psi^{(1)})(x),\\
&I_M(\psi^{n+1})(x)=e^{i\tau\Delta/4}I_M(\psi^{(2)})(x),
 \qquad x\in U.
 \end{split}
 \end{equation*}
 Noticing that $e^{i\tau\Delta}$ preserves $H_{\sin}^m$ norm, we find that
\begin{align*}
&\|e_I^n(x)\|_{H_{\sin}^m}=\|I_M(\psi^{(1)})(x)-P_M(\psi^{\{1\}})(x)\|_{H_{\sin}^m},\\
&\|e_I^{n+1}(x)\|_{H_{\sin}^m}=\|I_M(\psi^{(2)})(x)-P_M(\psi^{\{2\}})(x)\|_{H_{\sin}^m}.
\end{align*}
Following the analogous mathematical induction for SIFD (Theroem \ref{thm:sifd:sec4}),  we can assume that
error estimates (\ref{eq:tsspes:sec4}) holds for $n\leq \frac{T}{\tau}-1$. For $n+1$, using the  techniques and results in \cite{BaoCai3}, we have
\begin{align*}
&\|I_M(\psi^{(2)})(x)-P_M(\psi^{\{2\}})(x)\|_{L^2}
\lesssim \tau\|I_M(\psi^{(1)})(x)-P_M(\psi^{\{1\}})(x)\|_{L^2}+\tau h^m,\\
&\|I_M(\psi^{(2)})(x)-P_M(\psi^{\{2\}})(x)\|_{H^1_0}
\lesssim \tau \|I_M(\psi^{(1)})(x)-P_M(\psi^{\{1\}})(x)\|_{H^1_0}+\tau h^{m-1}.
\end{align*}
Hence for $n\leq\frac{T}{\tau}-1 $,
\be
\|e_I^{n+1}(x)\|_{L^2}\lesssim \tau \|e^n_I(x)\|_{L^2}+\tau h^m,\quad \|e_I^{n+1}(x)\|_{H_0^1}\lesssim \tau \|e^n_I(x)\|_{H_0^1}+\tau h^{m-1}.
\ee
Then, mathematical induction and discrete Gronwall inequality would imply that for all $n\leq \frac{T}{\tau}-1$ and small $\tau$,
\be
\|e_I^{n+1}(x)\|_{L^2}\lesssim h^m+\tau^2,\quad \|e_I^{n+1}(x)\|_{H_0^1}\lesssim h^{m-1}+\tau^2.
\ee
Hence $\|I_M(\psi^n)(x)-\psi^{[n]}(x)\|_{H_0^1}\lesssim h^{m-1}+\tau^2$, and discrete Sobolev inequality gives that $\|\psi^n\|_\infty\leq M_1+1$ ($n\leq\frac{T}{\tau}$) (cf. proof of Theorem \ref{thm:sifd:sec4}). This would complete the proof of Theorem \ref{thm:tssp:sec4}.
\end{proof}
\begin{remark}Similar as Remark \ref{rmk:extension:sec4}, results in Theorem \ref{thm:tssp:sec4} can be extended to higher dimensions (2D and 3D) and the proof presented in this section is still valid with small modification of Lemma \ref{lem:sta:sec4}, by using the fact that $H^2(\Bbb R^d)\subset L^\infty(\Bbb R^d)$ ($d=2,3$).
\end{remark}
For time splitting Fourier pseudospectral method and time splitting finite difference method (\ref{eq:tsfd:sec4}), similar error estimates to Theorem \ref{thm:tssp:sec4} can be established.
\subsection{Numerical results}
In this section, we report numerical results of the proposed numerical methods.
\begin{example} 1D defocusing condensate,
i.e. we choose $d=1$
and consider GPE
\be
i\p_t\psi=-\frac{1}{2}\p_{xx}\psi+\frac{x^2}{2}\psi+\beta|\psi|^2\psi,
\ee
 with positive $\beta=50$.
The initial condition is taken as
 \be
\label{eq:inite1:sec4} \psi(x,0)=\fl{1}{\pi^{1/4}}e^{-x^2/2}, \qquad x\in{\Bbb R}.
 \ee
We  solve this
problem on $[-16,16]$, i.e. $a=-16$ and $b=16$ with homogenous Dirichlet boundary conditions. Let $\psi$ be the `exact' solution which is
obtained numerically by using TSSP4 (fourth order time-splitting sine pseudospectral method (\ref{eq:4thsplit:sec4})) with a very fine mesh and time
step, e.g., $h=\fl{1}{1024}$ and $\tau=0.0001$, and $\psi_{h,\tau}$ be
the numerical solution obtained by using a method with mesh size
$h$ and time step $\tau$.
\end{example}

\bigskip

First we compare the discretization error in space. We choose a
very small time step, e.g., $\tau=0.0001$ for CNFD, TSSP4 and TSFD, $\tau=0.00001$ for TSSP2 such that the error from
the time discretization is negligible compared to the spatial
discretization error, and solve the GPE using different methods
and various spatial mesh sizes $h$. Tab.~\ref{tab:tab1:sec4} lists the numerical
errors $\|\psi(t)-\psi_{h,\tau}(t)\|_{l^2}$ at $t=1$ for various
spatial mesh sizes $h$.

\begin{table}[htbp]
\begin{center}
\begin{tabular}{ccccccc}\hline
 Mesh   &$h=\fl{1}{4}$ &$h=\frac{1}{8}$ &$h=\frac{1}{16}$
   &$h=\frac{1}{32}$ & $h=1/64$\\
\hline
TSSP2 &9.318E-2 &4.512E-7 &$<$5.0E-10 &$<$5.0E-10&$<$5.0E-10\\
TSSP4 &9.318E-2 &4.512E-7 &$<$5.0E-10 &$<$5.0E-10&$<$5.0E-10\\
TSFD &7.943E-1 &3.147E-1 &9.025E-2 &2.239E-2&5.574E-3\\
CNFD &7.943E-1 &3.147E-1 &9.026E-2 &2.240E-2 &5.583E-3\\
\hline
\end{tabular}
\end{center}
\caption{Spatial discretization error analysis:
$\|\psi(t)-\psi_{h,\tau}(t)\|_{l^2}$
at time $t=1$ under $\tau=0.0001$ for different numerical methods including TSSP2
(\ref{eq:tssp:sec4}), TSFD (\ref{eq:tsfd:sec4}), CNFD (\ref{eq:cnfd:sec4}), and TSSP4 (fourth order time integrator (\ref{eq:4thsplit:sec4}) with sine pseudospectral method). }\label{tab:tab1:sec4}
\end{table}

Secondly, we test the discretization error in time.  Tab.~\ref{tab:tab2:sec4} shows the numerical errors
$\|\psi(t)-\psi_{h,\tau}(t)\|_{l^2}$ at $t=1$ with a very small mesh
size $h=\fl{1}{1024}$ for different time steps $\tau$ and different
numerical methods.

\begin{table}[htbp]
\begin{center}
\begin{tabular}{cccccc}\hline
 Time step   &$\tau=0.01$ &$\tau=0.005$ &$\tau=0.0025$
   &$\tau=0.00125$\\
\hline
TSSP2 &4.522E-4 &1.129E-4 &2.821E-5 &7.051E-6\\
TSSP4 &1.091E-5 &6.756E-7 &4.213E-8 &2.630E-9\\
TSFD &3.332E-2 &8.261E-3 &2.071E-3 &5.323E-4\\
CNFD &1.099E-1 &2.884E-2 &7.268E-3 &1.835E-3\\
\hline
\end{tabular}
\end{center}
\caption{Time discretization error analysis:
$\|\psi(t)-\psi_{h,\tau}(t)\|_{l^2}$
at time $t=1$ under $h=\fl{1}{1024}$.}\label{tab:tab2:sec4}
\end{table}


\bigskip

  From Tabs.~\ref{tab:tab1:sec4}-\ref{tab:tab2:sec4}, one can make the following observations:
(i) Both TSSP2 and TSSP4 are spectral accurate in space and they share the same accuracy for fixed mesh size $h$, and resp., TSFD and CNFD are second-order
in space and they share the same accuracy for fixed mesh size $h$. (ii) TSSP2, TSFD and CNFD are second-order in time and TSSP4 is fourth-order in time. In general, for fixed time step $\tau$, the error from time discretization of TSSP2
is much smaller than that of TSFD and CNFD, and the error from time discretization of TSFD is much smaller than that of CNFD. From our computations, the error bounds for SIFD are similar to CNFD and we omit it here (cf. \cite{BaoCai2,BaoCai}). For more comparisons, we refer to \cite{BaoTang0} and references therein.

 Among the above numerical methods: (i) TSSP is explicit with computational cost per time step at $O(M\ln M)$
 with $M$ the total number of unknowns in 1D, 2D and 3D,  TSFD and SIFD are implicit with computational cost
 per time step $O(M)$ and $O(M\ln M)$ in 1D and 2D/3D, respectively, CNFD is implicit which is the most expensive one since it needs to solve
a fully coupled nonlinear system per time step. (ii) The storage requirement of TSSP is little less than those of
TSFD, CNFD and SIFD. (iii) TSSP, TSFD and CNFD are unconditionally stable and SIFD is conditionally stable.
(iv) TSSP and TSFD are time transverse invariant, where CNFD and SIFD are not.
(v) TSSP, TSFD and CNFD conserve the mass in the discretized level, where SIFD doesn't.
(vi) CNFD conserves the energy in the discretized level, where TSSP, TSFD and SIFD don't. However,
when the time step is small, TSSP, TSFD and SIFD conserve the energy very well in practical computation.
(vii)  Extension of all the numerical methods to 2D and 3D cases is straightforward
without additional numerical difficulty.
Based on these comparisons, in order to solve the GPE numerically for computing the dynamics of BEC,
when the solution is smooth, we recommend to use TSSP method, and resp., when the solution is not very smooth,
e.g. with random potential, we recommend to use TSFD.

\subsection{Extension to damped Gross-Pitaevskii equations}
In section \ref{subsubsec:blowup}, a damping term is introduced in GPE to describe the collapse of focusing BEC. Our numerical methods can be generalized to this damped GPE easily. For simplicity, we will only consider 1D case and extensions to 2D and 3D are straightforward. For $d=1$, the general damped NLSE becomes \cite{BaoJaksch,BaoJakschP46}
\begin{align} \label{eq:dgpe1d:sec4}
&i\psi_t=- \fl{1}{2} \psi_{xx}+ V(x)\psi
+\bt|\psi|^{2\sg}\psi -i\, f(|\psi|^2)\psi,
\quad a<x<b,\quad t>0,  \\
\label{eq:dgpei:sec4} &\psi(x,t=0)=\psi_0(x), \quad a\le x\le b, \qquad
\psi(a,t)=\psi(b,t)={ 0}, \quad t\ge 0.
\end{align}
Due to the high performance of TSSP (\ref{eq:tssp:sec4}), we will extend it to solve damped GPE and adopt the same notations and  mesh strategy. From time $t=t_n$ to time $t=t_{n+1}$, the damped GPE (\ref{eq:dgpe1d:sec4}) is solved in two steps.
One solves
\begin{equation} \label{eq:1step:sec4}
i\; \psi_t=- \fl{1}{2} \psi_{xx},
\end{equation}
for one time step, followed by solving
\begin{equation}
\label{eq:2step:sec4} i\; \psi_t(x,t)= V(x)\psi(x,t)+ \bt
|\psi(x,t)|^{2\sg} \psi(x,t)-i\, f(|\psi(x,t)|^2)\psi(x,t),
\end{equation}
again for the same time step. Equation (\ref{eq:1step:sec4}) is
discretized in space by the sine-spectral method and integrated in
time {\it exactly}. For $t\in[t_n,t_{n+1}]$, multiplying the ODE
(\ref{eq:2step:sec4}) by $\overline{\psi(x,t)}$, the conjugate of
$\psi(x,t)$, one obtains
\begin{equation} \label{eq:sstepa:sec4}
i\;\psi_t(x,t)\overline{\psi(x,t)}= V(x)|\psi(x,t)|^2+ \bt
|\psi(x,t)|^{2\sg+2} -i\, f(|\psi(x,t)|^2)|\psi(x,t)|^2.
\end{equation}
Subtracting the conjugate of Eq.~(\ref{eq:sstepa:sec4}) from
Eq.~(\ref{eq:sstepa:sec4}) and multiplying by $-i$ one obtains
\begin{equation}
\label{eq:sstepc:sec4}
\fl{\rd}{\rd t}|\psi(x,t)|^2=\overline{\psi_t(x,t)}\psi(x,t)+
\psi_t(x,t)\overline{\psi(x,t)} =-2 f(|\psi(x,t)|^2)|\psi(x,t)|^2.
\end{equation}
Let
\begin{equation} \label{eq:frho:sec4}
g(s)=\int \fl{1}{s\; f(s)}\;ds, \qquad h(s,s^\prime)=\left\{\ba{ll}
g^{-1}\left(g(s)-2s^\prime\right), &s>0, \ s^\prime\ge0, \\
0, &s=0, \ s^\prime\ge0. \ea\right. \\
\end{equation}
Then, if $f(s)\ge0$ for $s\ge 0$, we find
\begin{equation}
\label{eq:positive:sec4} 0\le h(s,s^\prime)\le s, \qquad \hbox{for}\quad s\ge
0,\quad  s^\prime\ge 0,
\end{equation}
and the solution of the ODE (\ref{eq:sstepc:sec4}) can be expressed as
(with $s^\prime=t-t_n$)
\begin{eqnarray} \label{eq:rhot:sec4} 0\le\rho(t)&=&\rho(t_n+s^\prime)
:=|\psi(x,t)|^2=h\left(|\psi(x,t_n)|^2, t-t_n\right):=
h\left(\rho(t_n),s^\prime\right) \nn\\
&\le&\rho(t_n)=|\psi(x,t_n)|^2, \qquad t_n\le t\le t_{n+1}.
\end{eqnarray}
Combining Eq.~(\ref{eq:rhot:sec4}) and Eq.~(\ref{eq:sstep:sec4}) we obtain
\begin{eqnarray} \label{eq:sstepd:sec4}
i\;\psi_t(x,t)&=&V(x)\psi(x,t)+
\bt \left[h\left(|\psi(x,t_n)|^2, t-t_n\right)\right]^{\sg} \psi(x,t)\nn\\
&&-i\, f\left(h\left(|\psi(x,t_n)|^2,t-t_n\right)\right)\psi(x,t),
\quad t_n\le t\le t_{n+1}. \qquad
\end{eqnarray}
Integrating (\ref{eq:sstepd:sec4}) from $t_n$ to $t$, we find
\begin{align}
\label{eq:spsi:sec4} \psi(x,t)=&\;
\exp\left\{i\left[-V(x)(t-t_n)-G\left(|\psi(x,t_n)|^2,t-t_n\right)\right]
-F\left(|\psi(x,t_n)|^2,t-t_n\right)\right\} \nn\\
&\quad \tm\; \psi(x,t_n), \qquad t_n\le t\le t_{n+1},
\end{align}
where we have defined
\begin{equation}
\label{eq:FG:sec4}
F(s,\eta)=\int_0^\eta f(h(s,s^\prime))\;ds^\prime\ge0,
\quad
G(s,\eta)=\int_0^\eta \bt\; \left[h(s,s^\prime)\right]^{\sg}\;ds^\prime.
\end{equation}
To find the time evolution between $t=t_n$ and $t=t_{n+1}$, we
combine the splitting steps via the standard second-order Strang
splitting (TSSP) for solving the damped GPE (\ref{eq:dgpe1d:sec4}).
 In detail, the steps for obtaining
$\psi^{n+1}_j$ from $\psi^{n}_j$ are given by
\begin{align}
&\psi^{(1)}_j=\exp\left\{-F\left(|\psi_j^{n}|^2,\frac{\tau}{2}\right)+i
\left[-V(x_j)\frac{\tau}{2}-G\left(|\psi_j^{n}|^2,\frac{\tau}{2}\right)\right]\right\}
\;\psi_j^{n},\nn
\\
&\psi_j^{(2)}=\frac{2}{M}\sum_{l=1}^{M-1} e^{-i
\tau\mu_l^2/2}\;\widetilde{(\psi^{(1)})}_l \; \sin(\mu_l(x_j-a)),\qquad
j\in\calT_M, \label{schmg}\\
&\psi^{n+1}_j=\exp\left\{-F\left(|\psi_j^{(2)}|^2,\frac{\tau}{2}\right)+i
\left[-V(x_j)\frac{\tau}{2}-G\left(|\psi_j^{(2)}|^2,\frac{\tau}{2}\right)\right]\right\}
\;\psi_j^{(2)},\nn
\end{align}
where $ \widetilde{u}_l$ are the sine-transform coefficients of a
complex vector $u=(u_0,\cdots, u_M)^T$ with $u_0=u_M={0}$ defined in (\ref{eq:sinetran:sec3}),
and
\begin{equation}
\label{eq:init1dg:sec4}
\psi^{0}_j=\psi(x_j,0)=\psi_0(x_j), \qquad j=0,1,2,\cdots,M.
\end{equation}
For some frequently used damping terms, the integrals in
(\ref{eq:frho:sec4}) and (\ref{eq:FG:sec4}) can be evaluated analytically (cf. \cite{BaoJaksch}).

\section{Theory for rotational BEC}
\label{sec:rotat}
\setcounter{equation}{0}\setcounter{figure}{0}\setcounter{table}{0}
In view of
potential applications,
the study of quantized vortices, which are
related to superfluid properties,
is one of the key issues. In fact, bulk
superfluids are distinguished from normal fluids by their ability
to support dissipationless flow. Such persistent currents are
intimately related to the existence of quantized vortices, which
are localized phase singularities with integer topological charge \cite{Fetter}.
The superfluid vortex is an example of a topological defect that
is well known in  superconductors
 and in liquid helium.
 Currently, one of the most popular ways to generate quantized vortices from
BEC ground state is the following:  impose a laser beam rotating with an angular
velocity
on the magnetic trap holding the atoms to
create a harmonic anisotropic potential. Various experiments have confirmed the observation of quantized vortices in  BEC under a rotational frame \cite{Abo,Madison,Caradoc}.
\subsection{GPE with an angular momentum rotation term}
At temperatures $T$ much smaller than the critical temperature
$T_c$, following the mean field theory (cf. section \ref{subsec:manybody}),
BEC in a rotational frame  is well described by
the macroscopic wave function $\psi(\bx,t)$, whose evolution is
governed by the Gross-Pitaevskii equation (GPE)
with an angular momentum rotation term
\cite{BaoShen,BaoShen2,BaoWang,Fetter,AftalionDu}, (w.l.o.g.) assuming the rotation
being around the $z$-axis:
\be\label{eq:gperota:sec5}
i\hbar\pl{\psi(\bx,t)}{t}=\left(-\fl{\hbar^2}{2m}\btd^2 + V(\bx)+N
g |\psi(\bx,t)|^2-\Omega L_z\right) \psi(\bx,t),
\ee
where  $\bx=(x,y,z)^T\in{\Bbb R}^3$ is the spatial
coordinate vector,
\; and
\be\label{eq:rota:sec5}
L_z=-i \hbar\left(x\partial_y -y\partial_x\right)
\ee
is the $z$-component of the angular momentum operator $L=(L_x,L_y,L_z)^T$ given by
$L=-i\hbar(\bx\wedge\nabla)$. The appearance of the angular momentum term means that we are using a reference frame where the trap is at rest. The energy
functional per particle $E(\psi)$ is defined as
\be\label{eq:penergy:sec5}
E(\psi) =\int_{{\Bbb R}^3} \left[
\fl{\hbar^2}{2m} \left|\btd \psi\right|^2 + V(\bx)|\psi|^2 +
\fl{Ng}{2} |\psi|^4-\Og \bar{\psi} L_z \psi \right]\; d\bx,
\ee
and wave function is normalized as
 \begin{equation}
\label{eq:norm:sec5}
\int_{{\Bbb R}^3} \; |\psi(\bx,t)|^2\;d\bx=1.
\end{equation}

\subsubsection{Dimensionless form} For the conventional harmonic potential case,
by introducing the dimensionless variables: $t\to t/\omega_{0}$
with $\omega_0=\min\{\omega_x,\omega_y,\omega_z\}$,
$\bx\to \bx x_s$ with $x_s=\sqrt{\hbar/m\og_0 }$,
$\psi\to \psi/x_s^{3/2}$, $\Omega\to \Omega \omega_0$ and
$E(\cdot)\to\hbar\og_0 E_{\beta,\Og}(\cdot)$, we get the dimensionless
GPE
\be \label{eq:gperot:sec5}
i\;\pl{\psi(\bx,t)}{t}=\left(-\fl{1}{2}\btd^2 +
V(\bx) + \kappa\; |\psi(\bx,t)|^2-\Omega
L_z\right)\psi(\bx,t),
\ee
where  $\kappa=\fl{g N}{x_s^3\hbar \og_0}=\fl{4\pi a_sN}{x_s}$,
$L_z=-i(x\p_y -y \p_x)$ ($L=-i(\bx\wedge\nabla)$),
$V(\bx)=\fl{1}{2}\left(\gamma_x^2 x^2+\gm_y^2 y^2+\gm_z^2 z^2\right)$ with
$\gm_x=\fl{\og_x}{\og_0}$,
$\gm_y=\fl{\og_y}{\og_0}$ and  $\gm_z=\fl{\og_z}{\og_0}$.

In a disk-shaped condensate with parameters $\og_x\approx \og_y$
and $\og_z\gg \og_x$ ($\Longleftrightarrow$ $\gamma_x=1$,
$\gm_y\approx1$ and
$\gm_z\gg 1$ with choosing $\omega_0=\omega_x$),
the 3D GPE (\ref{eq:gperot:sec5}) can be reduced to a
 2D GPE with $\bx=(x,y)^T$ (cf. section \ref{subsubsec:dred}):
 \be \label{eq:gpe2drot:sec5}
i\;\pl{\psi(\bx,t)}{t}=-\fl{1}{2}\btd^2 \psi+
V_2(x,y) \psi +c_2\kappa
|\psi|^2\psi-\Omega L_z\psi,
\ee
where $c_2= \sqrt{\gm_z/2\pi}$
and $V_2(x,y)=\fl{1}{2}\left(\gamma_x^2 x^2+\gm_y^2 y^2\right)$
 \cite{BaoWang,BaoWangP,BaoCai,Fetter}.

 Thus here we consider the dimensionless GPE under a rotational frame
in $d$-dimensions ($d=2,3$):
 \be \label{eq:gpegrot:sec5}
i\;\pl{\psi(\bx,t)}{t}=-\fl{1}{2}\btd^2 \psi+
V(\bx)\psi + \beta|\psi|^2\psi-\Omega L_z
\psi, \quad \bx\in {\Bbb R}^d, \quad t>0,
\ee
where
\be\label{eq:pot:sec5}
\beta=\kappa \begin{cases}
\sqrt{\gm_z/2\pi}, \\
1,
\end{cases}
\,
V(\bx)=\begin{cases}
 \fl{1}{2}\left(\gm_x^2x^2+\gm_y^2 y^2\right), & d=2, \\
 \fl{1}{2}\left(\gm_x^2x^2+\gm_y^2 y^2+\gm_z^2 z^2\right), &d=3.
\end{cases}
\ee
 Then the dimensionless energy functional  per particle
$E_{\beta,\Omega}(\psi)$ is defined as
\be\label{eq:engd1:sec5}
E_{\beta,\Og}(\psi) =\int_{{\Bbb
R}^d} \left[\fl{1}{2} \left|\btd \psi(\bx,t)\right|^2+
V(\bx)|\psi|^2 +\fl{\beta}{2}\; |\psi|^4-
\Omega \bar{\psi}\, L_z \psi\right]d\bx,
\ee
and the normalization is given by
\be
\|\psi(\cdot,t)\|_{L^2(\Bbb R^d)}^2=\|\psi(\cdot,0)\|_{L^2(\Bbb R^d)}^2=1.
\ee
The vortex structure of rotating BEC in 3D is very complicated \cite{Aftalion,Fetter} due to the presence of vortex lines, while in 2D, the structure is relatively simple as the vortex center is only a point. As a result, most investigations  start from the 2D case.

\subsection{Theory for ground states}
Similar to section \ref{sec:mathgpe},
the ground state wave function $\phi_g:=\phi_g(\bx)$ of
a rotating BEC satisfies the nonlinear eigenvalue problem (Euler-Lagrange equation)
\be
\label{eq:gss:sec5}
\mu\; \phi(\bx)=\left[-\fl{1}{2}\nabla^2+
V(\bx)
+ \bt |\phi|^2-\Omega L_z\right]\phi(\bx),
\quad \bx\in {\Bbb R}^d,\quad d=2,3,
\ee
under the normalization condition
\be
\label{eq:normgg:sec5}
\|\phi\|_2^2=\int_{{\Bbb R}^d} \; |\phi(\bx)|^2\;d\bx=1,
\ee
with eigenvalue (or chemical potential)
$\mu$ given by
\begin{equation}
\mu=E_{\beta,\Og}(\phi)+\fl{\bt}{2}\int_{{\Bbb R}^d}
\left|\phi(\bx)\right|^4\;d\bx.
\label{eq:engvf:sec5}
\end{equation}

The eigenfunction $\phi(\bx)$ of (\ref{eq:gss:sec5}) under the constraint
(\ref{eq:normgg:sec5}) with least energy is called  ground state.
The ground state can be found by minimizing the energy functional $E_{\beta,\Og}(\phi)$
(\ref{eq:engd1:sec5})
over the unit sphere $S=\{\phi \ |\ \|\phi\|_2=1, \ E_{\beta,\Og}(\phi)<\ift\}$:

(I) Find $(\mu_{\beta,\Og}^g, \phi_g\in S)$ such that \be
\label{eq:minim:sec5} E^g_{\beta,\Og}:=E_{\beta,\Og}^g(\phi_g)=
\min_{\phi\in S} \; E_{\beta,\Og}(\phi), \quad
\mu^g_{\beta,\Og}: =\mu^g_{\beta,\Og}(\phi_g).
\ee

Any eigenfunction $\phi(\bx)$ of (\ref{eq:gss:sec5}) under constraint
(\ref{eq:normgg:sec5}) whose energy $E_{\beta,\Og}(\phi)
>E_{\beta,\Og}(\phi_g)$ is usually called
as an excited state in the physics literature.

Existence/nonexistence results of ground state depend
on the magnitude $|\Og|$ of the angular velocity relative to
the trapping frequencies \cite{Seiringer,BaoWangP,Cai}.

For the existence and simple properties of ground state for rotating BEC, we have the following \cite{BaoWangP,Seiringer}.
\begin{theorem}\label{thm:rot:sec5} Suppose that $V(\bx)$ is given in (\ref{eq:pot:sec5}), then we have the conclusions below.

i) In 2D, if $\phi_{\beta,\Og}(x,y)\in S$ is a ground state
of the energy functional $E_{\beta,\Og}(\phi)$, then
$\phi_{\beta,\Og}(x,-y)\in S$  and $\phi_{\beta,\Og}(-x,y)\in S$ are
 ground states
of the energy functional $E_{\beta,-\Og}(\phi)$. Furthermore
\be
\label{eq:relat:sec5}
E_{\beta,\Og}^g = E_{\beta,-\Og}^g, \qquad
\mu_{\beta,\Og}^g=\mu_{\beta,-\Og}^g.
\ee

ii) In 3D, if $\phi_{\beta,\Og}(x,y,z)\in S$ is a ground state
of the energy functional $E_{\beta,\Og}(\phi)$, then
$\phi_{\beta,\Og}(x,-y,z)\in S$  and $\phi_{\beta,\Og}(-x,y,z)\in S$ are
 ground states
of the energy functional $E_{\beta,-\Og}(\phi)$, and (\ref{eq:relat:sec5})
is also valid.

iii) When  $|\Og|<\min\{\gamma_{x},\gamma_y\}$ and $\beta\ge0$ in 3D or $\beta>-C_b$ in 2D ($C_b$ given in (\ref{eq:bestcons:2d})), there exists a minimizer
for the minimization problem (\ref{eq:minim:sec5}), i.e. there
exist ground state.
\end{theorem}
Theorem \ref{thm:rot:sec5} is a direct consequence of the below observation and the theory for non-rotating BEC (cf. section \ref{sec:mathgpe}).
  \begin{lemma}\label{lem:EOM:sec5}Under the conditions of Theorem \ref{thm:rot:sec5}, the following results hold.

i) In 2D, we have
\begin{equation*}
E_{\beta,-\Og}(\phi(x,-y))=E_{\beta,\Og}(\phi(x,y)), \quad
E_{\beta,-\Og}(\phi(-x,y))=E_{\beta,\Og}(\phi(x,y)), \quad
\phi\in S.
\end{equation*}
ii) In 3D, we have
\begin{equation*}
E_{\beta,-\Og}(\phi(x,-y,z))=E_{\beta,\Og}(\phi(x,y,z)), \
E_{\beta,-\Og}(\phi(-x,y,z))=E_{\beta,\Og}(\phi(x,y,z)), \
\phi\in S.
\end{equation*}
iii) In 2D and 3D, we have \bea \label{eq:engm:sec5}
\lefteqn{\int_{{\Bbb R}^d} \left[\fl{1-|\Og|}{2} \left|\btd
\phi(\bx)\right|^2+
\left(V(\bx)-\fl{|\Og|}{2}(x^2+y^2)\right)|\phi|^2
+\fl{\beta}{2}\; |\phi|^4\right]d\bx
\le E_{\beta,\Og}(\phi) \nn} \\[2mm]
&\le& \int_{{\Bbb R}^d} \left[\fl{1+|\Og|}{2} \left|\btd
\phi(\bx)\right|^2+
\left(V(\bx)+\fl{|\Og|}{2}(x^2+y^2)\right)|\phi|^2
+\fl{\beta}{2}\; |\phi|^4\right]d\bx. \qquad \eea
\end{lemma}
For understanding the uniqueness
question, note that $E_{\beta,\Og}(\ap \phi_{\beta,\Og}^g)=E_{\beta,\Og}
(\phi_{\beta,\Og}^g)$ for all $\ap\in{\Bbb C}$ with $|\ap|=1$.
Thus an additional constraint has to be introduced to
show uniqueness.  For non-rotating BEC, i.e.
$\Og=0$, the unique positive minimizer is usually taken as the
ground state. In fact, the ground state is unique up to
a constant $\ap$ with $|\ap|=1$, i.e. density of the ground state
is unique, when $\Og=0$. For rotating BEC under  $|\Og|<\min\{\gamma_{x},\gamma_y\}$,
 the density of the ground state may be no longer unique
when $|\Og|> \Og^c$ with $\Og^c$ a critical angular rotation speed.

When rotational speed exceeds the trap frequency in $x$- or $y$- direction, i.e. $\Omega>\min\{\gamma_x,\gamma_y\}$,
 there will be no ground state \cite{Cai,BaoWangP}.
 \begin{theorem} \label{thm:rotnon:sec5}(nonexistence) Suppose that $V(\bx)$ is given in (\ref{eq:pot:sec5}), then there exists no minimizer  for problem (\ref{eq:minim:sec5}) if one of the following holds:

 (i) $\beta<0$ in 3D or $\beta<-C_b$ ($C_b$ given in (\ref{eq:bestcons:2d})) in 2D;

 (ii) $|\Og|>\min\{\gamma_x,\gamma_y\}$.
 \end{theorem}

\subsection{Critical speeds for quantized vortices}
\label{subsec:critical}
From Theorems \ref{thm:rot:sec5} and \ref{thm:rotnon:sec5}, only the case of $0\leq\Omega<\min\{\gamma_x,\gamma_y\}$ is interesting for considering the ground state of a rotating BEC and we will assume $\Omega\ge0$ in the subsequent discussions.  In particular,  the  vortex  appears in the ground state only if the rotational speed exceeds certain critical value. Various experiments and mathematical studies confirm the existence of such critical speeds.

In \cite{Seiringer}, for a radially symmetric  potential  in 2D, Seiringer proved that there exists critical velocity $\Omega_0>0$ that a state containing  vortices is energetically favorable for $\Omega>\Omega_0$. The point is that, as $\Omega$ increases from 0 to $\min\{\gamma_x,\gamma_y\}$, the first ground state with vortex should be the {\it central vortex} state, i.e., a state containing central vortex (vortex line) in the rotational center (axis).

In 2D with radially symmetric  potential $V(r)$ ($r=|\bx|$, $\bx\in\Bbb R^2$), the symmetric state ($m=0$) and central vortex state with index (or winding number or circulation) $0\neq m\in\Bbb Z$ is the solution of the nonlinear eigenvalue problem (\ref{eq:gss:sec5})-(\ref{eq:engvf:sec5}) with the form
\be\label{eq:centralvor:sec5}
\phi(\bx)=\phi_m(r)e^{im\theta},\quad \text{where}\quad(r,\theta)\quad \text{is the polar coordinate}.
\ee
In polar coordinate, the angular momentum term  (\ref{eq:rota:sec5}) becomes $L_z=-i\frac{\p}{\p\theta}$.
In order to find the above central vortex states with index $m$ ($m\ne0$),
$\phi(\bx)=\phi_m(r)e^{im\tht}$,
we need to find  a real nonnegative
 function $\phi_m(r):=\phi_{\beta,\Og}^m(r)$ which minimizes
the energy functional
\bea
\label{eq:ecen:sec5}
\lefteqn{E_{\beta,\Og}^m(\phi(r))=E_{\beta,\Og}(\phi(r)e^{im\tht})\nn}\\[2mm]
&=&\pi \int_0^\ift \left[|\phi^\prime(r)|^2
+\left(2V(r) + \fl{m^2}{r^2}\right)|\phi(r)|^2 + \bt |\phi(r)|^4
-2m\Og |\phi(r)|^2 \right]r\;dr\nn\\
&=&E_{\beta,0}^n(\phi(r)) -m\Og, \qquad \Og\ge0,\quad m\ge0,
\eea over the set $S_r=\{
\phi(r)\in{\Bbb R} \ |\ 2\pi \int_0^\ift |\phi(r)|^2 r\;dr=1,
\ E_{\beta,0}^m(\phi)<\ift, \
\phi^\prime(0)=0\ (m=0),\ \hbox{and resp.}\ \phi(0)=0 (m\ne0)
\}$.
The existence and
uniqueness of nonnegative minimizer for this minimization problem
can be obtained similarly as for the ground state when $\Og=0$
\cite{LiebSeiringerPra2000,Seiringer}. It is clear from (\ref{eq:ecen:sec5}) that  contribution of rotation in the  energy (\ref{eq:ecen:sec5}) is fixed with index $m$ and the central vortex state with index $m$ is independent of $\Omega$.  In fact, let us denote
\be\label{eq:redecen:sec5}
E^m(\phi(r))=E_{\beta,\Og}^m(\phi(r))+m\Omega,
\ee
and the central vortex state can be found as the minimizer of energy $E^m(\phi)$ which is independent of $\Omega$.

Similarly, in order
 to find the cylindrically symmetric state ($m=0$), and
resp. central vortex line states ($m\ne0$),
 in 3D with cylindrical
 symmetry,
i.e. $d=3$ and $V(\bx)=V(r,z)$
in (\ref{eq:gpegrot:sec5}),
we solve the nonlinear eigenvalue problem (\ref{eq:gss:sec5})-(\ref{eq:engvf:sec5}) with the special form of wave function as
\be
\label{eq:statv35:sec5}
\phi(\bx)= \phi_m(x,y,z)
=\phi_m(r,z)e^{im\tht},
\ee
where $(r,\theta,z)$ is the cylindrical coordinate, $m$ is an
integer and called as index when $m\ne0$.
 $\phi_m(r,z)$ is a real function independent of  angle.
It is equivalent to computing  a real nonnegative
 function $\phi_m(r,z):=\phi_{\beta,\Og}^m(r,z)$ which minimizes
the energy functional
\bea\label{eq:mincy:sec5}
\lefteqn{E_{\beta,\Og}^m(\phi(r,z))=E_{\beta,\Og}
(\phi(r,z)e^{im\tht}) }\\[2mm]
&=&\pi \int_0^\ift \int_{-\ift}^\ift \left[|\p_r \phi|^2 +|\p_z
\phi|^2 +\left(2V(r,z)+\fl{m^2}{r^2}-2m\Og\right)|\phi|^2 + \bt |\phi|^4\right]
r\;dzdr,\nn
 \eea
over the set $S_c=\{
\phi\in{\Bbb R} \ |\ 2\pi \int_0^\ift\int_{-\ift}^\ift
 |\phi(r,z)|^2 r\;dzdr=1, \ E_{\beta,0}^m(\phi)<\ift,\
\ \p_r \phi(0,r)=0 \ (m=0), \hbox{and resp.}\
\phi(0,z)=0 \ (m\ne0), \; -\ift<z<\ift \}$.

For 2D case, by carefully studying the energy functional (\ref{eq:ecen:sec5}), Seiringer \cite{Seiringer} established the estimates of critical velocity  $\Omega_m$ ($m\ge0$) when a central vortex $\phi(r)$ with index $m+1$ is energetically favorable to a central vortex state with index $m$ for $\Omega>\Omega_m$, i.e., $E_{\beta,\Omega}^{m+1}(\phi(r))<E_{\beta,\Omega}^m(\phi(r))$. The critical speed is given by
\be
\Omega_m=\min_{\phi\in S_r}{E^{m+1}(\phi(r))}-\min_{\phi\in S_r}{E^{m}(\phi(r))}>0.
\ee
The estimates for $\Omega_m$ is the following.
\begin{theorem}(cf. \cite{Seiringer})\label{thm:critical:sec5} In 2D, assume potential $V(r)=\frac{1}{2}r^2$ ($r=|\bx|$, $\bx\in\Bbb R^d$), for positive $\beta>0$ and $m\ge0$, then the following bounds on $\Omega_m$ hold,
\begin{align}
&\Og_m\leq (2m+1)\frac{2\pi e}{\beta}\left(1+\sqrt{\frac{2\beta}{\pi}}\right)\left(3+\left[\ln(\beta/(2\pi e^2))\right]\right),\\
&\Og_m\ge \frac{2m+1}{(m+2)\sqrt{1+\frac{\beta}{b_{m+1}(m+2)}}},
\end{align}
where $b_m=\frac{2\times 4^m\pi(m!)^2}{(2m)!}$. In addition, there is a relation between the critical velocities as  $\Og_{m+1}\leq \frac{2m+3}{2m+1}\Og_m$.
\end{theorem}

In particular, we notice that $\Omega_0$ is the critical velocity for the appearance of vortex in the ground state. For large $\Omega$, the central vortex state (\ref{eq:centralvor:sec5}) will no longer be the right ansatz for ground state due to the  symmetry breaking \cite{Seiringer}.
\begin{theorem}(cf. \cite{Seiringer})\label{thm:symbre:sec5} In 2D,  assume potential $V(r)=\frac{1}{2}\gamma_r^2r^2$ ($r=|\bx|$, $\bx\in\Bbb R^d$), for fixed rotational speed $0<\Omega<\gamma_r$, there exists a constant $\beta_{\Omega}$, such that if $\beta\ge\beta_{\Omega}$, no ground state of problem (\ref{eq:minim:sec5}) is an eigenfunction of the angular momentum operator $L_z$ (\ref{eq:rota:sec5}). Let $\phi_g$ be the ground state of (\ref{eq:minim:sec5}) when $\beta\ge\beta_{\Omega}$, then $|\phi_g|$ is not radially symmetric.
\end{theorem}

For small rotational speed $\Omega$ with radial potential $V(r)$, the minimizer of problem (\ref{eq:minim:sec5}) remains the same as the ground state of the non-rotating case, i.e., the ground state of small $\Omega$ is the same as that of $\Omega=0$. This observation is recently proved rigourously by Aftalion et al. \cite{AftalionJerr}.

\begin{theorem} (cf. \cite{AftalionJerr}) In 2D,  assuming potential $V(r)=\frac{1}{2}\gamma_r^2r^2$ ($r=|\bx|$, $\bx\in\Bbb R^d$),  let $\phi^g_{\beta,\Omega}(\bx)$ be the solution of the minimization problem (\ref{eq:minim:sec5}) with rotational speed $\Omega$ and $\phi^g_{\beta,0}(r)$ be the unique positive minimizer for (\ref{eq:minim:sec5}) with $\Omega=0$ (see section \ref{sec:mathgpe}). There exists $\beta_0\ge0$, such that when $\beta\ge\beta_0$, for $0\leq\Omega\leq \Omega_\beta$ (constant $\Omega_\beta$ depends on $\beta$),
$\phi^g_{\beta,\Omega}=e^{i\theta}\phi^g_{\beta,0}$ for some constant $\theta\in\Bbb R$.
\end{theorem}
All these results are valid for   general potential $V(r)$ \cite{Seiringer,AftalionJerr}.

When $\Omega$ increases, the number of vortices in the ground state will also increase and the vortices interact. At high rotational speed, vortices will form a lattice, known as an Abrikosov lattice \cite{AftalionDu,Fetter,Aftalion}.   For the fast rotational speeds, Correggi \cite{Correggi3} studied the  expansion of energy (\ref{eq:engd1:sec5}) and proved that the vortices will be asymptotically equidistributed, which means that the vortices will form a triangle lattice.
For general harmonic potentials in 2D, i.e., $\gamma_x\neq\gamma_y$, Ignat et al.  estimated the critical angular velocity for the existence of $n$ vortices in the ground state as well as the location of the vortex centers  \cite{Ignat1,Ignat2}.

If an anharmonic potential, i.e. $V(\bx)= V(r)=O(r^4)$ when $r\to\infty$, is considered instead of harmonic potentials,  there exists another phase transition \cite{Correggi1,Correggi2}. When $\Omega$ increases,   the ground state of a rotating BEC with anharmonic potential will first undergo a phase transition to the vortex state, and then become a vortex lattice. If the  velocity $\Omega$ keeps increasing, the vortex lattice will disappear, and the density will be depleted near the trap center. Then  all the vortices will be pushed away from the center and form a giant vortex (or vortex ring) \cite{Rougerie}. Thus, there are three typical critical speeds that can be identified with these kinds of phase transitions in rotating BEC \cite{Correggi1}.
In the study of the critical speeds, a more widely used scaling is different from the one adopted here \cite{AftalionDu} (cf. section \ref{sec:semiclass}).

\subsection{Well-posedness of Cauchy problem}
This section is devoted to the well-posedness of Cauchy problem for the rotating GPE (\ref{eq:gpegrot:sec5}). We use the same notations for function spaces as those in section \ref{sec:mathgpe}.

 Like the  non-rotating case ($\Omega=0$) (cf. section \ref{sec:mathgpe}), the key part is to establish the dispersive  estimates like Lemma \ref{lem:stri} for the  evolutionary operator $e^{itL_{R}}$ generated by the linear operator
\be
L_R=-\frac12\nabla^2+V(\bx)-\Omega L_z,\qquad \bx\in\Bbb R^d,\quad d=2,3.
 \ee
 For the special case $V(\bx)=\frac{\omega^2}{2}|\bx|^2$ and $\Omega=\omega$, Hao et al. obtained the dispersive estimates for  $e^{itL_R}$ using a generalization of Mehler's formula for the kernel of $e^{itL_{R}}$ \cite{Hao1,Hao2,Cazenave}. Later, Antonelli et al. found another approach to establish the dispersive estimates for $e^{itL_R}$ without those assumptions on $\Omega$ and $V$ in \cite{Hao1,Hao2}. It turns out that the rotational term $L_z$ does not affect the dispersive behavior (in short time), and the dispersive estimates for $e^{itL_R}$ are then analogous to those in Lemma \ref{lem:stri} for harmonic potentials. Then the well-posedness of the rotational GPE (\ref{eq:gpegrot:sec5}) follows from the classical arguments \cite{Cazenave,Sulem}.

 \begin{theorem} (cf. \cite{SparberA}) Assume  that $V(\bx)$ is given by (\ref{eq:pot:sec5})  and denote $\gamma_{\min}=\min\{\gamma_\alpha\}$ ($\gamma_\alpha>0$; $\alpha=x,y$ for $d=2$, and
$\alpha=x,y,z$ for $d=3$), then we have the following results.

(i) For any initial data $\psi(\bx,t=0)=\psi_0(\bx)\in X(\Bbb R^d)$ ($d=2,3$),
 there exists a
$T_{{\rm max}}\in(0,+\infty]$ such that the Cauchy problem of
(\ref{eq:gpegrot:sec5})
 has a unique maximal solution
$\psi\in C\left([0,T_{{\rm max}}),X\right)$. It is maximal in
the sense that if $T_{{\rm max}}<\infty$, then
$\|\psi(\cdot,t)\|_{X}\to\infty$ when  $t\to T^-_{{\rm
max}}$.

(ii) As long as the solution $\psi(\bx,t)$ remains in the energy
space $X$, the {\sl $L^2$-norm} $\|\psi(\cdot,t)\|_2$ and {\sl
energy} $E_{\beta,\Omega}(\psi(\cdot,t))$ in (\ref{eq:engd1:sec5}) are conserved for
$t\in[0,T_{\rm max})$.

(iii) The solution of the Cauchy problem for (\ref{eq:gpegrot:sec5}) is global in time, i.e.,
  $T_{\rm max}=\infty$, if $\beta\ge0$.

(iv) If $\beta<0$ in 2D and 3D, and if either:
\begin{enumerate}
 \renewcommand{\labelenumi}{(\arabic{enumi})}
\item $V(\bx)$ is axially symmetric, i.e.,
$\Omega\, L_zV(\bx)=0$;
\item  $\Omega\, L_zV(\bx)\neq0$ and $d\sqrt{(1-(\Omega/\gamma_{\min})^2}\ge1$ with $|\Omega|\leq\gamma_{\min}$ ($d=2,3$).
\end{enumerate}
Then there exists  $\psi_0\in X(\Bbb R^d)$  such that  finite time blow-up happens for the solution of the corresponding Cauchy problem (\ref{eq:gpegrot:sec5}).
 \end{theorem}

\subsection{Dynamical laws}
\label{subsec:dynrot}
In this section, we present results on the dynamical properties of rotating BEC governed by GPE with an angular momentum term (\ref{eq:gpegrot:sec5}).

For the dynamics of angular momentum expectation in rotating BEC,
we have  the following lemmas \cite{BaoDuZhang}:
\begin{lemma}
\label{lem:lemmalz:sec5}
Suppose $\psi({\bx},t)$ is the solution to the Cauchy problem of
(\ref{eq:gpegrot:sec5}), then we have
\be
\label{eq:ode_lz:sec5}
 \fl{\rd\langle L_z\rangle(t)}{\rd t} =
\left(\gm_x^2-\gm_y^2\right)\dt_{xy}(t),
\quad
\mbox{where}\;\;\dt_{xy}(t)=\int_{\mathbb{R}^d}xy|\psi({\bx},t)|^2d{\bx},
\quad t\geq 0\;.
\ee
Consequently, the angular momentum expectation and energy
for non-rotating part
are conserved, that is,
for any given initial data $\psi(\bx,0)=\psi_0({\bx})$,
\be
\label{conse}
\langle L_z\rangle(t)\equiv \langle L_z\rangle(0),\quad
E_{\beta,0}(\psi)\equiv E_{\beta,0}(\psi_0), \qquad t\ge 0,
\ee
at least for radially symmetric trap in 2D or
cylindrically symmetric trap in 3D, i.e. $\gm_x =\gm_y$.
\end{lemma}

For the condensate width defined by $\delta_{\alpha}(\cdot)$ in (\ref{eq:def_sigma:sec2}),
we can obtain similar results to that in Lemma \ref{lem:variance} \cite{BaoDuZhang}.
\begin{lemma}
Suppose $\psi({\bx},t)$ is the solution of the problem
(\ref{eq:gpegrot:sec5}), then we have
\bea
\label{eq:sigma_ODE2:sec5}
&&\ddot{\dt}_{\ap}(t) = \int_{{\Bbb R}^d}\biggl[(\p_y\ap-\p_x\ap)
\left(4i\Og\bar{\psi}(x\p_y+y\p_x)\psi+2\Og^2(x^2-y^2)|\psi|^2\right) \nn\\
&&\qquad \qquad \qquad +2|\p_\ap\psi|^2+\bt|\psi|^4-2\ap|\psi|^2
\p_\ap V(\bx) \biggr]d{\bx},\quad t\geq 0,\\
\label{eq:sigma_init0:sec5}
&&\dt_{\ap}(0) = \dt_{\ap}^{(0)} = \int_{\mathbb{R}^d}\ap^2|\psi_0({\bx})|^2
d{\bx},\qquad \ap = x, y, z,\\
\label{eq:sigma_init1:sec5}
&&\dot{\dt}_{\ap}(0) = \dt_{\ap}^{(1)} = 2\int_{\mathbb{R}^d}\ap \left[ -\Og
|\psi_0|^2\left(x\p_y-y\p_x\right)\ap+ {\rm Im}
\left(\bar{\psi}_0\p_\ap\psi_0\right)\right]\; d\bx.
\eea
\end{lemma}

\begin{lemma}
(i) In 2D with a radially symmetric trap, i.e. $d = 2$ and
$\gm_x=\gm_y:=\gm_r$ in (\ref{eq:gpegrot:sec5}), for any initial data
$\psi(x,y,0)=\psi_0=\psi_0(x,y)$, we have for any $t\ge0$,
\be
\label{eq:solution_dt_r:sec5}
\dt_r(t) = \fl{E_{\beta,\Og}(\psi_0)+\Og\langle
L_z\rangle(0)}{\gm_r^2}\left[1-\cos(2\gm_rt)\right]
+\dt_r^{(0)}\cos(2\gm_rt)+\fl{\dt_r^{(1)}}{2\gm_r}\sin(2\gm_r t),
\ee
where $\dt_r(t)= \dt_x(t)+\dt_y(t)$,
$\dt_r^{(0)}:=\dt_x(0)+\dt_y(0)$, and
$\dt_r^{(1)}:=\dot{\dt}_x(0)+\dot{\dt}_y(0)$.
Furthermore, when the initial condition $\psi_0(x,y)$  satisfies
\be
\label{eq:vortex_initial2:sec5}
\psi_0(x,y)=f(r)e^{im\theta}\quad{\rm with}\quad m\in{\mathbb Z}\quad
{\rm and} \quad f(0) = 0\quad{\rm when}\quad m\neq0,
\ee
we have, for any $t\geq 0$,
\begin{align}
\label{eq:solution_dt_xy:sec5}
\dt_x(t)=&\dt_y(t) = \fl{1}{2}\dt_r(t)\nonumber\\
=&\fl{E_{\beta,\Og}(\psi_0)+m\Og}{2\gm_x^2}\left[1-\cos(2\gm_xt)\right]
+\dt_x^{(0)}\cos(2\gm_xt)+\fl{\dt_x^{(1)}}{2\gm_x}\sin(2\gm_xt).
\end{align}
This and (\ref{eq:def_sigma:sec2}) imply that
\begin{equation*}
\sg_x=\sg_y=\sqrt{\fl{E_{\bt,\Og}(\psi_0)+m\Og}{2\gm_x^2}
\left[1-\cos(2\gm_xt)\right]
+\dt_x^{(0)}\cos(2\gm_xt)+\fl{\dt_x^{(1)}}{2\gm_x}\sin(2\gm_xt)}.
\end{equation*}
Thus in this case, the condensate widths $\sg_x(t)$
and $\sg_y(t)$ are periodic functions
with frequency doubling the trapping frequency.

(ii) For all other cases, we have, for any $t\geq 0$
\be
\label{eq:sigma_general:sec5}
\qquad \dt_\ap(t)=\fl{E_{\beta,\Og}(\psi_0)}{\gm^2_\ap}+\left(\dt_\ap^{(0)}-
\fl{E_{\beta,\Og}(\psi_0)} {\gm_\ap^2}\right)\cos(2\gm_\ap
t)+\fl{\dt_\ap^{(1)}}{2\gm_\ap} \sin(2\gm_\ap t)+f_\ap(t),
\ee
where $f_\ap(t)$ is the solution of the following second-order ODE:
\bea
\ddot{f}_\ap(t)+4\gm_\ap^2\;f_\ap(t)=F_\ap(t),
\qquad f_\ap(0)=\dot{f}_\ap(0) = 0, \eea with
\beas
F_\ap(t)
&=& \int_{{\mathbb R}^d}\bigg[2|\p_\ap\psi|^2-2|\nabla\psi|^2
-\bt|\psi|^4+\left(2\gm_\ap^2\ap^2-4V({\bx})\right)|\psi|^2
+4\Og \bar{\psi}L_z\psi\nonumber\\
&&\quad+(\p_y\ap-\p_x\ap)\left(4i\Og\bar{\psi}\left(x\p_y+y\p_x\right)\psi+
2\Og^2(x^2-y^2)|\psi|^2\right)\bigg]d{\bx}.
\eeas
\end{lemma}

Let $\phi_e(\bx)$ be a stationary state of the GPE (\ref{eq:gpegrot:sec5})
with a chemical potential $\mu_e$, i.e., $(\mu_e,\phi_e)$ satisfies the nonlinear eigenvalue problem (\ref{eq:gss:sec5})-(\ref{eq:normgg:sec5}). If the initial data $\psi_0(\bx)$ for the Cauchy problem of (\ref{eq:gpegrot:sec5}) is chosen as
a stationary state with a shift in its center, one can construct an exact
solution of the GPE with an angular momentum  term (\ref{eq:gpegrot:sec5}), which is similar to Lemma \ref{lem:shift:sec2} \cite{BaoDuZhang}.

\begin{lemma}Suppose $V(\bx)$ is given by (\ref{eq:pot:sec5}),
if the initial data $\psi_0(\bx)$  for the Cauchy problem of (\ref{eq:gpegrot:sec5}) is chosen as
\be
\label{eq:init5:sec5}
\psi_0(\bx)=\phi_e(\bx-\bx_0), \qquad
\bx \in {\mathbb R}^d,
\ee
where $\bx_0$ is a given point in ${\mathbb R}^d$, then the exact solution
of (\ref{eq:gpegrot:sec5}) satisfies:
\be
\label{eq:exacts1:sec5}
\psi(\bx,t)=\phi_e(\bx-\bx_c(t))\;e^{-i\mu_e t}\; e^{iw(\bx,t)},
\qquad \bx\in{\mathbb R}^d, \quad t\ge 0,
\ee
where for any time $t\ge0$, $w(\bx,t)$ is linear for $\bx$, i.e.
\be
\label{eq:exact3:sec5}
w(\bx,t) = {\bf c}(t) \cdot \bx + g(t), \qquad {\bf c}(t)=(c_1(t), \cdots,
c_d(t))^T, \qquad \bx\in {\mathbb R}^d, \quad t\ge0,
\ee
and $\bx_c(t)=(x_c(t),y_c(t))^T$ in 2D, and resp., $\bx_c(t)=(x_c(t),y_c(t),z_c(t))^T$ in 3D, satisfies the following second-order ODE system:
\bea
\label{eq:govern_eq1:sec5}
&&\ddot{x}_c(t)-2\Og\dot{y}_c(t)+\left(\gm_x^2-\Og^2\right)x_c(t) = 0, \\
\label{eq:govern_eq2:sec5}
&&\ddot{y}_c(t)+2\Og\dot{x}_c(t)+\left(\gm_y^2-\Og^2\right)y_c(t) = 0,
\qquad t\geq 0, \\
\label{eq:govern_eq3:sec5}
&&x_c(0) = x_0, \quad  y_c(0) = y_0, \qquad \dot{x}_c(0) =
\Og y_0, \qquad  \dot{y}_c(0) = -\Og x_0.
\eea
Moreover, if in 3D, another ODE needs to be added:
\be
\label{eq:govern_eq4:sec5}
\ddot{z}_c(t)+\gm_z^2 z_c(t) = 0,\qquad
z_c(0) = z_0, \qquad \dot{z}_c(0) = 0.
\ee
\end{lemma}
We note that the above ODE system (\ref{eq:govern_eq1:sec5})-(\ref{eq:govern_eq3:sec5}) can be solved analytically \cite{BaoDuZhang}.

\section{Numerical methods for rotational BEC}
\label{sec:numrotat}
\setcounter{equation}{0}\setcounter{figure}{0}\setcounter{table}{0}
In this section, we first present efficient and accurate numerical methods for computing ground states, as well as the central vortex states (\ref{eq:centralvor:sec5}) of rotating BEC modeled by the GPE with an angular momentum term (\ref{eq:gpegrot:sec5}). To compute the dynamics of rotating BEC based on GPE (\ref{eq:gpegrot:sec5}), there are new difficulties due to the angular momentum term $L_z$ in (\ref{eq:gpegrot:sec5}). We will show how to design efficient numerical methods for simulating the dynamics of the rotational GPE (\ref{eq:gpegrot:sec5}). In most cases, we consider the potential $V(\bx)$ given as (\ref{eq:pot:sec5}) if no further clarification.

\subsection{Computing  ground states}
In order to find the ground state for rotational GPE (\ref{eq:gpegrot:sec5}), we need to solve the minimization problem (\ref{eq:minim:sec5}). Analogous to the non-rotating case (cf. section \ref{sec:numgs}), we adopt the gradient flow with discrete normalization (GFDN) method (or imaginary time method) \cite{BaoWang2}.
 \begin{align} \label{eq:ngf1:sec6}
&\phi_t = -\fl{\dt E_{\beta,\Og}(\phi)}{\dt \overline{\phi}}=\fl{1}{2}\nabla^2
\phi - V(\bx) \phi -\bt\; |\phi|^2\phi+\Og\; L_z \phi,
\quad  t_n<t<t_{n+1},  \qquad  \\
\label{eq:ngf2:sec6}
&\phi(\bx,t_{n+1})\stackrel{\triangle}{=}
\phi(\bx,t_{n+1}^+)=\fl{\phi(\bx,t_{n+1}^-)}{\|\phi(\cdot,t_{n+1}^-)\|_2},
\quad \bx\in {\Bbb R}^d, \quad d=2,3,\quad n\ge 0,\\
\label{eq:ngf3:sec6}
&\phi(\bx,0)=\phi_0(\bx), \qquad \bx \in {\Bbb R}^d \qquad
\hbox{with}\quad \|\phi_0\|_2=1; \end{align}
where $0=t_0<t_1<t_2<\cdots<t_n<\cdots$ with $\tau_{n}=t_{n+1}-t_{n}>0$ and $\tau=\max_{n\ge 0} \; \tau_n$, and
$\phi(\bx, t_n^\pm)=\lim_{t\to t_n^\pm} \phi(\bx,t)$.
As stated in section \ref{sec:numgs}, the gradient
flow  (\ref{eq:ngf1:sec6}) can be viewed as applying the steepest descent
method to the energy functional $E_{\beta,\Og}(\phi)$ (\ref{eq:engd1:sec5}) without
constraint  and (\ref{eq:ngf2:sec6}) then projects the solution back to
the unit sphere in order to satisfy the constraint (\ref{eq:normgg:sec5}).

In non-rotating BEC, i.e. $\Og=0$,
the  unique real valued nonnegative ground state solution
$\phi_g(\bx)\ge0$ for all $\bx\in{\Bbb R}^d$  is
obtained by choosing a positive initial datum $\phi_0(\bx)\ge0$ for
$\bx\in{\Bbb R}^d$, e.g., the ground state solution of linear
Schr\"{o}dinger equation with a harmonic oscillator potential
\cite{Bao,BaoDu}. For rotating BEC, e.g., $|\Og|<\min\{\gamma_{x},\gamma_y\}$, our numerical
experiences suggest that the initial data can be
chosen as a linear combination of the ground state and central
vortex ground state with index $m=1$ of (\ref{eq:gpegrot:sec5}) when $\bt=0$ and $\Og=0$,
which are given explicitly in section \ref{subsec:numrot}.

\subsubsection{Backward Euler finite difference method} In order to derive a full discretization of
the GFDN (\ref{eq:ngf1:sec6})-(\ref{eq:ngf3:sec6}), we first
truncate  the physical domain of the problem to a
rectangle in 2D or a box in 3D with homogeneous Dirichlet
boundary condition, and then
apply backward Euler
for time discretization and second-order centered finite difference
for spatial derivatives. The backward Euler finite difference method is similar to
the BEFD discretization for non-rotating BEC in section \ref{subsec:BEFD} and we omit the details here for brevity \cite{BaoWang2}.

\subsubsection{Backward Euler Fourier pseudospectral method} Computing the ground state of rotating BEC is a very challenging problem, especially for fast rotational speed $\Omega$ close to $\min\{\gamma_x,\gamma_y\}$, where many vortices exist in the ground state. In order to achieve high resolution for the vortex structure, a very fine mesh must be used if BEFD is used for discretizing the GFDN (\ref{eq:ngf1:sec6})-(\ref{eq:ngf3:sec6}), when $\Omega$ is large.
Here, we propose a Fourier pseudospectral discretization in space for the GFDN (\ref{eq:ngf1:sec6})-(\ref{eq:ngf3:sec6}) to maintain the accuracy for high rotational speed $\Omega$ with less computational cost.

 We first truncate the problem (\ref{eq:ngf1:sec6})-(\ref{eq:ngf3:sec6}) in 2D on a rectangle
$U=[a,b]\tm[c,d]$, and resp. in 3D on a box
$U=[a,b]\tm[c,d]\tm[e,f]$ with homogeneous Dirichlet
boundary condition:
 \begin{align} \label{eq:ngf1t:sec6}
&\phi_t = \left[\fl{1}{2}\nabla^2
 - V(\bx) -\bt |\phi|^2+\Og\; L_z \right]\phi,
\qquad  t_n<t<t_{n+1}, \quad \bx\in U,  \\
\label{eq:ngf2t:sec6}
&\phi(\bx,t_{n+1})\stackrel{\triangle}{=}
\phi(\bx,t_{n+1}^+)=\fl{\phi(\bx,t_{n+1}^-)}{\|\phi(\cdot,t_{n+1}^-)\|_{L^2(U)}},
\quad \bx\in U, \quad n\ge 0,\\
\label{eq:ngf3t:sec6}
&\phi(\bx,\cdot)|_{\p U}=0,\quad \phi(\bx,0)=\phi_0(\bx), \quad \bx \in U \quad
\hbox{with}\quad \|\phi_0\|_{L^2(U)}=1; \end{align}
where $\phi(\bx, t_n^\pm)=\lim_{t\to t_n^\pm} \phi(\bx,t)$.

For the simplicity of notation, we only present the methods in 2D. Generalizations to $d=3$ are straightforward
for tensor product grids and the results remain valid without
modifications.  Choose mesh sizes $\Delta x:=\frac{b-a}{M}$ and
$\Delta y:=\frac{d-c}{N}$ with $M$ and $N$ two even positive integers and
denote the  grid points
as
\be\label{eq:grid:sec5}
x_j:=a+j\,\Delta x,\quad j=0,1,\ldots, M;
\qquad y_k:=c+k\,\Delta y,\quad  k=0,1,\ldots,N.\ee
Define the index sets
 \begin{eqnarray*} &&{\calT}_{MN}=\{(j,k)\ |\ j=1,2,\ldots,M-1,\
k=1,2,\ldots, N-1\}, \\
&&{\calT}_{MN}^0=\{(j,k)\ |\ j=0,1,2\ldots,M,\ k=0,1,2\ldots,N\}.
 \end{eqnarray*}
Let $\phi_{jk}^n$ be the numerical approximation of the solution
$\phi(x_j,y_k,t_n)$ of the GFDN (\ref{eq:ngf1t:sec6})-(\ref{eq:ngf3t:sec6}) for $(j,k)\in {\calT}_{MN}^0$ and $n\ge0$, and
denote $\phi^n\in {\Bbb C}^{(M+1)\tm(N+1)}$ to be the numerical
approximate solution at time $t=t_n$. Define
\be\label{eq:fourdomain:sec6}
\lambda^x_p=\frac{2p\pi}{b-a},\quad \lambda_q^y=\frac{2q\pi}{d-c},\quad p= -\frac M2,\ldots,\frac M2-1,\,
q=-\frac N2,\ldots, \frac N2-1.
\ee

Then the backward Euler Fourier pseudospectral (BEFP) discretization for solving GFDN (\ref{eq:ngf1t:sec6})-(\ref{eq:ngf3t:sec6}) reads as \cite{BaoWang2,BaoChernZhang,Zeng}
\begin{align}\label{eq:befp2d:sec6}
&\frac{\phi_{jk}^{(1)}-\phi_{jk}^n}{\tau}=\fl{1}{2 }\nabla^2_s\phi^{(1)}\big|_{jk}+i\Omega\left(y\p_x^s\phi^{(1)}
-x\p_y^s\phi^{(1)}\right)\big|_{jk}
-\left[\beta|\phi_{jk}^n|^2+V_{jk}\right]\phi_{jk}^{(1)},\\
&\phi_{jk}^{n+1}=\frac{\phi_{jk}^{(1)}}{\|\phi_{jk}^{(1)}\|_2},\quad
\phi_{jk}^0=\phi_0(x_j,y_k),\quad (j,k)\in\calT_{MN}^0,
\end{align}
 where $V_{jk}=V(x_j,y_k)$ for $(j,k)\in\calT_{MN}^0$,
 $\|\phi^{(1)}\|_2$ denotes the discrete $l^2$ norm  given by $\|\phi^{(1)}\|_2^2=\Delta x\Delta y\sum\limits_{j=0}^{M-1}\sum\limits_{k=0}^{N-1}|\phi^{(1)}_{jk}|^2$,
and $\nabla^2_s$,  $\p_x^s$ and $\p_y^s$ are the Fourier pseudospectral approximations of $\nabla^2$, $\p_x$ and $\p_y$, respectively,  which can be written as
\begin{align}
&\left(\nabla^2_s\phi\right)_{jk}= \frac{-1}{MN}\sum\limits_{p=-M/2}^{M/2-1}\,
\sum\limits_{q=-N/2}^{N/2-1}\left[(\lambda_p^x)^2+(\lambda_q^y)^2\right]
\widehat{\phi}_{pq}e^{i\lambda^x_p(x_j-a)}
e^{i\lambda^y_q(y_k-c)},\label{eq:pslap:sec5}\\
&\left(\p_x^s\phi\right)_{jk}= \frac{i}{MN}\sum\limits_{p=-M/2}^{M/2-1}\,
\sum\limits_{q=-N/2}^{N/2-1}\lambda_p^x\widehat{\phi}_{pq}e^{i\lambda^x_p(x_j-a)}
e^{i\lambda^y_q(y_k-c)},\label{eq:psdx:sec5}\\
&\left(\p^s_y\phi\right)_{jk}= \frac{i}{MN}\sum\limits_{p=-M/2}^{M/2-1}\,
\sum\limits_{q=-N/2}^{N/2-1}\lambda_q^y\widehat{\phi}_{pq}e^{i\lambda^x_p(x_j-a)}
e^{i\lambda^y_q(y_k-c)},\label{eq:psdy:sec5}
\end{align}
for $-M/2\leq p\leq M/2-1$, $-K/2\leq q\leq K/2-1$. Here $\widehat{\phi}_{pq}$ denotes the Fourier coefficients of mesh function $\phi_{jk}$ as
\be\label{eq:fourcoe:sec6}
\hat{\phi}_{pq}=\sum\limits_{j=0}^{M-1}\sum\limits_{k=0}^{N-1}\phi_{jk}e^{-i\frac{2jp\pi}{M}}
e^{-i\frac{2kq\pi}{N}}=\sum\limits_{j=0}^{M-1}\sum\limits_{k=0}^{N-1}\phi_{jk}e^{-i\lambda^x_p(x_j-a)}
e^{-i\lambda^y_q(y_k-c)}.
\ee
Similar to those in section 3.3, at every time step, we can design an iterative method to solve
the linear system (\ref{eq:befp2d:sec6}) for $\phi^{(1)}$ via discrete Fourier transform with a stabilization term.
We omit the details here for brevity.

\begin{remark}\label{lem:calenergy:sec5} For large $\Omega$, there exists many local minimums for the energy (\ref{eq:engd1:sec5}). To calculate the energy accurately, when the pseudospectral discretization BEFP (\ref{eq:befp2d:sec6}) is used, the terms involving derivatives in energy (\ref{eq:engd1:sec5}) should use the pseudospectral approximations
like (\ref{eq:pslap:sec5})-(\ref{eq:psdy:sec5}). For a numerical approximation $\phi^n_{jk}$ given by the BEFP (\ref{eq:befp2d:sec6}), the discretized energy $E^h_{\beta,\Omega}(\phi^n)$ can be computed as
\begin{align}
E^h_{\beta,\Omega}(\phi^n)=&\Delta x\Delta y\sum\limits_{j=0}^{M-1}\sum\limits_{k=0}^{N-1}\Big[\frac12|(\p_x^s\phi^n)_{jk}|^2
+\frac12|(\p_y^s\phi^n)_{jk}|^2+V(x_j,y_k)|\phi^n_{jk}|^2\nn\\
&
+i\Omega\left(x_j(\p_y^s\phi^n)_{jk}-y_k(\p_x^s\phi^n)_{jk}\right)\bar{\phi}^n_{jk}
+
\frac{\beta}{2}|\phi^n_{jk}|^2\Big].
\end{align}
 \end{remark}
\subsection{Central vortex states with polar/cylindrical symmetry}
\label{subsec:central}
As shown in Theorem \ref{thm:symbre:sec5}, if the potential is radially symmetric, the  ground state density of a rotating BEC may be no longer radially symmetric.
Thus, those simplified finite difference methods for computing ground states of non-rotating BEC with radially symmetric or cylindrically symmetric potentials  in section \ref{subsec:sympotgs} can not be directly used for computing ground states of rotational GPE (\ref{eq:gpegrot:sec5}).

For central vortex state (\ref{eq:centralvor:sec5}) and central vortex line state (\ref{eq:statv35:sec5}), when potential $V$ is radially symmetric in 2D or cylindrically symmetric in 3D, finding the central vortex state (\ref{eq:centralvor:sec5}) in 2D or the central line vortex state (\ref{eq:statv35:sec5}) in 3D, is almost the same as computing the radially symmetric (2D) or cylindrically symmetric (3D) ground states in section \ref{subsec:sympotgs}. Simplified backward Euler finite difference method can be directly used here.

\subsubsection{Formulation of the problem with cylindrical symmetry}
When we consider the harmonic potential $V(\bx)$ (\ref{eq:pot:sec5}), the polar and cylindrical symmetries lead to new efficient and accurate numerical methods.  Since angular momentum rotation does not affect the central vortex states, here we will only consider the GPE (\ref{eq:gpegrot:sec5}) with $\Omega=0$. In particular, for 3D,  we consider GPE
\begin{equation}\label{eq:gpegrot:sec6}
i\;\frac{\partial}{\partial t} \psi =
\left[-\frac{1}{2}\left(\frac{\p^2}{\partial
x^2}+\frac{\p^2}{\partial y^2}+\frac{\p^2}{\partial z^2}\right)
  + V(x,y,z)+
\beta |\psi|^2 \right]\psi,
\end{equation}
where $V(\bx)=\frac{1}{2}(\gamma_r^2(x^2+y^2)+\gamma_z^2z^2)+W(z)$,
and $\psi :=\psi(x,y,z,t)$ is the normalized wave function of the
condensate with
\begin{equation}\label{eq:norm:sec6}
\|\psi(x,y,z,t)\|_2^2 = \int_{{\Bbb R}^3} |\psi(x,y,z,t)|^2 \;dxdydz
= 1.
\end{equation}
To find cylindrically symmetric states ($m=0$) and  central
vortex line states with index or winding number $m$ ($m\ne0$) for
the BEC, we write \cite{BaoWangP} \be\label{eq:antz:sec6} \psi(x,y,z,t)=e^{-i\mu_m
t}\phi_m(x,y,z)=e^{-i\mu_m t}\phi_m(r,z) e^{im\tht}, \ee
where
$(r,\tht,z)$ is the cylindrical coordinates, $\mu_m$ is the
chemical potential, $\phi_m=\phi_m(r,z)$ is a function independent
of time $t$ and angle $\tht$.
Denote
\begin{equation}
  \label{eq:bm:sec6}
\begin{split}
&  B^r_m\phi:=\frac{1}{2}\left[-\frac{1}{r}\frac{\p}{\p r}\left(r\frac{\p}{\p
r}\right)
  + \gm_r^2r^2+\frac{m^2}{r^2}\right]\phi,\quad B^z\phi:=\frac{1}{2}\left[-\frac{\p^2}{\partial
      z^2}+\gm_z^2z^2\right]\phi,\\
&B_m := B^r_m+B^z.
\end{split}
\end{equation}
Plugging (\ref{eq:antz:sec6}) into the GPE
(\ref{eq:gpegrot:sec6}) and the normalization condition (\ref{eq:norm:sec6}), we
obtain  (see \cite{BaoZhang,BaoZhang2,BaoWangP,BaoShen2}  for more details)
\begin{align}\label{eq:neng1:sec6}&\mu_m\;\phi_m =
\left[B_m+W(z)+
\beta |\phi_m|^2 \right]\phi_m,\qquad (r,z)\in
(0,+\infty)\times(-\infty,+\infty), \\
\label{eq:neng2:sec6} &\phi_m(0,z)=0 \quad (\hbox{for}\ m\ne0), \qquad -\ift<z<\ift, \\
\label{eq:neng3:sec6} &\lim_{r\to\ift}\phi_m(r,z)=0, \quad -\ift<z<\ift,
\qquad \lim_{|z|\to \ift}\phi_m(r,z)=0, \quad 0\le r<\ift, \end{align}
under the normalization condition \be\label{eq:norm3d:sec6} \|\phi_m\|_c^2
=2\pi\int_0^\ift\int_{-\ift}^\ift |\phi_m(r,z)|^2\; r\; dzdr
=1.\ee
From a mathematical point of view, the symmetric states ($m=0$)
and central vortex line  states with index $m$ ($m\ne0$) of the
BEC are defined as the minimizer of the  nonconvex
minimization problem (\ref{eq:mincy:sec5}).

To compute the symmetric states and central vortex line states of BEC,  we use the gradient flow with discrete normalization (GFDN) method \cite{BaoShen2}:
\begin{eqnarray}\label{eq:12:sec6}
&&\frac{\partial}{\partial t} \phi(r,z,t)=-B_m\phi
-\left[W(z)+\beta|\phi|^2\right]\phi,\quad  t_n \le t <
t_{n+1}, \quad n\ge 0,\\
 \label{eq:14:sec6} &&\phi(0,z,t)=0 \quad (\hbox{for}\
m\ne0),
\qquad z\in \mathbb{R},, \quad t\ge0, \\
\label{eq:15:sec6} &&\lim_{r\to\ift}\phi(r,z,t)=0, \; z\in \mathbb{R},
\quad \lim_{|z|\to \ift}\phi(r,z,t)=0, \; r\in \mathbb{R^+}=(0,\infty),\\
\label{eq:13:sec6} &&\phi(r,z,t_{n+1}) :=\phi(r,z,t_{n+1}^+) =
\frac{\phi(r,z,t_{n+1}^-)} {\|\phi(\cdot,t_{n+1}^-)\|_c},
 \\
\label{eq:131:sec6} &&\phi(r,z,0) = \phi_0(r,z), \qquad \hbox{with}
\quad \|\phi_0(\cdot)\|_c=1;
\end{eqnarray}
where $0=t_0<t_1<\cdots$,
$\|\phi(\cdot)\|_c^2=2\pi\int_0^\ift\int_{-\ift}^\ift |\phi(r,z)|^2
\;r \; dzdr$, and $\phi(r,z,t_n^\pm)=\lim_{t\to t_n^\pm}
\phi(r,z,t)$.

For the time discretization of
(\ref{eq:12:sec6})-(\ref{eq:131:sec6}), we adopt the following
 backward Euler scheme with projection:

 Given $\phi^0$,
 find  $\tilde \phi^{n+1}$ and $\phi_{MN}^{n+1}$
such that
\bea\label{eq:dis3d1:sec6}
&&\frac{\tilde \phi^{n+1}(r,z)-\phi^{n}(r,z)}{\tau} =-B_m\;
\tilde \phi^{n+1} -\left(
W(z)+\beta\;|\phi^{n}|^2\right)\tilde \phi^{n+1},\qquad \\
\label{eq:dis3d2:sec6} &&\phi^{n+1}(r,z)=
\frac{\tilde \phi^{n+1}(r,z)}{\|\tilde \phi^{n+1}\|_c}.
 \eea

For $\beta=0$, it is shown in section \ref{sec:numgs} (cf. \cite{BaoDu}) that
 the scheme \eqref{eq:dis3d1:sec6} is energy diminishing. However,
 (\ref{eq:dis3d1:sec6}) involves non-constant
coefficients so it can not be solved by a direct fast spectral
solver. Therefore, we propose to solve (\ref{eq:dis3d1:sec6})  iteratively
(for $p=0,1,2,\ldots$) by introducing a stabilization term with
constant coefficient (cf. section \ref{subsec:besp}) \begin{align}& \frac{\tilde
\phi^{n+1,p+1}-\phi^{n}}{\tau} =-(B_m+\ap_n) \tilde
\phi^{n+1,p+1}+\left(\ap_n-W(z)-\beta |\phi^{n}|^2\right)\tilde
\phi^{n+1,p}, \label{eq:sol3d:sec6}\\
 &\tilde \phi^{n+1,0}=\phi^n,\quad \tilde \phi^{n+1}=\lim_{p\rightarrow \infty}
\tilde \phi^{n+1,p},  \quad
\phi^{n+1}=\frac{\tilde \phi^{n+1}}{\|\tilde \phi^{n+1}\|_c}. \label{eq:sol3d2:sec6}
\end{align}
The stabilization factor $\ap_n$ is chosen such that the convergence of
the iteration is `optimal' and $\ap_n$ should be chosen as (cf. section \ref{subsec:besp} and \cite{BaoChernLim})  $\ap_n=\frac{1}{2}\left(b^n_{\rm min}+b^n_{\rm
max}\right)$ with
 \begin{equation*} b^n_{\rm min}= \min_{(r,z)\in\mathbb{\bar R^+}\times\mathbb{R}}
 \left[W(z)+\beta |\phi^{n}(r,z)|^2\right],\;
b^n_{\rm max}= \max_{(r,z)\in\mathbb{\bar{R^+}}\times\mathbb{R}} \left[W(z)+\beta
|\phi^{n}(r,z)|^2\right].\end{equation*}

\subsubsection{Eigenfunctions of $B_m$}The numerical scheme for
(\ref{eq:12:sec6})-(\ref{eq:131:sec6}) requires solving, repeatedly, (\ref{eq:sol3d:sec6}).
Therefore,  it is most convenient
to  use eigenfunctions of $B_m$ as basis functions. Thanks to
(\ref{eq:bm:sec6}), we only need to find eigenfunctions of $B_m^r$ and $B^z$.
We shall construct these eigenfunctions by properly
scaling the  Hermite polynomials and
generalized Laguerre polynomials.

We start with $B^z$. Let $H_l(z)$ ($l=0,1,2,\ldots$) be the
standard Hermite polynomials of degree $l$ satisfying \cite{BaoShen,BaoShen2,BaoLiShen}
\be\label{eq:GHF1:sec6} H_l^{\prime\prime}(z)-2z\;H_l^\prime(z) +
2l\;H_l(z)=0,\; z\in\mathbb{R}, \qquad l=0,1,2,\ldots,\ee
\be\label{eq:GHF2:sec6}\int_{-\ift}^\ift  H_l(z) \;H_{l^\prime}(z)\;
e^{-z^2}\; dz = \sqrt{\pi}\;2^l\;l! \;\delta_{ll^\prime}, \quad
l,l^\prime=0,1,2,\ldots\;,\ee where $\delta_{ll^\prime}$ is the
Kronecker delta.

Define the scaled Hermite
functions  \be\label{eq:GHF3:sec6} h_l(z)=e^{-\gm_z z^2/2}\;
H_l\left(\sqrt{\gm_z} z\right)/\sqrt{2^l\;l!} (\gm_z/\pi)^{1/4},
\qquad z\in \mathbb{R}.\ee
It is clear that $\lim_{|z|\to\ift} h_l(z) =0$.

Plugging (\ref{eq:GHF3:sec6}) into (\ref{eq:GHF1:sec6}) and (\ref{eq:GHF2:sec6}), a simple
computation  shows \be\label{eq:GHF4:sec6}
B^z h_l(z)=-\frac{1}{2}h_l^{\prime\prime}(z)+\frac{1}{2}\gm_z^2z^2 h_l(z)=
\left(l+\frac{1}{2}\right)\gm_z\; h_l(z), \quad z\in {\Bbb R},
\quad
  l\ge0,\ee \be \label{eq:GHF5:sec6} \int_{-\ift}^\ift
h_l(z)\; h_{l^\prime}(z)\; dz =\delta_{ll^\prime}, \qquad
l,l^\prime=0,1,2,\ldots\;.\ee Hence $\{h_l\}_{l=0}^\ift$ are
eigenfunctions of the linear operator $B^z$ in  (\ref{eq:bm:sec6}).

We now consider $B_m^r$.
To this end,
we recall  the definition for the generalized Laguerre  polynomials.

For any fixed $m$ ($m=0,1,2,\ldots$), let $\hat{L}_k^m(r)$
($k=0,1,2,\ldots$) be the the generalized-Laguerre polynomials of
degree $k$ satisfying \cite{GS} \be\label{eq:GLF1:sec6} \left(r
\frac{\rd^2}{\rd r^2}+(m+1-r)\frac{\rd}{\rd r}\right)\hat{L}_k^m(r) +k
\;\hat{L}_k^m(r)=0, \qquad k=0,1,2,\ldots,\ee
\be\label{eq:GLF2:sec6}\int_0^\ift r^m\; e^{-r}\; \hat{L}_k^m(r)\;
\hat{L}_{k^\prime}^m(r)\;dr = C_k^m \;\delta_{kk^\prime}, \quad
k,k^\prime=0,1,2,\ldots,\ee  where
\[C_k^m=\Gamma(m+1)\left(\ba{c} k+m\\
k\\
\ea\right)=\prod_{j=1}^m (k+j), \qquad k=0,1,2,\ldots\;.\]

We
define the scaled generalized-Laguerre functions $L_k^m$ by
\be\label{eq:GLF3:sec6} L_k^m(r)=\frac{\gm_r^{(m+1)/2}}{\sqrt{\pi
C_k^m}}\; r^m\; e^{-\gm_r r^2/2}\; \hat{L}_k^m(\gm_r r^2).\ee

Plugging (\ref{eq:GLF3:sec6}) into (\ref{eq:GLF1:sec6}) and (\ref{eq:GLF2:sec6}), direct
computation  leads to
 \begin{align}\label{eq:GLF4:sec6}
&B^r_m L_k^m(r)=\gm_r(2k+m+1)
L_k^m(r),\\ &\label{eq:GLF5:sec6} 2\pi\int_0^\ift
L_k^m(r)\; L_{k^\prime}^m(r)\; r\; dr =\delta_{kk^\prime}.
\end{align}
 Hence $\{L_k^m\}_{k=0}^\ift$ are
eigenfunctions of  $B^r_m$.

Finally we derive from the above that \cite{BaoShen,BaoShen2,BaoLiShen}
 \bea\label{eq:GLH1:sec6}
B_m(L_k^m(r)h_l(z))
&=&h_l(z)B^r_m L_k^m(r)
+L_k^m(r) B^zh_l(z)\\
&=&\gm_r(2k+m+1)L_k^m(r)h_l(z)+\gm_z\left(l+\frac{1}{2}\right)L_k^m(r)h_l(z)\nn\\
&=&\left[\gm_r(2k+m+1)+\gm_z\left(l+\frac{1}{2}\right)\right]L_k^m(r)h_l(z).
\eea Hence, $\{L_k^m(r)h_l(z)\}_{k,l=0}^\ift$ are eigenfunctions
of the operator $B_m$ defined in (\ref{eq:bm:sec6}).

\subsubsection{Interpolation operators}
In order to efficiently deal with the term
$|\phi^{n}|^2\tilde \phi^{n+1,p}$ in \eqref{eq:sol3d:sec6}, a proper interpolation
operator should be used.  We shall define below scaled
interpolation operators in both $r$, $z$ directions and in the $(r,z)$
space.

Let $\{\hat z_s\}_{s=0}^N$ be the Hermite-Gauss points, i.e., they
are the $N+1$ roots of the Hermite polynomial $H_{N+1}(z)$, and
let $\{\hat \omega^z_s\}_{s=0}^N$ be the associated Hermite-Gauss
quadrature weights \cite{GS}. We have
 \be\label{eq:quad1d1:sec6} \sum_{s=0}^N
\hat{\og}_s^z \; \frac{H_l(\hat{z}_s)}{\pi^{1/4}\sqrt{2^l\;l!}}\;
\frac{H_{l^\prime}(\hat{z}_s)}{\pi^{1/4}\sqrt{2^{l^\prime}\;l^\prime!}}=
\delta_{ll^\prime}, \qquad l,l^\prime=0,1,\ldots,N.\ee
We then define the scaled Hermite-Gauss points and
weights by \be\label{eq:pot1d:sec6} z_s
=\frac{\hat{z}_s}{\sqrt{\gm_z}}, \qquad \og_s^z=\frac{
\hat{\og}_s^z \; e^{\hat{z}_s^2}}{\sqrt{\gm_z}}, \quad
s=0,1,2,\ldots,N.\ee
We derive from (\ref{eq:GHF3:sec6}) and (\ref{eq:quad1d1:sec6}) that
 \bea\label{eq:quad1d2:sec6} \sum_{s=0}^N
\og_s^z \; h_l(z_s)\; h_{l^\prime}(z_s)&=&\delta_{ll^\prime}, \qquad l,l^\prime=0,1,\ldots,N. \eea

Let us denote
\be\label{eq:yn:sec6}
Y_N^h=\text{span}\{h_k: k=0,1,\cdots,N\}.\ee
 We define
\begin{equation}\label{eq:in:sec6}
I^z_N: C(\mathbb{R}) \rightarrow Y_N^h \;\text{ such that }\;
  (I^z_Nf)(z_s)=f(z_s), 0\leq s\leq N, \;\forall f\in C(\mathbb{R}).
\end{equation}

Now, let $\{\hat{r}_j^m\}_{j=0}^M$ be the
generalized-Laguerre-Gauss points  \cite{GS,Shen2,BaoShen2}; i.e. they
are the $M+1$ roots of the polynomial $\hat{L}_{M+1}^m(r)$, and
let $\{\hat\og_j^m\}_{j=0}^M$ be the weights associated with the
generalized-Laguerre-Gauss quadrature \cite{GS,Shen2,BaoShen2}. Then, we
have \be\label{eq:quad2d1:sec6} \sum_{j=0}^M \hat\og_j^m \;
\frac{\hat{L}_k^m(\hat{r}_j^m)}{\sqrt{C_k^m}}\;
\frac{\hat{L}_{k^\prime}^m(\hat{r}_j^m)}{\sqrt{C_{k^\prime}^m}}=
\delta_{kk^\prime}, \qquad k,k^\prime=0,1,\ldots,M.\ee
We then
define the scaled generalized-Laguerre-Gauss points and weights by
 \be\label{eq:pot2d:sec6} r_j^m =\sqrt{\frac{\hat{r}_j^m}{\gm_r}},
\qquad \og_j^m=\frac{\pi\; \hat\og_j^m \; e^{\hat{r}_j^m}}{\gm_r\;
\left(\hat{r}_j^m\right)^m}, \quad j=0,1,\ldots,M.\ee We derive
from (\ref{eq:GLF3:sec6}) and (\ref{eq:quad2d1:sec6})
 that \bea\label{eq:quad2d2:sec6} \sum_{j=0}^M
\og_j^m \; L_k^m(r_j^m)\; L_{k^\prime}^m(r_j^m)&=&\sum_{j=0}^M
\frac{\pi\; \hat\og_j^m \; e^{\hat{r}_j^m}}{\gm_r\;
\left(\hat{r}_j^m\right)^m}\;
L_k^m\left(\sqrt{\hat{r}_j^m/\gm_r}\right)\;
L_{k^\prime}^m\left(\sqrt{\hat{r}_j^m/\gm_r}\right) \nn\\
&=&\sum_{j=0}^M \hat{\og}_j^m \;
\frac{\hat{L}_k^m(\hat{r}_j^m)}{\sqrt{C_k^m}}\;
\frac{\hat{L}_{k^\prime}^m(\hat{r}_j^m)}{\sqrt{C_{k^\prime}^m}}\nn\\
&=&\delta_{kk^\prime}, \qquad k,k^\prime=0,1,\ldots,M. \eea

Let us denote
\be\label{eq:xn:sec6}
X^m_M=\text{span}\{L_k^m: k=0,1,\cdots,M\}.\ee
 We define
\begin{equation}\label{eq:im:sec6}
I^m_M: C(\mathbb{\bar R_+}) \rightarrow X^m_M \;\text{ such that }\;
  (I^m_Mf)(r^m_j)=f(r^m_j), 0\leq j\leq M, \;\forall f\in
C(\mathbb{\bar R_+}).
\end{equation}

Finally, let
\be\label{eq:xmn:sec6}
X^m_{MN}=\hbox{span}\{L_k^m(r)h_l(z)\
:\ k=0,1,2,\ldots,M, \ l=0,1,2,\ldots,N\}.\ee
 we define $I^m_{MN}: C(\mathbb{\bar R_+}\times \mathbb{R})
 \rightarrow X^m_{MN}$
 such that
\begin{equation}\label{eq:imn:sec6}
  (I^m_{MN}f)(r^m_j,z_s)=f(r^m_j,z_s),\ j=0,1,\cdots,M,\
s=0,1,\cdots,N,\ \forall f\in C(\mathbb{\bar R_+}\times
\mathbb{R}).
\end{equation}
It is clear that $I^m_{MN}=I^m_M\circ I_N^z$.

 Note
that the computation of the weights $\{\omega^m_j,\omega^z_s\}$ from
(\ref{eq:pot2d:sec6}) and (\ref{eq:pot1d:sec6}) is not a
stable process for  large $m$, $M$ and $N$. However, they can be
computed in a stable way as suggested in  the Appendix of
\cite{Shen2}.

\subsubsection{A Hermite pseudospectral method in 1D}
\label{subsubsec:hermite}
In this section, we introduce a Hermite pseudospectral method for
computing ground states of 1D BEC. In fact, when $\gm_r\gg\gm_z$
in (\ref{eq:gpegrot:sec6}), the 3D GPE (\ref{eq:gpegrot:sec6}) can be approximated
by a 1D GPE (cf. section \ref{subsubsec:dred}). In this case, the stationary
states satisfy \be \label{eq:neng1d1:sec6}\mu\;\phi =
\left[-\frac{1}{2}\frac{\p^2}{\p z^2}
  + \frac{1}{2}\gm_z^2z^2+W(z)+
\beta |\phi|^2 \right]\phi, \ee under the normalization
condition \be\label{eq:norm1d:sec6} \|\phi\|_2^2 =\int_{-\ift}^\ift
|\phi(z)|^2\; dz =1,\ee where $\phi=\phi(z)$.   The stationary states can
be viewed as the Euler-Lagrange equations of the energy
functional $E(\phi)$, defined as \be\label{eq:eng1d:Sec6} E(\phi)
 =\int_{-\ift}^\ift\left[\frac{1}{2}\left|\p_z\phi\right|^2
 +\left(\frac{1}{2}\gm_z^2 z^2
 +W(z)\right)|\phi|^2+\frac{\beta}{2}\left|\phi\right|^4\right]\;dz,
  \ee
under the constraint (\ref{eq:norm1d:sec6}). Similarly, in this case, the
normalized gradient flow (\ref{eq:12:sec6})-(\ref{eq:131:sec6}) collapses to \cite{BaoShen,BaoShen2,BaoLiShen}
\begin{eqnarray}\label{eq:121d:sec6}
&&\frac{\partial}{\partial t} \phi(z,t)
  =-B^z\phi-W(z)\phi-\beta |\phi|^2 \phi, \qquad  \\
\label{eq:141d:sec6} &&\lim_{|z|\to\ift}\phi(z,t)=0, \qquad t\ge 0,\\
\label{eq:131d:sec6} &&\phi(z,t_{n+1}) :=\phi(z,t_{n+1}^+) =
\frac{\phi(z,t_{n+1}^-)} {\|\phi(\cdot,t_{n+1}^-)\|_2},
 \\
\label{eq:1311d:sec6} &&\phi(z,0) = \phi_0(z), \quad z\in\mathbb{R}
\qquad \hbox{with} \quad \|\phi_0(\cdot)\|_2=1, \qquad
\end{eqnarray}
where $\phi(z,t_n^\pm)=\lim_{t\to t_n^\pm} \phi(z,t)$,
$\|\phi(\cdot)\|_2^2=\int_{-\ift}^\ift |\phi(z)|^2  \; dz$.

Similarly, the scheme (\ref{eq:sol3d:sec6}) in this case becomes:
\be\label{eq:sol1d:sec6} \frac{\tilde
\phi^{n+1,p+1}(z)-\phi^{n}(z)}{\tau} =-(B^z+\ap_n) \tilde
\phi^{n+1,p+1}+\left(\ap_n-W(z)-\beta |\phi^{n}|^2\right)\tilde
\phi^{n+1,p}. \ee We now describe a pseudo-spectral method based
on the scaled Hermite functions $\{h_l(z)\}$ for
\eqref{eq:sol1d:sec6}-\eqref{eq:sol3d2:sec6}.

Let $(u,v)_{\mathbb{R}}=\int_{\mathbb{R}} u\,v dz$ and
$\phi_N^0\in Y_N^h$.
 For $n=0,1,\cdots$,  set $\tilde\phi_N^{n+1,0}=\phi_N^n$ and
$\ap_n=\frac{1}{2}\left(b^n_{\rm min}+b^n_{\rm
max}\right)$ with
\[b^n_{\rm min}= \min_{-\ift< z<\ift} \left[W(z)+\beta
|\phi_N^{n}(z)|^2\right], \qquad b^n_{\rm max}= \max_{-\ift< z<\ift}
\left[W(z)+\beta |\phi_N^{n}(z)|^2\right].\]
Then, the Hermite pseudospectral method for
(\ref{eq:sol1d:sec6})-\eqref{eq:sol3d2:sec6} is: \\
Find
$\tilde \phi_N^{n+1,p+1}\in Y_N^h$ such
that
\be\label{eq:sol1dn:sec6}
\begin{split}
&\left(\frac{\tilde \phi_N^{n+1,p+1}-\phi_N^{n}}{\tau}
+(B^z+\ap_n)
\tilde \phi_N^{n+1,p+1}, h_l\right)_{\mathbb{R}}\\
&\qquad =\left(I^z_N[(\ap_n-W-\beta
|\phi_N^{n})|^2)\tilde \phi_N^{n+1,p}],h_l\right)_{\mathbb{R}}, \; 0\le l\le
N,\;p=0,1,\cdots,\\
&\tilde \phi^{n+1}_N=\lim_{p\rightarrow \infty}
\tilde\phi^{n+1,p}_N,\quad
\phi^{n+1}_N=\frac{\tilde \phi^{n+1}_N}{\|\tilde \phi^{n+1}_N\|_2}.
\end{split}
\ee
We note that  $\tilde\phi^{n+1,p+1}_N$ can be easily determined from
\eqref{eq:sol1dn:sec6} as follows:

We write the expansion
\be \label{eq:GHF6:sec6}
\tilde \phi_N^{n+1,p+1}(z)=\sum_{l=0}^{N} \hat{\phi}_l^{n+1,p+1}\; h_l(z), \quad
\phi_N^{n}(z)=\sum_{l=0}^{N} \hat{\phi}_l^{n}\; h_l(z),\ee
and
\[g^{n,p}(z)=I^z_N\left[\left(\ap_n-W(z)-\beta
|\phi_N^{n}(z)|^2\right)\tilde\phi_N^{n+1,p}(z)\right]=
\sum_{l=0}^{N} \hat{g}_l^{n,p}\; h_l(z),\] where the coefficients
$\{\hat{g}_l^{n,p}\}_{l=0}^N$ can be computed from the known
function values $\{g^{n,p}(z_s)\}_{s=0}^N$ through the discrete
Hermite transform using \eqref{eq:quad1d2:sec6}, i.e., \be
\hat{g}_l^{n,p}=\sum_{s=0}^N g^{n,p}(z_s)\;
h_l(z_s)\;\omega^z_s.\ee Thanks to (\ref{eq:GHF4:sec6})-(\ref{eq:GHF5:sec6}), we
find from (\ref{eq:sol1dn:sec6}) that
 \be\label{eq:sol1d1:sec6}
\frac{\hat{\phi}_l^{n+1,p+1}-\hat{\phi}_l^{n}}{\tau} =
-\left[\gm_z\left(l+\frac{1}{2}\right)+\ap_n\right]\hat{\phi}_l^{n+1,p+1}
+ \hat{g}_l^{n,p},\quad l=0,1,\ldots,N,\ee
from which we derive
\be\label{eq:sol1d2:sec6}\hat{\phi}_l^{n+1,p+1}=\frac{\hat{\phi}_l^{n}+\tau \;\hat{g}_l^{n,p}}{1+\tau\left[\ap_n
+\gm_z\left(l+\frac{1}{2}\right)\right]}, \qquad
l=0,1,\ldots,N.\ee
Then, $\tilde\phi_N^{n+1}$ and $\phi_N^{n+1}$ can be determined from
the second equation in (\ref{eq:sol1dn:sec6}).

\subsubsection{A generalized-Laguerre pseudospectral method in
2D}

We now consider the
 2D BEC with radial symmetry. The physical motivation is that when
 $\gm_z\gg\gm_r$ in
(\ref{eq:gpegrot:sec6}), the 3D GPE (\ref{eq:gpegrot:sec6}) can be approximated by a
2D GPE (cf. section \ref{subsubsec:dred}). In this case, the radial symmetric
state ($m=0$) and central vortex state with index $m$ ($m\ne0$)
satisfy \bea \label{eq:neng2d1:sec6}&&\mu_m\;\phi_m =
\frac{1}{2}\left[-\frac{1}{r}\frac{\p}{\p r}\left(r\frac{\p}{\p
r}\right)
  + \gm_r^2r^2+\frac{m^2}{r^2}+
2\beta |\phi_m|^2 \right]\phi_m,\\
\label{eq:neng2d2:sec6} &&\phi_m(0)=0 \quad (\hbox{for}\ m\ne0), \qquad
\lim_{r\to\ift}\phi_m(r)=0, \eea under the normalization condition
\be\label{eq:norm2d:sec6} \|\phi_m\|_{r}^2 =2\pi\int_0^\ift |\phi_m(r)|^2\;
r\; dr =1,\ee where $\phi_m=\phi_m(r)$.
Again, this nonlinear eigenvalue problem
(\ref{eq:neng2d1:sec6})-(\ref{eq:norm2d:sec6}) can also be viewed as the
Euler-Lagrange equations of the energy functional $E(\phi_m)$,
defined by
 \be\label{eq:eng2d:sec6} E(\phi_m)
 =\pi \int_0^\ift\left[\left|\p_r\phi_m\right|^2+\left(\gm_r^2 r^2
 +\frac{m^2}{r^2}\right)\left|\phi_m\right|^2+\beta\left|\phi_m\right|^4\right]r\;dr,
  \ee
under the constraint (\ref{eq:norm2d:sec6}). Accordingly,
the normalized gradient flow (\ref{eq:12:sec6})-(\ref{eq:131:sec6}) collapses to \cite{BaoShen,BaoShen2,BaoLiShen}
\begin{eqnarray}\label{eq:122d:sec6}
&&\frac{\partial}{\partial t} \phi(r,t)
=-B^r_m\phi-\beta |\phi|^2 \phi, \qquad  \\
\label{eq:142d:sec6} &&\phi(0,t)=0 \quad (\hbox{for}\ m\ne0),
\qquad\lim_{r\to\ift}\phi(r,t)=0, \qquad t\ge 0,\\
\label{eq:132d:sec6} &&\phi(r,t_{n+1}) :=\phi(r,t_{n+1}^+) =
\frac{\phi(r,t_{n+1}^-)} {\|\phi(\cdot,t_{n+1}^-)\|_{r}},
 \\
\label{eq:1312d:sec6} &&\phi(r,0) = \phi_0(r), \quad 0\le r<\ift,
\qquad \hbox{with} \quad \|\phi_0(\cdot)\|_{r}=1, \qquad
\end{eqnarray}
where $\phi(r,t_n^\pm)=\lim_{t\to t_n^\pm} \phi(r,t)$,
$\|\phi(\cdot)\|_{r}^2=2\pi\int_0^\ift |\phi(r)|^2 \;r \; dr$. The
scheme (\ref{eq:sol3d:sec6}) in this case becomes: \be\label{eq:sol2d:sec6}
\frac{\tilde \phi^{n+1,p+1}-\phi^{n}(r)}{\tau} =-(B^r_m+\ap_n)
\tilde \phi^{n+1,p+1}+\left(\ap_n-\beta
|\phi^{n}|^2\right)\tilde \phi^{n+1,p}. \ee We now describe a
pseudo-spectral method based on the scaled generalized-Laguerre
functions $\{L_k^m(r)\}$ for (\ref{eq:sol2d:sec6})-(\ref{eq:sol3d2:sec6}).

Let $(u,v)_{r,\Bbb R^+}=\int_{\Bbb R^+} u\,v\,r\,dr$
 and $\phi_M^0\in X^m_M$.
 For $n=0,1,\cdots$,  set $\tilde\phi_M^{n+1,0}=\phi_M^n$ and
$\ap_n=\frac{1}{2}\left(b^n_{\rm min}+b^n_{\rm
max}\right)$ with
\[b^n_{\rm min}= \min_{0\le r<\ift} \left[\beta
|\phi_M^{n}(r)|^2\right], \qquad b^n_{\rm max}= \max_{0\le r<\ift}
\left[\beta |\phi_M^{n}(r)|^2\right].\]
Then, the generalized-Laguerre pseudospectral method for
\eqref{eq:sol2d:sec6}-\eqref{eq:sol3d2:sec6} is:

Find
$\tilde \phi_M^{n+1,p+1}\in X^m_M$ such
that
\be\label{eq:sol2dn:sec6}
\begin{split}
&\left(\frac{\tilde \phi_M^{n+1,p+1}-\phi_M^{n}}{\tau}
+(B^r_m+\ap_n)
\tilde \phi_M^{n+1,p+1}, L_k^m\right)_{r,\mathbb{R_+}}\\
&\qquad =\left(I^m_M[(\ap_n-\beta
|\phi_M^{n})|^2)\tilde \phi_M^{n+1,p}],L_k^m\right)_{r,\mathbb{R_+}}, \; 0\le k\le
M,\;p=0,1,\cdots,\\
&\tilde \phi^{n+1}_M=\lim_{p\rightarrow \infty}
\tilde\phi^{n+1,p}_M,\quad
\phi^{n+1}_M=\frac{\tilde \phi^{n+1}_M}{\|\tilde \phi^{n+1}_M\|_{r}}.
\end{split}
\ee
The function $\tilde\phi^{n+1,p+1}_M$ can be easily determined from
(\ref{eq:sol2dn:sec6}) as follows:

We write the expansion
\begin{align} \label{eq:GLF6:sec6}
&\tilde \phi_M^{n+1,p+1}(r)=\sum_{k=0}^{M} \hat{\phi}_k^{n+1,p+1}\;
L_k^m(r), \quad
\phi_M^{n}(r)=\sum_{k=0}^{M} \hat{\phi}_k^{n}\;L_k^m(r) ,\\
&g^{n,p}(z)=I^m_M\left[\left(\ap_n-\beta
|\phi_M^{n}(r)|^2\right)\tilde\phi_M^{n+1,p}(r)\right]=
\sum_{k=0}^{M} \hat{g}_k^{n,p}\; L_k^m(r),\nn\end{align} where the
coefficients $\{\hat{g}_k^{n,p}\}_{k=0}^M$ can be computed from
the known function values $\{g^{n,p}(r^m_j)\}_{j=0}^M$ through the
discrete generalized-Laguerre transform using (\ref{eq:quad2d2:sec6}),
i.e., \be \hat{g}_k^{n,p}=\sum_{j=0}^M g^{n,p}(r^m_j)\;
L_k^m(r^m_j)\;\omega^r_j.\ee Thanks to (\ref{eq:GLF4:sec6})-(\ref{eq:GLF5:sec6}),
we find from (\ref{eq:sol2dn:sec6}) that
 \be\label{eq:sol1d1b:sec6}
\frac{\hat{\phi}_k^{n+1,p+1}-\hat{\phi}_k^{n}}{\tau} =
-\left[\gm_r\left(2k+m+1\right)+\ap_n\right]\hat{\phi}_k^{n+1,p+1}
+ \hat{g}_k^{n,p},\quad k=0,\ldots,N,\ee
from which we derive
\be\label{eq:sol1d2b:sec6}\hat{\phi}_k^{n+1,p+1}=\frac{\hat{\phi}_k^{n}+\tau \;\hat{g}_k^{n,p}}{1+\tau\left[\ap_n
+\gm_r\left(2k+m+1\right)\right]}, \qquad
k=0,1,\ldots,N.\ee
Then, $\tilde\phi_M^{n+1}$ and $\phi_M^{n+1}$ can be determined from
the second equation in (\ref{eq:sol2dn:sec6}).

\subsubsection{A generalized-Laguerre-Hermite pseudospectral method in 3D}
We are now in position to describe the generalized-Laguerre-Hermite
pseudo-spectral method for computing symmetric and central vortex
line states of 3D BEC with cylindrical symmetry \cite{BaoShen,BaoShen2,BaoLiShen}.

Let $(u,v)_{r,\mathbb{R^+}\times\mathbb{R}}=\int_\mathbb{R}
\int_{\mathbb{R^+}}u\,v\, r\,dr\,dz$ and $\phi_{MN}^0\in
X^m_{MN}$.
 For $n=0,1,\cdots$,  set $\tilde\phi_{MN}^{n+1,0}=\phi_{MN}^n$ and
$\ap_n=\frac{1}{2}\left(b^n_{\rm min}+b^n_{\rm
max}\right)$ with $U=\mathbb{R^+}\times\mathbb{R}$,
\begin{equation*}b^n_{\rm min}= \min_{(r,z)\in U} \left[W(z)+\beta
|\phi_{MN}^{n}(r,z)|^2\right], \; b^n_{\rm max}=
\max_{(r,z)\in  U}
\left[W(z)+\beta |\phi_{MN}^{n}(r,z)|^2\right].\end{equation*}
Then,  the generalized-Laguerre-Hermite
pseudo-spectral method for  \eqref{eq:sol3d:sec6}-\eqref{eq:sol3d2:sec6} is: find
$\tilde \phi_{MN}^{n+1,p+1}\in X^m_{MN}$ such
that for $ \; 0\le k\le
M,\; 0\le l\le N,\;p=0,1,\cdots$,
\be\label{eq:sol3dn:sec6}
\begin{split}
&\left(\frac{\tilde \phi_{MN}^{n+1,p+1}-\phi_{MN}^{n}}{\tau} +(B_m+\ap_n)
\tilde \phi_{MN}^{n+1,p+1}, L_k^m(r)h_l(z)\right)_{r,\mathbb{R^+}\times\mathbb{R}}\\
&\qquad \quad=\left(I^m_{MN}[(\ap_n-W(z)-\beta
|\phi_{MN}^{n}|^2)\tilde \phi_{MN}^{n+1,p}],L_k^m(r)h_l(z)
\right)_{r,\mathbb{R^+}\times\mathbb{R}},\\
&\tilde \phi^{n+1}_{MN}=\lim_{p\rightarrow \infty}
\tilde\phi^{n+1,p}_{MN},\quad \phi^{n+1}_{MN}=\frac{\tilde
\phi^{n+1}_{MN}}{\|\tilde \phi^{n+1}_{MN}\|_c}.
\end{split}
\ee
The function $\tilde\phi^{n+1,p+1}_{MN}$ can be easily determined from
(\ref{eq:sol3dn:sec6}) as follows:

We write the expansion
\be \label{eq:GLHF6:sec6}
\tilde \phi_{MN}^{n+1,p+1}=\sum_{k=0}^{M}\sum_{l=0}^N
 \hat{\phi}_{kl}^{n+1,p+1}
L_k^m(r) h_l(z), \;
\phi_{MN}^{n}=\sum_{k=0}^{M}
\sum_{l=0}^N\hat{\phi}_{kl}^{n}L_k^m(r)h_l(z),\ee
and
\begin{equation*}g^{n,p}(r,z)=I^m_{MN}\left[\left(\ap_n-W(z)-\beta
|\phi_{MN}^{n}|^2\right)\tilde\phi_{MN}^{n+1,p}\right]=
\sum_{k=0}^{M} \sum_{l=0}^N\hat{g}_{kl}^{n,p}\; L_k^m(r)h_l(z),\end{equation*}
where the coefficients $\{\hat{g}_{kl}^{n,p}\}$ can be computed
from the known function values $\{g^{n,p}(r^m_j,z_s)\}$ through
the discrete generalized-Laguerre transform and discrete Hermite
transform using (\ref{eq:quad2d2:sec6}) and  (\ref{eq:quad1d2:sec6}), i.e., \be
\hat{g}_{kl}^{n,p}=\sum_{s=0}^N\sum_{j=0}^M g^{n,p}(r^m_j,z_s)\;
L_k^m(r^m_j)\;h_l(z_s)\;\omega^r_j\;\omega^z_s.\ee Thanks to
(\ref{eq:GLF4:sec6})-(\ref{eq:GLF5:sec6}) and (\ref{eq:GHF4:sec6})-(\ref{eq:GHF5:sec6}), we find
from
 (\ref{eq:sol3dn:sec6}) that
\be\label{eq:sol3d1:sec6}
\frac{\hat{\phi}_{kl}^{n+1,p+1}-\hat{\phi}_{kl}^{n}}{\tau} =
-\left[\gm_r(2k+m+1)\gm_z\left(l+\frac{1}{2}\right)+\ap_n\right]
\hat{\phi}_{kl}^{n+1,p+1}
+ \hat{g}_{kl}^{n,p},\ee
from which we derive
\be
\hat{\phi}_{kl}^{n+1,p+1}=\frac{\hat{\phi}_{kl}^{n}+\tau \;\hat{g}_{kl}^{n,p}}{1+\tau\left[\ap_n
+\gm_r(2k+m+1)+\gm_z\left(l+\frac{1}{2}\right)\right]}.\ee
Then,  $\tilde\phi_{MN}^{n+1}$ and $\phi_{MN}^{n+1}$ can be determined from
the second equation in \eqref{eq:sol3dn:sec6}.

\subsection{Numerical methods for dynamics}
In this section, we consider different numerical methods for solving the GPE (\ref{eq:gperot:sec5})
with an angular momentum rotation term in $d$-dimensions ($d=2,3$) for the dynamics of rotating BEC:
 \bea \label{eq:gpegrot2:sec6}
&&i\;\pl{\psi(\bx,t)}{t}=-\fl{1}{2}\btd^2 \psi+ V(\bx)\psi +
\beta|\psi|^2\psi-\Omega L_z \psi, \quad \bx\in {\Bbb R}^d,
\quad t>0, \qquad \\
\label{eq:gpegp:sec6}
&&\psi(\bx,0)=\psi_0(\bx), \qquad \bx\in {\Bbb R}^d,\qquad\|\psi_0\|_2=1,
\eea
 where $L_z=-i (x\p_y-y\p_x)$ and $V(\bx)$ in $d$ dimensions is given in (\ref{eq:pot:sec5}).

In fact, many efficient and accurate numerical methods have been proposed for discretizing the above GPE,
such as the time-splitting Fourier pseudospectral method via the alternating direction implicit (ADI) to decouple the angular momentum rotation term \cite{BaoWang}, time-splitting finite element method based on polar and cylindrical coordinates in 2D and 3D, respectively, such that the angular momentum rotation term becomes constant coefficient \cite{BaoDuZhang}, time-splitting generalized Laguerre-Hermite pseudospectral method via eigen-expansion of the linear operator \cite{BaoLiShen}, finite difference time domain methods \cite{BaoCai2}, etc. Each method has its own advantages and disadvantages.
Here we present the detailed algorithms for some of these methods.


\subsubsection{Time splitting Fourier pseudospectral method via an ADI technique}

Due to the external trapping potential $V(\bx)$, the
   solution $\psi(\bx,t)$ of (\ref{eq:gpegrot2:sec6})-(\ref{eq:gpegp:sec6})
   decays to zero exponentially fast when $|\bx|\to \infty$.
Thus in practical computation, we can truncate the problem
(\ref{eq:gpegrot2:sec6})-(\ref{eq:gpegp:sec6}) into a bounded computational domain:
 \begin{align}
&i\p_t\psi({\bx},t) = -\fl{1}{2}\nabla^2\psi +
\left[V(\bx)+\bt|\psi|^2 - \Og L_z\right] \psi,\quad {\bx}\in U,\;
t>0,\label{eq:GPE2:sec6} \\
 &\psi({\bx},0) = \psi_0({\bx}), \quad
{\bx}\in U;\label{eq:initial_data:sec6}
\end{align}
with periodic boundary condition.
Here  we choose
$U=[a,b]\tm[c,d]$ in 2D, and resp.,
$U=[a,b]\tm[c,d]\tm[e,f]$ in 3D,
 with $|a|$, $|b|$, $|c|$, $|d|$, $|e|$  and $|f|$
sufficiently large.

We choose a time step size $\tau>0$. For $n=0,1,2,\cdots$, similar to the case of non-rotating GPE in section \ref{subsubsec:TSSP}, from
time $t=t_n=n\tau$ to $t=t_{n+1}=t_n+\tau$, the GPE
(\ref{eq:GPE2:sec6}) is  solved in two splitting steps. One solves
first \cite{BaoWang} \be \label{eq:fstep:sec6} i\;\p_t\psi({\bx},t) = -\fl{1}{2}\nabla^2
\psi(\bx,t) - \Og L_z \psi(\bx,t)\ee for the time step of length
$\tau$, followed by solving
 \be \label{eq:sstep:sec6}
i\;\p_t\psi({\bx},t) = V(\bx) \psi(\bx,t) +
\bt|\psi(\bx,t)|^2\psi(\bx,t),
 \ee
for the same time step. (\ref{eq:sstep:sec6}) can be integrated exactly \cite{BaoDuZhang,BaoWang}, and we find for $\bx \in U$ and
$t_n\le t\le t_{n+1}$:
\be \label{eq:solode:sec6}
\psi(\bx,t)=
e^{-i[V(\bx)+\bt|\psi(\bx,t_n)|^2](t-t_n)}\;
\psi(\bx,t_n).
\ee
 To discretize (\ref{eq:fstep:sec6}) in space, compared with
non-rotating BEC (cf. section \ref{sec:numdym}), i.e. $\Og=0$ in
(\ref{eq:gpegrot2:sec6}), the
main difficulty is that the coefficients in $L_z$ are {\sl not}
constants which causes big trouble in applying sine or Fourier
pseudospectral discretization. Due to the special structure in the
angular momentum rotation term $L_z$, we will apply
the alternating direction implicit (ADI) technique and decouple
the operator $-\fl{1}{2}\nabla^2 - \Og L_z$ into two one dimensional
operators such that each operator becomes a summation of terms
with constant coefficients in that dimension. Therefore, they can
be discretized in space by Fourier pseudospectral method and
the ODEs in phase space can be integrated  analytically.

{\it Discretization in 2D.} When $d=2$ in (\ref{eq:fstep:sec6}),
we choose mesh sizes $\btu x>0$ and
$\btu y>0$ with $\btu x=(b-a)/M$ and $\btu y=(d-c)/N$ for $M$ and
$N$ even positive integers, and let the grid points be
\be x_j=a + j \btu x, \quad j=0,1,2,\cdots, M; \quad y_k=c +k
\btu y, \quad k=0,1,2,\cdots,N.\ee Let $\psi_{jk}^n$ be the
approximation of $\psi(x_j,y_k,t_n)$ and $\psi^n$ be the solution
vector with component $\psi_{jk}^n$.

  From time $t=t_n$ to $t=t_{n+1}$, we solve (\ref{eq:fstep:sec6}) first
\be\label{eq:fstep1:sec6} i\;\p_t\psi({\bx},t) = -\fl{1}{2}\p_{xx}
\psi(\bx,t) - i\Og y\p_x \psi(\bx,t), \ee for the time step of
length $\tau$, followed by solving \be\label{eq:fstep2:sec6}
i\;\p_t\psi({\bx},t) = -\fl{1}{2}\p_{yy} \psi(\bx,t) + i\Og x\p_y
\psi(\bx,t), \ee for the same time step. Using the standard
second order Strang splitting, a time splitting Fourier pseudospectral (TSSP) method
  for solving (\ref{eq:GPE2:sec6})-(\ref{eq:initial_data:sec6}) can be written as \cite{BaoWang}:
 \begin{equation} \label{eq:tssp2d:sec6}\begin{split}
&\psi_{jk}^{(1)}=\sum_{p=-M/2}^{M/2-1}
 e^{-i\tau(\mu_p^2
+2\Og y_k \mu_p)/4}\; \widehat{(\psi_k^n)}_p\; e^{i\mu_p(x_j-a)},
    \ (j,k)\in\calT_{MN}^0, \\
&\psi_{jk}^{(2)}=\sum_{q=-N/2}^{N/2-1}
 e^{-i\tau(\ld_q^2
-2\Og x_j \ld_q)/4}\; \widehat{(\psi_j^{(1)})}_q\;
e^{i\ld_q(y_k-c)},
    \ (j,k)\in\calT_{MN}^0, \\
&\psi^{(3)}_{jk}=
e^{-i\tau [V(x_j,y_k)+\bt|\psi_{jk}^{(2)}|^2]}\;\psi_{jk}^{(2)},\qquad (j,k)\in\calT_{MN}^0,
  \\
&\psi_{jk}^{(4)}=\sum_{q=-N/2}^{N/2-1}
 e^{-i\tau(\ld_q^2
-2\Og x_j \ld_q)/4}\; \widehat{(\psi_j^{(3)})}_q\;
e^{i\ld_q(y_k-c)},
    \ (j,k)\in\calT_{MN}^0, \\
&\psi_{jk}^{n+1}=\sum_{p=-M/2}^{M/2-1}
 e^{-i\tau(\mu_p^2
+2\Og y_k \mu_p)/4}\; \widehat{(\psi_k^{(4)})}_p\;
e^{i\mu_p(x_j-a)},
    \ (j,k)\in\calT_{MN}^0,
\end{split}
\end{equation}
where  for each fixed $k$, $\widehat{(\psi_k^\alpha)}_p$
($p=-\frac{M}{2},\cdots,\frac{M}{2}-1$) with an index $\ap$,
the Fourier coefficients of the vector
$\psi_k^\ap=(\psi_{0k}^\ap$, $\psi_{1k}^\ap$, $\cdots$,
$\psi_{(M-1)k}^\ap)^T$, are defined as \be \label{eq:Fouv1:sec6}
\widehat{(\psi_k^\ap)}_p=\frac{1}{M}\sum_{j=0}^{M-1}
 \psi^\ap_{jk}\;e^{-i\mu_p(x_j-a)},  \quad\mu_p=\frac{2p\pi}{b-a},\quad
 p=-\fl{M}{2},\cdots,\fl{M}{2}-1;
\ee similarly, for each fixed $j$, $\widehat{(\psi_j^\ap)}_q$
($q=-\frac{N}{2},\cdots,\frac{N}{2}-1$), the Fourier coefficients of the vector
$\psi_j^\ap=(\psi_{j0}^\ap$, $\psi_{j1}^\ap$, $\cdots$,
$\psi_{j(N-1)}^\ap)^T$, are defined as \be \label{eq:Fouv2:sec6}
\widehat{(\psi_j^\ap)}_q=\frac{1}{N}\sum_{k=0}^{N-1}
 \psi^\ap_{jk}\;e^{-i\ld_q(y_k-c)},  \quad \ld_p=\frac{2q\pi}{d-c},\quad
 q=-\fl{N}{2},\cdots,\fl{N}{2}-1.
\ee
For the TSSP (\ref{eq:tssp2d:sec6}), the total memory requirement is $O(MN)$ and the total
computational cost per time step is $O(MN\ln (MN))$. The scheme is
time reversible
just as it holds for the GPE (\ref{eq:GPE2:sec6}), i.e. the
scheme is
unchanged if we interchange $n\leftrightarrow n+1$ and $\tau\leftrightarrow -\tau$ in (\ref{eq:tssp2d:sec6}).  Also, a main advantage of the
numerical method is its time-transverse invariance, just as it
holds for the GPE (\ref{eq:GPE2:sec6}) itself. If a constant $\ap$ is
added to the external potential $V$, then the discrete wave
functions $\psi_{jk}^{n+1}$ obtained from (\ref{eq:tssp2d:sec6}) get
multiplied by the phase factor $e^{-i\ap(n+1)\tau}$, which
leaves the discrete quadratic observable  $|\psi_{jk}^{n+1}|^2$
unchanged.

{\it Discretization in 3D.}  When $d=3$ in (\ref{eq:fstep:sec6}), we choose mesh sizes $\btu x>0$, $\btu
y>0$ and $\btu z>0$ with $\btu x=(b-a)/M$, $\btu y =(d-c)/N$ and
$\btu z=(f-e)/L$ for even positive integers $M$, $N$  and $L$, and
let the grid points be
\be x_j=a + j \btu x, \quad y_k=c +k
\btu y, \ (j,k)\in\calT_{MN}^0; \quad z_l=e+l \btu z, \ 0\le l\le L.\ee Let
$\psi_{jkl}^n$ be the approximation of $\psi(x_j,y_k,z_l,t_n)$ and
$\psi^n$ be the solution vector with component $\psi_{jkl}^n$.

  Similar as those for 2D case, from time $t=t_n$ to $t=t_{n+1}$,
  we solve (\ref{eq:fstep:sec6}) first
\be\label{eq:fstep3:sec6} i\;\p_t\psi({\bx},t) =
\left(-\fl{1}{2}\p_{xx}-\frac{1}{4}\p_{zz} - i\Og y\p_x\right)
\psi(\bx,t), \ee for the time step of length $\tau$, followed by
solving \be\label{eq:fstep4:sec6} i\;\p_t\psi({\bx},t) =
\left(-\fl{1}{2}\p_{yy} -\frac{1}{4}\p_{zz} + i\Og x\p_y\right)
\psi(\bx,t), \ee for the same time step. A second order time splitting Fourier pseudospectral (TSSP) method
  for solving (\ref{eq:GPE2:sec6})-(\ref{eq:initial_data:sec6}) can be written as \cite{BaoWang}:
  \begin{align}
&\psi_{jkl}^{(1)}=\sum_{p=-M/2}^{M/2-1}\;\sum_{s=-L/2}^{L/2-1}
 e^{-i\tau(2\mu_p^2+\gm_s^2
+4\Og y_k \mu_p)/8}\ \widehat{(\psi_k^n)}_{ps}\ e^{i\mu_p(x_j-a)}\
e^{i\gm_s(z_l-e)}, \nn\\
&\psi_{jkl}^{(2)}=\sum_{q=-N/2}^{N/2-1}\;\sum_{s=-L/2}^{L/2-1}
 e^{-i\tau(2\ld_q^2+\gm_s^2
-4\Og x_j \ld_q)/8}\ \widehat{(\psi_j^{(1)})}_{qs}\
e^{i\ld_q(y_k-c)}\ e^{i\gm_s(z_l-e)},\nn \\
&\psi^{(3)}_{jkl}=e^{-i\tau
[V(x_j,y_k,z_l)+\bt|\psi_{jkl}^{(2)}|^2]}\;\psi_{jkl}^{(2)},\quad
\ (j,k,l)\in\calT_{MNL}^0,\label{eq:tssp3d:sec6}\\
&\psi_{jkl}^{(4)}=\sum_{q=-N/2}^{N/2-1}\;\sum_{s=-L/2}^{L/2-1}
 e^{-i\tau(2\ld_q^2+\gm_s^2
-4\Og x_j \ld_q)/8} \widehat{(\psi_j^{(3)})}_{qs}\
e^{i\ld_q(y_k-c)}\ e^{i\gm_s(z_l-e)},\nn \\
&\psi_{jkl}^{n+1}=\sum_{p=-M/2}^{M/2-1}\;\sum_{s=-L/2}^{L/2-1}
 e^{-i\tau(2\mu_p^2+\gm_s^2
+4\Og y_k \mu_p)/8}\ \widehat{(\psi_k^{(4)})}_{ps}\
e^{i\mu_p(x_j-a)}\ e^{i\gm_s(z_l-e)},\nn
\end{align}
 where
 \[\calT_{MNL}^0=\{(j,k,l) \ |\ j=0,1,\ldots,M,\ k=0,1,\ldots,N,\ l=0,1,\ldots,L\},\]
 and for each
fixed $k$, $\widehat{(\psi_k^\ap)}_{ps}$ ($-M/2\le p\le M/2-1$,
$-L/2\le s\le L/2-1$) with an index $\ap$,
the Fourier coefficients of the vector $\psi_{jkl}^\ap$
($0\le j< M$, $0\le l<L$), are defined as \[
\widehat{(\psi_k^\ap)}_{ps}=\frac{1}{ML}\sum_{j=0}^{M-1}
\sum_{l=0}^{L-1}
 \psi^\ap_{jkl}\ e^{-i\mu_p(x_j-a)} \ e^{-i\gm_s(z_l-e)},  \
 -\fl{M}{2}\le p<\fl{M}{2}, \ -\frac{L}{2}\le s<\frac{L}{2};
\] similarly, for each fixed $j$, $\widehat{(\psi_j^\ap)}_{qs}$
($-N/1\le q \le N/2-1$, $-L/2\le s\le L/2-1$) with an index $\ap$,
the Fourier
coefficients of the vector $\psi_{jkl}^\ap$ ($k=0,\cdots,N$,
$l=0,\cdots,L$), are defined as \[
\widehat{(\psi_j^\ap)}_{qs}=\frac{1}{NL}\sum_{m=0}^{N-1}
\sum_{l=0}^{L-1}
 \psi^\ap_{jkl}\ e^{-i\ld_q(y_k-c)} \ e^{-i\gm_s(z_l-e)},  \
 -\fl{N}{2}\le q<\fl{N}{2}, \ -\frac{L}{2}\le s < \frac{L}{2},
 \] with $\gm_s=\frac{2\pi s}{f-e}$ for
$s=-L/2,\cdots, L/2-1$.
For the scheme (\ref{eq:tssp3d:sec6}), the total memory requirement is
$O(MNL)$ and the total computational cost per time step is
$O(MNL\ln (MNL))$. Furthermore,
 the discretization is time
reversible and time transverse invariant in the discretized
level.

\subsubsection{Time-splitting finite element method
based on polar/cylindrical coordinates}
 As noticed in
\cite{BaoDuZhang,BaoCai2}, the angular momentum rotation is a constant coefficient in
2D with polar coordinates and 3D with cylindrical coordinates. Thus
the original problem of GPE with an angular momentum rotation term
defined in ${\Bbb R}^d$ ($d=2,3$) for rotating BEC can also be
truncated on a disk in 2D and a cylinder in 3D as bounded
computational domain with homogeneous Dirichlet boundary condition:
\begin{align}
&i\p_t\psi({\bx},t) = -\fl{1}{2}\nabla^2\psi +
\left[V(\bx)+\bt|\psi|^2 - \Og L_z\right] \psi,\quad {\bx}\in U,\;
t>0,\label{eq:GPE27:sec6} \\
 &\psi(\bx,t)=0, \quad \bx\in\Gamma=\p U, \quad t\ge0, \qquad \psi({\bx},0) = \psi_0({\bx}), \quad
{\bx}\in U;\label{eq:initial_data8:sec6}
\end{align}
where we choose
$U=\{\bx=(x,y)\ |\ r=\sqrt{x^2+y^2}<R\}$ in 2D, and resp.,
$U=\{\bx=(x,y,z)\ |\ r=\sqrt{x^2+y^2}<R, \ Z_1<z<Z_2\}$ in 3D
 with $R$, $|Z_1|$, $|Z_2|$ sufficiently large.

We choose a time step size $\tau>0$. For $n=0,1,2,\cdots$,  from
time $t=t_n=n\tau$ to $t=t_{n+1}=t_n+\tau$, the GPE
(\ref{eq:GPE27:sec6}) is solved in two splitting steps. One solves
first \be \label{eq:fstep7:sec6} i\;\p_t\psi({\bx},t) = -\fl{1}{2}\nabla^2
\psi(\bx,t) - \Og L_z \psi(\bx,t)\ee for the time step of length
$\tau$, followed by solving
 \be \label{eq:sstep7:sec6}
i\;\p_t\psi({\bx},t) = V(\bx) \psi(\bx,t) +
\bt|\psi(\bx,t)|^2\psi(\bx,t),
 \ee
for the same time step. (\ref{eq:sstep7:sec6}) can be integrated exactly
\cite{BaoWang}, and we find for $\bx \in U$ and
$t_n\le t\le t_{n+1}$:
\be \label{eq:solode7:sec6}
\psi(\bx,t)=
e^{-i[V(\bx)+\bt|\psi(\bx,t_n)|^2](t-t_n)}\;
\psi(\bx,t_n).
\ee

{\sl Discretization in 2D.}
To solve (\ref{eq:fstep7:sec6}), we try to formulate the equation in a
variable separable form. When $d=2$,
 we use the polar coordinate $(r,\theta)$,
and discretize in the $\theta$-direction by a Fourier pseudo-spectral
method, in the $r$-direction by a finite element method (FEM) and in
time by a Crank-Nicolson (C-N) scheme. Assume \cite{BaoDuZhang}
\be
\label{Fourier-form_2d}
\psi(r,\theta,t) = \sum_{l=-L/2}^{L/2-1}\
\widehat{\psi}_{l} (r, t)\ e^{il\theta},
\ee
where $L$ is an even positive integer and $\widehat{\psi}_{l} (r, t)$ is
the Fourier coefficient for the $l$-th mode. Plugging (\ref{Fourier-form_2d})
into (\ref{eq:fstep7:sec6}), noticing the orthogonality of the Fourier
functions, we obtain for $-\fl{L}{2}\leq l\leq\fl{L}{2}-1$ and $0<r<R$:
\bea
\label{eq_hat_psi_2d}
&&i\p_t\widehat{\psi}_{l}(r,t)=-
\fl{1}{2r}\frac{\p}{\p r}
\left(r\frac{\p\widehat{\psi}_{l}(r,t)}{\p r}\right)
+\left(\fl{l^2}{2r^2}-l\Og\right)\widehat{\psi}_{l}(r,t),\\
\label{femit}
 &&\widehat{\psi}_{l}(R,t) = 0\quad (\hbox{for all}\ l) , \qquad
\widehat{\psi}_{l}(0,t) = 0 \quad  (\hbox{for}\ l\ne0).
\eea
Let $P^k$ denote all polynomials with degree at most $k$,
$M>0$ be a chosen integer,
$0=r_0<r_1<r_2<\cdots<r_M=R$  be  a partition
for the interval $[0,R]$ with a mesh size $h=\max_{0\le m<M}\;
\{r_{m+1}-r_m\}$. Define a FEM subspace by
$$U^h=
\left\{u^h\in C[0,R]\ |\ \left.u^h\right|_{[r_m,r_{m+1}]}\in P^k,
\; 0\le m<M, \ u^h(R)=0\right\}$$
for $l=0$, and for $l\ne 0$,
$$U^h= \left\{u^h\in C[0,R]\; |\;
 \left.u^h\right|_{[r_m,r_{m+1}]}\in P^k,
\; 0\le m<M, \; u^h(0)=u^h(R)=0\right\}\,,$$
then we obtain
the FEM  approximation for (\ref{eq_hat_psi_2d})-(\ref{femit}): Find
$\hat{\psi}_l^h=\hat{\psi}_l^h(\cdot,t)\in U^h$
such that for all $\phi^h\in U^h$
and  $t_n\le t\le t_{n+1}$,
\be
\label{fem2}
i\frac{d}{dt} A(\hat{\psi}_l^h(\cdot,t),\phi^h)
=B(\hat{\psi}_l^h(\cdot,t),\phi^h)+l^2
C(\hat{\psi}_l^h,\phi^h)-l\Og
A(\hat{\psi}_l^h,\phi^h), \ee where \beas
&&A(u^h,v^h)=\int_0^R r\; u^h(r)\; v^h(r)\; dr, \qquad
B(u^h,v^h)=\int_0^R \frac{r}{2}\; \frac{d u^h(r)}{dr}\; \frac{d
v^h(r)}{dr}\; dr,  \\
&&C(u^h,v^h)=\int_0^R \frac{1}{2r}\; u^h(r)\; v^h(r)\; dr, \qquad
u^h,\ v^h\in U^h.
\eeas
The ODE system (\ref{fem2}) is then discretized by the
standard  Crank-Nicolson scheme in time.
Although an implicit time discretization is applied for (\ref{fem2}),
the 1D nature of the problem makes
the coefficient matrix for the linear system
band-limited. For example, if the piecewise linear polynomial is used,
i.e. $k=1$ in $U^h$, the matrix is tridiagonal.
Fast algorithms can be applied to solve the resulting linear systems.

In practice, we always use the second-order Strang splitting \cite{Strang}, i.e.
from time $t=t_n$ to $t=t_{n+1}$: i) first evolve (\ref{eq:sstep7:sec6}) for half
time step $\tau/2$ with initial data given at $t=t_n$; ii) then evolve
(\ref{eq:fstep7:sec6}) for one time step $\tau$ starting with the new  data;
iii) and evolve (\ref{eq:sstep7:sec6}) for half time step $\tau/2$ with
the newer data. For the discretization considered here, the total memory requirement is
$O(ML)$ and the total computational cost per time step is $O(ML\ln L)$.
Furthermore, it conserves the total density
in the discretized level.

{\sl Discretization  in 3D.} When $d=3$ in (\ref{eq:fstep7:sec6}), we use the cylindrical  coordinate
$(r,\theta,z)$, and discretize in
the $\theta$-direction by the Fourier pseudo-spectral method,
in the $z$-direction
by the sine pseudo-spectral method, and in the $r$-direction by finite
element or finite difference method and
in time by the C-N scheme. Assume that,
\be
\label{Fourier-form_3d}
\psi(r,\theta,z,t) = \sum_{l=-L/2}^{L/2-1}\ \sum_{k=1}^{K-1}
\widehat{\psi}_{l,k} (r,t)\ e^{il\theta}\ \sin(\mu_k(z-a)) ,
\ee
where $L$ and $K$ are two even positive integers,
$\mu_k=\frac{\pi k}{b-a}$ ($k=1,\cdots,K-1$) and
$\widehat{\psi}_{l,k} (r, t)$ is the Fourier-sine coefficient for
the $(l,k)$th
mode. Plugging (\ref{Fourier-form_3d}) into (\ref{eq:fstep7:sec6}) with $d=3$,
noticing the orthogonality of the Fourier-sine modes, we obtain, for
$-\fl{L}{2}\leq l\leq\fl{L}{2}-1$, $1 \leq k\leq K-1$
and $0<r<R$, that \cite{BaoDuZhang}:
\be \label{eq_hat_psi_3d}
i\p_t\widehat{\psi}_{l,k}(r,t)=-
\fl{1}{2r}\frac{\p}{\p r}
\left(r\frac{\p\widehat{\psi}_{l,k}(r,t)}{\p r}\right)
+\left(\fl{l^2}{2r^2}+\frac{\mu_k^2}{2}
-l\Og\right)\widehat{\psi}_{l,k}(r,t),
\ee
with essential boundary conditions
\be \label{femit3d}
\widehat{\psi}_{l,k}(R,t) = 0\ (\hbox{for all} \ l), \qquad
\widehat{\psi}_{l,k}(0,t) = 0 \ (\hbox{for}\ l\ne0).
\ee
The discretization of (\ref{eq_hat_psi_3d})-(\ref{femit3d}) is similar
as that for (\ref{eq_hat_psi_2d})-(\ref{femit}) and it is omitted here.

For the algorithm in 3D, the total memory requirement is $O(MLK)$
and the total computational cost per time step is $O(MLK\ln(LK))$.

\subsection{A generalized Laguerre-Fourier-Hermite pseudospectral method}
Like section \ref{subsec:central}, for polar coordinate in 2D and cylindrical coordinate
in 3D, a similar Laguerre-Hermite  pseudospectral method  can be designed for computing dynamics for rotating BEC (\ref{eq:gpegrot2:sec6})-(\ref{eq:gpegp:sec6}). Here, we assume that the potential $V(\bx)$ in  (\ref{eq:gpegrot2:sec6})-(\ref{eq:gpegp:sec6}) is given as \cite{BaoLiShen}
\be
V(\bx)=V_h(\bx)+W(\bx),\quad V_h(\bx)=\begin{cases}\frac{1}{2}(\gamma_r^2(x^2+y^2)+\gamma_z^2z^2),&d=3,\\
\frac{1}{2}\gamma_r^2(x^2+y^2),&d=2,
\end{cases}\ee
where $W(\bx)$ is a real potential.

Denoting \bea \label{eq:Bpp:sec6} B_\perp\phi
&=&\left[-\fl{1}{2}\left(\frac{\p^2}{\p x^2}+\frac{\p^2}{\p
y^2}\right) + \frac{1}{2}\gm_r^2(x^2+y^2) - \Og L_z\right]\phi, \\
 \label{eq:Bz:sec6}
B_z\phi&=&\left[-\frac{1}{2}\frac{\p^2}{\p
z^2}+\frac{1}{2}\gm_z^2z^2\right]\phi,\\
A\phi&=&\left[W(\bx)+\beta|\phi|^2\right]\phi,\quad
\label{eq:Bb:sec6} B\phi=\left\{\ba{ll} B_\perp\phi, &\quad d=2, \\
(B_\perp+B_z)\phi, &\quad d=3, \\ \ea \right.
 \eea then the GPE ({\ref{eq:gpegrot2:sec6}) becomes \be \label{eq:GPEAB:sec6}
i\p_t\psi({\bx},t) =A\psi+B\psi,\qquad {\bx}\in {\Bbb{R}}^d,\quad
t>0.\ee
For
$n=0,1,2,\ldots,$ let $\psi^n:=\psi^n(\bx)$ be the approximation of
$\psi(\bx,t_n)$. A standard Strang splitting second-order symplectic time integrator
 for (\ref{eq:GPEAB:sec6}) is as follows
\be\label{eq:timesp1:sec6} \psi^{(1)}=e^{-i\tau A/2}\psi^n,\qquad
\psi^{(2)}=e^{-i\tau B}\psi^{(1)},\qquad \psi^{n+1}=e^{-i\tau A/2}\psi^{(2)}. \ee
Thus the key for an efficient implementation of
(\ref{eq:timesp1:sec6}) is to solve efficiently the following two
subproblems: \be \label{eq:timesp2:sec6}i\p_t\psi({\bx},t) = A\psi(\bx,t)=
\left[W(\bx)+ \bt|\psi(\bx,t)|^2\right]\psi(\bx,t),\qquad
{\bx}\in {\mathbb{R}}^d, \ee and \be\label{eq:timesp3:sec6}\begin{split}
&i\p_t\psi({\bx},t) =B\psi(\bx,t) =\left[-\fl{1}{2}\nabla^2 + V_h({\bx})
- \Og L_z\right]\psi(\bx,t),\qquad {\bx}\in {\mathbb{R}}^d,\\
&\lim_{|\bx|\to+\ift} \psi(\bx,t)=0.
\end{split}
\ee The decaying condition in (\ref{eq:timesp3:sec6}) is necessary for
satisfying the mass conservation.

\subsubsection{Discretization in 2D} In the 2D case, we use
 the polar coordinates $(r,\tht)$, and write the solutions
of (\ref{eq:timesp3:sec6}) as $\psi(r,\tht,t)$ . Therefore, for
 $t\ge t_s$ ($t_s$ is any given time), (\ref{eq:timesp3:sec6}) collapses to \cite{BaoLiShen}
\bea\label{eq:GPE2D:sec6}\begin{split}
&i\p_t\psi(r,\tht,t)=\left[-\frac{1}{2r}\frac{\p}{\p
r}\left(r\frac{\p}{\p r}\right) -\frac{1}{2r^2}\frac{\p^2}{\p
\tht^2}  + \frac{1}{2}\gm_r^2r^2 +
i\Og \p_\tht\right]\psi(r,\tht,t)\\
&\qquad \qquad \quad:=B_\perp\psi(r,\tht,t),\\
&\psi(r,\tht+2\pi,t)=\psi(r,\tht,t), \qquad r\in(0,\ift), \quad
\tht\in(0,2\pi),\quad\lim_{r\to\ift}
\psi(r,\tht,t)=0.
\end{split}
\eea
For any fixed $m$ ($m=0,\pm1,\pm2,\ldots$), recalling the scaled generalized-Laguerre functions $L_k^n$  (\ref{eq:GLF3:sec6}) ($n\ge0$),  a simple calculation
shows that
\be\label{eq:Bperp2:sec6} B_\perp
\left(L_k^{|m|}(r)\,e^{im\tht}\right)=\mu_{km}\;L_k^{|m|}(r)e^{im\tht},\qquad k=0,1,2,\ldots\,.\ee
where \be\label{eq:mukm:sec6} \mu_{km}=\gm_r(2k+|m|+1)-m\Og, \qquad
k=0,1,2,\ldots\,. \ee
 This immediately implies
that
 $\{L_k^{|m|}(r)\,e^{im\tht},\ k=0,1,\cdots,\
m=0,\pm1,\pm2,\cdots \}$ are eigenfunctions of the linear operator
$B_\perp$.

For fixed even integer $M>0$ and integer $K>0$, let $X_{KM}={\rm
span}\{L_k^{|m|}(r)\,e^{im\tht}\ :\ k=0,1,\ldots,K,\
m=-M/2,-M/2+1,\ldots,-1,0,1,\ldots,M/2-1\}$. The
generalized-Laguerre-Fourier spectral method for (\ref{eq:GPE2D:sec6}) is to
find $\psi_{KM}(r,\tht,t)\in X_{KM}$, i.e. \be\label{eq:expan2D:sec6}
\psi_{KM}(r,\tht,t)=\sum_{m=-M/2}^{M/2-1}\left[e^{im\tht}\sum_{k=0}^K
\hat{\psi}_{km}(t)L_k^{|m|}(r)\right],\quad 0\le r<\ift,\ 0\le
\tht\le 2\pi, \ee
such that
\bea\label{eq:sol2D5:sec6}
i\pl{\psi_{KM}(r,\tht,t)}{t}&=&\left[-\frac{1}{2r}\frac{\p}{\p
r}\left(r\frac{\p}{\p r}\right) -\frac{1}{2r^2}\frac{\p^2}{\p
\tht^2}  + \frac{1}{2}\gm_r^2r^2 +
i\Og \p_\tht\right]\psi(r,\tht,t)\nn\\
&=&B_\perp \psi_{KM}(r,\tht,t), \qquad 0<r<\ift, \quad 0<\tht<2\pi.
\eea Noting that $\lim_{r\to\ift}L_k^{|m|}(r) =0$ for
$k=0,1,2,\ldots$ and $m=0,\pm1,\pm2,\ldots$ \cite{GS}; hence,
$\lim_{r\to\ift}\psi_{KM}(r,\tht,t)=0$ is automatically satisfied.
In addition, the expansions in $r$- and $\tht$-directions for
(\ref{eq:expan2D:sec6}) {\sl do not} commute. Plugging (\ref{eq:expan2D:sec6}) into
(\ref{eq:sol2D5:sec6}), thanks to (\ref{eq:Bperp2:sec6}), noticing the orthogonality
of the Fourier series, for $k=0,1,\ldots,K$ and
$m=-M/2,-M/2-1,\ldots,-1,0,1,\ldots,M/2-1$, we find
\be\label{eq:sol2D6:sec6} i\frac{\rd\hat{\psi}_{km}(t)}{\rd t}=\mu_{km} \;
\hat{\psi}_{km}(t)=\left[\gm_r(2k+|m|+1)-m\Og\right]
\hat{\psi}_{km}(t).\ee The above linear ODE can be integrated {\sl
exactly} and the solution is given by \be\label{eq:sol2D7:sec6}
\hat{\psi}_{km}(t)=e^{-i\mu_{km}(t-t_s)} \;\hat{\psi}_{km}(t_s),
\qquad t\ge t_s. \ee
Plugging (\ref{eq:sol2D7:sec6}) into (\ref{eq:expan2D:sec6}),
we obtain the solution of (\ref{eq:sol2D5:sec6}) as \begin{align}\label{eq:sol2D8:sec6}
\psi_{KM}(r,\tht,t)=&e^{-iB_\perp (t-t_s)}\psi_{KM}(r,\tht,t_s)\nn\\
=&\sum_{m=-M/2}^{M/2-1}\left[e^{im\tht}\sum_{k=0}^K
e^{-i\mu_{km}(t-t_s)}\,\hat{\psi}_{km}(t_s)\,L_k^{|m|}(r)\right],
\quad t\ge t_s, \end{align}
with \be\label{eq:expan2d1:sec6}
\hat{\psi}_{km}(t_s)=\frac{1}{2\pi}\int_0^{2\pi}
\left[e^{-im\tht}\int_0^\ift
\psi_{KM}(r,\tht,t_s)L_k^{|m|}(r)r\,dr\right] d\tht. \ee
To summarize, a second-order time-splitting
generalized-Laguerre-Fourier spectral method for the GPE
(\ref{eq:gpegrot2:sec6}) with $d=2$ is  as follows:

Let $\psi^0=\Pi_{KM}\psi_0$ where $\Pi_{KM}$ is the $L^2$
 projection operator from  $L^2((0,\infty)\times
 (0,2\pi))$ onto $X_{KM}$, we determine $\psi^{n+1}$
 $(n=0,1,\cdots)$ by \cite{BaoLiShen}
 \be\label{eq:tssp2Drd:sec6}\begin{split}
&\psi^{(1)}(r,\tht)=e^{-i\tau[W(r,\tht)+\beta|\psi^n(r,\tht)|^2]/2}\psi^n(r,\tht),\\
&\psi^{(2)}(r,\tht)=\sum_{m=-M/2}^{M/2-1}\left[e^{im\tht}\sum_{k=0}^K
e^{-i\tau \mu_{km}}\,\widehat{\psi^{(1)}}_{km}\,L_k^{|m|}(r)\right],\\
&\psi^{n+1}(r,\tht)=e^{-i\tau [W(r,\tht)+\beta|\psi^{(2)}(r,\tht)|^2]/2}\psi^{(2)}(r,\tht), \end{split}\ee with
\be\label{eq:int1:sec6}
\widehat{\psi^{(1)}}_{km}=\frac{1}{2\pi}\int_0^{2\pi}
\left[e^{-im\tht}\int_0^\ift
\psi^{(1)}(r,\tht)\,L_k^{|m|}(r)r\,dr\right]d\tht.\ee

The scheme \eqref{eq:tssp2Drd:sec6} is not suitable in practice due to the
difficulty to compute  the initial data
$\psi^0_{KM}=\Pi_{KM}\psi_0$
and the integrals in \eqref{eq:int1:sec6}.
We now present an efficient implementation  by
choosing $\psi^0_{KM}(r,\theta)$ as the interpolation of
$\psi(r,\theta,0)$  on a  suitable grid, and approximating
\eqref{eq:int1:sec6} (for all $m$) by
a quadrature rule on this grid.

 It is clear that the
optimal quadrature rule, hence the collocation points,   for the
$r$-integral in  \eqref{eq:int1:sec6} depends on $m$ \cite{BaoShen,BaoShen2}. However, we
have to use the same set of collocation points for all $m$ to form a
tensorial grid in the $(r,\theta)$ domain. Therefore, let
$\{\hat{r}_j\}_{j=0}^{K+M/2}$ be the Laguerre-Gauss points
\cite{GS,Shen2}; i.e. they are the $K+M/2+1$ roots of the standard
Laguerre polynomial $\hat{L}_{K+M/2+1}^0(r):=\hat{L}_{K+M/2+1}(r)$.
Let $\{\hat\og_j\}_{j=0}^{K+M/2}$ be the corresponding weights
associated with the generalized-Laguerre-Gauss quadrature (\ref{eq:quad2d1:sec6}).
 We then
define the scaled generalized-Laguerre-Gauss points and weights
 $r_j$ and $\og_j$ ($j=0,1,\ldots,K+M/2$) as in (\ref{eq:pot2d:sec6})
and the appendix of \cite{Shen2}.

Let $\theta_s=\frac{2s\pi}M$ $(s=0,1,\cdots, M-1)$. For any given
set of values $\{\psi_{js}$, $0\le j\le K+M/2$; $0\le s\le
  M-1\}$, we can define a unique function $\psi$ in $X_{KM}$
interpolating this set, i.e.,
\be
\begin{split}
&\psi(r,\theta)=\sum_{m=-M/2}^{M/2-1}\sum_{k=0}^{K}
\widehat{\psi}_{km}\,L_k^{|m|}(r)e^{im\theta}\quad \text{such that}\\
&\psi(r_j,\theta_s)=\psi_{js},\quad  {0\le j\le K+M/2;\; 0\le s\le
  M-1}.
\end{split}\ee
 By using the discrete orthogonality
 relation  (\ref{eq:quad2d2:sec6}) for the scaled generalized Laguerre
 functions and the discrete Fourier orthogonality relation
\be\label{eq:four_orth:sec6}
\frac 1M\sum_{s=0}^{M-1} e^{ik\theta_s}
e^{-ik'\theta_s}=\delta_{kk'},
\quad |k|,|k'|\le M/2,\ee
we find that
\be\label{eq:trans7:sec6}
\widehat{\psi}_{km}=\frac{1}{M}
\sum_{s=0}^{M-1}\left[e^{-im\tht_s}\sum_{j=0}^{K+M/2}
\og_j\, \psi_{js} \,L_k^{|m|}(r_j)\right],\ee
and that
\begin{equation*}
\|\psi\|_{2}^2:=\int_0^{2\pi}\int_0^\infty |\psi|^2 r\,
 dr\,d\theta=2\pi\sum_{m=-M/2}^{M/2-1}\sum_{k=0}^{K}
|\widehat{\psi}_{km}|^2=\frac
 {2\pi}M\sum_{j=0}^{K+M/2}\sum_{s=0}^{M-1}|\psi_{js}|^2 \omega_j.
 \end{equation*}

We can now describe
the second-order time-splitting generalized-Laguerre-Fourier
pseudospectral (TSGLFP2) method for the GPE (\ref{eq:gpegrot2:sec6}) with $d=2$
as follows:

Let
$\psi^0_{js}=\psi_0(r_j,\theta_s)$ for $0\le j\le K+M/2$ and $0\le
s\le M-1$.  For $n=0,1,2,\cdots$, we compute
 $\psi^{n+1}_{js}$ ($0\le j\le K+M/2, \ 0\le s\le M-1$) by \cite{BaoLiShen}
\be\label{eq:tssp2Dprd:sec6}\begin{split}
&\psi^{(1)}_{js}=e^{-i\tau[W(r_j,\tht_s)+\beta|\psi_{js}^n|^2]/2}\psi^n_{js},\\
&\psi^{(2)}_{js}=\sum_{m=-M/2}^{M/2-1}\left[e^{im\tht_s}\sum_{k=0}^{K}
e^{-i\tau\mu_{km}}\,\widehat{(\psi^{(1)})}_{km}\,L_k^{|m|}(r_j)\right],\\
&\psi^{n+1}_{js}=e^{-i\tau[W(r_j,\tht_s)+\beta|\psi^{(2)}_{js}|^2]/2}\psi^{(2)}_{js},
\end{split}\ee
where $\{\widehat{(\psi^{(1)})}_{km}\}$ are the expansion
coefficients of $\psi^{(1)}$ given by (\ref{eq:trans7:sec6}).

\subsubsection{Discretization in 3D} In the 3D case, by
using the cylindrical coordinates $(r,\tht,z)$, we can write the
solutions of (\ref{eq:timesp3:sec6}) as $\psi(r,\tht,z,t)$. Therefore, for
 $t\ge t_s$ ($t_s$ is any given time), (\ref{eq:timesp3:sec6}) collapses to \cite{BaoLiShen}
\begin{equation*}\begin{split} &\qquad
i\p_t\psi(r,\tht,z,t)=\frac{1}{2}\left[-\frac{1}{r}\frac{\p}{\p
r}\left(r\frac{\p}{\p r}\right) -\frac{1}{r^2}\frac{\p^2}{\p
\tht^2}-\frac{\p^2}{\p z^2} + \gm_r^2r^2+\gm_zz^2 +
2i\Og \p_\tht\right]\psi,\\
&\qquad \qquad \qquad \qquad =
\left(B_\perp+B_z\right)\psi(r,\tht,z,t)=B\,\psi(r,\tht,z,t),\\
&\qquad \psi(r,\tht+2\pi,z,t)=\psi(r,\tht,z,t), \qquad 0<r<\ift,
\quad
0<\tht<2\pi,\quad z\in{\mathbb R},\\
 &\qquad \lim_{r\to\ift}
\psi(r,\tht,z,t)=0, \qquad -\ift<z<\ift,\quad t\ge t_s.
\end{split}
\end{equation*}
Let  the scaled Hermite functions $h_l(z)$ ($l=0,1,\ldots,$) be given in
(\ref{eq:GHF3:sec6}). For any fixed $m$ ($m=0,\pm1,\pm2,\ldots$), we find that \cite{BaoLiShen}
\begin{align}\label{eq:GLH1:sec67}
&B\left(L_k^{|m|}(r)\,e^{im\tht}\, h_l(z)\right)
=\left(\mu_{km}+\ld_l\right) L_k^{|m|}(r)\,e^{im\tht}\,h_l(z),\quad\lambda_l=(l+\frac12)\gamma_z.
\end{align} Hence, $\{L_k^{|m|}(r)\,e^{im\tht}\,h_l(z),\ k,l=0,1,\cdots,\
m=0,\pm1,\pm2,\cdots\}$ are eigenfunctions of the linear operator
$B=B_\perp+B_z$ defined in (\ref{eq:Bb:sec6}) for $d=3$.

Then a second-order time-splitting generalized-Laguerre-Fourier-Hermite spectral
method for the GPE (\ref{eq:gpegrot2:sec6}) with $d=3$ can be constructed analogously to (\ref{eq:tssp2Drd:sec6}). Here, we only present pseudospectral method generalizing TSGLFP2 (\ref{eq:tssp2Dprd:sec6}).

Define the scaled Hermite-Gauss points $z_p$ and weights $\omega_p^z$ ($0\leq p\leq L$) by (\ref{eq:pot1d:sec6}). For any given set of values $\{\psi_{jsp},\; {0\le j\le K+M/2;\; 0\le s\le
  M-1;\; 0\le p\le L}\}$, we can define a unique function $\psi$ in
  $Y_{KML}={\rm span}\{L_k^{|m|}(r)\,e^{im\tht}\,h_l(z)\ :\
0\le k\le K,\ -M/2\le m\le M/2-1, \
0\le l\le L\}$ interpolating this set, i.e.,
\be
\begin{split}
&\psi(r,\theta,z)=\sum_{m=-M/2}^{M/2-1}\sum_{k=0}^{K}\sum_{l=0}^L
\widehat{\psi}_{kml}\,L_k^{|m|}(r)e^{im\theta}h_l(z)\quad\text{such that}\\
&\psi(r_j,\theta_s,z_p)=\psi_{jsp},\quad  {0\le j\le K+M/2;\; 0\le s\le
  M-1;\; 0\le p\le L}.
\end{split}\ee

 By using the discrete orthogonality
relations (\ref{eq:quad2d2:sec6}), (\ref{eq:four_orth:sec6}) and (\ref{eq:quad1d2:sec6}),
we find that
\be\label{eq:trans3D9:sec6}\widehat{\psi}_{kml}=\frac{1}{M}\sum_{p=0}^L\left[h_l(z_p)
\og_p^z\sum_{s=0}^{M-1}\left(e^{-im\tht_s}\sum_{j=0}^{K+M/2}
\og_j\, \psi_{jsp} \,L_k^{|m|}(r_j)\right)\right],\ee
and that
\be\label{eq:norm2:sec6}
\begin{split}
\|\psi\|_2^2&:=\int_{-\ift}^\ift\int_0^\ift\int_0^{2\pi}|\psi(r,\tht,z)|^2
r\;d\tht dr dz\\
&=2\pi\sum_{m=-M/2}^{M/2-1}\sum_{k=0}^{K}\sum_{l=0}^L
|\widehat{\psi}_{kml}|^2=\frac
{2\pi}M\sum_{j=0}^{K+M/2}\sum_{s=0}^{M-1}\sum_{p=0}^L|\psi_{jsp}|^2
  \omega_j\omega_p^z.
\end{split}\ee

Then the second-order time-splitting
generalized-Laguerre-Fourier-Hermite pseudospectral (TSGLFHP2)
method for the GPE (\ref{eq:gpegrot2:sec6}) with $d=3$ is as follows:

Let
$\psi^0_{jsp}=\psi_0(r_j,\theta_s,z_p)$ for $0\le j\le K+M/2$, $0\le
s\le M-1$ and $0\le p\le L$.  For $n=0,1,\cdots$, we compute
 $\psi^{n+1}_{jsp}$ by \cite{BaoLiShen}
\be\label{eq:tssp3Dp:sec6}\begin{split}
&\psi^{(1)}_{jsp}=e^{-i\tau[W(r_j,\tht_s,z_p)+\beta|\psi_{jsp}^n|^2]/2}\psi^n_{jsp},\\
&\psi^{(2)}_{jsp}=\sum_{l=0}^L\left[h_l(z_p)
\sum_{m=-M/2}^{M/2-1}\left(e^{im\tht_s}\sum_{k=0}^K
e^{-i\tau(\mu_{km}+\ld_l)}\,\widehat{(\psi^{(1)})}_{kml}\,L_k^{|m|}(r_j)\right)\right],\\
&\psi^{n+1}_{jsp}=e^{-i\tau[W(r_j,\tht_s,z_p)+\beta|\psi^{(2)}_{jsp}|^2]/2}\psi^{(2)}_{jsp},
\end{split}\ee
where $\{\widehat{(\psi^{(1)})}_{kml}\}$ are the expansion
coefficients of $\psi^{(1)}$ given by \eqref{eq:trans3D9:sec6}.

\subsection{Numerical results}
\label{subsec:numrot}
In this section, we report numerical examples for  ground states and central vortex states as well as dynamics for rotating BEC.

\begin{figure}[t!]
\centerline{\psfig{figure=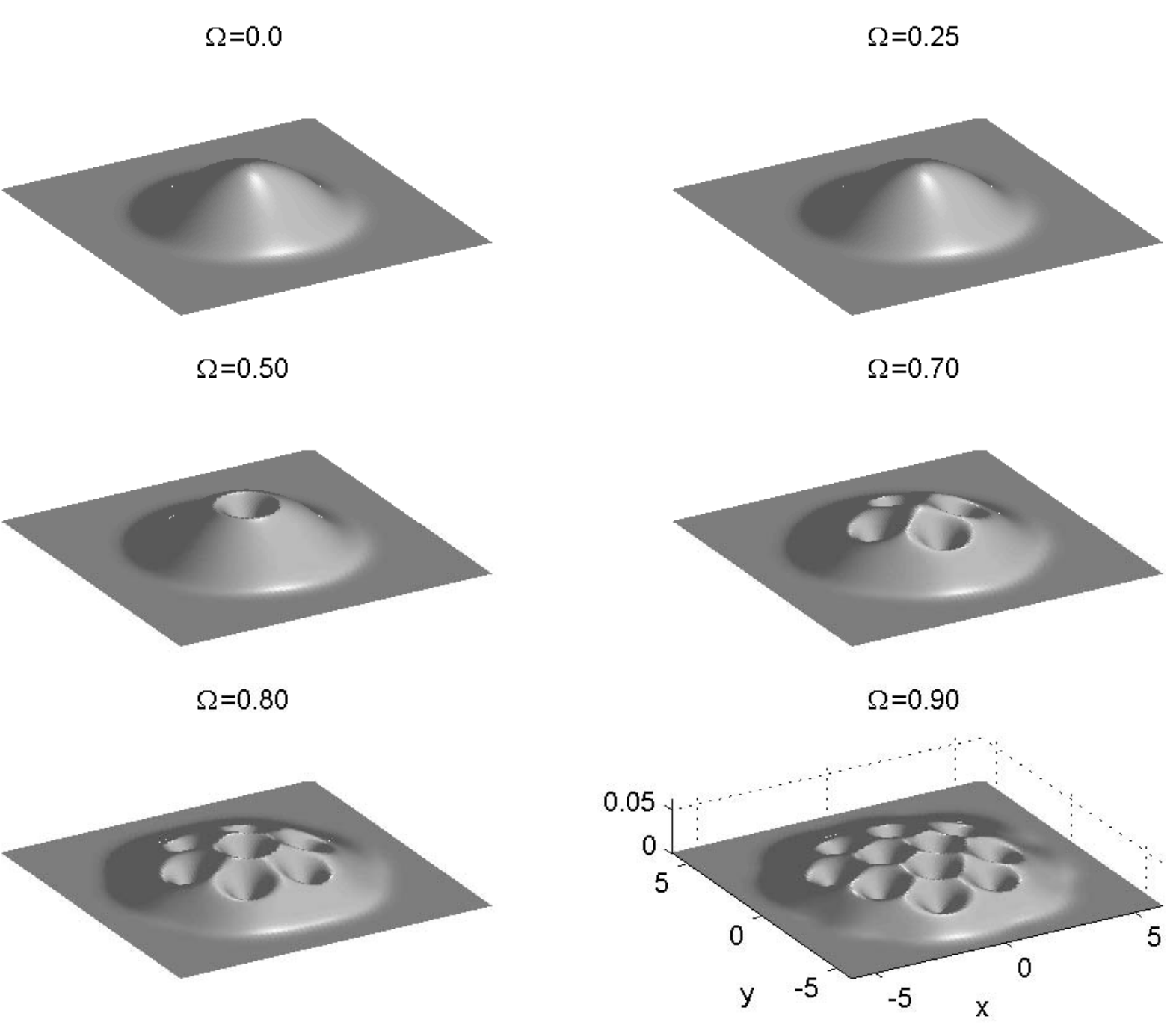,height=16cm,width=11cm,angle=0}}

\caption{Surface plots of ground state density function
$|\phi^g_{\Omega}(x,y)|^2$ in 2D with $\gamma_x=\gm_y=1$ and $\bt=100$ for
different $\Omega$ in Example \ref{exm:1:sec6}.}
\label{fig:1:sec6}
\end{figure}

\begin{figure}[tb!]
\centerline{a)\psfig{figure=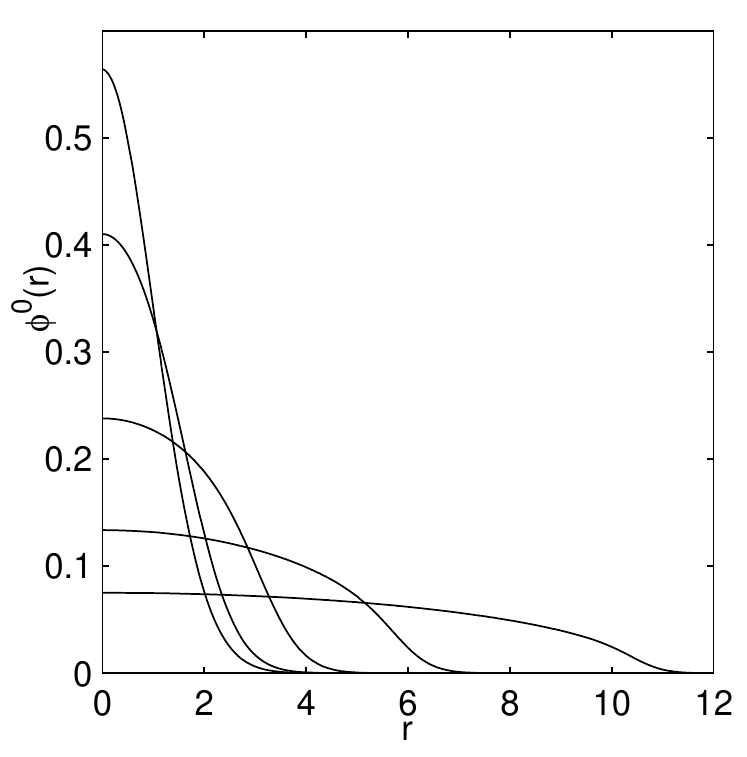,height=5cm,width=5cm,angle=0}
\qquad b)\psfig{figure=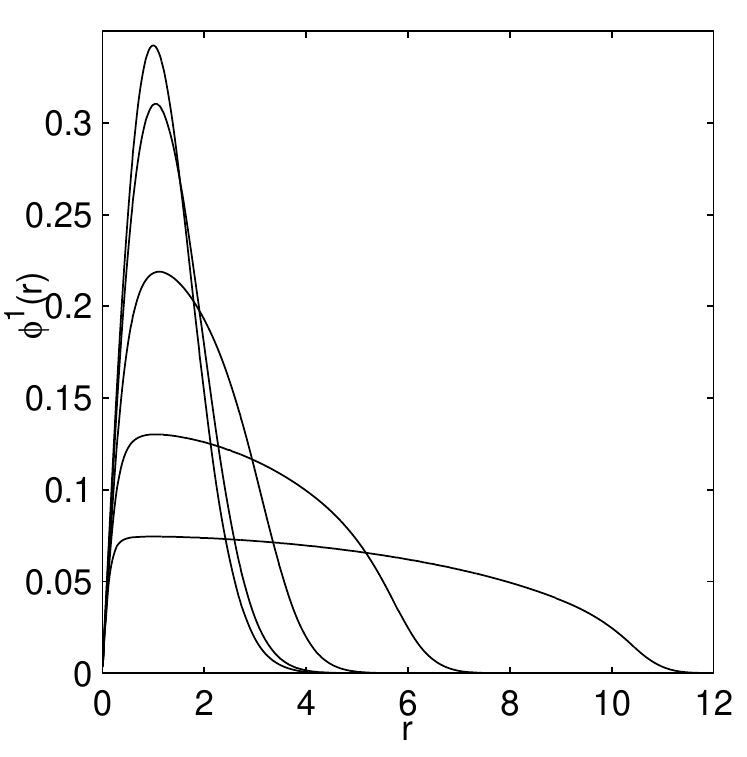,height=5cm,width=5cm,angle=0}}
\centerline{c) \psfig{figure=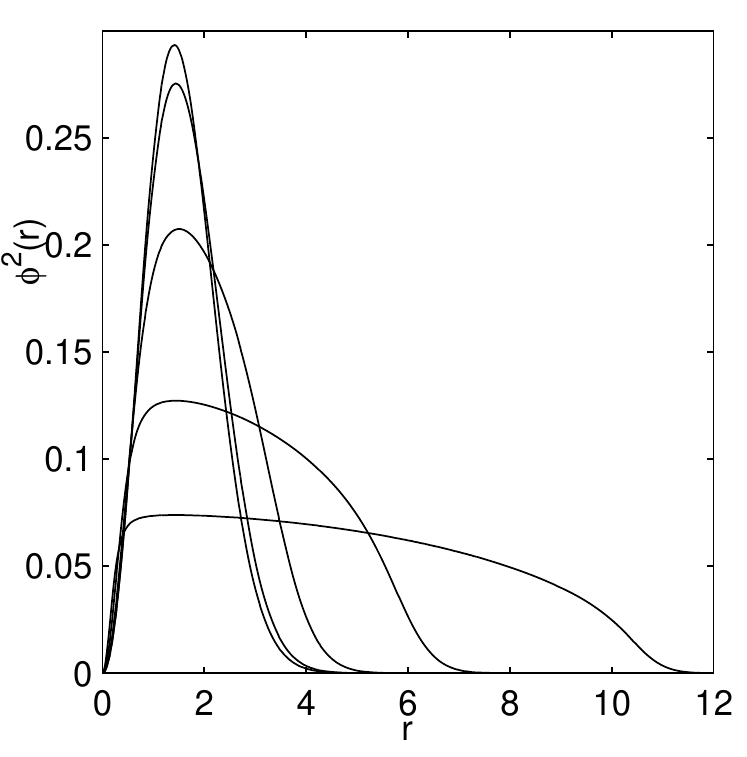,height=5cm,width=5cm,angle=0}
\qquad d) \psfig{figure=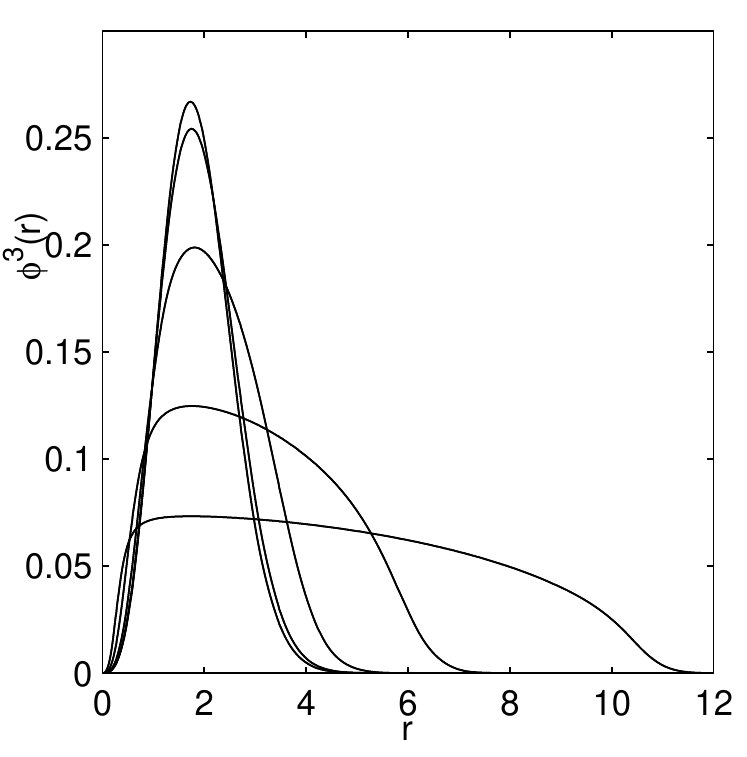,height=5cm,width=5cm,angle=0}}

\caption{Symmetric and central vortex states in 2D with
$\gamma_x=\gm_y=1$
for $\beta=0$, $10$, $100$, $1000$, $10000$ (in the order of
decreasing of peak) in Example \ref{exm:1:sec6}. Symmetric state $\phi^0(r)$: a);
and central
vortex states $\phi^m(r)$: b). $m=1$, c). $m=2$ and d). $m=3$.}\label{fig:2:sec6}

\end{figure}

\begin{example}\label{exm:1:sec6} Ground, symmetric and central vortex
states, as well as their energy configurations, in 2D, i.e. we
take $d=2$ and $\gamma_x=\gm_y=1$ in (\ref{eq:gpegrot:sec5}). Fig.~\ref{fig:1:sec6} plots surface
 of the ground state
$\phi^g(x,y):=\phi_{\Og}^g(x,y)$ with $\bt=100$ for
different $\Og$.  Fig.~\ref{fig:2:sec6}  plots the symmetric
state $\phi^0(r):=\phi_{0}^0(r)$ and first three central
vortex states $\phi_m(r):=\phi_{0}^m(r)$ ($m=1,2,3$) for
different interaction rate $\bt$. Backward Euler  finite difference method is used here
with a bounded computational domain
$U=[-6,6]\tm [-6,6]$ and initial data for GFDN (\ref{eq:ngf1:sec6})-(\ref{eq:ngf3:sec6})
is chosen as
$\phi_0(x,y) =
\fl{(1-\Og) \phi_{\rm ho}(x,y) + \Og \phi_{\rm ho}^v(x,y)}
{\|(1-\Og) \phi_{\rm ho}(x,y) + \Og \phi_{\rm ho}^v(x,y)\|_2}$,
$(x,y)\in U$, where
 $\phi_{\rm ho}^v(x,y) = \fl{x+iy}{\sqrt{\pi}}\;
e^{-(x^2+y^2)/2}$ and $\phi_{\rm
ho}(x,y) = \fl{1}{\sqrt{\pi}}\; e^{-(x^2+y^2)/2}$.
The steady state solution is obtained numerically when $\|\phi^{n+1}-\phi^n\|_\infty:=\max_{(j,l)}
|\phi_{j\,l}^{n+1}-\phi_{j\,l}^n|<\epsilon=10^{-7}$.
\end{example}

\begin{figure}[h!]
\centerline{\psfig{figure=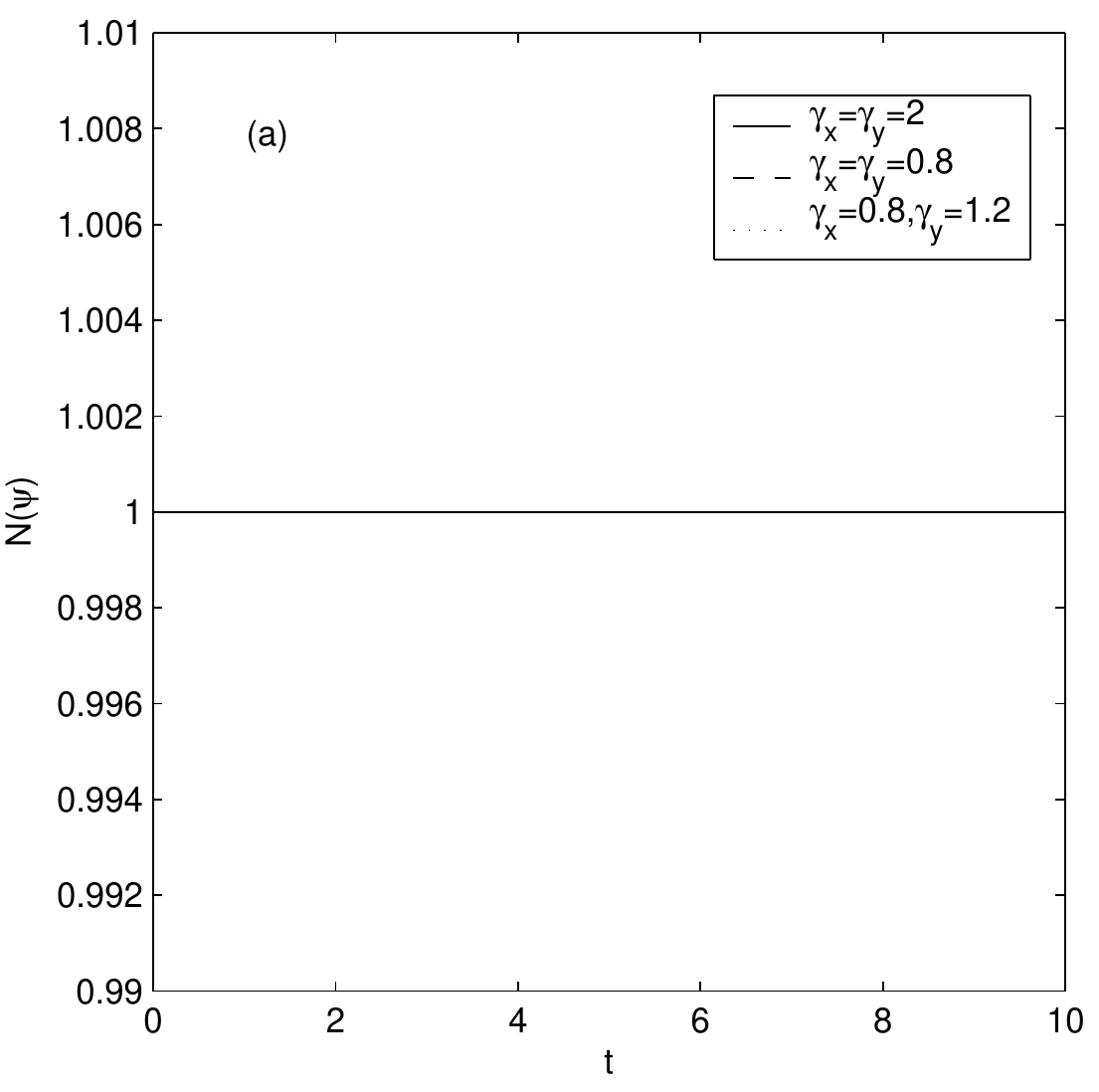,height=5cm,width=5cm,angle=0}
\psfig{figure=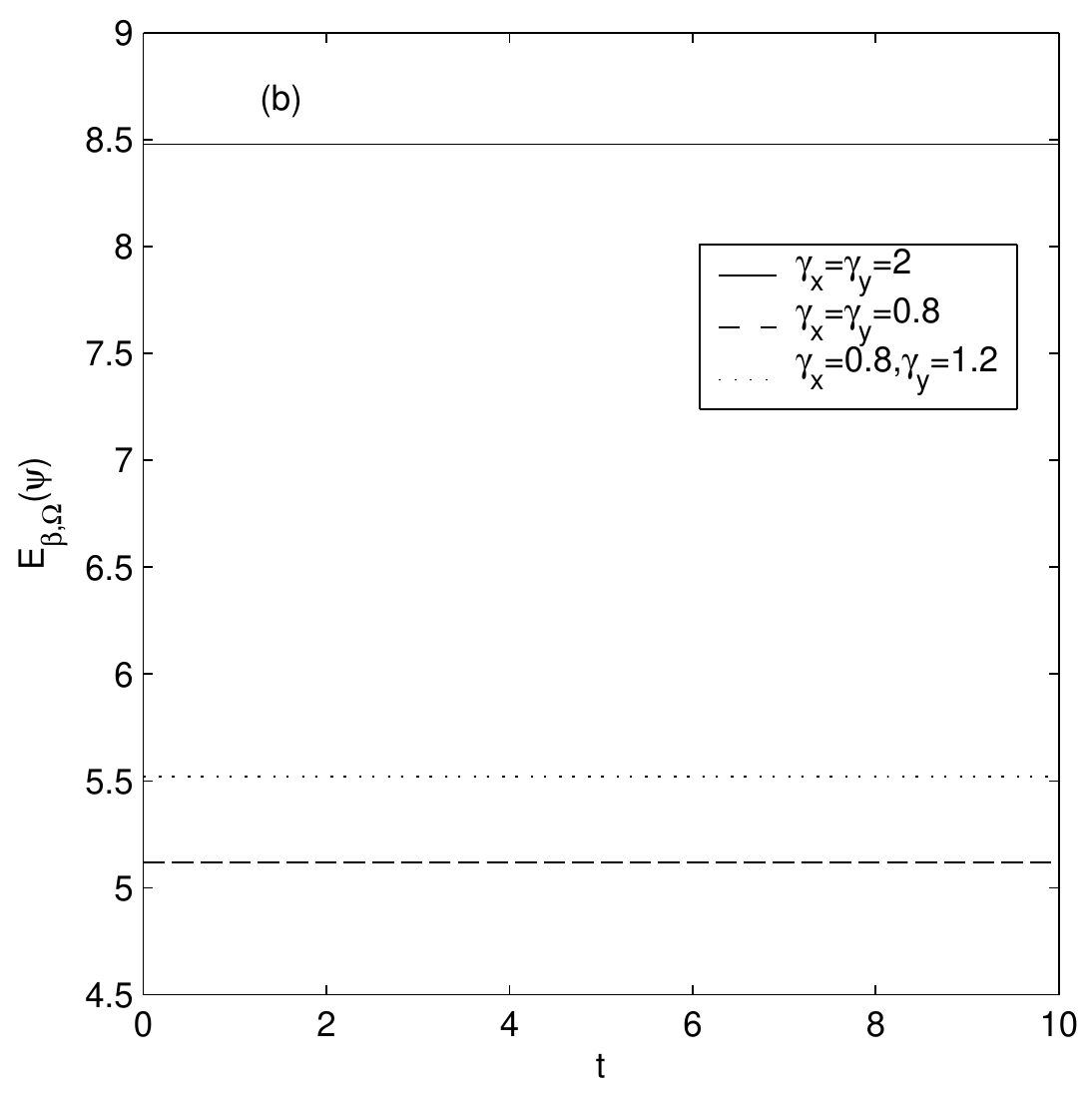,height=5cm,width=5cm,angle=0}}
\centerline{\psfig{figure=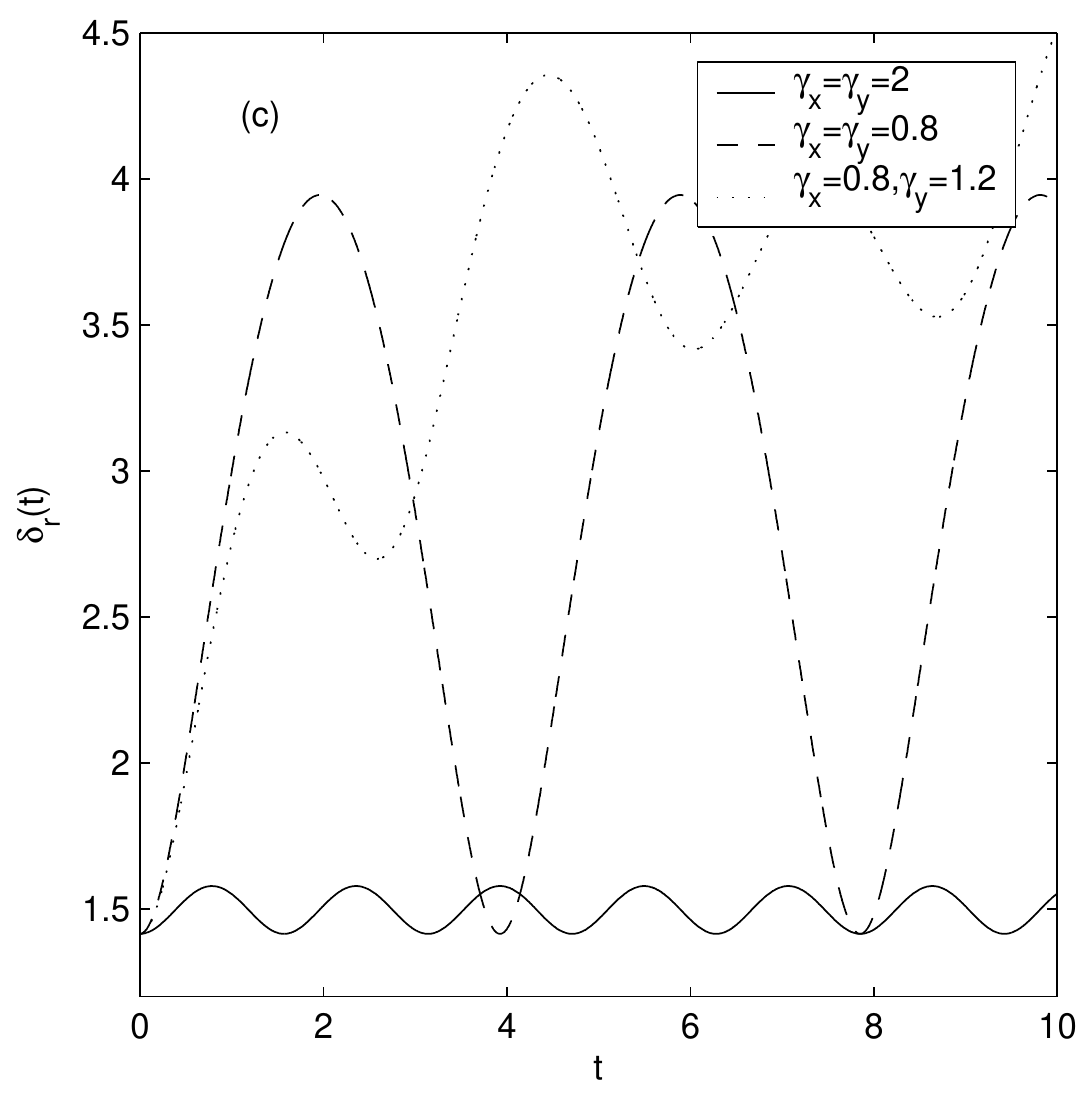,height=5cm,width=5cm,angle=0}
\psfig{figure=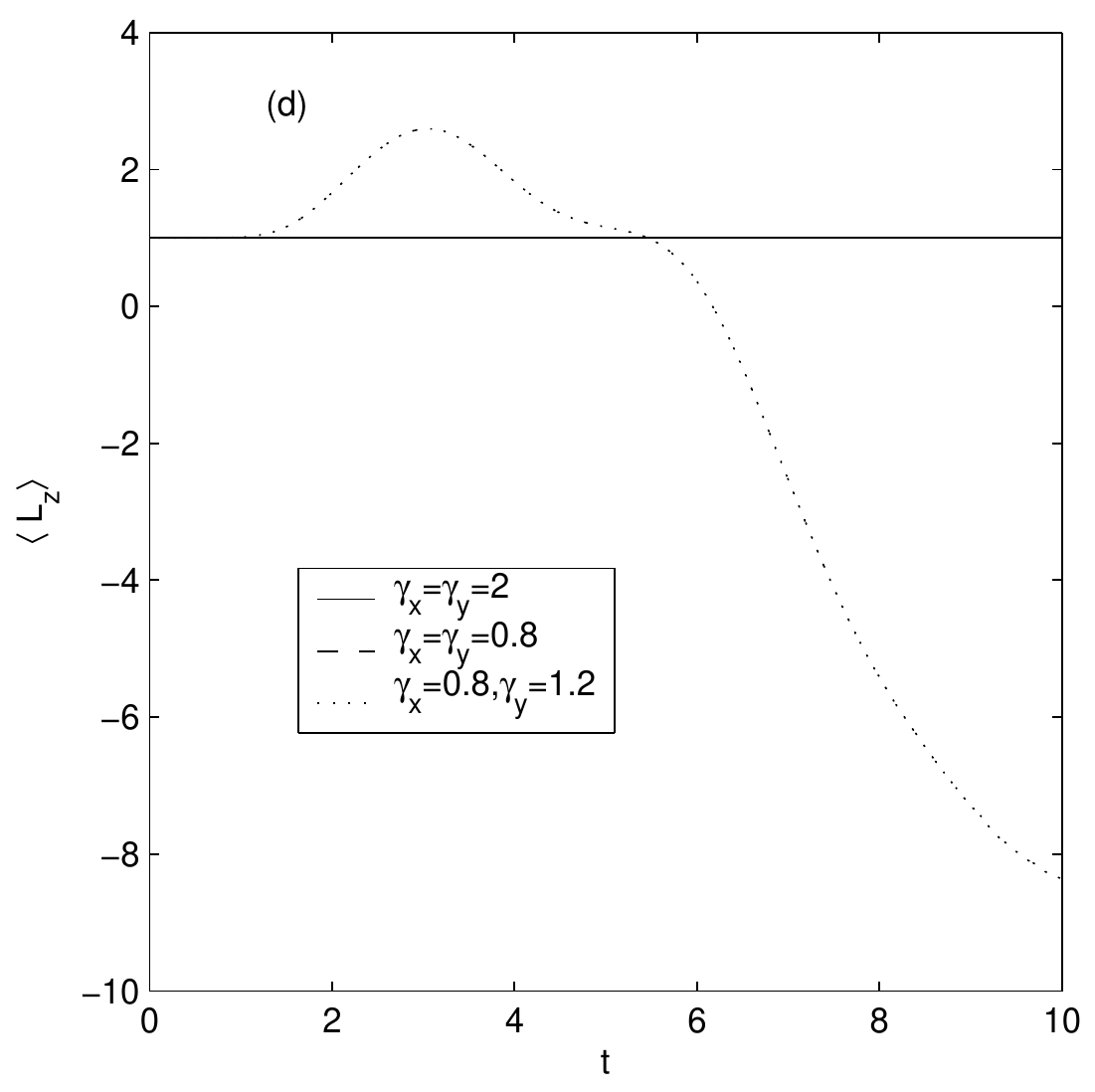,height=5cm,width=5cm,angle=0}} \caption{Time
evolution of a few quantities for the dynamics of rotating BEC in 2D
with three sets of parameters: (a) normalization $N(\psi)$, (b)
energy $E_{\beta,\Og}(\psi)$, (c) condensate width $\dt_r(t)$, and
(d) angular momentum expectation $\langle L_z\rangle (t)$. }
\label{fig:3:sec6}
\end{figure}

\begin{example}\label{exm:2:sec6} Dynamics of a rotating BEC in 2D, i.e. we
take $d=2$, $\beta=100$, $\Og=0.5$ and $W(\bx)\equiv 0$ in
(\ref{eq:gpegrot2:sec6}). The initial data in (\ref{eq:gpegrot2:sec6}) is chosen as
\be\label{inita2d4}
\psi_0(x,y)=\frac{x+iy}{\sqrt{\pi}}e^{-(x^2+y^2)/2}, \qquad (x,y)\in
{\mathbb R}^2. \ee We solve the problem by the scheme
(\ref{eq:tssp2Dprd:sec6}) with $\tau =0.0005$, $M=128$ and $K=200$.
 Fig.~\ref{fig:3:sec6} depicts time evolution of
 the normalization $N(\psi)$, energy
  $E_{\beta,\Og}(\psi)$, condensate width $\dt_r(t)$ and
angular momentum expectation $\langle L_z\rangle (t)$ for three sets
of parameters in (\ref{eq:gpegrot2:sec6}): (i) $\gm_x=\gm_y=2$, (ii)
$\gm_x=\gm_y=0.8$, and (iii) $\gm_x=0.8$, $\gm_y=1.2$.
\end{example}

From Fig.~\ref{fig:3:sec6}, we can draw the following conclusions: (i)
the normalization $N(\psi)$ and energy $E_{\beta,\Og}(\psi)$ are
conserved well in the computation (cf. Fig.~\ref{fig:3:sec6}a\&b); (ii)
the angular momentum expectation $\langle L_z\rangle (t)$ is
conserved when $\gm_x=\gm_y$ (cf. Fig.~\ref{fig:3:sec6}d), i.e. the
trapping is radially symmetric, which again confirms the analytical
results in section \ref{subsec:dynrot} ; (iii) the condensate width $\dt_r(t)$ is a
periodic function when $\gm_x=\gm_y$ (cf. Fig.~\ref{fig:3:sec6}c), which
again confirms the analytical results in section \ref{subsec:dynrot}.

\section{Semiclassical scaling and limit}
\label{sec:semiclass}\setcounter{equation}{0}\setcounter{figure}{0}\setcounter{table}{0}
In section \ref{subsec:gpe}, we have introduced the scaling in the GPE (\ref{eq:GPE}) to obtain the dimensionless form
which has been widely adopted in physics literatures. For BEC with or without the rotational frame (cf. section \ref{sec:rotat}), the dimensionless GPE in $d$-dimensions ($d=1,2,3$) (lower dimensions with $d=1,2$ are treated
as from 3D GPE by dimension reduction) can be written as
\be\label{eq:gpe:sec7}
i\p_t\psi(\bx,t)=\left[-\frac12\nabla^2+V(\bx)-\Omega L_z+\beta|\psi|^2\right]\psi,
\quad \bx\in U\subseteq\Bbb R^d,\quad t>0,
\ee
with normalization condition
\be\label{eq:normg:sec7}
\|\psi(\cdot,t)\|_2^2=\int_{\Bbb R^d}|\psi(\bx,t)|^2\,d\bx=1,
\ee
where $\psi:=\psi(\bx,t)$ is the macroscopic wave function, $U=[0,1]^d$ for box potentials, $U=\Bbb R^d$ for harmonic potential and other confining potentials (cf. section \ref{subsec:gpe}), $L_z=-i(y\p_x-x\p_y)$ for $d=2,3$, and $\Omega=0$ for $d=1$. The energy $E(\psi)$ for (\ref{eq:gpe:sec7}) is given by
\be\label{eq:denergy:sec7}
E(\psi(\cdot,t))=\int_{U} \left[\fl{1}{2} \left|\btd \psi(\bx,t)\right|^2+
V(\bx)|\psi|^2 +\fl{\beta}{2}\; |\psi|^4-
\Omega \bar{\psi}\, L_z \psi\right]d\bx.
\ee

The ground state  $\phi_g$ of the GPE (\ref{eq:gpe:sec7}) is  the minimizer of  the
energy $E(\phi)$ (\ref{eq:denergy:sec7}) over the unit sphere $S=\{\phi \ |\ \|\phi\|_2=1, \ E(\phi)<\ift\}$.
It can also be  characterized by the nonlinear eigenvalue problem:
\begin{align}
\label{eq:charactereq:sec7}
&\mu\;\phi({\bx}) =  -\fl{1}{2}\Dt\phi({\bx})+V({\bx})\phi({\bx})
-\Omega  L_z \phi+\bt|  \phi({\bx})|^2\phi({\bx}), \qquad {\bx}\in U,\\
&\phi(\bx)|_{\p U}=0,\nn
\end{align}
under  the normalization condition (\ref{eq:normg:sec7}) with $\psi=\phi$.
Here, the nonlinear
eigenvalue (or chemical potential) $\mu$ can be computed from its corresponding eigenfunction
$\phi({\bx})$ by
\bea
\label{eq:mu-energy:sec7}
\mu & = & \mu(\phi)
=\int_{U}\left[\fl{1}{2}|\nabla\phi({\bx})|^2+
V({\bx})|\phi({\bx})|^2+\bt|\phi({\bx})|^4-
\Omega \bar{\phi}\, L_z \phi\right]d{\bx} \nn \\
&=&  E(\phi)+\int_{U}\fl{\bt}{2}|\phi({\bx})|^4d{\bx}.
\eea

\subsection{Semiclassical scaling in the whole space}
When $U = {\Bbb R^d}$,  $\bt \gg 1$ and $V(\bx) = V_0(\bx) +
W(\bx)$ satisfies
\be
\label{eq:gpon:sec7}
 V_0(\lambda\bx) = |\lambda|^\alpha V_0(\bx),\quad \forall \lambda\in\Bbb R,\quad
\lim\limits_{|\bx|\to\infty}V_0(\bx)=\ift,\quad \lim_{|\bx|\to \ift}  \fl{W(\bx)}{V_0(\bx)}=0,
\ee
where $\bx\in\Bbb R^d$ and  $\alpha>0$, another  scaling  (under the normalization (\ref{eq:normg:sec7})
 with $\psi$ being
replaced by $\psi^\vep$) -- semiclassical scaling -- for (\ref{eq:gpe:sec7}) is also very useful in
practice by choosing,
 $t\to t \vep^{(\ap-2)/(\ap+2)}$, $\bx\to \bx
\vep^{-2/(2+\ap)}$, 
and $\psi = \psi^\vep\; \vep^{d/(2+\ap)}$ with
$\vep=1/\bt^{(\ap+2)/2(d+\ap)}$ \cite{BaoZhang,BaoDuZhang,BaoChai}:
\be
\label{eq:semiclass:sec7}
i\vep  \fl{\partial\psi^\vep({\bx},t)}{\partial t}
=\left[-\fl{\vep^2}{2}\nabla^2 -\vep^{\frac{2\alpha}{2+\alpha}}\Omega L_z\right]\psi^\vep + (V_0(\bx)+W^\vep(\bx))\psi^\vep
+ |\psi^\vep|^2\psi^\vep, \ {\bx}\in {\Bbb R^d},
\ee
where  $W^\vep(\bx) = \vep^{2\ap/(2+\ap)} W(\bx/\vep^{2/(2+\ap)})$ and
the energy functional $E^\vep(\psi^\vep)$ is defined as
\begin{equation} \label{eq:energysem:sec7}
E^\vep(\psi^\vep)=\int_{{\Bbb R}^d}\left[\fl{\vep^2}{2}|\nabla
\psi^\vep|^2-
\vep^{\frac{2\alpha}{2+\alpha}}\Omega \overline{\psi^\vep} L_z \psi^\vep+(V_0 + W^\vep)  |\psi^\vep|^2
+\fl{|\psi^\vep|^4}{2}\right]d{\bx}=O(1).
\end{equation}

Similarly, the nonlinear eigenvalue problem (\ref{eq:charactereq:sec7})
(under the normalization (\ref{eq:normg:sec7}) with $\psi=\phi^\vep$)  reads
\be
\label{eq:charac:sec7} \mu^\vep\phi^\vep({\bx}) =  -\fl{\vep^2}{2}\Dt\phi^\vep
+(V_0({\bx})+W^\vep(\bx))\phi^\vep-\vep^{\frac{2\alpha}{2+\alpha}}\Omega L_z\phi^\vep+|  \phi^\vep|^2\phi^\vep, \quad {\bx}
\in {\Bbb R}^d ,
\ee
where eigenvalue $\mu^\vep$
can be computed from its corresponding
eigenfunction $\phi^\vep$ by
\be \mu^\vep=\mu^\vep(\phi^\vep)=E^\vep(\phi^\vep)+\frac{1}{2}\int_{\Bbb R^d}|\phi^\vep|^4\,d\bx=O(1).\ee

Based on this re-scaling, it is  easy to get the leading asymptotics of the energy
functional  $E(\psi)$ in (\ref{eq:denergy:sec7})  and the chemical potential
(\ref{eq:mu-energy:sec7}) when $\bt\gg1$ from this scaling \cite{BaoZhang,BaoDuZhang,BaoChai}:
\bea
\label{eq:energyl:sec7}
&&E(\psi)= \vep^{-2\ap/(2+\ap)} E^\vep(\psi^\vep) =
O\left(\vep^{-2\ap/(2+\ap)}\right)=O\left(\bt^{\ap/(d+\ap)}\right), \\
\label{eq:muresal5:sec7}
&&\mu(\phi)= \vep^{-2\ap/(2+\ap)} \mu^\vep(\phi^\vep) =
O\left(\vep^{-2\ap/(2+\ap)}\right)=O\left(\bt^{\ap/(d+\ap)}\right).
\eea
In \cite{Carles,Sparber}, a different rescaling for the nonlinear
Schr\"{o}dinger equation subject to smooth, lattice-periodic
potentials was used in the semiclassical regime. There they
studied Bloch waves dynamics in BEC on optical
lattices.


\subsection{Semiclassical scaling in bounded domain}
When $U=(0,1)^d\subset\Bbb R^d$ is a bounded domain,  $\bt\gg 1$, we use  the following scaling  (under the normalization (\ref{eq:normg:sec7})) with $\psi$ being
replaced by $\psi^\vep$) -- semiclassical scaling -- for (\ref{eq:gpe:sec7}) by choosing $t\to t \vep^{-1}$,  and $\psi = \psi^\vep$ with
$\vep=1/\sqrt{\bt}$ \cite{BaoChai}:
\be
\label{eq:semiclassb:sec7}
i\vep  \fl{\partial\psi^\vep({\bx},t)}{\partial t}
=\left[-\fl{\vep^2}{2}\nabla^2 -\vep^2\Omega L_z\right]\psi^\vep + V^\vep(\bx)\psi^\vep
+ |\psi^\vep|^2\psi^\vep, \ {\bx}\in U,
\ee
where  $V^\vep(\bx) = \vep^2 V(\bx)$ and
the energy functional $E^\vep(\psi^\vep)$ is defined as
\begin{equation} \label{eq:energysemb:sec7}
E^\vep(\psi^\vep)=\int_{U}\left[\fl{\vep^2}{2}|\nabla
\psi^\vep|^2-
\vep^2\Omega \overline{\psi^\vep} L_z \psi^\vep+V^\vep  |\psi^\vep|^2
+\fl{|\psi^\vep|^4}{2}\right]d{\bx}=O(1).
\end{equation}

We can derive the leading asymptotics of the energy
functional  $E(\psi)$ in (\ref{eq:denergy:sec7})  and the chemical potential
(\ref{eq:mu-energy:sec7}) when $\bt\gg1$ from this scaling in the bounded  domain case \cite{BaoChai}:
\bea
\label{eq:energylb:sec7}
&&E(\psi)= \vep^{-2} E^\vep(\psi^\vep) =
O\left(\vep^{-2}\right)=O(\bt), \\
\label{eq:muresal5b:sec7}
&&\mu(\phi)= \vep^{-2} \mu^\vep(\phi^\vep) =
O\left(\vep^{-2}\right)=O(\bt).
\eea

For comparison, Tabs.~\ref{tab:1:sec7} and \ref{tab:2:sec7} display dimensionless units and several important quantities to obtain the GPE (\ref{eq:gpe:sec7}) under the standard physical scaling and (\ref{eq:semiclass:sec7}) or (\ref{eq:semiclassb:sec7}) under the semiclassical scaling
for a BEC in the whole space with harmonic potential (\ref{eq:hp}) ($\og_x=\min\{\og_x,\og_y,\og_z\}$) and in a bounded domain with box potential (\ref{eq:box3d}), respectively. Again, the dimensionless GPE in lower dimensions with $d=1,2$ are treated
as from 3D GPE by dimension reduction.

{\renewcommand\baselinestretch{1.2}\selectfont
 \begin{table}\begin{center}
\begin{tabular}{ccc}\hline\hline
Quantities &Physical scaling &Semiclassical scaling\\ \hline\hline
time unit $t_s$ &$\frac{1}{\og_x}$&$\frac{1}{\og_x}$\\
length unit $x_s$ &$\sqrt{\frac{\hbar}{m\og_x}}:=a_0$ &$\frac{a_0}{\vep^{1/2}}$\\
energy unit $E_s$ &$\hbar\og_x$ &$\frac{\hbar\og_x}{\vep}$\\
wave amplitude unit $\psi_s$ &$a_0^{-3/2}$
 &$a_0^{-3/2}\vep^{d/4}$ \\
\hline
healing length $\xi_h$ &$O(\beta^{-1/(2+d)})$ &$O(\vep)$ \\
energy $E_g$ &$O(\beta^{2/(d+2)})$&$O(1)$\\
chemical potential $\mu_g$ &$O(\beta^{2/(d+2)})$&$O(1)$\\
Thomas-Fermi radius $R^{\rm TF}_g$ &$O(\beta^{1/(d+2)})$&$O(1)$\\
wave amplitude $\phi_g^{\rm max}$ &$O(\beta^{-d/2(d+2)})$& $O(1)$\\
\hline\hline
\end{tabular}
\end{center}
\caption{Comparison of dimensionless units and several important quantities under
the standard physical scaling and semiclassical scaling. Here $t_s$ is time unit, $x_s$ is length unit,
$E_s$ is energy unit, $\psi_s$ is wave function unit, where $m$, $\hbar$, $a_s$ and $N$ are the mass, Planck constant, $s$-wave scattering length and total particle number, respectively (cf. \ref{subsec:gpe}).
$\xi_h$ is the healing length \cite{Pethick}, $E_g$ is the energy of ground state, $\mu_g$ is the chemical potential of ground state,
$R_g^{\rm TF}$ is the Thomas-Fermi radius of the ground state, $\phi_g^{\rm max}$ is the maximum value of ground state.
(a) For a BEC in the whole space with a harmonic potential (\ref{eq:hp}) ($\og_x=\min\{\og_x,\og_y,\og_z\}$),
$\beta$ is given in (\ref{eq:dhp:sec1}) and $\vep=1/\beta^{2/(2+d)}$.}\label{tab:1:sec7}
\end{table}}

\setcounter{table}{0}

\begin{table}\begin{center}
\begin{tabular}{ccc}\hline\hline
Quantities &Physical scaling &Semiclassical scaling\\\hline
\hline
time unit $t_s$ &$\frac{mL^2}{\hbar}$ & $\frac{mL^2}{\hbar\vep^2}$\\
length unit $x_s$ &$L$ &$L$\\
energy unit $E_s$ &$\frac{\hbar^2}{mL^2}$ &$\frac{\hbar^2}{m\vep^2L^2}$\\
wave amplitude unit $\psi_s$ &$L^{-3/2}$ &$L^{-3/2}$ \\
\hline
healing length $\xi_h$& $O(\beta^{-1/2})$ &$O(\vep)$ \\
energy $E_g$ &$O(\beta)$&$O(1)$\\
Chemical potential $\mu_g$ &$O(\beta)$&$O(1)$\\
Thomas-Fermi radius $R_g^{\rm TF}$  &$O(1)$&$O(1)$\\
wave amplitude $\phi_g^{\rm max}$  &$O(1)$& $O(1)$\\
\hline\hline
\end{tabular}
\end{center}
\caption{(Con't) (b) For a BEC in the box potential (\ref{eq:box3d}) with $L$ the size of box potential.
$\beta=4\pi a_s N/L$ and $\vep=\beta^{-1/2}$.}\label{tab:2:sec7}
\end{table}
}


\subsection{Semiclassical limits and geometric optics}

Suppose $V^\vep(\bx)=V_0(\bx)+W^\vep(\bx)$ in (\ref{eq:semiclass:sec7}), and we set
\be\psi^\vep(\bx,t)=\sqrt{\rho^\vep(\bx,t)}\exp\left(\fl{i}{\vep}S^\vep(\bx,t)\right), \label{eq:anssemi:sec7}\ee
where $\rho^\vep=|\psi^\vep|^2$ and $S^\vep$ are the density and phase of the wave function, respectively.
Inserting (\ref{eq:anssemi:sec7}) into the GPE (\ref{eq:semiclass:sec7})  and
separating real and imaginary parts give \cite{GMMP,GaMa,BaoDuZhang,Carles01,ZhangP}
\bea
&&\rho_t^\vep+{\rm div}\, (\rho^\vep\; \btd S^\vep)+\Omega \hat{L}_z\rho^\vep=0, \label{eq:eqtrans:sec7}\\
&&S_t^\vep+\fl{1}{2}|\btd S^\vep|^2 +
\rho^\vep+V^\vep(\bx)+\Omega \hat{L}_zS^\vep
=\fl{\vep^2}{2}\;\fl{1}{\sqrt{\rho^\vep}}\ \Delta \sqrt{\rho^\vep},
\label{eq:eqHJ:sec7}
 \eea
 where $\hat{L}_z=(x\p_y-y\p_x)$.
The equation (\ref{eq:eqtrans:sec7}) is the transport equation for the
atom density and (\ref{eq:eqHJ:sec7}) the Hamilton-Jacobi equation for
the phase.

 By formally passing to the limit
$\vep\to 0$ (cf. \cite{SAGardiner}), we obtain the system \bea \label{eq:eqHJ1:sec7}
&&\rho^0_t+{\rm div}\, (\rho^0\; \btd S^0)+\Omega \hat{L}_z\rho^0=0,\\
\label{eq:HJ2:sec7}
&&S^0_t+\fl{1}{2}|\btd S^0|^2 +\rho^0+V_0(\bx)+\Omega \hat{L}_zS^0=0.
 \eea

It is well known that this limit process is only correct in the
defocusing case $\beta>0$ before caustic onset, i.e. in
time-intervals where the solution of the Hamilton-Jacobian equation
(\ref{eq:eqHJ:sec7}) coupled with the atom-number conservation equation
(\ref{eq:eqtrans:sec7}) is smooth. After the breakdown of regularity,
oscillations occur, which make the term
$\fl{\vep^2}{2}\;\fl{1}{\sqrt{\rho^\vep}}\ \btu \sqrt{\rho^\vep}$ at least $O(1)$
such that the validity of the formal limit process is destroyed.
The limiting behavior after caustics
onset is not understood yet except in 1D case
without confinement, see \cite{ShenJin}. Also, the focusing case
$\beta < 0$ is not fully understood yet.

Furthermore, by defining the
current densities \bea {\bf J}^{\vep}(\bx,
t)=\rho^{\vep}\nabla S^{\vep} = \vep\,{\rm
Im}\left[\overline{\psi^{\vep}}\nabla\psi^{\vep}\right],
 \eea
  we can rewrite (\ref{eq:eqtrans:sec7})-(\ref{eq:eqHJ:sec7}) as a
coupled Euler system with third-order dispersion  terms \cite{GMMP,GaMa,BaoDuZhang,Carles01,ZhangP} {\small\begin{align}
\label{eq:6rho_1:sec7}&\p_t\rho^{\vep}+{\rm div}{\bf
J}^{\vep}+\Og\widehat{L}_z\rho^{\vep} = 0, \\
\label{eq:6J:sec7} &\p_t{\bf J}^{\vep}+{\rm div}\left(\fl{{\bf
J}^{\vep}\otimes{\bf
J}^{\vep}}{\rho^{\vep}}\right)+\rho^{\vep}\nabla
V^\vep(\bx)+\frac12\nabla
\left(\rho^{\vep}\right)^2+\Og\widehat{L}_z{\bf J}^{\vep}=\fl{\vep^2}{4}\nabla\left(\rho^{\vep}
\nabla^2\ln\rho^{\vep}\right).\quad \end{align}}

 Letting
$\vep\to0^+$ in (\ref{eq:6rho_1:sec7})-(\ref{eq:6J:sec7}), formally we get an Euler
system coupling through the pressures \cite{GMMP,GaMa,BaoDuZhang,Carles01,ZhangP} \bea
\label{eq:6rho_3:sec7}&&\p_t\rho^{0}+{\rm div}{\bf
J}^{0}+\Og\widehat{L}_z\rho^{0} = 0, \\
\label{eq:6J3:sec7} &&\p_t{\bf J}^{0}+{\rm div}\left(\fl{{\bf
J}^{0}\otimes{\bf
J}^{0}}{\rho^{0}}\right)+\rho^{0}\nabla V_0(\bx)+\frac12\nabla
\left(\rho^{0}\right)^2 +\Og\widehat{L}_z{\bf J}^{0} =0, \qquad \quad \eea where ${\bf
J}^{0}(\bx, t)=\rho^{0}\nabla S^{0}$. The system
(\ref{eq:6rho_3:sec7})-(\ref{eq:6J3:sec7}) is a coupled isotropic Euler system
with quadratic pressure-density constitutive relations in the
rotational frame. The formal asymptotics is supposed to hold up to
caustic onset time \cite{GaMa,GMMP}.

\section{Mathematical theory and numerical methods for dipolar BEC}
\label{sec:dipole}
\setcounter{equation}{0}\setcounter{figure}{0}\setcounter{table}{0}
In the last several years, there has been a quest for realizing a
novel kind of quantum gases with the dipolar interaction, acting
between particles having a permanent magnetic or electric dipole
moment. In 2005, the first dipolar BEC
with $^{52}$Cr atoms was successfully realized in experiments at the
Stuttgart University \cite{Griesmaier}.
Later in 2011, a dipolar BEC  with $^{164}$Dy atoms, whose
dipole-dipole interaction is much stronger than that of  $^{52}$Cr,
was achieved in experiments at the
Stanford University \cite{Lum}. Very recently in 2012, a dipolar BEC of $^{168}$Er atoms
has been produced in Innsbruck University \cite{Aikawa}.
 These successes of
experiments have renewed interests in theoretically studying dipolar BECs.
\subsection{GPE with dipole-dipole interaction}
At temperature $T$ much smaller than the critical temperature $T_c$,
a dipolar BEC is well described by the macroscopic wave function
$\psi=\psi(\bx,t)$ whose evolution is governed by the
3D Gross-Pitaevskii equation (GPE)
\cite{Yi,Santos,Baranov} \be \label{eq:ngpe:sec8} i\hbar \p_t
\psi(\bx,t)=\left[-\fl{\hbar^2}{2m}\nabla^2+V(\bx)+g|\psi|^2+
\left(V_{\rm dip}\ast |\psi|^2\right)\right]\psi, \quad \bx\in{\Bbb
R}^3, \ t>0, \ee where $\bx=(x,y,z)^T\in {\Bbb R^3}$ is
the Cartesian coordinate and a harmonic trap potential $V(\bx)$ is considered here.
$g=\frac{4\pi \hbar^2 a_s}{m}$ describes local (or short-range)
interaction between dipoles in the condensate with $a_s$ the
$s$-wave scattering length. The long-range dipolar
interaction potential between two dipoles is given by
\be\label{eq:kel0:sec8} V_{\rm dip}(\bx)= \frac{\mu_0\mu_{\rm
dip}^2}{4\pi}\,\fl{1-3(\bx\cdot \bf
n)^2/|\bx|^2}{|\bx|^3}=\frac{\mu_0\mu_{\rm
dip}^2}{4\pi}\,\fl{1-3\cos^2(\theta)}{|\bx|^3}, \qquad \bx\in{\Bbb
R}^3,\ee  where $\mu_0$ is the vacuum magnetic permeability,
$\mu_{\rm dip}$ is permanent magnetic dipole moment (e.g. $\mu_{\rm
dip}=6\mu_{_B}$ for $^{52}$C$_{\rm r}$ with $\mu_{_B}$ being the
Bohr magneton), ${\bf n}=(n_1,n_2,n_3)^T\in {\Bbb R}^3$ is the
dipole axis (or dipole moment) which is a given unit vector, i.e.
$|{\bf n}|=\sqrt{n_1^2+n_2^2+n_3^3}=1$, and $\theta$ is the angle
between the dipole axis ${\bf n}$ and the vector $\bx$. The wave
function is normalized according to \be\label{eq:norm00:Sec8}
\|\psi\|_2^2:=\int_{{\Bbb R}^3} |\psi(\bx,t)|^2\;d\bx =N,\ee where $N$
is the total number of dipolar particles in the dipolar BEC.

By introducing the dimensionless variables, $t\to \frac{t}{\og_{0}}$
with $\og_0=\min\{\og_x,\og_y,\og_z\}$,  $\bx \to x_s\bx$ with $
x_s=\sqrt{\fl{\hbar}{m\og_{0}}}$, $\psi\to \frac{\sqrt{N}
\psi}{x_s^{3/2}}$,
 we obtain the dimensionless GPE  in 3D from (\ref{eq:ngpe:sec8})
 as :
 \be \label{eq:ngpe1:sec8} i\p_t \psi(\bx,t)=\left[-\fl{1}{2}\nabla^2+V(\bx)+\beta
|\psi|^2+\lambda \left(U_{\rm dip}\ast|\psi|^2\right)\right]\psi,
\quad \bx\in{\Bbb R}^3, \; t>0,\ee where $\beta
=\frac{Ng}{\hbar\og_0 x_s^3}= \fl{4\pi a_sN}{x_s}$, $\lambda
=\fl{mN\mu_0\mu_{\rm dip}^2}{3\hbar^2 x_s}$,
$V(\bx)=\fl{1}{2}(\gamma_{x}^2x^2+ \gamma_y^2y^2+\gamma_{z}^2 z^2)$
is the dimensionless harmonic trapping potential with
$\gm_x=\frac{\og_x}{\og_0}$, $\gm_y=\frac{\og_y}{\og_0}$ and
$\gm_z=\frac{\og_z}{\og_0}$, and the dimensionless long-range
dipolar interaction potential $U_{\rm dip}(\bx)$ is given as
\be\label{eq:kel:sec8} U_{\rm dip}(\bx)= \frac{3}{4\pi}\,\fl{1-3(\bx\cdot
\bf n)^2/|\bx|^2}{|\bx|^3}=\frac{3}{4\pi}\,
\fl{1-3\cos^2(\theta)}{|\bx|^3}, \qquad \bx\in{\Bbb R}^3.\ee
Although the kernel $U_{\rm dip}$
is highly singular near the origin,
the convolution is well-defined for $\rho\in L^p(\Bbb R^3)$
with $U_{\rm dip}*\rho\in L^p(\Bbb R^3)$  for $p\in(1,\infty)$ \cite{Carles}.

Denote the
differential operators $\partial_{\bn}=\bn\cdot\nabla$ and
$\partial_{\bn\bn}=\partial_{\bn}\partial_{\bn}$, and  notice the
identity \cite{BaoCaiWang} \be \label{eq:decop1:sec8} U_{\rm dip}(\bx)=\fl{3}{4\pi
|\bx|^3}\left(1-\frac{3(\bx\cdot {\bf n})^2}{|\bx|^2}\right) = -
\delta (\br)-3\p_{\bn\bn}\left( \frac{1}{4\pi |\bx|}\right),\qquad
\bx\in {\Bbb R}^3,\ee  we can
re-formulate  the GPE (\ref{eq:ngpe1:sec8})  as the following
Gross-Pitaevskii-Poisson system (GPPS) \cite{BaoCaiWang,BaoBenCai,CaiRosen} \begin{align}
\label{eq:gpe:sec8} &i \p_t
\psi(\bx,t)=\left[-\fl{1}{2}\nabla^2+V(\bx)+(\beta-\lambda)
|\psi|^2-3\lambda \p_{\bn\bn} \varphi \right]\psi, \quad \bx\in{\Bbb
R}^3,
\quad t>0,  \\
\label{eq:poisson:sec8}&\qquad \nabla^2 \varphi(\bx,t) =
-|\psi(\bx,t)|^2,\qquad \bx\in{\Bbb R}^3, \qquad
\lim\limits_{|\bx|\to\infty}\varphi(\bx,t)=0,\qquad  t\ge0.
 \end{align}
 The above GPPS in 3D conserves the {\sl mass}, or the {\sl normalization} condition,
 \be
\label{eq:norm3d:sec8}N(\psi(\cdot,t)):=\|\psi(\cdot,t)\|_{2}^2=\int_{{\Bbb
R}^3}|\psi(\bx,t)|^2\;d\bx\equiv \int_{{\Bbb
R}^3}|\psi(\bx,0)|^2\;d\bx=1, \quad t\ge0,\ee
 and  {\sl energy} per particle with $\varphi=\frac{1}{4\pi|\bx|}*|\psi|^2$,
\be\label{eq:energydp:sec8}
E_{\rm 3D}(\psi)=\int_{\Bbb
R^3}\left[\frac12|\nabla\psi|^2+V(\bx)|\psi|^2+\frac{\beta-\lambda}{2}
|\psi|^4+\frac{3\lambda}{2}\left|\p_{\bn}\nabla\varphi\right|^2\right]\,d\bx.
\ee

From
(\ref{eq:decop1:sec8}), it is straightforward to get the Fourier
transform of $U_{\rm dip}(\bx)$ as \be\label{eq:four11:sec8}
\widehat{(U_{\rm dip})}(\xi)=-1+\frac{3\left(\bn\cdot
\xi\right)^2}{|\xi|^2}, \qquad \xi\in {\Bbb R}^3. \ee
\subsection{Dimension reduction}
In many physical experiments of dipolar BECs,  the condensates
 are confined with  strong harmonic trap in one  or two axis directions,
 resulting in a disk- or cigar-shaped dipolar BEC, respectively.
 Mathematically speaking, this corresponds to the anisotropic potentials
 $V(\bx)$ of the form:

{\sl Case I} (disk-shaped), potential is strongly confined in
vertical $z$ direction
with \be\label{eq:case1:sec8} V(\bx)=V_2(x,y)+\frac{z^2}{2\vep^4}, \qquad
\bx\in{\Bbb R}^3, \ee

 {\sl Case II} (cigar-shaped), potential is strongly confined in
 horizontal $\bx_{\perp}=(x,y)^T\in {\Bbb R}^2$
 plane with
\be\label{eq:case2:sec8} V(\bx)=V_1(z)+\frac{x^2+y^2}{2\vep^4}, \qquad
\bx\in {\Bbb R}^3, \ee
where $0<\vep\ll 1$ ($\vep=1/\gamma_z$ in Case I and $\vep=1/\gamma_r$
with $\gamma_x=\gamma_y=\gamma_r$ in Case II) is a small parameter describing
the strength of confinement. In such cases, the above GPPS in 3D can be formally reduced
to 2D and 1D, respectively \cite{CaiRosen}.

In  {\sl Case I}, when $\vep\to 0^+$,  evolution of the solution
$\psi(\bx,t)$ of GPPS (\ref{eq:gpe:sec8})-(\ref{eq:poisson:sec8}) in $z$-direction
would essentially occur in the ground state mode of
$-\frac12\p_{zz}+\frac{z^2}{2\vep^4}$, which is spanned by
$w_\vep(z)=\vep^{-1/2}\pi^{-1/4}e^{-\frac{z^2}{2\vep^2}}$
\cite{CaiRosen,BaoBenCai}. By taking the ansatz
\be\label{eq:an1:sec8}\psi(\bx_{\perp},z,t)=e^{-it/2\vep^2}\phi(\bx_{\perp},t)w_\vep(z),
\qquad (\bx_{\perp},z)^T=(x,y,z)^T\in{\Bbb R}^3, \quad t\ge0, \ee the 3D GPPS
(\ref{eq:gpe:sec8})-(\ref{eq:poisson:sec8}) can be formally reduced to a {\sl
{quasi-2D equation} I} \cite{CaiRosen,BaoBenCai}: \bea \label{eq:gpe2d:sec8} i\p_t
\phi=\left[-\frac12\nabla^2+V_2+\beta_{2D} |\phi|^2-\frac{3\lambda}{2}(
\p_{\bn_\perp\bn_\perp} -n_3^2\nabla^2)\varphi^{2D} \right]\phi,  \eea
where  $\beta_{2D}=\frac{\beta-\lambda+3\lambda
n_3^2}{\sqrt{2\pi}\,\vep}$,
$\bn_\perp=(n_1,n_2)^T$, $\p_{\bn_\perp}=\bn_\perp\cdot\nabla=\bn_\perp\cdot(\p_x,\p_y)^T$,
$\p_{\bn_\perp\bn_\perp}=\p_{\bn_\perp}(\p_{\bn_\perp})$,
$\nabla^2=\p_{xx}+\p_{yy}$ and
 \be\label{eq:u2d1:sec8}
 \varphi^{2D}(x,y,t)=U^{2D}_\vep*|\phi|^2,\quad
 U^{2D}_\vep(x,y)=\frac{1}{2\sqrt{2}\pi^{3/2}}
 \int_{\Bbb R}\frac{e^{-s^2/2}}{\sqrt{x^2+y^2+
\vep^2s^2}}\,ds. \ee In addition, as
$\vep\to0^+$, $\varphi^{2D}$ can be approximated by
$\varphi^{2D}_\infty$ \cite{CaiRosen} as :
\be\label{eq:u2d2:sec8} \varphi^{2D}_\infty(\bx_{\perp},t)=U_{\rm dip}^{2D}*|\phi|^2, \quad
\hbox{with} \quad U_{\rm
dip}^{2D}(\bx_{\perp})=\frac{1}{2\pi|\bx_{\perp}|},  \ee which can be re-written as a fractional
Poisson equation  \be\label{eq:u2d2k:sec8}
(-\nabla^2)^{1/2}\varphi^{2D}_\infty(\bx_{\perp}t)=|\phi(\bx_{\perp},t)|^2, \quad \lim\limits_{|\bx_{\perp}|\to\infty}\varphi^{2D}_\infty(\bx_{\perp},t)=0,
\qquad t\ge0. \ee Thus an alternative {\sl{quasi-2D equation} II}
can be obtained as: \be \label{eq:gpe2d2:sec8} i \p_t
\phi=\left(-\frac{1}{2}\nabla^2+V_2+\beta_{2D}|\phi|^2-\frac{3\lambda}{2}(
\p_{\bn_\perp\bn_\perp} -n_3^2\nabla^2)(-\nabla^2)^{-\frac{1}{2}}(|\phi|^2)
\right)\phi.\ee

Similarly, in  {\sl Case II}, evolution of the solution
$\psi(x,y,z,t)$ of GPPS (\ref{eq:gpe:sec8})-(\ref{eq:poisson:sec8}) in $(x,y)$ plane
would essentially occur in the ground state mode of
$-\frac12(\p_{xx}+\p_{yy})+\frac{x^2+y^2}{2\vep^4}$, which
is spanned by
$w_\vep(x,y)=\vep^{-1}\pi^{-1/2}e^{-\frac{x^2+y^2}{2\vep^2}}$
\cite{BaoBenCai,CaiRosen}. Again, by taking the ansatz
\be\label{eq:an2:sec8}\psi(x,y,z,t)=e^{-i t/\vep^2}\phi(z,t)w_\vep(x,y), \quad t\ge0,\ee the 3D GPPS
(\ref{eq:gpe:sec8})-(\ref{eq:poisson:sec8}) can be formally reduced to a {\sl
{quasi-1D equation}}:
 \be \label{eq:gpe1d:sec8} i \p_t
\phi=\left[-\frac{1}{2}\p_{zz}+V_1+\beta_{1D} |\phi|^2 -\frac{3\lambda(
3n_3^2-1)}{8\sqrt{2\pi}\,\vep}\p_{zz}\varphi^{1D} \right]\phi, \quad
z\in{\Bbb R}, \ t>0,\ee where $\beta_{1D}=\frac{\beta+\frac{1}{2}\lambda
(1-3n_3^2)}{2\pi\vep^2}$ and
\be \label{eq:poisson1d:sec8}
\varphi^{1D}(z,t)=U_\vep^{1D}*|\phi|^2, \qquad
U^{1D}_\vep(z)=\frac{\sqrt{2}e^{z^2/2\vep^2}}{\sqrt{\pi}\,\vep}\int_{|z|}^\infty
e^{-s^2/2\vep^2}\,ds.
 \ee
\begin{remark} To describe a rotating dipolar BEC, we only need to include the angular momentum term (\ref{eq:rota:sec5}) in the dipolar GPE (\ref{eq:ngpe:sec8}). Therefore, dimensionless rotating dipolar GPEs in 3D and quasi-2D regime are straightforward.
\end{remark}
\subsection{Theory for ground states}
In this section, we report results for ground state of dipolar BECs. Denote unit sphere
\be
S=X\cap\left\{u\in L^2(\Bbb R^d)\big|\,\|u\|_{L^2(\Bbb R^d)}=1\right\},
\ee
where $X$ is the energy space associated with corresponding potential (\ref{eq:funcspace:sec2}).

\subsubsection{3D case}In 3D, the ground state of GPPS (\ref{eq:gpe:sec8})-(\ref{eq:poisson:sec8}) is the minimizer of energy $E_{3D}$ (\ref{eq:energydp:sec8}) over the nonconvex set $S$ \cite{BaoCaiWang}.
\begin{theorem} Assume $V(\bx)\ge0$ for $\bx\in{\Bbb R}^3$ and
$\lim\limits_{|\bx|\to\infty}V(\bx)=\infty$ (i.e., confining
potential) in GPPS (\ref{eq:gpe:sec8})-(\ref{eq:poisson:sec8}), then we have:

(i) If $\beta\ge0$ and $-\frac{1}{2}\beta\leq \lambda\leq \beta$,
there exists a ground state $\phi_g\in S$, and the positive ground
state $|\phi_g|$ is unique. Moreover, $\phi_g=e^{i\theta_0}|\phi_g|$
for some constant $\theta_0\in\Bbb R$.

(ii) If $\beta<0$, or $\beta\ge0$ and $\lambda<-\frac{1}{2}\beta$ or
$\lambda>\beta$, there exists no ground state, i.e.,
$\inf\limits_{\phi\in S}E_{\rm 3D}(\phi)=-\infty$.
\end{theorem}

By splitting the total energy $E_{\rm 3D}(\cdot)$ in (\ref{eq:energydp:sec8}) into
kinetic, potential, interaction and dipolar energies, i.e. \be
E_{3D}(\phi)=E_{\rm kin}(\phi)+E_{\rm pot}(\phi)+E_{\rm int}(\phi)+E_{\rm
dip}(\phi),\ee where {\small\begin{align}  E_{\rm
kin}(\phi)=&\frac{1}{2}\int_{{\Bbb R}^3} |\nabla\phi(\bx) |^2 d
\bx,\   E_{\rm pot}(\phi) = \int_{{\Bbb R}^3} V(\bx)|\phi(\bx)|^2 d
\bx, \ E_{\rm int}(\phi) =
\frac{\beta}{2}\int_{{\Bbb R}^3} |\phi(\bx) |^4 d \bx,\nn \\
E_{\rm dip} (\phi)=&\frac{\lambda}{2}\int_{{\Bbb R}^3} \left(U_{\rm
dip}\ast|\phi|^2\right) |\phi(\bx)|^2 d \bx=
\frac{\lambda}{2}\int_{{\Bbb R}^3}
|\phi(\bx)|^2\left[-|\phi(\bx)|^2-3
 \p_{\bn\bn} \varphi\right]d \bx\label{eq:dipp03:sec8}\\=&\frac{\lambda}{2}\int_{{\Bbb R}^3} \left[-|\phi(\bx)|^4+3
(\nabla^2\varphi)(\p_{\bn\bn}\varphi)\right]d
\bx=\frac{\lambda}{2}\int_{{\Bbb R}^3} \left[-|\phi(\bx)|^4+3
\left|\p_{\bf n}\nabla \varphi\right|^2\right]d \bx, \nn
\end{align}}
with $\varphi=\frac{1}{4\pi|\bx|}*|\phi|^2$, we have the following
Viral identity \cite{BaoCaiWang}:

\begin{proposition}
Suppose $V(\lambda\bx)=\lambda^2V(\bx)$ for all $\lambda\in\Bbb R$ and $\phi_g$ is the ground state of a dipolar BEC, i.e., the minimizer of energy
 (\ref{eq:energydp:sec8})
under the normalization constraint (\ref{eq:norm3d:sec8}), then we have \be \label{eq:virial:sec8}
2E_{\rm kin}(\phi_e)- 2E_{\rm pot}(\phi_e)+3E_{\rm
int}(\phi_e)+3E_{\rm dip}(\phi_e)=0.\ee
\end{proposition}
\subsubsection{Quasi-2D case I}
Associated to the quasi-2D equation I (\ref{eq:gpe2d:sec8})-(\ref{eq:u2d1:sec8}),
the energy is \be\label{eq:ener2d:sec8} E_{\rm 2D}(\phi)=\int_{\Bbb
R^2}\left[\frac12|\nabla\phi|^2+V_2(\bx_{\perp})|\phi|^2+
\beta_{2D}|\phi|^4-\frac{3\lambda}{4}|\phi|^2
\widetilde{\varphi^{2D}}\right]\,d\bx_{\perp}, \quad \phi\in X,\ee  where $\beta_{2D}=\frac{\beta-\lambda+3\lambda
n_3^2}{\sqrt{2\pi}\,\vep}$ and
\be\widetilde{\varphi^{2D}}=\left(\p_{\bn_\perp\bn_\perp}-n_3^2\nabla^2\right)\varphi^{2D},
\qquad \varphi^{2D}=U_\vep^{2D}*|\phi|^2.\ee

The ground state $\phi_g\in S$ of  (\ref{eq:gpe2d:sec8}) is  the minimizer
of the nonconvex minimization problem \cite{BaoBenCai}: \be \mbox{Find } \phi_g\in
S,\quad\mbox{such that }E_{\rm 2D}(\phi_g)=\min\limits_{\phi\in
S}E_{\rm 2D}(\phi). \ee

\begin{theorem}\label{thm:thm1:sec8}
Assume $0\leq V_2(\bx_{\perp})$  and
 $\lim\limits_{|\bx_{\perp}|\to\infty}V_2(\bx_{\perp})=\infty$, then
we have

(i) There exists a ground state $\phi_g\in S$ of the system
(\ref{eq:gpe2d:sec8})-(\ref{eq:u2d1:sec8}) if one of the following conditions holds
\begin{quote}
(A$1$) $\lambda\ge0$ and $\beta-\lambda> -\sqrt{2\pi} C_b\,\vep$;\\
(A$2$) $\lambda<0$ and
$\beta+\frac{1}{2}(1+3|2n_3^2-1|)\lambda>-\sqrt{2\pi} C_b\,\vep$,
\end{quote}
where $C_b$ is given in (\ref{eq:bestcons:2d}).

(ii) The positive ground state $|\phi_g|$ is unique under one of the
following conditions:

\begin{quote}
(A1$^\prime$) $\lambda\ge0$ and $\beta-\lambda\ge 0$;\\
(A2$^\prime$) $\lambda<0$ and
$\beta+\frac12(1+3|2n_3^2-1|)\lambda\ge0$.
\end{quote}
 Moreover, any ground state is of the form $\phi_g=e^{i\theta_0}|\phi_g|$
for some constant $\theta_0\in\Bbb R$.

(iii) If $\beta+\frac12\lambda(1-3n_3^2)<-\sqrt{2\pi} C_b\,\vep$,
there exists no ground state of Eq.  (\ref{eq:gpe2d:sec8}).
\end{theorem}

\subsubsection{Quasi-2D case II}
Associated to the quasi-2D equation II (\ref{eq:gpe2d2:sec8}), the energy is
\be\label{eq:ener2d2:sec8} \tilde{E}_{\rm 2D}(\phi)=\int_{\Bbb
R^2}\left[\frac12|\nabla\phi|^2+
V_2(\bx_{\perp})|\phi|^2+\beta_{2D}|\phi|^4
-\frac{3\lambda}{4}|\phi|^2 \varphi\right]\,d\bx, \; \phi\in
X,\ee
 where
\be\label{eq:vphu2d22:sec8}
\varphi(\bx_{\perp})=\left(\p_{\bn_\perp\bn_\perp}-n_3^2\nabla^2\right)((-\nabla^2)^{-1/2}|\phi|^2).
\ee The ground state $\phi_g\in S$ of the equation (\ref{eq:gpe2d2:sec8})
is defined as the minimizer of the nonconvex minimization problem:
\be \mbox{Find } \phi_g\in S,\quad\mbox{such that
}\tilde{E}_{\rm 2D}(\phi_g)=\min\limits_{\phi\in
S}\tilde{E}_{\rm 2D}(\phi). \ee

For the above ground state, we have  the following results \cite{BaoBenCai}.

\begin{theorem}\label{thm:thm1':sec8} Assume $0\leq V_2(\bx_{\perp})$
 and
$\lim\limits_{|\bx_{\perp}|\to\infty}V_2(\bx_{\perp})=\infty$, then we have

(i) There exists a ground state $\phi_g\in S$ of the equation
(\ref{eq:gpe2d2:sec8}) if one of the following conditions holds
\begin{quote}
(B1) $\lambda=0$ and $\beta> -\sqrt{2\pi}C_b\,\vep$;\\
(B2)  $\lambda>0$, $n_3=0$ and $\beta-\lambda> - \sqrt{2\pi}C_b\,\vep$;\\
(B3)  $\lambda<0$, $n_3^2\ge\frac12$ and
$\beta-(1-3n_3^2)\lambda>-\sqrt{2\pi} C_b\,\vep$.
\end{quote}

(ii)  The positive ground state $|\phi_g|$ is unique under one of
the following conditions
\begin{quote}
(B1$^\prime$) $\lambda=0$ and $\beta\ge 0$;\\
(B2$^\prime$) $\lambda>0$, $n_3=0$ and $\beta\ge\lambda$;\\
(B3$^\prime$) $\lambda<0$, $n_3^2\ge\frac12$ and
$\beta-(1-3n_3^2)\lambda\ge0$.
\end{quote}
Moreover, any ground state $\phi_g=e^{i\theta_0}|\phi_g|$ for some
constant $\theta_0\in\Bbb R$.

(iii) There exists no ground state of the equation (\ref{eq:gpe2d2:sec8}) if
one of the following conditions holds
\begin{quote}
(B1$^{\prime\prime}$) $\lambda>0$ and $n_3\neq0$;\\
(B2$^{\prime\prime}$) $\lambda<0$ and  $n_3^2<\frac12$;\\
(B3$^{\prime\prime}$) $\lambda=0$ and $\beta<- \sqrt{2\pi}C_b\,\vep$.
\end{quote}
\end{theorem}

\subsubsection{Quasi-1D case}
Associated to the quasi-1D equation (\ref{eq:gpe1d:sec8}), the energy is
\be\label{eq:ener1d:sec8} E_{\rm 1D}(\phi)=\int_{\Bbb
R}\left[\frac12|\p_{z}\phi|^2+V_1(z)|\phi|^2+\frac12\beta_{1D}|\phi|^4
+\frac{3\lambda(1-3n_3^2)}{16\sqrt{2\pi}\,\vep}|\phi|^2
\varphi\right]\,dz, \ee
 where $\beta_{1D}=\frac{\beta+\frac12\lambda(1-3n_3^2)}{2\pi\vep^2}$ and
\be\label{eq:kernel1d:sec8} \varphi(z)=\p_{zz}(U_\vep^{1D}*|\phi|^2),\quad
U_\vep^{1D}(z)=
\frac{2e^{\frac{z^2}{2\vep^2}}}{\sqrt{\pi}}\int_{|z|}^\infty
e^{-\frac{s^2}{2\vep^2}}\,ds. \ee Again, the ground state $\phi_g\in
S$ of the equation (\ref{eq:gpe1d:sec8}) is defined as the minimizer of
the nonconvex minimization problem: \be \mbox{Find } \phi_g\in
S,\quad\mbox{such that }E_{\rm 1D}(\phi_g)=\min\limits_{\phi\in S}\;
E_{\rm 1D}(\phi). \ee

For the above ground state, we have  the following results \cite{BaoBenCai}.

\begin{theorem}\label{thm:thm1'':sec8}(Existence and uniqueness of ground state) Assume
$0\leq V_1(z)$ and
$\lim_{|z|\to\infty}V_1(z)=\infty$, for any parameter
$\beta$, $\lambda$ and $\vep$, there exists a ground state
$\phi_g\in S$ of the  quasi-1D equation
(\ref{eq:gpe1d:sec8})-(\ref{eq:poisson1d:sec8}), and the positive ground state
$|\phi_g|$ is unique under one of the following conditions:

(C1)  $\lambda(1-3n_3^2)\ge0$ and $\beta-(1-3n_3^2)\lambda\ge0$;

 (C2)  $\lambda(1-3n_3^2)<0$ and  $\beta+\frac{\lambda}{2}(1-3n_3^2)\ge0$.

\noindent Moreover, $\phi_g=e^{i\theta_0}|\phi_g|$ for some constant
$\theta_0\in\Bbb R$.
\end{theorem}
\subsection{Well-posedness for dynamics}
In this section, we study the well-posedness for dynamics of dipolar BECs.
\subsubsection{3D case}
In 3D, we have the following results for GPPS (\ref{eq:gpe:sec8})-(\ref{eq:poisson:sec8}) \cite{BaoCaiWang}.
\begin{theorem} (Well-posedness) Suppose the real-valued trap
potential $V(\bx)\in C^\infty(\Bbb R^3)$ such that $V(\bx)\ge0$ for
$\bx\in{\Bbb R}^3$ and $D^\alpha V(\bx)\in L^\infty(\Bbb R^3)$ for
all $\alpha\in{\Bbb N}_0^3$ with $|\alpha|\ge 2$. For any initial
data $\psi(\bx,t=0)=\psi_0(\bx)\in X$,
 there exists
$T_{\mbox{\rm max}}\in(0,+\infty]$ such that the problem
(\ref{eq:gpe:sec8})-(\ref{eq:poisson:sec8})
 has a unique maximal solution
$\psi\in C\left([0,T_{\rm max}),X\right)$. It is maximal in the
sense that if $T_{\rm max}<\infty$, then
$\|\psi(\cdot,t)\|_{X}\to\infty$ when  $t\to T^-_{\rm max}$.
Moreover, the {\sl mass} $N(\psi(\cdot,t))$ and {\sl energy}
$E_{\rm 3D}(\psi(\cdot,t))$ defined in (\ref{eq:norm3d:sec8}) and (\ref{eq:energydp:sec8}),
respectively, are conserved for $t\in[0,T_{\rm max})$. Specifically,
if $\beta\ge0$ and $-\frac12\beta\leq\lambda\leq\beta$, the solution
to (\ref{eq:gpe:sec8})-(\ref{eq:poisson:sec8}) is global in time, i.e.,
  $T_{\rm max}=\infty$.
\end{theorem}

\begin{theorem}(Finite time blow-up) If $\beta<0$, or $\beta\ge0$ and
$\lambda<-\frac{1}{2}\beta$ or $\lambda>\beta$,
 and assume $V(\bx)$ satisfies $3V(\bx)+ \bx\cdot \nabla V(\bx)\ge0$ for
$\bx\in{\Bbb R}^3$. For any initial data
$\psi(\bx,t=0)=\psi_0(\bx)\in X$ to the problem
(\ref{eq:gpe:sec8})-(\ref{eq:poisson:sec8}), there exists finite time blow-up, i.e.,
$T_{\rm max}<\infty$, if one of the following holds:

(i) $E_{\rm 3D}(\psi_0)<0$;

(ii) $E_{\rm 3D}(\psi_0)=0$ and ${\rm Im}\left(\int_{\Bbb
R^3}\bar{\psi}_0(\bx)\ (\bx\cdot\nabla\psi_0(\bx))\,d\bx\right)<0$;

(iii) $E_{\rm 3D}(\psi_0)>0$ and ${\rm Im}\left(\int_{\Bbb R^3}
\bar{\psi}_0(\bx)\ (\bx\cdot\nabla\psi_0(\bx))\,d\bx\right)
<-\sqrt{3E_{\rm 3D}(\psi_0)}\|\bx\psi_0\|_{L^2}$.
\end{theorem}
\subsubsection{Quasi-2D case I}
For quasi-2D equation I (\ref{eq:gpe2d:sec8})-(\ref{eq:u2d1:sec8}), we have the following results \cite{BaoBenCai}.
\begin{theorem}\label{thm:thm1dy:sec8}
(Well-posedness of Cauchy problem) Suppose the real-valued trap
potential satisfies $V_2(\bx_{\perp})\ge0$ for $\bx_{\perp}\in{\Bbb R}^2$ and
\be\label{eq:cond:v2:sec8}
V_2(\bx_{\perp})\in C^\infty(\Bbb R^2) \hbox{ and
}D^{{\bf k}} V_2(\bx_{\perp})\in L^\infty(\Bbb R^2),\qquad \hbox{for all }
{\bf k}\in{\Bbb N}_0^2\  \hbox{with}\  |{\bf k}|\ge 2,\ee then we have

(i) For any initial data $\phi(\bx_{\perp},t=0)=\phi_0(\bx_{\perp})\in X$,
 there exists a
$T_{\rm max}\in(0,+\infty]$ such that the problem
(\ref{eq:gpe2d:sec8})-(\ref{eq:u2d1:sec8})
 has a unique maximal solution
$\phi\in C\left([0,T_{\rm max}),X\right)$. It is maximal in
the sense that if $T_{\rm max}<\infty$, then
$\|\phi(\cdot,t)\|_{X}\to\infty$ when  $t\to T^-_{\rm
max}$.

(ii) As long as the solution $\phi(\bx_{\perp},t)$ remains in the energy
space $X$, the {\sl $L^2$-norm} $\|\phi(\cdot,t)\|_2$ and {\sl
energy} $E_{\rm 2D}(\phi(\cdot,t))$ in (\ref{eq:ener2d:sec8}) are conserved for
$t\in[0,T_{\rm max})$.

(iii) Under either condition (A1) or (A2) in  Theorem \ref{thm:thm1:sec8}
with constant $C_b$ being replaced by $C_b/\|\phi_0\|_2^2$,
the solution of (\ref{eq:gpe2d:sec8})-(\ref{eq:u2d1:sec8}) is global in time, i.e.,
  $T_{\rm max}=\infty$.
\end{theorem}
\begin{theorem}(Finite time blow-up) For any initial data
$\phi(\bx_{\perp},t=0)=\phi_0(\bx_{\perp})\in X$ with $\int_{\Bbb
R^2}|\bx_{\perp}|^2|\phi_0(\bx_{\perp})|^2\,d\bx_{\perp}<\infty$,  if conditions (A1) and (A2)
with constant $C_b$ being replaced by $C_b/\|\phi_0\|_2^2$
  are not satisfied and assume $V_2(\bx_{\perp})$ satisfies $2V_2(\bx_{\perp})+ \bx_{\perp}\cdot
 \nabla V_2(\bx_{\perp})\ge0$,  and let $\phi:=\phi(\bx_{\perp},t)$ be
the solution of the problem (\ref{eq:gpe2d:sec8}), there exists finite time
blow-up, i.e., $T_{\rm{max}}<\infty$,  if $\lambda=0$, or
$\lambda>0$ and $n_3^2\ge\frac12$, and one of the following holds:

(i) $E_{\rm 2D}(\phi_0)<0$;

(ii) $E_{\rm 2D}(\phi_0)=0$ and ${\rm Im}\left(\int_{\Bbb
R^2}\bar{\phi}_0(\bx_{\perp})\ (\bx_{\perp}\cdot\nabla\phi_0(\bx_{\perp}))\,d\bx_{\perp}\right)<0$;

(iii) $E_{\rm 2D}(\phi_0)>0$ and ${\rm Im}\left(\int_{\Bbb R^2}
\bar{\phi}_0(\bx_{\perp})\ (\bx_{\perp}\cdot\nabla\phi_0(\bx_{\perp}))\,d\bx_{\perp}\right)
<-\sqrt{2E_{\rm 2D}(\phi_0)}\|\bx_{\perp}\phi_0\|_{2}$.
\end{theorem}
\subsubsection{Quasi-2D case II}
For quasi-2D equation II (\ref{eq:gpe2d2:sec8}),
noticing the nonlinearity
$\phi(\p_{\bn_\perp\bn_\perp}-n_3^2\nabla^2)((-\nabla^2)^{-1/2}|\phi|^2)$
is actually a derivative nonlinearity,  it would bring
significant difficulty in analyzing the dynamic behavior. The Cauchy problem of the Schr\"{o}dinger equation with derivative
nonlinearity has been investigated extensively  in the literatures
\cite{Ho,Kpv1}. We are able to prove the existence results in the energy
space with the special structure of our nonlinearity \cite{BaoBenCai}.
\begin{theorem}
(Existence  for Cauchy problem) Suppose the real potential
$V_2(\bx_{\perp})$ satisfies (\ref{eq:cond:v2:sec8}) and
$\lim_{|\bx_{\perp}|\to\infty}V_2(\bx_{\perp})=\infty$,
and initial value $\phi_0(\bx_{\perp})\in X$,
either condition  (B2) or (B3)  in Theorem \ref{thm:thm1':sec8} holds with
constant $C_b$ being replaced by $C_b/\|\phi_0\|_2^2$, then
there exists a solution $\phi\in L^\infty([0,\infty);X)\cap
W^{1,\infty}([0,\infty);X^\ast)$ for the Cauchy problem of
(\ref{eq:gpe2d2:sec8}). Here $X^\ast$ denotes the dual space of $X$.
 Moreover, there holds for $L^2$ norm and energy
$\tilde{E}_{\rm 2D}$ (\ref{eq:ener2d2:sec8}) conservation, i.e. \be
\|\phi(\cdot,t)\|_{L^2(\Bbb R^2)}=\|\phi_0\|_{L^2(\Bbb R^2)},\quad
\tilde{E}_{\rm 2D} (\phi(\cdot,t))\leq \tilde{E}_{\rm 2D}(\phi_0),
\quad\forall t\ge0. \ee
\end{theorem}
Next, we discuss possible finite time blow-up for the continuous
solutions of the quasi-2D equation II (\ref{eq:gpe2d2:sec8}). To this
purpose, the following assumptions are introduced:

(A) Assumption on the trap and coefficient of the cubic term, i.e.
$V_2(\bx_{\perp})$ satisfies $3V_2(\bx_{\perp})+ \bx_{\perp}\cdot \nabla V_2(\bx_{\perp})\ge0$,
$\frac{\beta-\lambda+3\lambda
n_3^2}{\sqrt{2\pi}\,\vep}\ge-\frac{C_b}{\|\phi_0\|_2^2}$, with
$\phi_0$ being the initial data of equation (\ref{eq:gpe2d2:sec8});

(B) Assumption on the trap and coefficient of the nonlocal term, i.e. $V_2(\bx_{\perp})$
satisfies $2V_2(\bx_{\perp})+ \bx_{\perp}\cdot \nabla V_2(\bx_{\perp})\ge0$, $\lambda=0$ or $\lambda>0$ and $n_3^2\ge\frac12$.

\begin{theorem}(Finite time blow-up)
For any initial data $\phi(\bx_{\perp},t=0)=\phi_0(\bx_{\perp})\in X$ with finite variance $\delta_V^0=\int_{\Bbb
R^2}|\bx_{\perp}|^2|\phi_0(\bx_{\perp})|^2\,d\bx_{\perp}<\infty$,
if conditions (B1), (B2) and (B3) with constant $C_b$ being replaced by
$C_b/\|\phi_0\|_2^2$ are not satisfied,
let $\phi:=\phi(\bx_{\perp},t)\in C([0,T_{\rm max}),X)$ solution of the problem (\ref{eq:gpe2d2:sec8}) with
$L^2$ norm and energy conservation, then there exists finite time
blow-up, i.e., $T_{\rm max}<\infty$, if one of the following
condition holds:

(i) $\tilde{E}_{\rm 2D}(\phi_0)<0$, and either Assumption (A) or (B) holds;

(ii) $\tilde{E}_{\rm 2D}(\phi_0)=0$ and ${\rm Im}\left(\int_{\Bbb
R^2}\bar{\phi}_0(\bx_{\perp})\ (\bx_{\perp}\cdot\nabla\phi_0(\bx_{\perp}))\,d\bx_{\perp}\right)<0$,
and either Assumption (A) or (B) holds;

(iii) $\tilde{E}_{\rm 2D}(\phi_0)>0$, and ${\rm Im}\left(\int_{\Bbb R^2}
\bar{\phi}_0(\bx_{\perp})\ (\bx_{\perp}\cdot\nabla\phi_0(\bx_{\perp}))\,d\bx_{\perp}\right)
<-(3\tilde{E}_{\rm 2D}^0)^{1/2}\delta_V^0$ if Assumption (A)
holds, or ${\rm Im}\left(\int_{\Bbb R^2} \bar{\phi}_0(\bx_{\perp})\
(\bx_{\perp}\cdot\nabla\phi_0(\bx_{\perp}))\,d\bx_{\perp}\right)
<-(2\tilde{E}_{\rm 2D}^0)^{1/2}\delta_V^0$ if Assumption (B)
holds. Here $\tilde{E}_{\rm 2D}^0=\tilde{E}_{\rm 2D}(\phi_0)$.
\end{theorem}
\subsubsection{Quasi-1D case}
Concerning the Cauchy problem,  similar
results as Theorem \ref{thm:thm1dy:sec8} can be obtained for the equation
(\ref{eq:gpe1d:sec8}) \cite{BaoBenCai}.

\begin{theorem}
(Well-posedness for Cauchy problem) Suppose the real-valued trap
potential satisfies $V_1(z)\ge0$ for $z\in{\Bbb R}$ and $V_1(z)\in
C^\infty(\Bbb R)$, $D^k V_1(z)\in L^\infty(\Bbb R)$ for all
integers $k\ge 2$, for any initial data
$\phi(z,t=0)=\phi_0(z)\in X$,
 there exists
 a unique solution
$\phi\in C\left([0,\infty),X\right)\cap C^1\left([0,\infty),X^\ast\right)$
to the Cauchy problem of equation (\ref{eq:gpe1d:sec8}).
\end{theorem}

\subsection{Convergence rate of dimension reduction}
In this section, we  discuss the dimension reduction of 3D GPPS to
lower dimensions. In lower dimensions,
 we require that in the quasi-2D case, $\beta=O(\vep)$, $\lambda=O(\vep)$, and in
 the quasi-1D case, $\beta=O(\vep^2)$, $\lambda=O(\vep^2)$, i.e. we are considering
 the weak interaction regime, then we would get an $\vep$-independent limiting equation.
 In this regime, we will see that  GPPS will reduce to a regular GPE in lower dimensions \cite{BaoBenCai}.

\subsubsection{Reduction to 2D}
 We consider the weak interaction regime, i.e.,   $\beta\to \vep\beta$, $\lambda\to\vep\lambda$.
In {\sl{Case I}} (\ref{eq:case1:sec8}),  for full 3D GPPS
(\ref{eq:gpe:sec8})-(\ref{eq:poisson:sec8}), introduce the re-scaling $z\to \vep z$,
$\psi\to \vep^{-1/2}\psi^\vep$ which preserves the normalization,
then
 \bea\label{eq:rescalgpe2:sec8}
 i\p_t\psi^\vep(\bx_{\perp},z,t)=\left[\bH_{\bx_{\perp}}^V+\frac{1}{\vep^2}\bH_z+(\beta-\lambda)|\psi^\vep|^2-3
 \vep\lambda\p_{\bn_\vep\bn_\vep}\varphi^\vep\right] \psi^\vep,
 \eea
 where $\bx_{\perp}=(x,y) \in \Bbb R^2$ and
 \bea
 &&\bH_{\bx_{\perp}}^V=-\frac{1}{2}(\p_{xx}+\p_{yy})+V_2(x,y),\qquad \bH_z=-\frac12\p_{zz}+\frac{z^2}{2},\\
 &&\bn_\vep=(n_1,n_2,n_3/\vep),\qquad \p_{\bn_\vep}=\bn_\vep\cdot\nabla,\qquad
 \p_{\bn_\vep\bn_\vep}=\p_{\bn_\vep}(\p_{\bn_\vep}),\\
 &&(-\p_{xx}-\p_{yy}-\frac{1}{\vep^2}\p_{zz})\varphi^\vep=\frac{1}{\vep}|\psi^\vep|^2,\qquad
 \lim\limits_{|\bx|\to \infty}\varphi^\vep(\bx_{\perp},z,t)=0.\label{eq:rescalps:sec8}
 \eea
 It is well-known that $\bH_z$ has eigenvalues $\mu_k=k+1/2$ with corresponding
 eigenfunction $w_k(z)$ ($k=0,1,\ldots$), where $\{w_k\}_{k=0}^\infty$ forms an
 orthornormal   basis  of $L^2(\Bbb R)$ \cite{GS}, specifically, $w_0(z)=
 \frac{1}{\pi^{1/4}}e^{-z^2/2}$.  It is convenient to
 consider the initial data concentrated on the ground mode of $\bH_z$, i.e.,
 \be\label{eq:init:sec8}
 \psi^\vep(\bx_{\perp},z,0)=\phi_0(\bx_{\perp})w_0(z),\quad \phi_0\in X(\Bbb R^2) \mbox{ and } \|\phi_0\|_{L^2(\Bbb R^2)}=1.
 \ee

In {\sl{Case I}} (\ref{eq:case1:sec8}),
 when $\vep\to0^+$, quasi-2D equation I (\ref{eq:gpe2d:sec8}), II (\ref{eq:gpe2d2:sec8})
 will yield  an $\vep$-independent equation in the weak interaction regime,
 \be\label{eq:2d:weak:sec8}
 i\p_t \phi(\bx_{\perp},,t)=\bH_{\bx_{\perp}}^V\phi+\frac{\beta-(1-3n_3^2)\lambda}{\sqrt{2\pi}}
 |\phi|^2\phi,\quad \bx_{\perp}=(x,y)\in\Bbb R^2,
 \ee
 with initial condition $\phi(\bx_{\perp},0)=\phi_0(\bx_{\perp})$.
 \begin{theorem}\label{thm:redthm1:sec8}(Dimension reduction to 2D) Suppose $V_2$ satisfies
condition (\ref{eq:cond:v2:sec8}), $-\frac{\beta}{2}\leq \lambda \leq\beta$
and $\beta\ge0$, let $\psi^\vep\in C([0,\infty);X(\Bbb R^3))$ and
 $\phi\in C([0,\infty);X(\Bbb R^2))$ be the unique solutions of equations (\ref{eq:rescalgpe2:sec8})-(\ref{eq:init:sec8})
 and (\ref{eq:2d:weak:sec8}), respectively, then for any $T>0$, there exists $C_{T}>0$ such that
\be \left\|\psi^\vep(\bx_{\perp},z,t)-e^{-i\frac{\mu_0
t}{\vep^2}}\phi(\bx_{\perp},t)w_0(z)\right\|_{L^2(\Bbb R^3)}\leq
C_{T}\,\vep, \qquad \forall t\in [0,T]. \ee
\end{theorem}
\subsubsection{Reduction to 1D}
In this case, we again consider the weak interaction regime $\beta\to \vep^{2}\beta$,
 $\lambda\to\vep^{2}\lambda$.
In {\sl{Case II}} (\ref{eq:case2:sec8}),  for the full 3D GPPS
(\ref{eq:gpe:sec8})-(\ref{eq:poisson:sec8}), introducing the re-scaling $x\to
\vep x$, $y\to \vep y$, $\psi\to \vep^{-1}\psi^\vep$ which preserves the
normalization,  then
 \be\label{eq:rescalgpe1:sec8}
 i\p_t\psi^\vep(\bx_{\perp},z,t)=\left[\bH_z^V+\frac{1}{\vep^2}\bH_{\bx_{\perp}}+(\beta-\lambda)|\psi^\vep|^2-3
 \vep\lambda\p_{\tilde{\bn}_{\vep}\tilde{\bn}_\vep}\varphi^\vep \right]\psi^\vep,
 \ee
 where $\bx_{\perp}=(x,y)\in\Bbb R^2$, $z\in\Bbb R$ and
 \bea
 &&\bH_z^V=-\frac{1}{2}\p_{zz}+V_1(z),\quad \bH_{\bx_{\perp}}=-\frac12(\p_{xx}+\p_{yy}+x^2+y^2),\\
 &&\tilde{\bn}_\vep=(n_1/\vep,n_2/\vep,n_3),\quad \p_{\tilde{\bn}_\vep}=\tilde{\bn}_\vep
 \cdot\nabla,\quad \p_{\tilde{\bn}_\vep\tilde{\bn}_\vep}=\p_{\tilde{\bn}_\vep}(\p_{\tilde{\bn}_\vep}),\\
 &&(-\frac{1}{\vep^2}\p_{xx}-\frac{1}{\vep^2}\p_{yy}-\p_{zz})\varphi^\vep=\frac{1}{\vep^2}|\psi^\vep|^2,\qquad
 \lim\limits_{|\bx|\to \infty}\varphi^\vep(\bx_{\perp},z,t)=0.\label{eq:rescalps1:sec8}
 \eea
 Note that the ground state mode of $\bH_{\bx_{\perp}}$ is given by $w_0(x)w_0(y)$ with eigenvalue 1,
 and the initial data is then assumed to be
\be\label{eq:init1:sec8}
 \psi^\vep(\bx_{\perp},z,0)=\phi_0(z)w_0(x)w_0(y),\quad \phi_0\in X(\Bbb R) \mbox{ and } \|\phi_0\|_{L^2(\Bbb R)}=1.
 \ee

In {\sl{Case II}} (\ref{eq:case2:sec8}),
 when $\vep\to0^+$, the quasi-1D equation  (\ref{eq:gpe1d:sec8}) will lead to an $\vep$-independent equation in the weak interaction regime,
 \be\label{eq:1d:weak:sec8}
 i\p_t \phi(z,t)=\bH_z^V\phi+\frac{\beta+\frac12\lambda (1-3n_3^2)}{2\pi}|\phi|^2\phi,\qquad z\in\Bbb
 R,\quad t>0,
 \ee
 with the initial condition $\phi(z,0)=\phi_0(z)$.

We can prove the following results \cite{BaoBenCai}.
 \begin{theorem}(Dimension reduction to 1D)
 Suppose the real-valued trap
potential satisfies  $V_1(z)\ge0$ for $z\in{\Bbb R}$ and $V_1(z)\in
C^\infty(\Bbb R)$, $D^k V_1(z)\in L^\infty(\Bbb R)$ for all
$k\ge 2$. Assume $-\frac{\beta}{2}\leq \lambda \leq\beta$ and
$\beta\ge0$, and let $\psi^\vep\in C([0,\infty);X(\Bbb R^3))$ and $\phi\in
C([0,\infty);X(\Bbb R))$ be the unique solutions of the equations
(\ref{eq:rescalgpe1:sec8})-(\ref{eq:init1:sec8}) and (\ref{eq:1d:weak:sec8}),
respectively, then for any $T>0$, there exists $C_{T}>0$ such that
\be
\left\|\psi^\vep(\bx_{\perp},z,t)-e^{-it/\vep^2}\phi(z,t)w_0(x)w_0(y)\right\|_{L^2(\Bbb
R^3)}\leq C_{T}\,\vep, \qquad \forall t\in [0,T]. \ee
\end{theorem}

\subsection{Numerical methods for computing  ground states}
In this section, we present efficient and accurate numerical methods for computing ground states  of dipolar BECs, based on the new formulation GPPS (\ref{eq:gpe:sec8})-(\ref{eq:poisson:sec8}) of dipolar GPE (\ref{eq:ngpe1:sec8}) in 3D.

The difficulty of computing dipolar GPE mainly comes from the dipolar term.
In most of the numerical
methods used in the literatures for theoretically and/or numerically
studying the ground states  of dipolar BECs, the way to
deal with the convolution in (\ref{eq:ngpe1:sec8}) is usually to use the
Fourier transform \cite{Ronen,Xiong}.
However, due to the high singularity in the dipolar interaction
potential (\ref{eq:kel:sec8}), there are two drawbacks in these numerical
methods: (i) the Fourier transforms of the dipolar interaction
potential (\ref{eq:kel:sec8}) and the density function $|\psi|^2$ are
usually carried out in the continuous level on the whole space
${\Bbb R}^3$  and in the discrete
level on a bounded computational domain $U$, respectively, and due
to this mismatch, there is a locking phenomena in practical
computation as observed in \cite{Ronen}; (ii) the second term in the
Fourier transform of the dipolar interaction potential is
$\frac{0}{0}$-type for $0$-mode, i.e when $\xi=0$ (see
(\ref{eq:four11:sec8}) for details), and it is artificially omitted when
$\xi=0$ in practical computation
\cite{Ronen,Yi1,Xiong} thus this may cause
some numerical problems too. With  formulation (\ref{eq:gpe:sec8})-(\ref{eq:poisson:sec8}),
 new numerical methods for computing ground states and
dynamics of dipolar BECs can be constructed, which can avoid the above two drawbacks and
thus they are more accurate than those currently used in the
literatures.

Based on the new mathematical formulation for the energy associated with GPPS (\ref{eq:gpe:sec8})-(\ref{eq:poisson:sec8}) in
(\ref{eq:energydp:sec8}), we will present an efficient and accurate backward
Euler sine pseudospectral (BESP) method for computing the ground states of
a dipolar BEC.

In practice, the whole space problem is usually truncated into a
bounded computational domain $U=[a,b]\tm[c,d]\tm[e,f]$ with
homogeneous  Dirichlet boundary condition.  We adopt the method
of gradient flow with discrete normalization (GFDN) as in section \ref{sec:numgs}: choose a time step
$\tau>0$ and set $t_n=n\; \tau$ for $n=0,1,\ldots$  Applying
the steepest decent method to the energy functional $E_{\rm 3D}(\phi)$ in
(\ref{eq:energydp:sec8}) without the normalization constraint (\ref{eq:norm3d:sec8}), and then
projecting the solution back to the unit sphere $S$ at the end of
each time interval $[t_n,t_{n+1}]$ in order to satisfy the
constraint (\ref{eq:norm3d:sec8}). Then GFDN for computing ground state of the GPPS (\ref{eq:gpe:sec8})-(\ref{eq:poisson:sec8}) is \cite{BaoCaiWang}: \bea \label{eq:ngf1:sec8} &&\p_t
\phi(\bx,t)=\left[\fl{1}{2 }\btd^2
-V(\bx)-(\beta-\lambda)|\phi(\bx,t)|^2+
3\lambda \p_{\bf{nn}}\varphi(\bx,t)\right]\phi(\bx,t),    \qquad\\
\label{eq:ngf21:sec8}&&\nabla^2 \varphi(\bx,t) = -|\phi(\bx,t)|^2, \qquad
\qquad \bx\in U,\quad  t_n \leq t < t_{n+1}, \\
\label{eq:ngf2:sec8} &&\phi(\bx,t_{n+1}):=
\phi(\bx,t_{n+1}^+)=\fl{\phi(\bx,t_{n+1}^-)}{\|\phi(\cdot,t_{n+1}^-)\|_2},
\qquad \bx\in U, \quad n\ge 0,\\
&&\left.\phi(\bx,t)\right|_{\bx\in\p U}=\left.\varphi(\bx,t)\right|_{\bx\in\p U}=0,
\qquad t\ge0,\\
 \label{eq:ngf3:sec8}
 &&\phi(\bx,0)=\phi_0(\bx), \qquad
\qquad \hbox{with}\quad \|\phi_0\|_2=1; \eea where
 $\phi(\bx,
t_n^\pm)=\lim\limits_{t\to t_n^\pm} \phi(\bx,t)$.

Let $M$, $K$ and $L$ be even positive integers and define the index
sets \beas &&{\calT}_{MKL}=\{(j,k,l)\ |\ j=1,2,\ldots,M-1,\
k=1,2,\ldots, K-1, \ l=1,2,\ldots,L-1\}, \\
&&{\calT}_{MKL}^0=\{(j,k,l)\ |\ j=0,1,\ldots,M,\ k=0,1,\ldots, K, \
l=0,1,\ldots,L\}. \eeas Choose the spatial mesh sizes as
$h_x=\frac{b-a}{M}$, $h_y=\frac{d-c}{K}$ and $h_z=\frac{f-e}{L}$ and
define
\[x_j:=a+j\;h_x,\qquad  y_k = c+ k\; h_y,\qquad
z_l = e+ l\; h_z, \qquad (j,k,l)\in {\calT}^0_{MKL}.\] Denote the
space
\[Y_{MKL}={\rm
span}\{\Phi_{jkl}(\bx), \quad (j,k,l)\in{\calT}_{MKL}\},\] with
\be
\Phi_{jkl}(\bx)=\sin\left(\mu_j^x(x-a)\right)\sin\left(\mu_k^y(y-c)\right)
\sin\left(\mu_l^z(z-e)\right), \quad \bx\in U, \ee
\[\mu_j^x =
\fl{\pi j}{b-a}, \qquad \mu_k^y = \fl{\pi k}{d-c}, \qquad \mu_l^z =
\fl{\pi l}{f-e}, \qquad (j,k,l)\in{\calT}_{MKL}; \] and  $P_{MKL}:
Y=\{\varphi\in C(U)\ |\ \varphi(\bx)|_{\bx\in\p U}=0\}\to
Y_{MKL}$ be the standard project operator, i.e.
\[(P_{MKL}v)(\bx)=\sum_{p=1}^{M-1}\sum_{q=1}^{K-1}\sum_{s=1}^{L-1}
\widehat{v}_{pqs}\; \Phi_{pqs}(\bx), \quad \bx\in U,\qquad \forall
v\in Y,
\]
with \be\label{eq:FST:sec8} \widehat{v}_{pqs}=\int_{U} v(\bx)\;
\Phi_{pqs}(\bx)\;d\bx, \qquad (p,q,s)\in{\calT}_{MKL}. \ee Then a
backward Euler sine spectral discretization
for (\ref{eq:ngf1:sec8})-(\ref{eq:ngf3:sec8}) reads:\\
 Find
$\phi^{n+1}(\bx)\in Y_{MKL}$ (i.e. $\phi^{+}(\bx)\in Y_{MKL}$) and
$\varphi^{n}(\bx)\in Y_{MKL}$ such that
\begin{align}
\frac{\phi^{+}(\bx)-\phi^n(\bx)}{\tau}=&-P_{MKL}\left\{\left[V(\bx)+(\beta-\lambda)|\phi^n(\bx)|^2-3\lambda
\p_{{\bf{nn}}}\varphi^n(\bx)\right]\phi^{+}(\bx)\right\}\nn\\
&+\frac{1}{2}\nabla^2
\phi^{+}(\bx), \\
\nabla^2\varphi^n(\bx)=&-P_{MKL}\left(|\phi^n(\bx)|^2\right),\quad
\phi^{n+1}(\bx)=\frac{\phi^{+}(\bx)}{\|\phi^{+}(\bx)\|_2}, \quad
\bx\in U, \end{align} where $n\ge0$ and
$\phi^0(\bx)=P_{MKL}\left(\phi_0(\bx)\right)$ is given.

The above discretization can be solved in phase space and it is not
suitable in practice due to the difficulty of computing the
integrals in (\ref{eq:FST:sec8}). We now present an efficient implementation
by choosing $\phi^0(\bx)$ as the interpolation of $\phi_0(\bx)$ on
the grid points $\{(x_j,y_k,z_l), \ (j,k,l)\in{\calT}_{MKL}^0\}$,
i.e. $\phi^0(x_j,y_k,z_l) =\phi_0(x_j,y_k,z_l)$ for $(j,k,l)\in{\calT}_{MKL}^0$, and approximating the integrals in (\ref{eq:FST:sec8}) by a
quadrature rule on the grid points. Let $\phi_{jkl}^n$ and
$\varphi_{jkl}^n$ be the approximations of $\phi(x_j,y_k,z_l,t_n)$
and $\varphi(x_j,y_k,z_l,t_n)$, respectively, which are the solutions
of (\ref{eq:ngf1:sec8})-(\ref{eq:ngf3:sec8}); denote $V_{jkl}=V(x_j,y_k.z_l)$, $\rho_{jkl}^n=|\phi^n_{jkl}|^2$
and  choose $\phi_{jkl}^0=\phi_0(x_j,y_k,z_l)$ for $(j,k,l)\in {\calT}_{MKL}^0$. For $n=0,1,\ldots$, a backward Euler sine
pseudospectral (BESP) discretization for (\ref{eq:ngf1:sec8})-(\ref{eq:ngf3:sec8}) reads \cite{BaoCaiWang}:
 \begin{align} &\fl{\phi_{jkl}^+-\phi_{jkl}^n}{\tau}=
-\left[V_{jkl}+(\beta-\lambda)
\left|\phi_{jkl}^n\right|^2 -3\lambda\left.\left(\p_{\bn\bn}^s
\varphi^n\right)\right|_{jkl}\right] \phi^+_{jkl}\nn\\
&\qquad\qquad\qquad+\fl{1}{2} \left.\left(\nabla_s^2
\phi^+\right)\right|_{jkl},\quad(j,k,l)\in
{\calT}_{MKL},
\label{eq:discretized1:sec8} \\
&-\left.\left(\nabla_s^2 \varphi^n\right)\right|_{jkl}=
|\phi_{j,k,l}^n|^2=\rho_{jkl}^n, \quad
\phi_{jkl}^{n+1}=\fl{\phi_{jkl}^+}{\|\phi^+\|_{2}}, \quad (j,k,l)\in
{\calT}_{MKL}, \\
&\phi_{0kl}^{n+1}=\phi_{Mkl}^{n+1}=\phi_{j0l}^{n+1}=
\phi_{jKl}^{n+1}=\phi_{jk0}^{n+1}=\phi_{jkL}^{n+1}=0,\quad
(j,k,l)\in {\calT}_{MKL}^0,\\
&\varphi_{0kl}^{n}
=\varphi_{Mkl}^{n}=\varphi_{j0l}^{n}=\varphi_{jKl}^{n}=
\varphi_{jk0}^{n}=\varphi_{jkL}^{n}=0, \quad (j,k,l)\in
{\calT}_{MKL}^0;\label{eq:discretized2:sec8} \end{align}  where $\nabla_s^2$ and
$\p_{\bn\bn}^s$ are sine pseudospectral
 approximations of $\nabla^2$ and $\p_{\bn\bn}$, respectively,
 defined for $(j,k,l)\in\calT_{MKL}$ as
 {\small\begin{align}&\left.\left(\nabla_s^2 \phi^n\right)\right|_{jkl} =
-\sum_{p=1}^{M-1}\sum_{q=1}^{K-1}\sum_{s=1}^{L-1}
\lambda_{pqs}\widetilde{(\phi^n)}_{pqs}
\sin\left(\frac{jp\pi}{M}\right)\sin\left(\frac{kq\pi}{K}\right)
\sin\left(\frac{ls\pi}{L}\right), \nn \\
&\left.\left(\p_{\bn\bn}^s
\varphi^n\right)\right|_{jkl}=\sum_{p=1}^{M-1}\sum_{q=1}^{K-1}\sum_{s=1}^{L-1}
\frac{\widetilde{(\rho^n)}_{pqs}}{(\mu_p^x)^2+(\mu_q^y)^2+(\mu_s^z)^2}
\left.\left(\p_{\bn\bn}\Phi_{pqs}(\bx)\right)\right|_{(x_j,y_k,z_l)}, \label{eq:dstpp:sec8}
 \end{align}}
with $\lambda_{pqs}=(\mu_p^x)^2+(\mu_q^y)^2+(\mu_s^z)^2$, $\widetilde{(\phi^n)}_{pqs}$ ($(p,q,s)\in{\calT}_{MKL})$ the
discrete sine transform coefficients of the vector $\phi^n$  for $(p,q,s)\in {\calT}_{MKl}$ as
{\small \be\label{eq:dst11:sec8}\widetilde{(\phi^n)}_{pqs}=
\frac{8}{MKL}\sum_{j=1}^{M-1}\sum_{k=1}^{K-1}\sum_{l=1}^{L-1}
\phi^n_{jkl}
\sin\left(\frac{jp\pi}{M}\right)\sin\left(\frac{kq\pi}{K}\right)
\sin\left(\frac{ls\pi}{L}\right), \ee } and the discrete norm is defined as
\[ \|\phi^+\|_{2}^2 = h_xh_yh_z\sum_{j=1}^{M-1}\sum_{k=1}^{N-1}\sum_{l=1}^{L-1}
|\phi_{jkl}^+|^2.\] Similar as those in section \ref{subsec:besp} (cf. \cite{BaoChernLim}), the linear
system (\ref{eq:discretized1:sec8})-(\ref{eq:discretized2:sec8}) can be iteratively
solved in phase space very efficiently via discrete sine transform
and we omit the details here for brevity.
\subsection{Time splitting scheme for dynamics}
Similarly, based on the new Gross-Pitaevskii-Poisson type system
(\ref{eq:gpe:sec8})-(\ref{eq:poisson:sec8}), we will present an efficient and
accurate time-splitting sine pseudospectral (TSSP) method for
computing the dynamics of a dipolar BEC.

Again, in practice, the whole space problem is truncated into a
bounded computational domain $U=[a,b]\tm[c,d]\tm[e,f]$ with
homogeneous  Dirichlet boundary condition. From time $t=t_n$ to time
$t=t_{n+1}$, the Gross-Pitaevskii-Poisson type system
(\ref{eq:gpe:sec8})-(\ref{eq:poisson:sec8}) is solved in two steps.  One solves
first \be\label{eq:fgpe:sec8}i \p_t
\psi(\bx,t)=-\fl{1}{2}\nabla^2\psi(\bx,t), \quad \bx\in U, \quad
\left.\psi(\bx,t)\right|_{\bx\in\p U}=0, \quad t_n\le t\le
t_{n+1},\ee for the time step of length $\tau$, followed by
solving \bea\label{eq:ode11:sec8} &&i \p_t
\psi(\bx,t)=\left[V(\bx)+(\beta-\lambda) |\psi(\bx,t)|^2-3\lambda
\p_{\bn\bn} \varphi(\bx,t)
\right]\psi(\bx,t),  \\
\label{eq:poisson11:sec8}&&\nabla^2 \varphi(\bx,t) = -|\psi(\bx,t)|^2,\qquad
 \bx\in U, \qquad t_n\le t \le t_{n+1}; \\ \label{eq:bond123:sec8}
 &&\left.\varphi(\bx,t)\right|_{\bx\in\p U}=0, \qquad \left.\psi(\bx,t)\right|_{\bx\in\p U}=0,
 \qquad t_n\le t \le t_{n+1};\eea
for the same time step. Eq. (\ref{eq:fgpe:sec8}) will be discretized in
space by sine pseudospectral method and integrated in time {\sl
exactly}. For $t\in[t_n,t_{n+1}]$, the equations
(\ref{eq:ode11:sec8})-(\ref{eq:bond123:sec8}) leave $|\psi|$ and $\varphi$ invariant
in $t$  and therefore they collapses
to  {\small\begin{align}\label{eq:ode111:sec8} &i \p_t
\psi(\bx,t)=\left[V(\bx)+(\beta-\lambda) |\psi(\bx,t_n)|^2-3\lambda
\p_{\bn\bn} \varphi(\bx,t_n)
\right]\psi(\bx,t),  \ t_n\le t \le t_{n+1},  \\
\label{eq:poisson113:sec8}&\nabla^2 \varphi(\bx,t_n) =
-|\psi(\bx,t_n)|^2,\qquad
 \bx\in U.
\end{align}}Again, equation (\ref{eq:poisson113:sec8}) will be discretized in
space by sine pseudospectral method \cite{BaoZhang,ShenTang,ShenTangWang,BaoCaiWang} and the linear
ODE (\ref{eq:ode111:sec8}) can be integrated in time {\sl exactly}.

Let $\psi_{jkl}^n$ and $\varphi_{jkl}^n$ be the approximations of
$\psi(x_j,y_k,z_l,t_n)$ and $\varphi(x_j,y_k,z_l,t_n)$,
respectively, which are the solutions of (\ref{eq:gpe:sec8})-(\ref{eq:poisson:sec8});
and choose $\psi^0_{jkl}=\psi_0(x_j,y_k,z_l)$ for $(j,k,l)\in {\calT}_{MKL}^0$. For $n=0,1,\ldots$,  a second-order TSSP method for
solving (\ref{eq:gpe:sec8})-(\ref{eq:poisson:sec8}) via the standard Strang
splitting is \cite{BaoCaiWang}
{\small \begin{align}\label{eq:tssp1:sec8}
 &\psi^{(1)}_{jkl}=\sum_{p=1}^{M-1}\sum_{q=1}^{K-1}\sum_{s=1}^{L-1}
e^{-i\tau\frac{(\mu_p^x)^2+(\mu_q^y)^2+(\mu_r^z)^2}{4}
}\;\widetilde{(\psi^n)}_{pqs}\sin\left(\frac{jp\pi}{M}
\right)\sin\left(\frac{kq\pi}{K}\right)
\sin\left(\frac{ls\pi}{L}\right), \nn \\
 &\psi^{(2)}_{jkl}=
 e^{-i\tau\left[V(x_j,y_k,z_l)+(\bt-\lambda)|\psi_{jkl}^{(1)}|^2-3\lambda
\left.\left(\p_{\bn\bn}^s\varphi^{(1)}\right)\right|_{jkl}\right]
}\; \psi_{jkl}^{(1)},
\quad (j,k,l)\in{\calT}_{MKL}^0, \\
&\psi^{n+1}_{jkl}=\sum_{p=1}^{M-1}\sum_{q=1}^{K-1}\sum_{s=1}^{L-1}
e^{-i\tau\frac{(\mu_p^x)^2+(\mu_q^y)^2+(\mu_r^z)^2}{4}
}\;\widetilde{(\psi^{(2)})}_{pqs}\sin\left(\frac{jp\pi}{M}
\right)\sin\left(\frac{kq\pi}{K}\right)
\sin\left(\frac{ls\pi}{L}\right); \nn
 \end{align}}
where $\widetilde{(\psi^n)}_{pqs}$ and
$\widetilde{(\psi^{(2)})}_{pqs}$ ($(p,q,s)\in{\calT}_{MKL}$) are
the discrete sine transform coefficients of the vectors $\psi^n$ and
$\psi^{(2)}$, respectively (defined similarly as those in
(\ref{eq:dst11:sec8})); and
$\left.\left(\p_{\bn\bn}^s\varphi^{(1)}\right)\right|_{jkl}$
 can be computed as in (\ref{eq:dstpp:sec8}) with
$\rho^n_{jkl}=\rho^{(1)}_{jkl}:=|\psi^{(1)}_{jkl}|^2$  for
$(j,k,l)\in {\calT}_{MKL}^0$.

The above method is explicit and unconditionally stable. The memory
cost is $O(MKL)$ and the computational cost per time step is
$O\left(MKL\ln(MKL)\right)$.

\subsection{Numerical results}
In this section, we first compare our new
methods and the standard method used in the literatures
\cite{Yi1,Xiong} to evaluate numerically the dipolar
energy and then report ground states and dynamics of dipolar BECs by
using our new numerical methods.

\begin{table}[htb]
\begin{center}
\begin{tabular}{c|cc|cc|cc}
  \hline
  &\multicolumn{2}{c|}{Case I}&\multicolumn{2}{c|}{Case II}&\multicolumn{2}{c}{Case III}\\ \cline{2-7}
  &DST &DFT &DST &DFT &DST &DFT\\\hline
 $h=1$ &2.756E-2 &2.756E-2 &3.555E-18 &1.279E-4 &0.1018
 &0.1020\\
 $h=0.5$ &1.629E-3 &1.614E-3  &9.154E-18 &1.278E-4
 &9.788E-5
 &2.269E-4\\
 $h=0.25$ &1.243E-7  &1.588E-5 &7.454E-17 &1.278E-4
&6.406E-7 &1.284E-4  \\
  \hline
\end{tabular}
\end{center}
\caption{Comparison for evaluating dipolar energy under different
mesh sizes $h$.} \label{tab:1:sec8}
\end{table}

\begin{example}\label{exm:1:sec8} Comparison of different methods.
Let \be \phi:=\phi(\bx) = \pi^{-3/4} \gamma_{x}^{1/2} \gamma_z^{1/4}
e^{-\fl{1}{2}\left( \gamma_x (x^2+y^2)+\gamma_z z^2\right)},\qquad
\bx\in{\Bbb R}^3.\ee Then the
 dipolar energy $E_{\rm dip}(\phi)$ in (\ref{eq:dipp03:sec8})
 can be  evaluated analytically as \cite{Tikhonenkov}
\be E_{\rm dip}(\phi)= -\fl{\lambda\gamma_x \sqrt{\gamma_z}}{4\pi
\sqrt{2\pi} }\left\{
\begin{array}{ll}
  \fl{1+2\chi^2}{1-\chi^2}-\fl{3\chi^2 \rm{arctan} \left(
\sqrt{\chi^2-1}\right)}{(1-\chi^2)\sqrt{\chi^2-1}}, & \chi>1, \\
  0, & \chi =1, \\
  \fl{1+2\chi^2}{1-\chi^2}-\fl{1.5\chi^2
}{(1-\chi^2)\sqrt{1-\chi^2}} \rm{ln} \left(
\fl{1+\sqrt{1-\chi^2}}{1-\sqrt{1-\chi^2}}\right), & \chi <1, \\
\end{array}
\right. \ee
 with $\chi =\sqrt{ \fl{\gamma_z}{\gamma_x}}$. This provides a
 perfect example to test the efficiency of different numerical
 methods to deal with the dipolar potential. Based on our new
 formulation, the dipolar energy can be evaluated via
 discrete sine transform (DST) as
\bea \label{eq:dstdip1:sec8}E_{\rm dip}(\phi)\approx \frac{\lambda
h_xh_yh_z}{2}
\sum_{j=1}^{M-1}\sum_{k=1}^{K-1}\sum_{l=1}^{L-1}|\phi(x_j,y_k,z_l)|^2\left[
-|\phi(x_j,y_k,z_l)|^2-3\left.\left(\p_{\bn\bn}^s\varphi^n
\right)\right|_{jkl}\right], \nn\eea where
$\left.\left(\p_{\bn\bn}^s\varphi^n \right)\right|_{jkl}$ is
computed as in (\ref{eq:dstpp:sec8}) with
$\rho^n_{jkl}=|\phi(x_j,y_k,z_l)|^2$ for $(j,k,l)\in {\calT}_{MKL}^0$. In the literatures
\cite{Yi1,Xiong}, this
dipolar energy is usually calculated via discrete Fourier transform
(DFT) as {\small \begin{equation*}E_{\rm dip}(\phi)\approx \frac{\lambda
h_xh_yh_z}{2}
\sum_{j=0}^{M-1}\sum_{k=0}^{K-1}\sum_{l=0}^{L-1}|\phi(x_j,y_k,z_l)|^2\left[
{\calF}^{-1}_{jkl}\left(\widehat{(U_{\rm
dip})}(2\mu_p^x,2\mu_q^y,2\mu_s^z)\cdot {\calF}_{pqs}(|\phi|^2)\right) \right], \end{equation*}}
where ${\calF}$ and
${\calF}^{-1}$ are the discrete Fourier and inverse Fourier
transforms over the grid points $\{(x_j,y_k,z_l), \ (j,k,l)\in {\calT}_{MKL}^0\}$, respectively \cite{Xiong}. We take $\lambda=8\pi/3$,
the bounded computational domain  $U=[-16,16]^3$, $M=K=L$ and thus
$h=h_x=h_y=h_z=\frac{32}{M}$. Tab.~\ref{tab:1:sec8} lists the errors
$e:=\left|E_{\rm dip}(\phi)-E_{\rm dip}^h\right|$ with $E_{\rm
dip}^h$ computed numerically via either DST (\ref{eq:dstdip1:sec8}) or
DFT with mesh size $h$ for three cases:
\begin{itemize}
\item Case I. $\gamma_x=0.25$ and $\gamma_z=1$, $\chi
=2.0$ and $E_{\rm dip}(\phi) = 0.0386708614$;
\item Case II. $\gamma_x=\gamma_z=1$,
$\chi =1.0$ and  $E_{\rm dip}(\phi)= 0$;
\item Case III. $\gamma_x=2$ and $\gamma_z=1$,
$\chi =\sqrt{0.5}$ and  $E_{\rm dip}(\phi)=-0.1386449741$.
\end{itemize}
\end{example}

\begin{example} Ground states of dipolar BEC.
 Here we report the ground
states of a dipolar BEC (e.g., ${}^{52}$Cr \cite{Parker}) with
different parameters and trapping potentials by using  the numerical method
(\ref{eq:discretized1:sec8})-(\ref{eq:discretized2:sec8}). In our computation and
results, we always use the dimensionless quantities. We take
$M=K=L=128$, time step $\tau=0.01$, dipolar direction
$\bn=(0,0,1)^T$ and the bounded computational domain $U=[-8,8]^3$
for all cases except $U=[-16,16]^3$ for the cases
$\frac{N}{10000}=1,\;5,\;10$ and $U=[-20,20]^3$ for the cases
$\frac{N}{10000}=50,\;100$ in Tab.~\ref{tab:2:sec8}. The ground state
$\phi_g$ is reached numerically when
$\|\phi^{n+1}-\phi^n\|_\infty:=\max\limits_{0\le j\le M,\ 0\le k\le
K,\ 0\le l\le L} |\phi^{n+1}_{jkl}-\phi^n_{jkl}|\le \epsilon:=10^{-6}$
in (\ref{eq:discretized1:sec8})-(\ref{eq:discretized2:sec8}). Tab.~\ref{tab:2:sec8} shows
the energy $E^g:=E_{\rm 3D}(\phi_g)$, chemical potential
$\mu^g:=\mu(\phi_g)$, kinetic energy $E_{\rm kin}^g:=E_{\rm
kin}(\phi_g)$, potential energy $E_{\rm pot}^g:=E_{\rm
pot}(\phi_g)$, interaction energy $E_{\rm int}^g:=E_{\rm
int}(\phi_g)$, dipolar energy $E_{\rm dip}^g:=E_{\rm dip}(\phi_g)$,
condensate widths $\sg_x^g:=\sigma_x(\phi_g)$ and
$\sg_z^g:=\sigma_z(\phi_g)$ in (\ref{eq:def_sigma:sec2}) and central density
$\rho_g({\bf 0}):=|\phi_g(0,0,0)|^2$ with harmonic potential
$V(x,y,z)= \fl{1}{2}\left(x^2+y^2+0.25z^2\right)$ for different
$\beta=0.20716N$ and $\lambda=0.033146N$ with $N$ the total number
of particles in the condensate; and Tab.~\ref{tab:3:sec8} lists similar
results  with $\beta=207.16$ for different values of $-0.5\le
\frac{\lambda}{\beta}\le 1$. In addition, Fig.~\ref{fig:1:sec8} depicts
the ground state $\phi_g(\bx)$, e.g. surface plots of
$|\phi_g(x,0,z)|^2$ and isosurface plots of  $|\phi_g(\bx)|=0.01$,
 of a dipolar BEC with $\beta = 401.432$ and $\lambda
=0.16\beta$ for harmonic potential $V(\bx)=
\fl{1}{2}\left(x^2+y^2+z^2\right)$, double-well potential
$V(\bx)=\fl{1}{2}\left(x^2+y^2+z^2\right)+4e^{-z^2/2}$ and optical
lattice potential
$V(\bx)=\fl{1}{2}\left(x^2+y^2+z^2\right)+100\left[\sin^2\left(\fl{\pi}{2}x\right)
+\sin^2\left(\fl{\pi}{2}y\right)+\sin^2\left(\fl{\pi}{2}z\right)
\right]$.
\end{example}

\begin{table}[htb]
\begin{center}
\begin{tabular}{cccccccccc}
\hline \\
$\frac{N}{10000}$ &$E^g$ &$\mu^g$  &$E_{\rm kin}^g$ &$E_{\rm pot}^g$
&$E_{\rm int}^g$ &$E_{\rm dip}^g$
&$\sigma_x^g$ &$\sigma_z^g$& $\rho_g({\bf 0})$\\
\hline
0.1 &1.567 &1.813 &0.477 &0.844 &0.262 &-0.015 &0.796 &1.299 &0.06139 \\
0.5 &2.225 &2.837 &0.349 &1.264 &0.659 &-0.047 &0.940 &1.745 &0.02675 \\
1   &2.728 &3.583 &0.296 &1.577 &0.925 &-0.070 &1.035 &2.009 &0.01779\\
5   &4.745 &6.488 &0.195 &2.806 &1.894 &-0.151 &1.354 &2.790 &0.00673 \\
10  &6.147 &8.479 &0.161 &3.654 &2.536 &-0.204 &1.538 &3.212 &0.00442     \\
50  &11.47 &15.98 &0.101 &6.853 &4.909 &-0.398 &2.095 &4.441 &0.00168  \\
100 &15.07 &21.04 &0.082 &9.017 &6.498  &-0.526 &2.400 &5.103 &0.00111  \\

\hline
\end{tabular}
 \end{center}
\caption{Different quantities of the ground states of a dipolar BEC
 for $\beta=0.20716N$ and $\lambda=0.033146N$ with different number of particles
$N$.} \label{tab:2:sec8}
 \end{table}

\begin{table}[htb]
\begin{center}
\begin{tabular}{cccccccccc}
\hline \\
 $\fl{\lambda}{\beta}$  &  $E^g$  & $\mu^g$ & $E_{\rm kin}^g$ &   $E_{\rm pot}^g$ &  $E_{\rm int}^g$
 &  $E_{\rm dip}^g$ &   $\sigma_x^g$ & $\sigma_z^g$ & $\rho_g({\bf 0})$
 \\ \hline
-0.5 &2.957 &3.927 &0.265 &1.721 &0.839 &0.131 &1.153 &1.770 &0.01575 \\
-0.25 &2.883 &3.817 &0.274 &1.675 &0.853 &0.081 &1.111 &1.879 &0.01605 \\
0 &2.794 &3.684 &0.286 &1.618 &0.890 &0.000 &1.066 &1.962 &0.01693 \\
0.25 &2.689 &3.525 &0.303 &1.550 &0.950 &-0.114 &1.017 &2.030 &0.01842 \\
0.5 &2.563 &3.332 &0.327 &1.468 &1.047 &-0.278 &0.960 &2.089 &0.02087 \\
0.75 &2.406 &3.084 &0.364 &1.363 &1.212 &-0.534 &0.889 &2.141 &0.02536\\
1.0 &2.193 &2.726 &0.443 &1.217 &1.575 &-1.041 &0.786  &2.189 &0.03630 \\
 \hline

\end{tabular}
 \end{center}
 \caption{ Different quantities of the ground states of a dipolar BEC with different
 values of $\frac{\lambda}{\beta}$ with $\beta=207.16$.}\label{tab:3:sec8}
 \end{table}
\begin{figure}[h!]
\centerline{
\psfig{figure=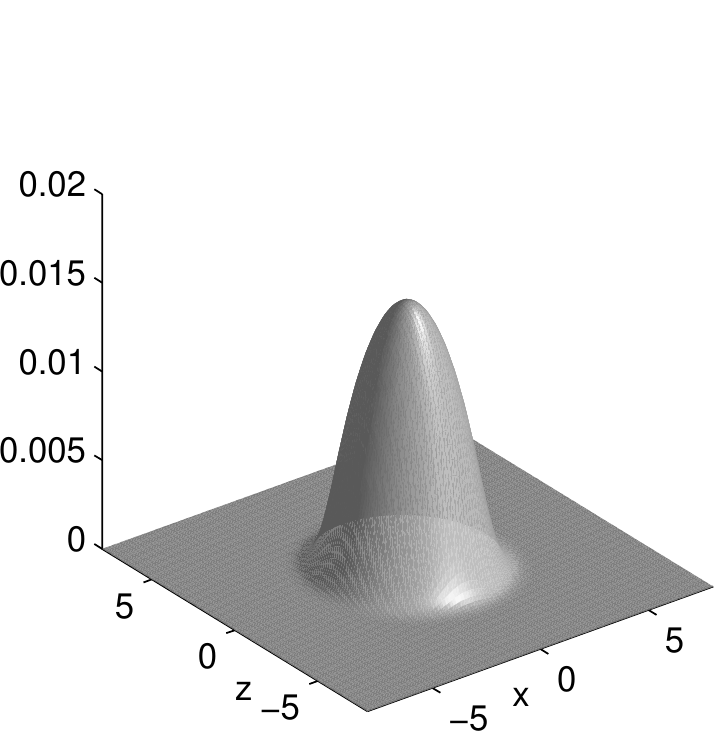,height=4.5cm,width=4.5cm,angle=0}\qquad
\psfig{figure=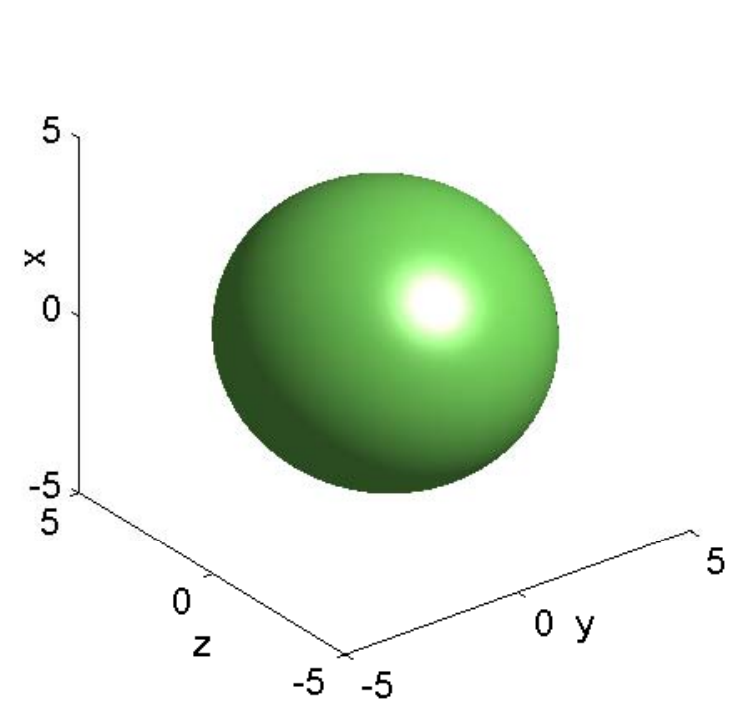,height=4.5cm,width=4.5cm,angle=0}
 }
 \centerline{
\psfig{figure=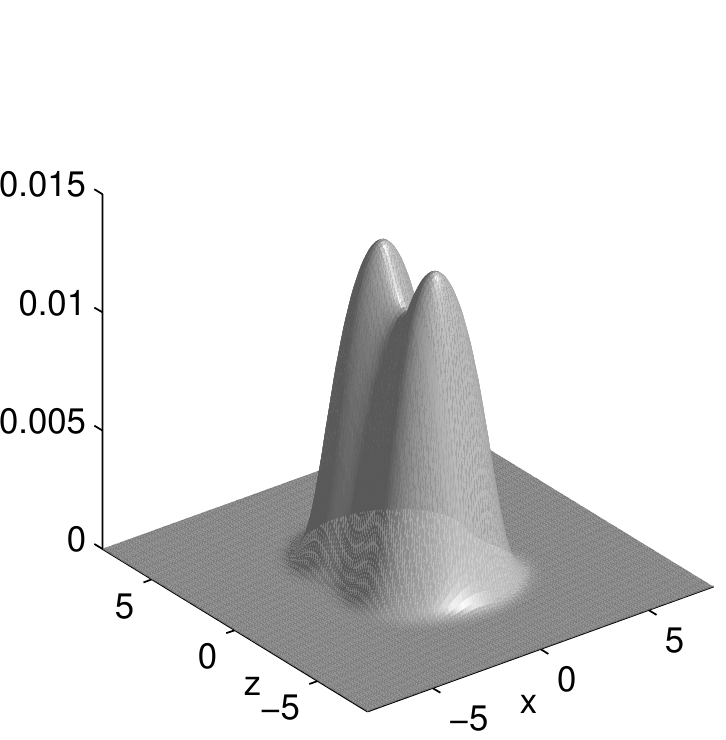,height=4.5cm,width=4.5cm,angle=0}\qquad
\psfig{figure=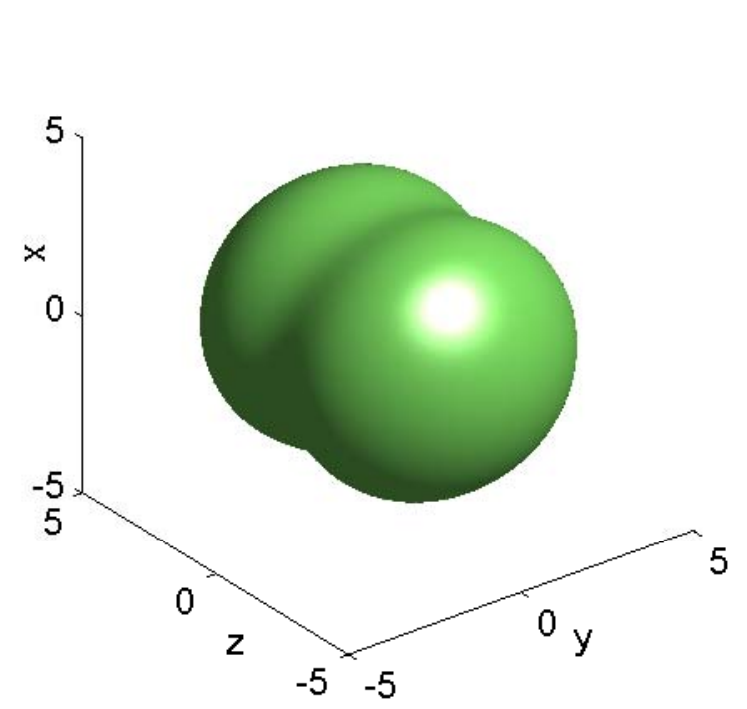,height=4.5cm,width=4.5cm,angle=0}
 }
 \centerline{
\psfig{figure=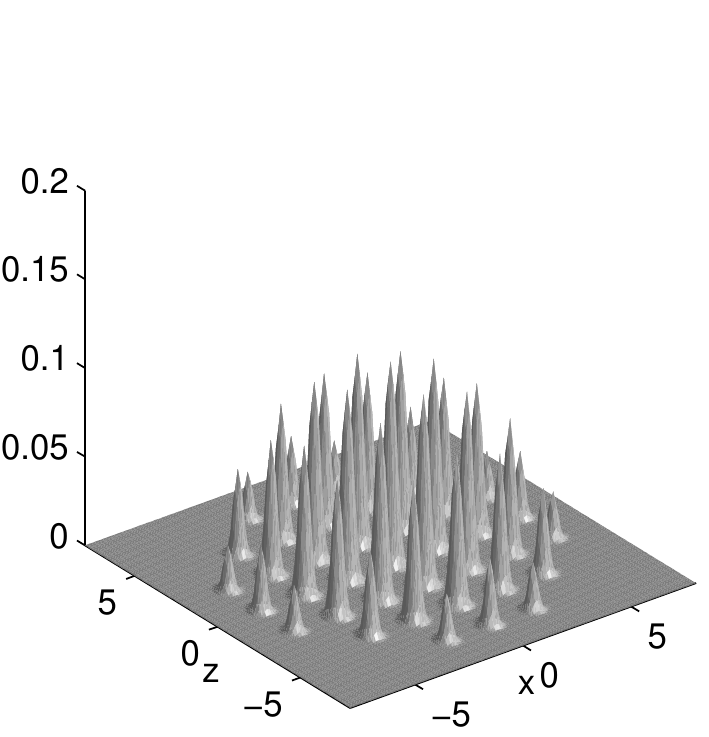,height=4.5cm,width=4.5cm,angle=0}\qquad
\psfig{figure=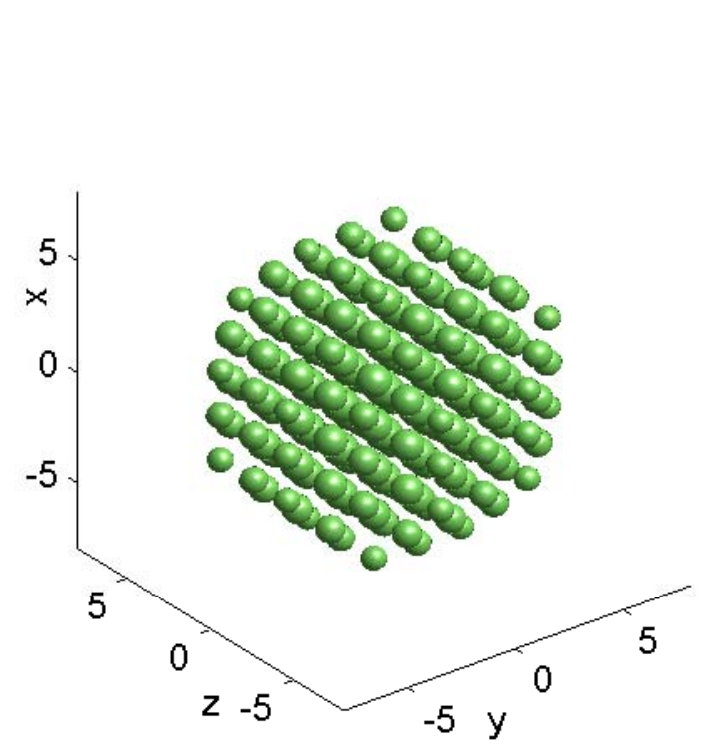,height=4.5cm,width=4.5cm,angle=0}
 }
  \caption{
Surface plots of $|\phi_g(x,0,z)|^2$ (left column) and isosurface
plots of $|\phi_g(x,y,z)|=0.01$ (right column) for the ground state
of a dipolar BEC with $\beta= 401.432$ and $\lambda = 0.16 \beta$
for harmonic potential (top row), double-well potential (middle row)
and optical lattice potential (bottom row).
  } \label{fig:1:sec8}
\end{figure}
\begin{figure}[h!]
\centerline{
\psfig{figure=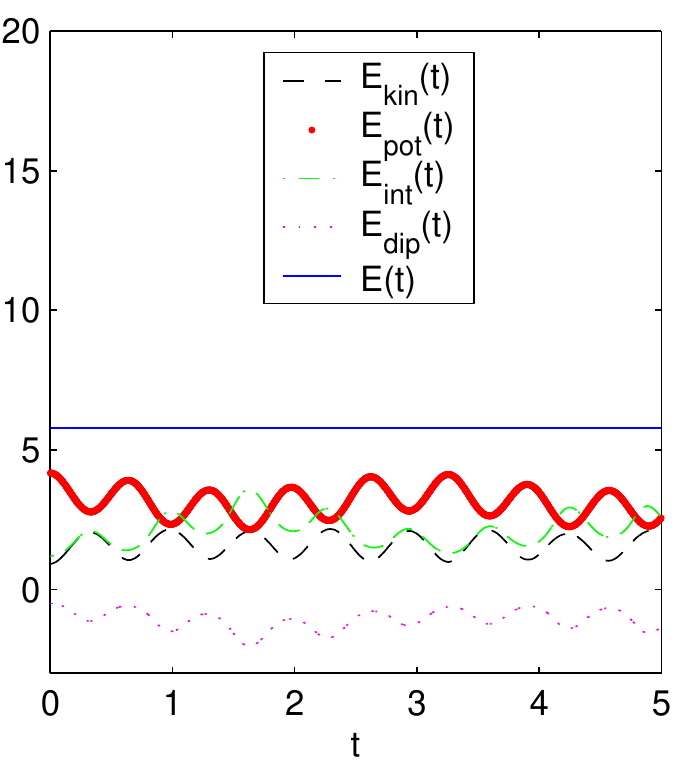,height=4.5cm,width=4.5cm,angle=0}
\qquad
\psfig{figure=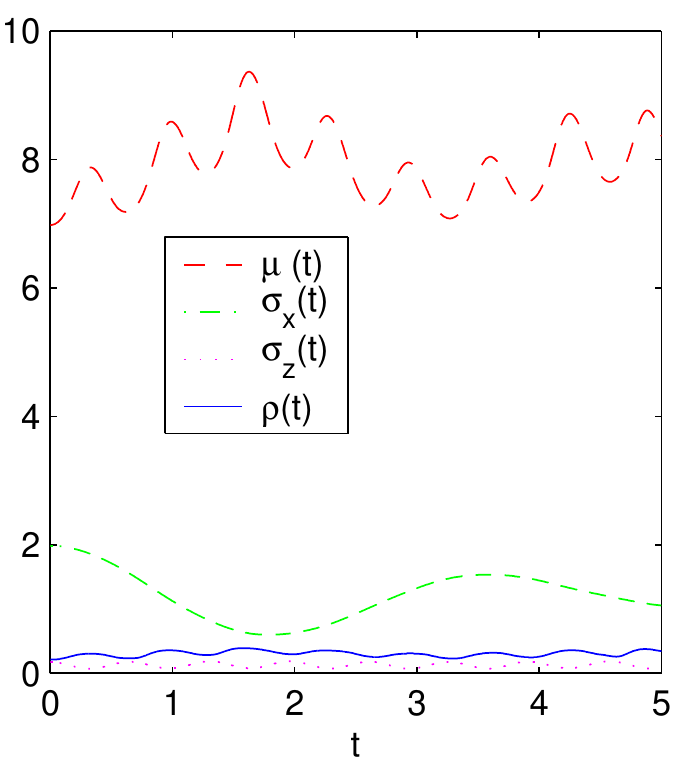,height=4.5cm,width=4.5cm,angle=0} }
\centerline{
\psfig{figure=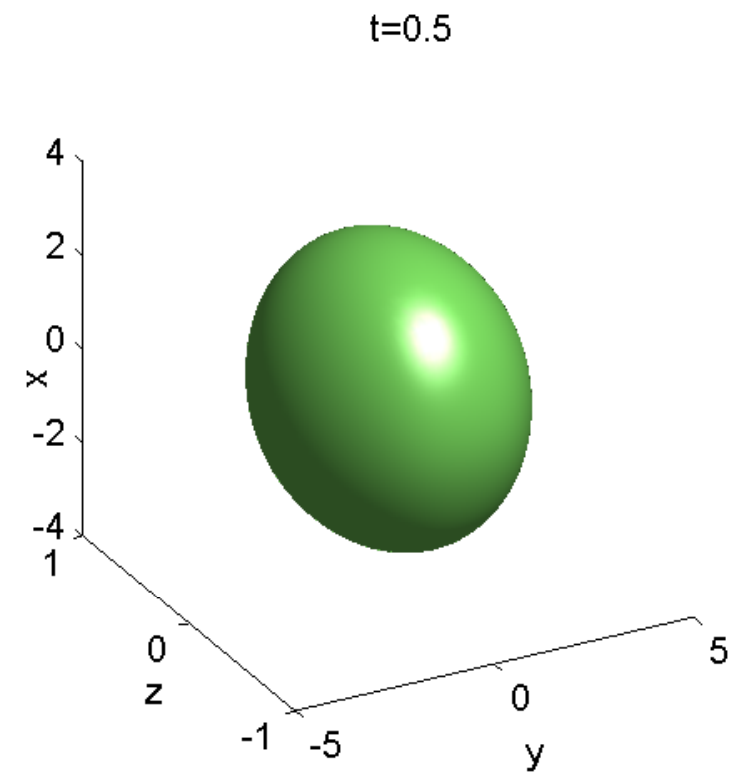,height=4.5cm,width=4.5cm,angle=0}\qquad
\psfig{figure=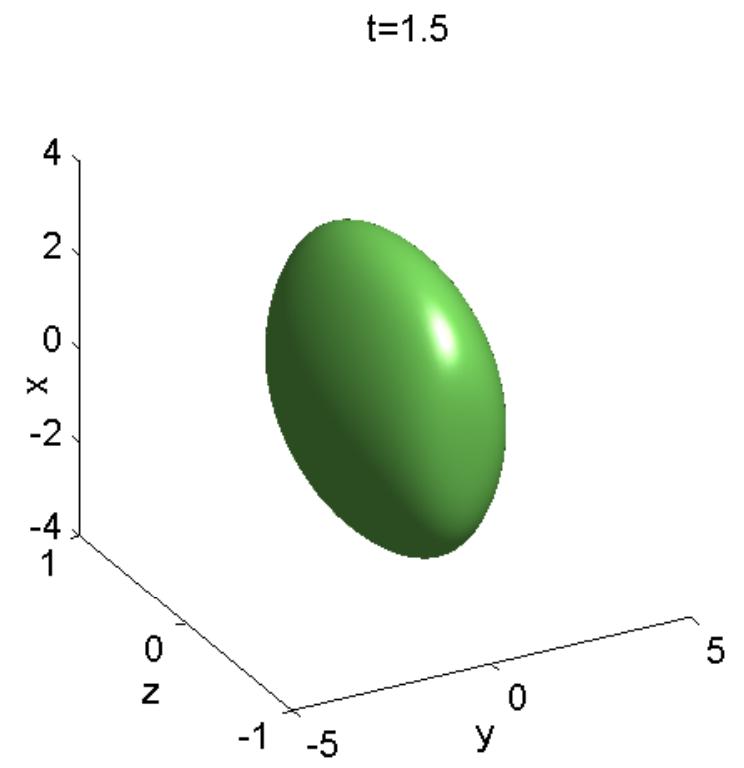,height=4.5cm,width=4.5cm,angle=0}
 }
 \centerline{
\psfig{figure=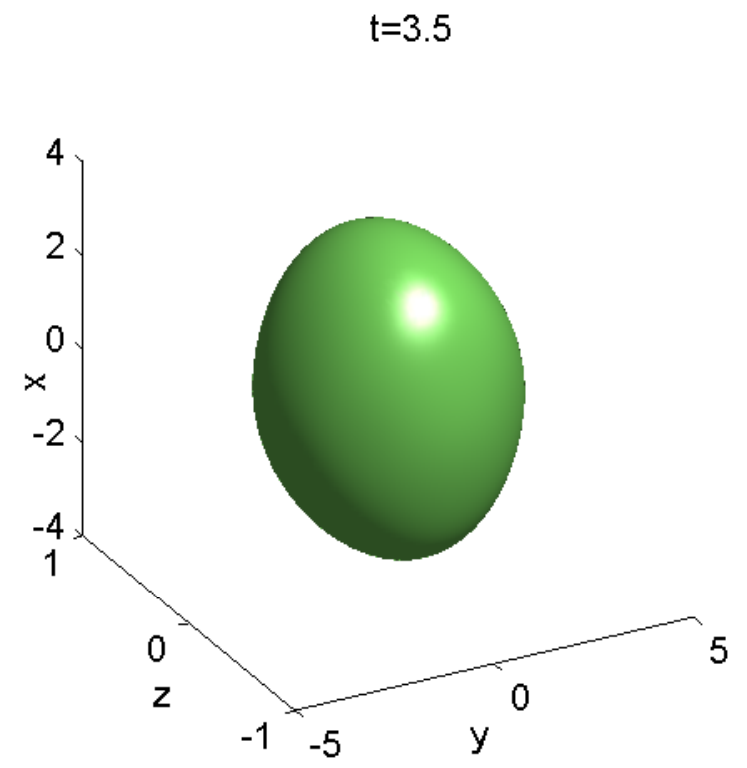,height=4.5cm,width=4.5cm,angle=0}\qquad
\psfig{figure=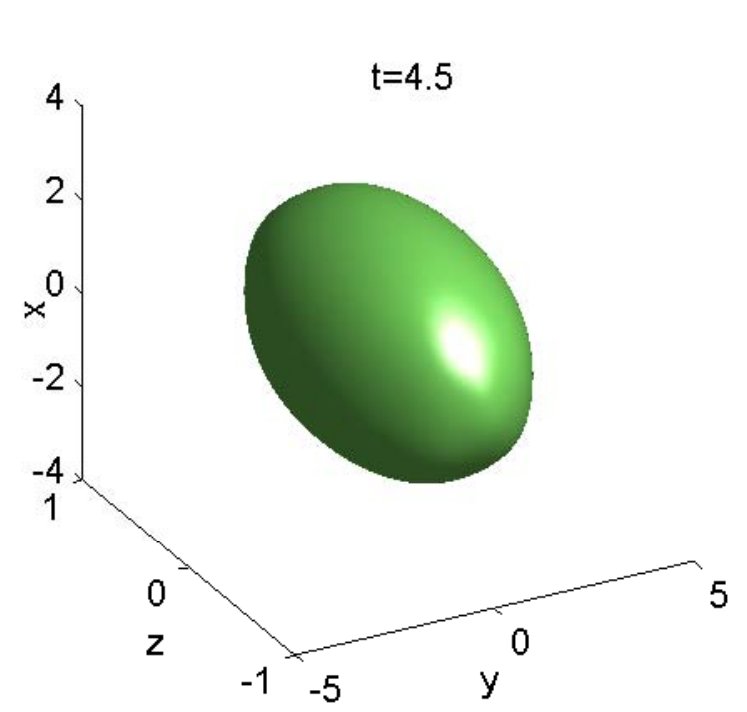,height=4.5cm,width=4.5cm,angle=0}
 }
  \caption{Time evolution of different quantities  and isosurface plots of the
density function $\rho(\bx,t):=|\psi(\bx,t)|^2=0.01$ at different
times for a dipolar BEC
  when the dipolar direction is suddenly changed from $\bn=(0,0,1)^T$ to
  $(1,0,0)^T$ at time $t=0$.
  } \label{fig:2:sec8}
\end{figure}

\begin{example} Dynamics of a dipolar BEC. Here we compute
 the dynamics of a dipolar BEC (e.g., ${}^{52}$Cr
\cite{Parker}) by using our  numerical method (\ref{eq:tssp1:sec8}). Again, in the computation and
results, we always use the dimensionless quantities. We take the
bounded computational domain $U=[-8,8]^2\times[-4,4]$,
$M=K=L=128$, i.e. $h=h_x=h_y=1/8,h_z=1/16$,  time step $\tau
=0.001$. The initial data $\psi(\bx,0)=\psi_0(\bx)$ is chosen as
the ground state of a dipolar BEC computed numerically by our
numerical method with $\bn=(0,0,1)^T$,
$V(\bx)=\fl{1}{2}(x^2+y^2+25z^2)$, $\beta=103.58$ and
$\lambda=0.8\beta=82.864$.

We study the dynamics of suddenly
changing the dipolar direction from $\bn=(0,0,1)^T$ to
$\bn=(1,0,0)^T$ at $t=0$ and keeping all other quantities unchanged.
Fig.~\ref{fig:2:sec8} depicts the time evolution of the energy
$E_{\rm 3D}(t):=E_{\rm 3D}(\psi(\cdot,t))$, chemical potential
$\mu(t)=\mu(\psi(\cdot,t))$, kinetic energy $E_{\rm kin}(t):=E_{\rm
kin}(\psi(\cdot,t))$, potential energy $E_{\rm pot}(t):=E_{\rm
pot}(\psi(\cdot,t))$, interaction energy $E_{\rm int}(t):=E_{\rm
int}(\psi(\cdot,t))$, dipolar energy $E_{\rm dip}(t):=E_{\rm
dip}(\psi(\cdot,t))$, condensate widths
$\sg_x(t):=\sg_x(\psi(\cdot,t))$,
$\sg_z(t):=\sg_z(\psi(\cdot,t))$, and central density
$\rho(t):=|\psi({\bf 0},t)|^2$, as well as the isosurface of the
density function $\rho(\bx,t):=|\psi(\bx,t)|^2=0.01$ for different
times.
\end{example}

From the above numerical results, we can see that the numerical methods based on the GPPS
(\ref{eq:gpe:sec8})-(\ref{eq:poisson:sec8}) are much more efficient and accurate than those used in the literatures based on (\ref{eq:ngpe:sec8}).

\subsection{Extensions in lower dimensions}
Here, we consider the numerical methods for computing ground states and dynamics for dipolar BECs
in 2D and 1D. The difficulties arise from the nonlocal terms, i.e., the dipolar terms in quasi-2D equation I (\ref{eq:gpe2d:sec8}), quasi-2D equation II (\ref{eq:gpe2d2:sec8}) and quasi-1D equation (\ref{eq:gpe1d:sec8}). It is obvious that those methods introduced in sections \ref{sec:numgs} and \ref{sec:numdym} can be extended here, provided that the nonlocal terms can be computed properly.

We propose to compute the convolution terms in (\ref{eq:gpe2d:sec8}), (\ref{eq:gpe2d2:sec8}) and (\ref{eq:gpe1d:sec8}) by Fourier transform. Unlike the 3D case, there are  no singularities for convolution kernels at origin, thus  discrete Fourier transform  is accurate in these cases.

\begin{lemma}(Kernels $U_\vep^{2D}$ in (\ref{eq:u2d1:sec8}) and $U_\vep^{1D}$(\ref{eq:poisson1d:sec8}) )
 For any real function $f(\bx)$ in the Schwartz
space ${\mathcal{S}}(\Bbb R^2)$, we have
\be
\widehat{U_\vep^{2D}*f}(\xi)=\hat{f}(\xi)\,\widehat{U_\vep^{2D}}(\xi)=\frac{\hat{f}(\xi)}{\pi}\int_{\Bbb
R}\frac{e^{-\vep^2s^2/2}}{|\xi|^2+s^2}ds, \qquad f\in
{\mathcal{S}}(\Bbb R^2). \ee
For any $g(z)$ in the Schwartz space ${\mathcal{S}}(\Bbb R)$, we
have \be\label{eq:u1dcov:sec8}
\widehat{U_\vep^{1D}*g}(\xi)=\hat{g}(\xi)\widehat{U_\vep^{1D}}(\xi)=
\frac{\sqrt{2}\,\vep\hat{g}(\xi)}{\sqrt{\pi}}\int_0^\infty\frac{e^{-\vep^2
s/2}}{|\xi|^2+s}ds, \quad \xi\in{\Bbb R}. \ee
Here $\widehat{f}$ and $\widehat{g}$ denote the Fourier transforms of $f$ and $g$, respectively.
\end{lemma}
The Fourier transforms of  $U_\vep^{2D}$ and $U_\vep^{1D}$ can be written in terms of the second kind Bessel functions \cite{CaiRosen}.

 \section{Mathematical theory and numerical methods for two component BEC}
 \label{sec:2bec}
\setcounter{equation}{0}\setcounter{figure}{0}\setcounter{table}{0}
 In view of potential applications, such
as the generation of bright beams of coherent matter waves (atom laser),
a central goal in the study of BEC has been the formation of condensate with the
number of atoms being as large as possible. It is thus of particular
interest to study a scenario where this goal is achieved by
uniting two (or more) independently grown condensates to form
one large single condensate. The first experiment involving the
uniting of  multiple-component BEC was performed with
atoms evaporatively cooled in the $|F=2, m_f=2\rangle$ and $|1, -1\rangle$
spin states of $^{87}$Rb \cite{Myatt}. It demonstrated the possibility of producing
long-lived multiple condensate systems, and that the condensate
wave function is dramatically affected by the  presence of
inter-component interactions.

\subsection{Coupled Gross-Pitaevskii equations}

At temperatures $T$ much smaller than the critical temperature
$T_c$ \cite{PitaevskiiStringari},  a two-component BEC with
an internal atomic Josephson junction (or an external driving field) can be well described by the coupled Gross-Pitaevskii equations (CGPEs)  \cite{LiebSeiringer2, LiebSeiringer1,da3,ZhangBaoLi,BaoCai0},
\be\label{eq:cgpe109:sec9}
\begin{split} &i\hbar\partial_t \psi_1=\left[-\frac{\hbar^2}{2m}\nabla^2
+V(\bx)+\hbar \delta
+g_{11}|\psi_1|^2+g_{12}|\psi_2|^2\right]\psi_1+\lambda \hbar
\psi_2, \\
&i\hbar\partial _t \psi_2=\left[-\frac{\hbar^2}{2m}\nabla^2
+V(\bx)+g_{21}|\psi_1|^2+g_{22}|\psi_2|^2\right]\psi_2+\lambda \hbar
\psi_1,\qquad \bx\in\Bbb{R}^3.\end{split} \ee
Here $\Psi:=\Psi(\bx,t)=(\psi_1(\bx,t),\psi_2(\bx,t))^T$ is the
complex-valued macroscopic wave function, $V(\bx)$ is the
real-valued external trapping potential, $\lambda$ is the effective
Rabi frequency describing the strength to realize the internal atomic
Josephson junction (JJ) by a Raman transition, $\delta$ is the Raman transition constant.
The interactions of particles are described by
$g_{jl}=\frac{4\pi \hbar^2 a_{jl}}{m}$ with $a_{jl}=a_{lj}$ ($j,l=1,2$) being
the $s$-wave scattering lengths between the $j$th and $l$th component (positive for
repulsive interaction and negative for attractive interaction).
It is necessary to ensure that the wave function is properly
normalized. Especially, we require \be \label{eq:norm09:sec9}
\int_{{\Bbb R}^3}
\left[|\psi_1(\bx,t)|^2+|\psi_2(\bx,t)|^2\right]\,d\bx=N=N_1^0+N_2^0, \ee
where
\[N_j^0=\int_{{\Bbb R}^3} |\psi_j(\bx,0)|^2d\bx, \]
is the particle number of the $j$th ($j=1,2$) component at time $t=0$
and $N$ the total number of particle in
the two-component BEC.

By properly nondimensionalization and dimension reduction, we can obtain
the following dimensionless CGPEs in $d$-dimensions ($d=1,2,3$) for a two-component BEC
\cite{ZhangBaoLi,BaoCai0}
\be\label{eq:cgpe1:sec9}
\begin{split} &i\partial_t \psi_1=\left[-\frac 12\nabla^2
+V(\bx)+\delta
+(\beta_{11}|\psi_1|^2+\beta_{12}|\psi_2|^2)\right]\psi_1+\lambda
\psi_2, \qquad\bx\in\Bbb R^d,\\
&i\partial _t \psi_2=\left[-\frac 12\nabla^2
+V(\bx)+(\beta_{12}|\psi_1|^2+\beta_{22}|\psi_2|^2)\right]\psi_2+\lambda
\psi_1,\qquad \bx\in\Bbb{R}^d.\end{split} \ee
Here $\Psi:=\Psi(\bx,t)=(\psi_1(\bx,t),\psi_2(\bx,t))^T$ is the dimensionless
complex-valued macroscopic wave function, $V(\bx)$ is the dimensionless
real-valued external trapping potential, $\beta_{11}$, $\beta_{12}=\beta_{21}$,
$\beta_{22}$ are dimensionless interaction constants, $\delta$ and $\lambda$ are dimensionless constants.
In addition, the wave function is normalized as
 \be \label{eq:norm:sec9}
\|\Psi\|_2^2:=\int_{{\Bbb R}^d}
\left[|\psi_1(\bx,t)|^2+|\psi_2(\bx,t)|^2\right]\,d\bx=1. \ee

The dimensionless CGPEs (\ref{eq:cgpe1:sec9}) conserves the total mass or
normalization, i.e.
 \be\label{eq:mass1:sec9}
N(t):=\|\Psi(\cdot,t)\|^2=N_1(t)+N_2(t) \equiv \|\Psi(\cdot,0)\|^2
=1, \quad t\ge0, \ee with \be\label{eq:Njt:sec9}
N_j(t)=\|\psi_j(\bx,t)\|_2^2:= \|\psi_j(\bx,t)\|_2^2=\int_{{\Bbb
R}^d} |\psi_j(\bx,t)|^2\,d\bx, \quad t\ge0, \quad j=1,2,\ee and
the energy
\begin{eqnarray} E(\Psi)&=&\int_{{\Bbb R}^d}\biggl[\frac
12 \left(|\nabla\psi_1|^2+
|\nabla\psi_2|^2\right)+V(\bx)(|\psi_1|^2+|\psi_2|^2)+\delta
|\psi_1|^2+\frac 12 \beta_{11}|\psi_1|^4\nonumber\\
&&\qquad +\frac 12\beta_{22}|\psi_2|^4
+\beta_{12}|\psi_1|^2|\psi_2|^2+2\lambda\cdot\text{Re}
(\psi_1\bar{\psi}_2)\biggl]d\bx. \label{eq:energy:sec9}
\end{eqnarray}
 In addition, if there is no internal
Josephson junction  in (\ref{eq:cgpe1:sec9}), i.e. $\ld=0$, the mass of each
component is also conserved, i.e. \be\label{eq:Njtt1:sec9} N_1(t)\equiv
\int_{{\Bbb R}^d} |\psi_1(\bx,0)|^2\,d\bx:=\ap, \quad
N_2(t)\equiv \int_{{\Bbb R}^d}
|\psi_2(\bx,0)|^2\,d\bx:=1-\ap, \ee with $0\le \ap\le 1$
a given constant.

\subsection{Ground states for the case without Josephson junction}

  If there is no external driving field in (\ref{eq:cgpe1:sec9}), i.e. $\ld=0$,
for any given $\ap\in[0,1]$, the ground state
$\Phi_g^\ap(\bx)=(\phi_1^\ap(\bx),\phi_2^\ap(\bx))^T$ of the
two-component BEC is defined as the minimizer of the following
nonconvex minimization problem:\\
  Find $\left(\Phi_g^\ap \in S_\ap\right)$, such that
  \begin{equation}\label{eq:minim2:sec9}
    E_g^\ap := E_0\left(\Phi_g^\ap\right) = \min_{\Phi \in S_\ap}
    E_0\left(\Phi\right),
  \end{equation}
where $S_\ap$ is a nonconvex set defined as \be\label{eq:nonconset2:sec9}
S_\ap:=\left\{\Phi=(\phi_1,\phi_2)^T \ | \ \|\phi_1\|_2^2=\ap,\
\|\phi_2\|_2^2=1-\ap,\  E_0(\Phi)<\infty \right\},\ee and the energy
functional $E_0(\Phi)$ is defined as
\begin{eqnarray} E_0(\Phi)&=&\int_{{\mathbb R}^d}\biggl[\frac
12 \left(|\nabla\phi_1|^2+
|\nabla\phi_2|^2\right)+V(\bx)(|\phi_1|^2+|\phi_2|^2)+\delta
|\phi_1|^2+\frac 12 \beta_{11}|\phi_1|^4\nonumber\\
&&\qquad +\frac 12\beta_{22}|\phi_2|^4
+\beta_{12}|\phi_1|^2|\phi_2|^2\biggl]d\bx. \label{eq:energy0:sec9}
\end{eqnarray}
 Again, it is
easy to see that the ground state $\Phi_g^\ap$  satisfies the
following Euler-Lagrange equations, \be \label{eq:grd2:sec9} \begin{split}
&\mu_1\,\phi_1=\left[-\frac 12\nabla^2 +V(\bold{x})+\delta
+(\beta_{11}|\phi_1|^2+\beta_{12}|\phi_2|^2)\right]\phi_1,\qquad \bx\in {\mathbb R}^d,\\
&\mu_2 \,\phi_2=\left[-\frac 12\nabla^2
+V(\bold{x})+(\beta_{12}|\phi_1|^2+\beta_{22}|\phi_2|^2)\right]\phi_2,\qquad
\bx\in {\mathbb R}^d,
\end{split}
\ee under the two constraints \be \label{eq:norm2:sec9}
\|\phi_1\|_2^2:=\int_{{\mathbb R}^d} |\phi_1(\bx)|^2\,d\bx=\ap,\qquad
\|\phi_2\|_2^2:=\int_{{\mathbb R}^d} |\phi_2(\bx)|^2\,d\bx=1-\ap, \ee
with $\mu_1$ and $\mu_2$ being the Lagrange multipliers or chemical
potentials corresponding to the two constraints (\ref{eq:norm2:sec9}).
Again, the above time-independent CGPEs (\ref{eq:grd2:sec9}) can also be
obtained from the CGPEs (\ref{eq:cgpe1:sec9}) with $\ld=0$ by substituting
the ansatz
\begin{equation}\label{eq:anst2:sec9} \psi_1(\bx,t)=e^{-i\mu_1 t}\phi_1(\bx),\qquad
\psi_2(\bx,t)=e^{-i\mu_2 t}\phi_2(\bx).
\end{equation}

Considering the case $\alpha\in(0,1)$ in minimization problem (\ref{eq:minim2:sec9}), denote
\[
\beta_{11}^\prime:=\alpha\beta_{11},\quad \beta_{22}^\prime
:=(1-\alpha)\beta_{22},\quad
\beta_{12}^\prime:=\sqrt{\alpha(1-\alpha)} \beta_{12},\quad
\alpha^\prime:=\alpha(1-\alpha),
\]
and
\[B=\left(\ba{cc}
\beta_{11} &\beta_{12}\\
\beta_{21} &\beta_{22}\\
\ea\right), \qquad B^\prime=\left(\ba{cc}
\beta_{11}^\prime &\beta_{12}^\prime\\
\beta_{21}^\prime &\beta_{22}^\prime\\
\ea\right). \]
Then the  following conclusions can be drawn \cite{BaoCai0}.
\begin{theorem} (Existence and uniqueness of (\ref{eq:minim2:sec9}))
Suppose  $V(\bx)\ge 0$ satisfying
$\lim_{|\bx|\to\infty}V(\bx)=\infty$ and at least one of the following condition holds,
\begin{enumerate}\renewcommand{\labelenumi}{(\roman{enumi})}
 \item $d=1$;
\item  $d=2$ and  $\beta^\prime_{11}
\ge -C_{b}$, $\beta^\prime_{22}\ge -C_b$, and  $\beta^\prime_{12}
\ge-\sqrt{(C_b+\beta^\prime_{11})(C_b+\beta_{22}^\prime)}$;
\item $d=3$ and  $B$ is
either positive semi-definite or nonnegative,
\end{enumerate}
then there exists a ground state $\Phi_g=(\phi_1^g,\phi_2^g)^T$ of
(\ref{eq:minim2:sec9}). In addition,
 $\widetilde{\Phi}_g:=
 (e^{i\theta_1}|\phi_1^g|,e^{i\theta_2}|\phi_2^g|)$ is also a
ground state of (\ref{eq:minim2:sec9}) with two constants $\theta_1$ and
$\theta_2$. Furthermore, if the matrix $B^\prime$ is positive
semi-definite,  the ground state $(|\phi_1^g|,|\phi_2^g|)^T$ of
(\ref{eq:minim2:sec9}) is unique.
 In contrast, if one of the following conditions holds,
\begin{enumerate}\renewcommand{\labelenumi}{(\roman{enumi})}
\item  $d=2$ and $\beta_{11}^\prime<-C_b$ or $\beta_{22}^\prime<-C_b$
or $\beta^\prime_{12}<-\frac{1}{2\sqrt{\alpha^\prime}}
\left(\alpha\beta^\prime_{11}+
(1-\alpha)\beta_{22}^\prime+C_b\right)$;
\item $d=3$ and $\beta_{11}<0$ or $\beta_{22}<0$ or $\beta_{12}<-\frac{1}
{2\alpha^\prime}(\alpha^2\beta_{11}+(1-\alpha)^2\beta_{22})$.
\end{enumerate}
there exists no ground states of (\ref{eq:minim2:sec9}).
\end{theorem}

\subsection{Ground states for the case with Josephson junction}
The ground state $\Phi_g(\bx)=(\phi_1^g(\bx),\phi_2^g(\bx))^T$ of
the two-component BEC with an internal
Josephson junction (\ref{eq:cgpe1:sec9})
is defined as the minimizer of the following nonconvex minimization
problem:\\
Find $\left(\Phi_g \in S\right)$, such that
  \begin{equation}\label{eq:minimize:sec9}
    E_g := E\left(\Phi_g\right) = \min_{\Phi \in S}
    E\left(\Phi\right),
  \end{equation}
where $S$ is a nonconvex set defined as \be\label{eq:nonconset:sec9}
S:=\left\{\Phi=(\phi_1,\phi_2)^T \big| \int_{{\Bbb
R}^d}\left(|\phi_1(\bx)|^2+|\phi_2(\bx)|^2\right)d\bx=1,\
E(\Phi)<\infty \right\}.\ee It is easy to see that the ground state
$\Phi_g$  satisfies the following Euler-Lagrange equations, \be
\label{eq:grd1:sec9} \begin{split} &\mu\,\phi_1=\left[-\frac 12\nabla^2
+V(\bold{x})+\delta
+(\beta_{11}|\phi_1|^2+\beta_{12}|\phi_2|^2)\right]\phi_1+\lambda
\phi_2,\quad \bx\in\Bbb R^d,\\
&\mu \,\phi_2=\left[-\frac 12\nabla^2
+V(\bold{x})+(\beta_{12}|\phi_1|^2+\beta_{22}|\phi_2|^2)\right]\phi_2+\lambda
\phi_1,\qquad \bx\in {\mathbb R}^d,
\end{split}
\ee under the constraint \be \label{eq:norm1:sec9}
\|\Phi\|_2^2:=\|\Phi\|_2^2=\int_{{\mathbb R}^d}
\left[|\phi_1(\bx)|^2+|\phi_2(\bx)|^2\right]\,d\bx=1, \ee
 with the eigenvalue  $\mu$ being the Lagrange multiplier
 or chemical potential corresponding to the constraint (\ref{eq:norm1:sec9}), which can be
computed as
\begin{eqnarray} \mu&=&\mu(\Phi)=\int_{{\mathbb R}^d}\biggl[\frac
12 \left(|\nabla\phi_1|^2+
|\nabla\phi_2|^2\right)+V(\bx)(|\phi_1|^2+|\phi_2|^2)+\delta
|\phi_1|^2+ \beta_{11}|\phi_1|^4\nonumber\\
&&\qquad +\beta_{22}|\phi_2|^4
+2\beta_{12}|\phi_1|^2|\phi_2|^2+2\lambda\cdot\text{Re}
(\phi_1\bar{\phi}_2)\biggl]d\bx. \label{eq:chemp3:sec9}
\end{eqnarray}
  In
fact, the above time-independent CGPEs (\ref{eq:grd1:sec9}) can also be
obtained from the CGPEs (\ref{eq:cgpe1:sec9}) by substituting the ansatz
\begin{equation} \label{eq:anst1:sec9} \psi_1(\bx,t)=e^{-i\mu t}\phi_1(\bx),\qquad
\psi_2(\bx,t)=e^{-i\mu t}\phi_2(\bx).
\end{equation}
The eigenfunctions of  the nonlinear eigenvalue problem (\ref{eq:grd1:sec9})
under the normalization (\ref{eq:norm1:sec9}) are usually called as
stationary states of the two-component BEC (\ref{eq:cgpe1:sec9}). Among
them, the  eigenfunction with minimum energy is the ground state and
those whose energy are larger than that of the ground state are
usually called as excited states.

  It is easy to see that the ground state $\Phi_g$ defined in
  (\ref{eq:minimize:sec9}) is equivalent to the following:\\
Find $\left(\Phi_g \in S\right)$, such that
 \bea\label{eq:minim3:sec9}
  E\left(\Phi_g\right)= \min_{\Phi \in S}
    E\left(\Phi\right)
    =\min_{\ap\in[0,1]}\;E(\ap),
    \qquad E(\ap)=\min_{\Phi \in S_\ap}\; E(\Phi).
  \eea

Denote
\begin{equation}
\mathcal{D}=\left\{\Phi=(\phi_1,\phi_2)^T\ | \,  V\,|\phi_j|^2\in
L^1(\mathbb{R}^d), \ \phi_j\in H^1(\mathbb{R}^d)\cap
L^4(\mathbb{R}^d),\ j=1,2\right\},
\end{equation}
then the ground state $\Phi_g$ of (\ref{eq:minimize:sec9}) is also given by the following:\\
  Find $\left(\Phi_g \in \mathcal{D}_1\right)$, such that
  \begin{equation}\label{eq:minim6:sec9}
    E_g := E\left(\Phi_g\right) = \min_{\Phi \in \mathcal{D}_1}
    E\left(\Phi\right),
  \end{equation}
where
\begin{equation}
\mathcal{D}_1=\mathcal{D}\cap \left\{\Phi=(\phi_1,\phi_2)^T\ |\
\|\Phi\|^2=\int_{\mathbb{R}^d}
(|\phi_1(\bx)|^2+|\phi_2(\bx)|^2)\,d\bx=1\right\}.
\end{equation}

In addition, we introduce the auxiliary  energy functional
\begin{align}\label{eq:mini1:sec9}
\widetilde{E}(\Phi)=&\int_{\mathbb{R}^d}\bigg\{\frac 12
\left(|\nabla\phi_1|^2+
|\nabla\phi_2|^2\right)+\left[V(\bx)\left(|\phi_1|^2
+|\phi_2|^2\right)+\delta|\phi_1|^2\right]\\&
+\left(\frac 12 \beta_{11}|\phi_1|^4+\frac 12\beta_{22}|\phi_2|^4
+\beta_{12}|\phi_1|^2|\phi_2|^2\right)-2|\lambda|\cdot
|\phi_1|\cdot|\phi_2|\,\bigg\}\,d\bx,\nn
\end{align}
and the auxiliary nonconvex minimization problem:\\
  Find $\left(\Phi_g \in \mathcal{D}_1\right)$, such that
  \begin{equation}\label{eq:minim7:sec9}
     \widetilde{E}\left(\Phi_g\right) = \min_{\Phi \in \mathcal{D}_1}
    \widetilde{E}\left(\Phi\right).
  \end{equation}

For $\Phi=(\phi_1,\phi_2)^T$, we write $E(\phi_1,\phi_2)=E(\Phi)$
and $\widetilde{E}(\phi_1,\phi_2)=\widetilde{E}(\Phi)$. Then we have
the following lemmas \cite{BaoCai0}:

\begin{lemma}\label{lem:equival:sec9} For the minimizers
$\Phi_g(\bx)=(\phi_1^g(\bx),\phi_2^g(\bx))^T$  of the
nonconvex minimization problems (\ref{eq:minim6:sec9}) and (\ref{eq:minim7:sec9}),
we have

(i). If $\Phi_g$ is a minimizer of (\ref{eq:minim6:sec9}), then
$\phi_1^g(\bx)=e^{i\tht_1}|\phi_1^g(\bx)|$ and
$\phi_2^g(\bx)=e^{i\tht_2}|\phi_2^g(\bx)|$ with $\tht_1$ and
$\tht_2$ two constants satisfying $\theta_1=\theta_2$ if
$\lambda<0$; and  $\theta_1=\theta_2\pm\pi$ if $\lambda>0$. In
addition,
$\widetilde{\Phi}_g=\left(e^{i\tht_3}\phi_1^g,e^{i\tht_4}\phi_2^g\right)^T$
with $\tht_3$ and $\tht_4$ two constants satisfying
$\theta_3=\theta_4$ if $\lambda<0$; and $\theta_3=\theta_4\pm\pi$ if
$\lambda>0$ is also a minimizer of (\ref{eq:minim6:sec9}).

(ii). If $\Phi_g$ is a minimizer of (\ref{eq:minim7:sec9}), then
$\phi_1^g(\bx)=e^{i\tht_1}|\phi_1^g(\bx)|$ and
$\phi_2^g(\bx)=e^{i\tht_2}|\phi_2^g(\bx)|$ with $\tht_1$ and
$\tht_2$ two constants. In addition,
$\widetilde{\Phi}_g=\left(e^{i\tht_3}\phi_1^g,e^{i\tht_4}\phi_2^g\right)^T$
with $\tht_3$ and $\tht_4$ two constants  is also a minimizer of
(\ref{eq:minim7:sec9}).

(iii). If $\Phi_g$ is a minimizer of (\ref{eq:minim6:sec9}), then
 $\Phi_g$ is also a minimizer of (\ref{eq:minim7:sec9}).

 (iv). If $\Phi_g$ is a minimizer of (\ref{eq:minim7:sec9}), then
 $\widetilde{\Phi}_g=\left(|\phi_1^g|,-\text{sign}(\lambda)|\phi_2^g|\right)^T$
 is a minimizer of (\ref{eq:minim6:sec9}).
\end{lemma}
For the auxiliary minimization problem (\ref{eq:minim7:sec9}), we have the following results generalizing
the single component BEC case in section \ref{sec:mathgpe}.

\begin{theorem}\label{thm:mres:sec9}(Existence and uniqueness of (\ref{eq:minim7:sec9}) \cite{BaoCai0})
Suppose  $V(\bx)\ge 0$ satisfying
$\lim\limits_{|\bx|\to\infty}V(\bx)=\infty$, then  there exists a
minimizer
 $\Phi^\infty=(\phi_1^\infty,\phi_2^\infty)^T\in \mathcal{D}_1$ of
 (\ref{eq:minim7:sec9}) if one of the following conditions holds,
 \begin{enumerate}\renewcommand{\labelenumi}{(\roman{enumi})}
 \item $d=1$;
\item  $d=2$ and   $\beta_{11}
\ge -C_{b},\beta_{22}\ge -C_b$, $\beta_{12}\ge -C_b
-\sqrt{C_b+\beta_{11}}\sqrt{C_b+\beta_{22}}$;
\item $d=3$ and  $B$ is
either positive semi-definite or nonnegative,
\end{enumerate}
where $C_b$ is given in (\ref{eq:bestcons:2d}).
In addition, if the matrix $B$
is positive semi-definite and at least one of the
 parameters  $\delta$, $\lambda$, $\gm_1$ and $\gm_2$ are nonzero,
 then the minimizer
$(|\phi_1^\infty|,|\phi_2^\infty|)^T$ is unique.
\end{theorem}

Combining Theorem \ref{thm:mres:sec9} and Lemma \ref{lem:equival:sec9}, we  draw the conclusions \cite{BaoCai0}:
\begin{theorem} (Existence and uniqueness of (\ref{eq:minimize:sec9}))\label{thm:con:sec9}
Suppose  $V(\bx)\ge 0$ satisfying
$\lim\limits_{|\bx|\to\infty}V(\bx)=\infty$ and at least one of
the following condition holds,
\begin{enumerate}\renewcommand{\labelenumi}{(\roman{enumi})}
 \item $d=1$;
\item  $d=2$ and  $\beta_{11}
\ge -C_{b}$,  $\beta_{22}\ge -C_b$, and $\beta_{12}\ge -C_b-
\sqrt{C_b+\beta_{11}}\sqrt{C_b+\beta_{22}}$;
\item $d=3$ and  $B$ is
either positive semi-definite or nonnegative,
\end{enumerate}
there exists a
ground state $\Phi_g=(\phi_1^g,\phi_2^g)^T$ of (\ref{eq:minimize:sec9}). In
addition,
 $\widetilde{\Phi}_g:=
 (e^{i\theta_1}|\phi_1^g|,e^{i\theta_2}|\phi_2^g|)$ is also a
ground state of (\ref{eq:minimize:sec9}) with $\theta_1$ and $\theta_2$ two
constants satisfying $\theta_1-\theta_2=\pm\pi$ when $\lambda>0$ and
$\theta_1-\theta_2=0$ when $\lambda<0$, respectively. Furthermore,
if the matrix $B$ is positive semi-definite and at least one of the
 parameters  $\delta$, $\lambda$, $\gm_1$ and $\gm_2$ are nonzero, then the ground state
$(|\phi_1^g|,-{\rm sign}(\ld)|\phi_2^g|)^T$ is unique.
 In contrast, if one of the following conditions holds,
\begin{enumerate}\renewcommand{\labelenumi}{(\roman{enumi})}
\item  $d=2$ and $\beta_{11}<-C_b$ or $\beta_{22}<-C_b$ or $\beta_{12}<
-C_b-\sqrt{C_b+\beta_{11}}\sqrt{C_b+\beta_{22}}$ ;
\item $d=3$ and $\beta_{11}<0$ or $\beta_{22}<0$ or $\beta_{12}<0$ with
$\beta_{12}^2> \beta_{11}\beta_{22}$;
\end{enumerate}
there exists no ground state of (\ref{eq:minimize:sec9}).
\end{theorem}
When either $|\delta|$ or $|\lambda|$ goes to infinity, the two component ground state problem (\ref{eq:minimize:sec9}) will collapse to a single component ground state problem in section \ref{sec:mathgpe} \cite{BaoCai0}.

For fixed  $\beta_{11}\ge0$ and
$\beta_{22}\ge0$, when $\beta_{12}\to\infty$,
 the phase of two components of the ground state
 $\Phi_g=(\phi_1^g,\phi_2^g)^T$ will be segregated \cite{CaffeLin,Caliari,DuLin}, i.e. $\Phi_g$ will converge to a state such
that $\phi_1^g\cdot\phi_2^g=0$.

\begin{remark} If the potential  $V(\bx)$ in the two equations in
(\ref{eq:cgpe1:sec9}) is chosen to be different in different equations, i.e.
$V_j(\bx)$ in the $j$th ($j=1,2$) equation, and they satisfy
$V_j(\bx)\ge 0$, $\lim\limits_{|\bx|\to\infty}V_j(\bx)=\infty$
($j=1,2$), then the conclusions in the above Lemmas and Theorems
\ref{thm:mres:sec9}-\ref{thm:con:sec9} are still valid under the similar conditions.
\end{remark}

\subsection{Dynamical properties}
Well-posedness of Cauchy problem of the CGPEs (\ref{eq:cgpe1:sec9}) in energy space is quite similar to that of  single GPE (cf. section \ref{sec:mathgpe}), and we omit the results here. If there is no internal Josephson junction, i.e. $\lambda=0$,  the density of each
component is conserved. With an internal Josephson junction, we have
the following lemmas for the dynamics of the
density of each component \cite{ZhangBaoLi}:
\begin{lemma}\label{lem:6LEMMA1:sec9}
Suppose $(\psi_1({\bx},t),\psi_2({\bx},t))$ is the solution of the
CGPEs (\ref{eq:cgpe1:sec9}) with potential $V(\bx)+\delta$ for the first component $\psi_1$ replaced by $V_1(\bx)$ and potential $V(\bx)$ for the second component $\psi_2$ replaced by $V_2(\bx)$; then we have, for $j=1,2$
\be \label{eq:ode11:sec9}
\ddot{N}_j(t)=-2\ld^2\left[2N_j(t)-1\right]+F_j(t),
\qquad t\ge0, \ee with initial conditions \bea \label{eq:6L1_1:sec9}
&&N_j(0)=N_j^{(0)}=\int_{{\Bbb
R}^d}|\psi_j^0(\bx)|^2\;d{\bx}=\frac{N_j^0}{N},\\
\label{eq:6L1_15:sec9} &&\dot{N}_j(0)=N_j^{(1)}= 2\ld \int_{{\Bbb
R}^d}{\rm
Im}\left[\psi_j^0(\bx)\left(\overline{\psi_{k_j}^0(\bx)}\right)\right]
\;d{\bx}; \eea where $k_1=2$, $k_2=1$ and for $t\ge0$,  \beas &&F_j(t)=\ld\int_{{\Bbb
R}^d}\left(\overline{\psi_j}\psi_{k_j}+\psi_j\overline{\psi_{k_j}}\right)
\bigg[V_{k_j}({\bx})-V_j({\bx}) \\
&&\qquad\qquad\qquad-(\bt_{jj}-\bt_{k_jj})|\psi_j|^2
+(\bt_{k_jk_j}-\bt_{jk_j})|\psi_{k_j}|^2\bigg] d\bx,\quad t\ge
0.\qquad\qquad \eeas
\end{lemma}

From this lemma, we have \cite{ZhangBaoLi}

\begin{lemma}\label{6LEMMA9}
 If $\delta =0$ and  $\beta_{11} =
\beta_{12} =\beta_{21}= \beta_{22}$ in (\ref{eq:cgpe1:sec9}),
for any initial data $\Psi(\bx,t=0)=(\psi_1^0(\bx), \psi_2^0(\bx))^T$, we have,
for $t\ge0$, \bea\label{6L1_0}
N_j(t)=\left\|\psi_j(\cdot,t)\right\|_2^2=
\left(N_j^{(0)}-\fl{1}{2}\right)\cos(2\ld
t)+\fl{N_j^{(1)}}{2\ld}\sin(2\ld t)+\fl{1}{2}, \quad  j=1,2. \eea
Thus in this case, the density of each component is a periodic
function with period $T = \pi/|\ld|$ depending only on $\ld$.
\end{lemma}

\subsection{Numerical methods for computing ground states}
To find the ground state, we first present a continuous normalized gradient flow (CNGF) method discussed in section \ref{subsubsec:gfdn} and then propose a GFDN method based on discretization of CNGF.
\subsubsection{Continuous normalized gradient flow and its discretization}
In order to compute the ground state of two-component BEC with an
internal Josephson junction (\ref{eq:minimize:sec9}), we construct the following CNGF \cite{BaoCai0}:
\bea\label{eq:cgf1:sec9}\begin{split} &\qquad
\frac{\partial\phi_1(\bx,t)}{\partial t}=\left[\frac 1 2\nabla^2
-V(\bx)-\delta-(\beta_{11}|\phi_1|^2+\beta_{12}|\phi_2|^2)\right]\phi_1-\lambda
\phi_2+\mu_\Phi(t)\phi_1,  \\
&\qquad \frac{\partial\phi_2(\bx,t)}{\partial t}=\left[\frac
12\nabla^2-V(\bx)-(\beta_{12}|\phi_1|^2+\beta_{22}|\phi_2|^2)\right]\phi_2-\lambda
\phi_1+\mu_\Phi(t)\phi_2,
\end{split}
\eea where $\Phi(\bx,t)=(\phi_1(\bx,t),\phi_2(\bx,t))^T$ and
$\mu_\Phi(t)$ is chosen such that the above CNGF is mass or
normalization conservative and it is given as
\begin{eqnarray}
\mu_\Phi(t)&=&\frac{1}{\|\Phi(\cdot,t)\|_2^2}\int_{{\mathbb R}^d}
\biggl[\frac 12 \left(|\nabla\phi_1|^2+
|\nabla\phi_2|^2\right)+V(\bx)(|\phi_1|^2+|\phi_2|^2)+\delta|\phi_1|^2\nn \\
&&\qquad + \beta_{11}|\phi_1|^4+\beta_{22}|\phi_2|^4
+2\beta_{12}|\phi_1|^2|\phi_2|^2+2\lambda\ \text{Re}(
\phi_1\bar{\phi}_2)\biggr]\,d\bx\nn\\
&=&\frac{\mu(\Phi(\cdot,t))}{\|\Phi(\cdot,t)\|_2^2}, \qquad t\ge0.
\end{eqnarray}

For the above CNGF, we have \cite{BaoCai0}

\begin{theorem}
For any given initial data \be\label{eq:init1:sec9}
\Phi(\bx,0)=(\phi_1^0(\bx),\phi_2^0(\bx))^T:=\Phi^{(0)}(\bx), \qquad
\bx\in{\mathbb R}^d,\ee satisfying $\|\Phi^{(0)}\|_2^2=1$, the CNGF
(\ref{eq:cgf1:sec9}) is mass or normalization conservative and energy
diminishing, i.e.
\begin{equation}
\|\Phi(\cdot,t)\|_2^2\equiv \|\Phi^{(0)}\|_2^2=1,\qquad
E(\Phi(\cdot,t))\leq E(\Phi(\cdot,s)),\qquad 0\leq s\leq t.
\end{equation}
\end{theorem}
For practical computation, here we also present a second-order in
both space and time full discretization for the above CNGF
(\ref{eq:cgf1:sec9}). For  simplicity of notation, we introduce the method
for the case of one spatial dimension (1D) in a bounded domain
$U=(a,b)$ with homogeneous Dirichlet boundary condition
\be\label{eq:dbc1:sec9} \Phi(a,t)=\Phi(b,t)={\bf 0}, \qquad t\ge0. \ee
Generalizations to higher dimensions are straightforward for tensor
product grids.

Let
$\Phi_{j}^n=(\phi_{1,j}^n,\phi_{2,j}^n)^T$ be the numerical
approximation of $\Phi(x_j,t_n)$ and $\Phi^n$ be the solution vector
with component $\Phi_{j}^n$. In addition, denote
$\Phi_j^{n+1/2}=(\phi_{1,j}^{n+1/2},\phi_{2,j}^{n+1/2})^T$ with \be
\phi_{l,j}^{n+1/2}=\frac
12\left(\phi_{l,j}^{n+1}+\phi_{l,j}^n\right),\qquad
j=0,1,2,\ldots,M,\qquad l=1,2. \ee Then a second-order  full
discretization for the CNGF (\ref{eq:cgf1:sec9}) is given, for
$j=1,2,\ldots,M-1$ and $n\ge0$,  as
\begin{align*}
\frac{\phi_{1,j}^{n+1}-\phi_{1,j}^n}{\tau}=&\frac{\phi_{1,j+1}^{n+1/2}
-2\phi_{1,j}^{n+1/2}+\phi_{1,j-1}^{n+1/2}}{2h^2}
-\left[V(x_j)+\delta-\mu_{\Phi,h}^{n+1/2}\right]\phi_{1,j}^{n+1/2}-\lambda
\phi_{2,j}^{n+1/2} \\&-\frac 12\left[
\beta_{11}\left(|\phi_{1,j}^{n+1}|^2+|\phi_{1,j}^n|^2\right)+
\beta_{12}\left(|\phi_{2,j}^{n+1}|^2+|\phi_{2,j}^n|^2\right)
\right]\phi_{1,j}^{n+1/2},\\
\frac{\phi_{2,j}^{n+1}-\phi_{2,j}^n}{\tau}=&\frac{\phi_{2,j+1}^{n+1/2}-2\phi_{2,j}^{n+1/2}
+\phi_{2,j-1}^{n+1/2}}{2h^2}
-\left[V(x_j)-\mu_{\Phi,h}^{n+1/2}\right]\phi_{2,j}^{n+1/2}-\lambda
\phi_{1,j}^{n+1/2} \\&-\frac 12\left[
\beta_{12}\left(|\phi_{1,j}^{n+1}|^2+|\phi_{1,j}^n|^2\right)+
\beta_{22}\left(|\phi_{2,j}^{n+1}|^2+|\phi_{2,j}^n|^2\right)\right]
\phi_{2,j}^{n+1/2},
\end{align*}
where \be
\mu_{\Phi,h}^{n+1/2}=\frac{D_{\Phi,h}^{n+1/2}}{h\sum\limits_{j=0}^{M-1}
\left(|\phi_{1,j}^{n+1/2}|^2+|\phi_{2,j}^{n+1/2}|^2\right)}, \qquad
n\ge0,\ee with
\begin{eqnarray}
D_{\Phi,h}^{n+1/2}&=&h\sum\limits_{j=0}^{M-1}\biggl\{\sum\limits_{l=1}^2
\left(\frac{1}{2h^2}|\phi_{l,j+1}^{n+1/2}-\phi_{l,j}^{n+1/2}|^2+V(x_j)
|\phi_{l,j}^{n+1/2}|^2\right)
+\delta |\phi_{1,j}^{n+1/2}|^2\nonumber\\
&&+\frac
12\beta_{11}(|\phi_{1,j}^{n+1}|^2|+|\phi_{1,j}^n|^2)|\phi_{1,j}^{n+1/2}|^2+\frac
12\beta_{22}(|\phi_{2,j}^{n+1}|^2+|\phi_{2,j}^n|^2)|\phi_{2,j}^{n+1/2}|^2\nonumber\\
&&+\frac 12\beta_{12}\left[(|\phi_{2,j}^{n+1}|^2+|\phi_{2,j}^n|^2)
|\phi_{1,j}^{n+1/2}|^2+(|\phi_{1,j}^{n+1}|^2|+|\phi_{1,j}^n|^2)
|\phi_{2,j}^{n+1/2}|^2\right] \nonumber\\&&+2\lambda\
\text{Re}\left(
\phi_{1,j}^{n+1/2}\bar{\phi}_{2,j}^{n+1/2}\right)\biggl\}.
\end{eqnarray}
The boundary condition (\ref{eq:dbc1:sec9}) is  discretized as
 \be\label{eq:cngfbdry:sec9}
 \phi_{1,0}^{n+1}=\phi_{1,M}^{n+1}=\phi_{2,0}^{n+1}=\phi_{2,M}^{n+1}=0, \quad n=0,1,2,\ldots.
 \ee
The initial data (\ref{eq:init1:sec9}) is discretized as \be\label{eq:cngfini:sec9}
\phi_{1,j}^0=\phi_1^0(x_j), \qquad \phi_{2,j}^0=\phi_2^0(x_j),\qquad
j=0,1,\ldots,M. \ee
In the above full discretization, at every time step, we need to
solve a fully nonlinear system which  is very tedious in practical
computation. Below we present a more efficient discretization for
the CNGF (\ref{eq:cgf1:sec9}) for computing the ground states.

\subsubsection{Gradient flow with discrete normalization}

Another more efficient way to discretize  the CNGF (\ref{eq:cgf1:sec9}) is
through the construction of the following GFDN \cite{BaoCai0}:

\be\label{eq:dngf1:sec9} \begin{split}
\frac{\partial\phi_1}{\partial t}=\left[\frac 1 2\nabla^2
-V(\bx)-\dt-(\beta_{11}|\phi_1|^2+\beta_{12}|\phi_2|^2)\right]\phi_1
-\lambda\phi_2, \qquad\qquad\quad \\
\frac{\partial\phi_2}{\partial t}=\left[\frac
12\nabla^2-V(\bx)-(\beta_{12}|\phi_1|^2+\beta_{22}|\phi_2|^2)\right]\phi_2-\lambda
\phi_1,\quad t\in(t_n,t_{n+1}), \end{split} \ee followed by a
projection step as
\begin{eqnarray}\label{eq:dnproj:sec9}
\phi_l(\bx,t_{n+1}):=
\phi_l(\bx,t_{n+1}^+)=\sg_l^{n+1}\;\phi_l(\bx\, ,t_{n+1}^-), \qquad
l=1,2, \quad n\ge0,
\end{eqnarray}
where $\phi_l(\bx,t_{n+1}^{\pm})=\lim\limits_{t\to
t_{n+1}^{\pm}}\phi_l(\bx,t)$ ($l=1,2$)  and $\sg_l^{n+1}$ ($l=1,2$)
are chosen such that \be \label{eq:pcont3:sec9}
 \|\Phi(\bx,t_{n+1})\|^2=\|\phi_1(\bx,t_{n+1})\|_2^2+
 \|\phi_2(\bx,t_{n+1})\|_2^2=1, \qquad n\ge0.
\ee The above GFDN (\ref{eq:dngf1:sec9})-(\ref{eq:dnproj:sec9}) can be viewed as
applying the first-order splitting method to the CNGF (\ref{eq:cgf1:sec9})
and the projection step (\ref{eq:dnproj:sec9}) is equivalent to solving the
following ordinary differential equations (ODEs) \be
\frac{\partial\phi_1(\bx,t)}{\partial t}= \mu_\Phi(t)\phi_1, \qquad
\frac{\partial\phi_2(\bx,t)}{\partial t}= \mu_\Phi(t)\phi_2, \qquad
\quad t_n\leq t\leq t_{n+1}, \ee which immediately suggests that the
projection constants in (\ref{eq:dnproj:sec9}) are chosen as \be
\label{eq:pcont1:sec9}\sg_1^{n+1}=\sg_{2}^{n+1}, \qquad n\ge0. \ee Plugging
(\ref{eq:pcont1:sec9}) and (\ref{eq:dnproj:sec9}) into (\ref{eq:pcont3:sec9}), we obtain \be
\sg_1^{n+1}=\sg_2^{n+1}=\frac{1}{\|\Phi(\cdot,t_{n+1}^-)\|_2}=
\frac{1}{\sqrt{\|\phi_1(\cdot,t_{n+1}^-)\|_2^2
+\|\phi_2(\cdot,t_{n+1}^-)\|_2^2}}, \; n\ge0. \ee

Then, BEFD in section \ref{subsec:BEFD} can be used to discretize the GFDN (\ref{eq:dngf1:sec9})-(\ref{eq:dnproj:sec9}) and we omit the detailed scheme here, as the generalization is straightforward.

\subsection{Numerical methods for computing dynamics}
To compute dynamics of a two component BEC, finite difference time domain methods in section \ref{subsec:fdtd} can be directly extended to solve the CGPEs (\ref{eq:cgpe1:sec9}). Here we focus on the time splitting methods.
 For $n = 0, 1, \dots$,  from time
$t=t_n=n \tau$ to $t=t_{n+1}=t_n+\tau$, the CGPEs (\ref{eq:cgpe1:sec9})
are solved in three splitting steps \cite{ZhangBaoLi,Wang,WangXu}. One first
solves \bea \label{eq:ODE1:Sec9} i\fl{\p\psi_j}{\p
t}=-\fl{1}{2}\nabla^2\psi_j, \qquad j = 1, 2, \eea
for the time step of length $\tau$, followed by solving \bea
\label{eq:ODE3:sec9} i\fl{\p\psi_j}{\p t} =
V_j({\bx})\psi_j+\sum_{l=1}^2\bt_{jl} |\psi_l|^2\psi_j,\qquad j =
1, 2, \eea for the same time step with $V_1(\bx)=V(\bx)+\delta$ and $V_2(\bx)=V(\bx)$, and then by solving \be
\label{eq:ODE2:sec9} i\fl{\p\psi_1}{\p t} = -\ld\psi_{2}, \quad i\fl{\p\psi_2}{\p t}=-\ld\psi_{1},\ee  for the same time step. For time $t\in[t_n,t_{n+1}]$,
the ODE system (\ref{eq:ODE3:sec9}) leaves $|\psi_1({\bx},t)|$ and
$|\psi_2({\bx},t)|$
 invariant in $t$, and thus it
can be integrated {\sl exactly} to obtain \cite{BaoJinP,BaoJakschP,BaoZhang,BaoZhang2,Zhang0}, for
$j=1,2$ and $t\in[t_n, t_{n+1}]$ \bea \label{eq:solution3:sec9}
\psi_j({\bx},t) = \psi_j(\bx,t_n)\exp\left[-i\left(V_j({\bx})+
\sum_{l=1}^2\bt_{jl}\left|\psi_l(\bx,t_n)\right|^2\right)
(t-t_n)\right]. \eea
 For the ODE system (\ref{eq:ODE2:sec9}), we can rewrite it as
\bea i\fl{\p\Psi}{\p t}= -\ld A\Psi,\qquad\mathrm{with}\quad
A=\left(\begin{array}{cc}
0 & 1\\
1 & 0
\end{array}\right)\quad {\mathrm and}\quad
\Psi=\left(\begin{array}{c}
\psi_1\\
\psi_2
\end{array}\right).
\eea Since $A$ is a real and symmetric matrix, it can be
diagonalized and integrated {\sl exactly}, and then we obtain
\cite{Bao,ZhangBaoLi}, for $t\in[t_n, t_{n+1}]$ \begin{equation*} \Psi(\bx,t)=e^{i\ld
A\;(t-t_n)}\Psi(\bx,t_n) =\left(\ba{ll}
\cos\left(\ld(t-t_n)\right) &i\sin\left(\ld(t-t_n)\right)\\
i\sin\left(\ld(t-t_n)\right) &\cos\left(\ld(t-t_n)\right)\\
\ea\right)\Psi(\bx,t_n). \end{equation*}
Then, time splitting spectral method introduced in sections \ref{sec:numdym} and \ref{sec:numrotat} can be applied to compute the dynamics of the CGPEs (\ref{eq:cgpe1:sec9}), by a suitable composition of the above three steps (cf. section \ref{subsec:ts}). The detailed scheme is omitted here for brevity.

\subsection{Numerical results}
In this section, we will report
the ground states of (\ref{eq:minimize:sec9}), computed by our numerical methods.

\begin{figure}[t!]
\centerline{
\psfig{figure=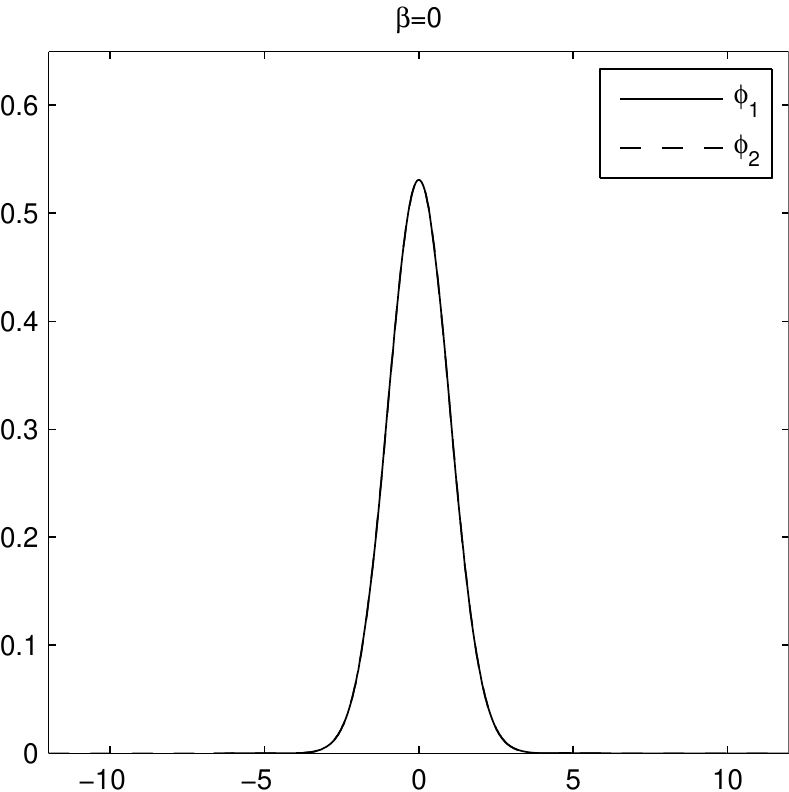,height=5cm,width=5cm,angle=0} \quad
\psfig{figure=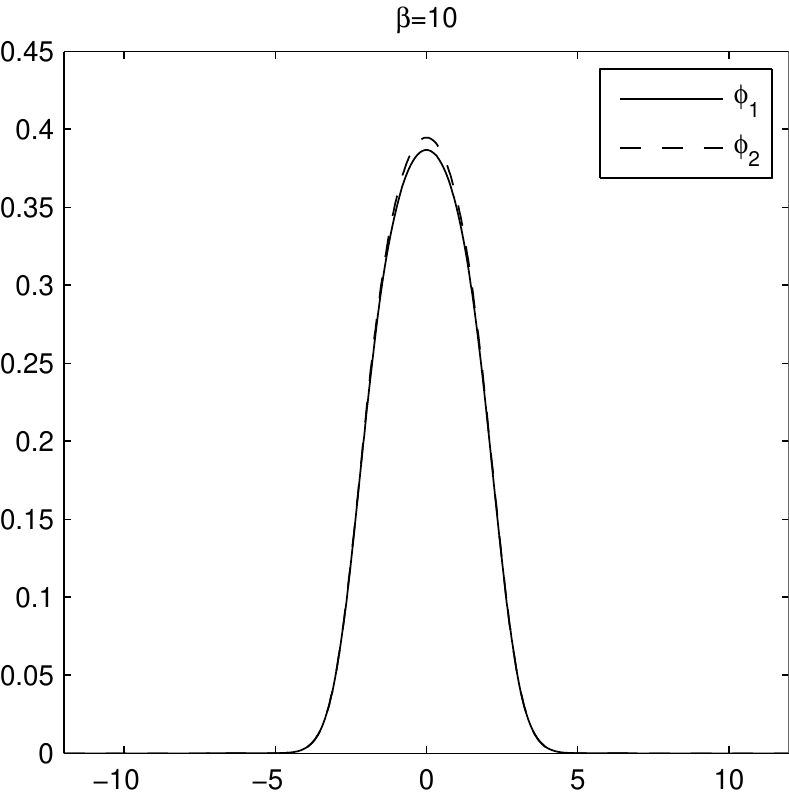,height=5cm,width=5cm,angle=0}}
\centerline{
\psfig{figure=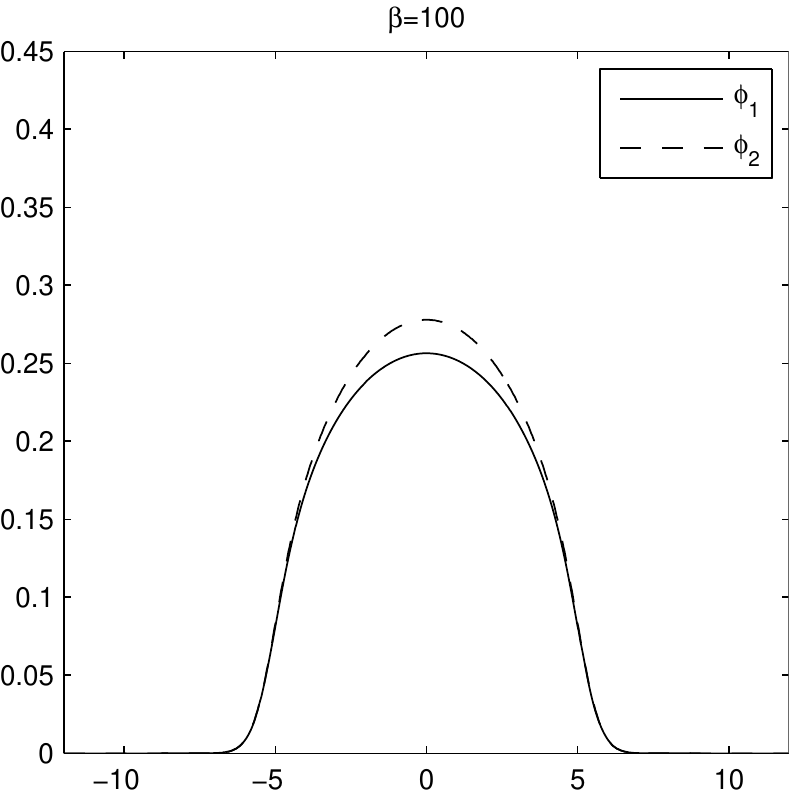,height=5cm,width=5cm,angle=0}
\quad
\psfig{figure=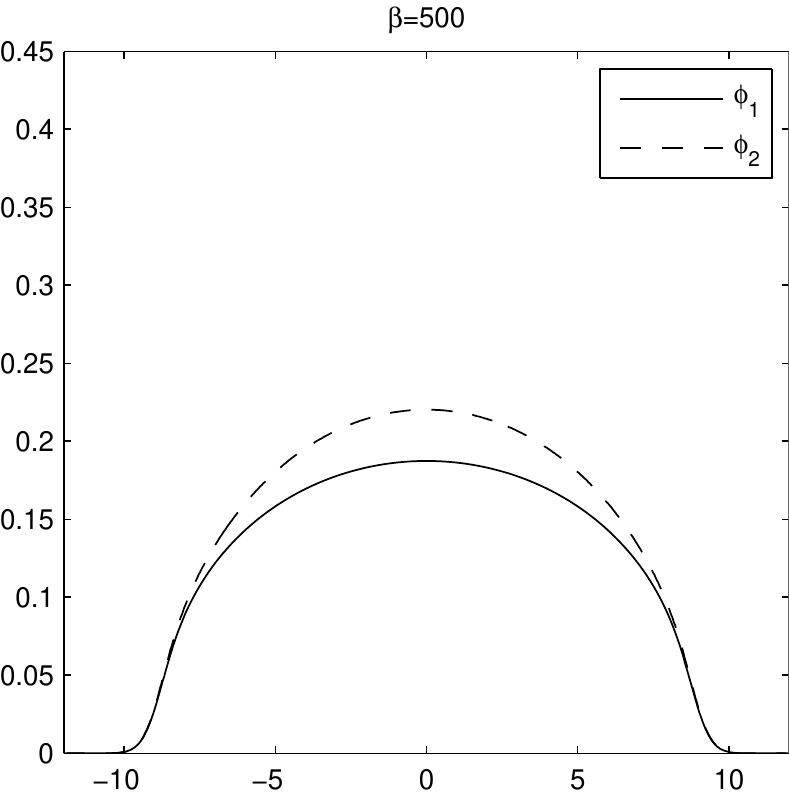,height=5cm,width=5cm,angle=0}}

\caption{Ground states $\Phi_g=(\phi_1,\phi_2)^T$ in Example~\ref{exm:1:sec9} when
$\dt=0$ and $\ld=-1$ for different $\beta$. } \label{fig:fig1:sec9}
\end{figure}
\begin{figure}[t]
\centerline{
\psfig{figure=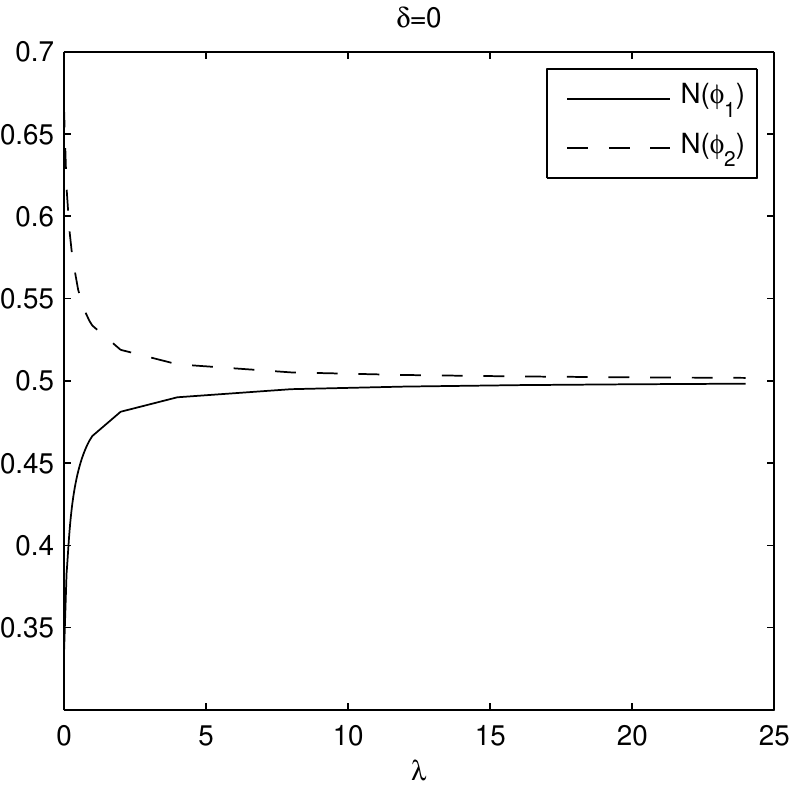,height=5cm,width=5cm,angle=0} \quad
\psfig{figure=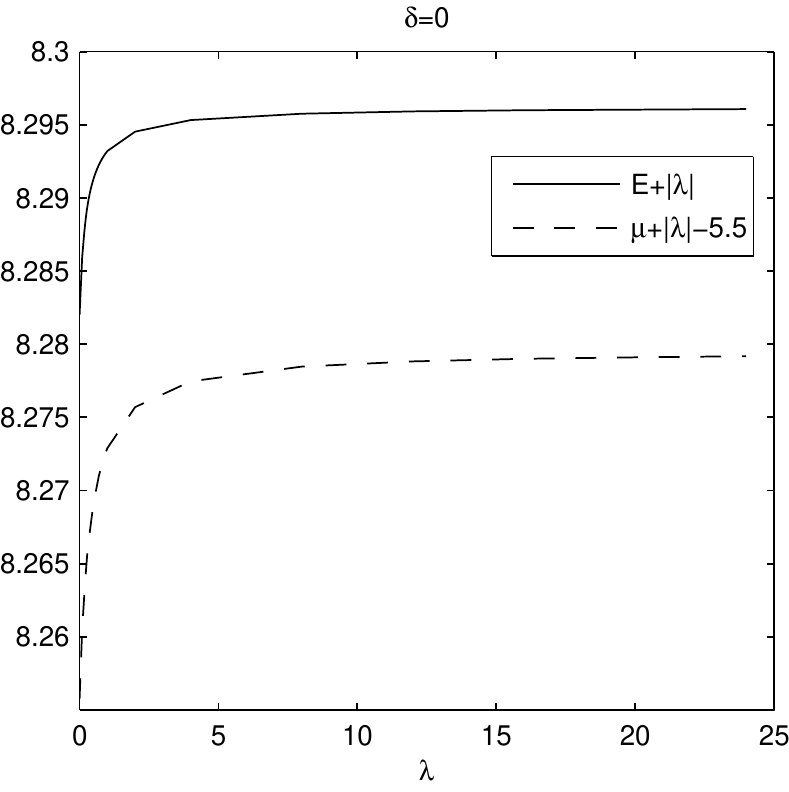,height=5cm,width=5cm,angle=0}}
\centerline{ \psfig{figure=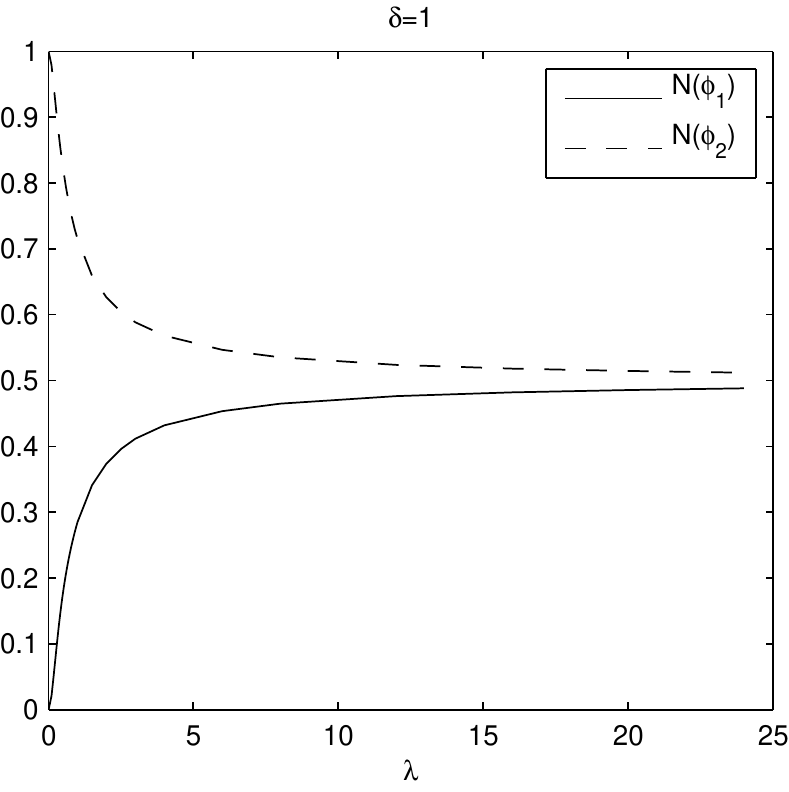,height=5cm,width=5cm,angle=0}
\quad \psfig{figure=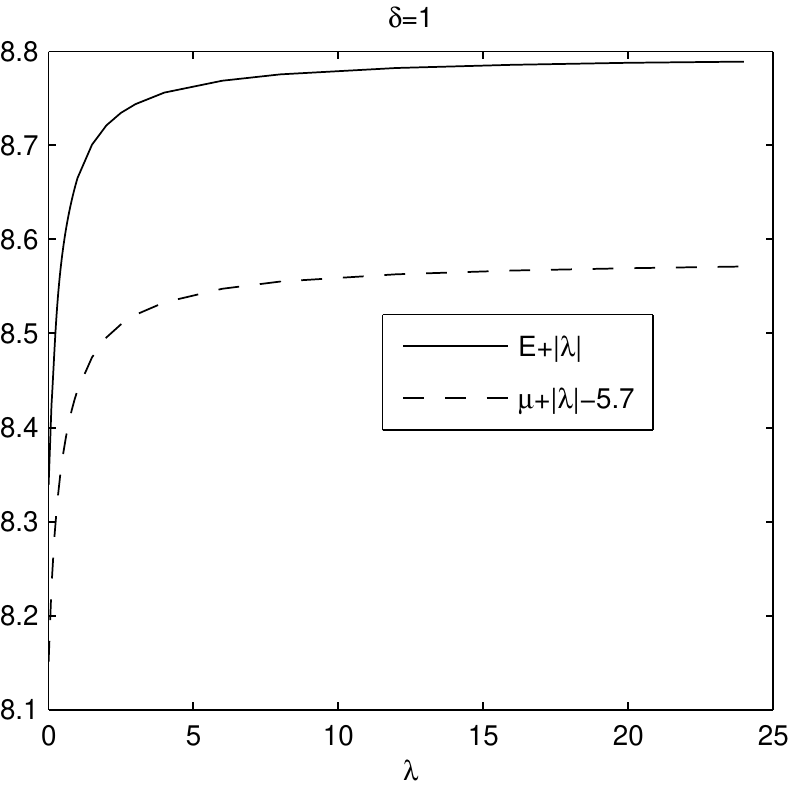,height=5cm,width=5cm,angle=0}}
\caption{Mass of each component $N(\phi_j)=\|\phi_j\|^2$ ($j=1,2$),
energy $E:=E(\Phi_g)$ and chemical potential
 $\mu:=\mu(\Phi_g)$ of the ground states in Example~\ref{exm:1:sec9} when $\beta=100$
 and $\dt=0,1$
 for different $\ld$. } \label{fig:fig4:sec9}
\end{figure}
\begin{figure}[t]
\centerline{
\psfig{figure=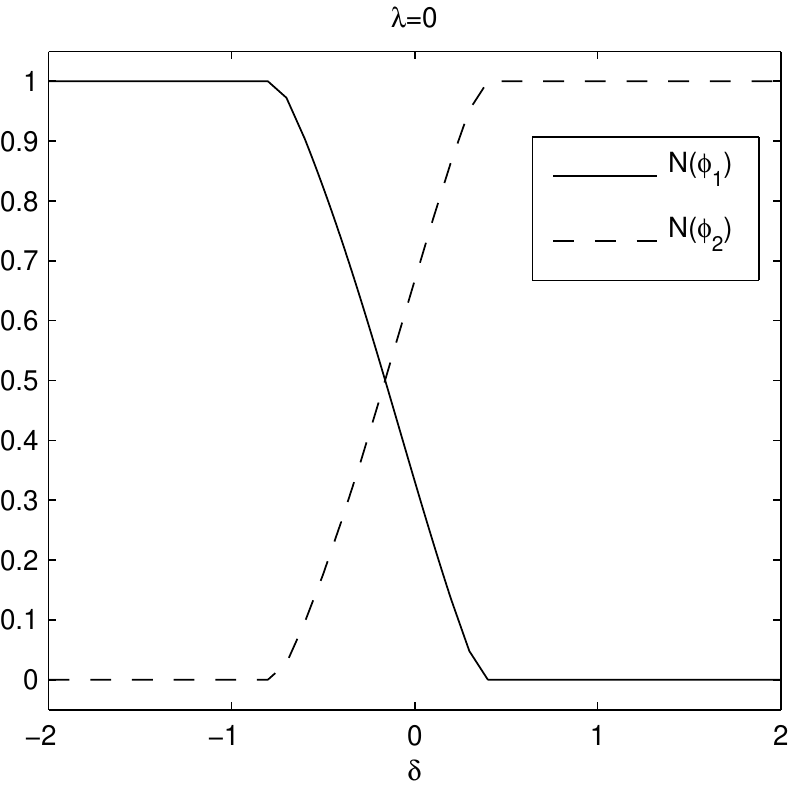,height=5cm,width=5cm,angle=0} \quad
\psfig{figure=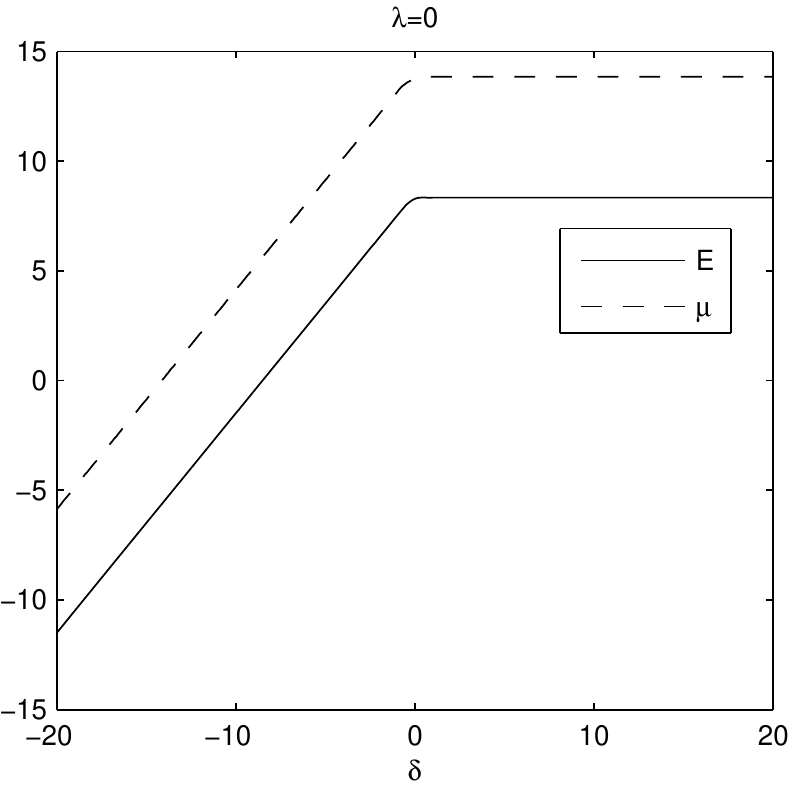,height=5cm,width=5cm,angle=0}}
\centerline{
\psfig{figure=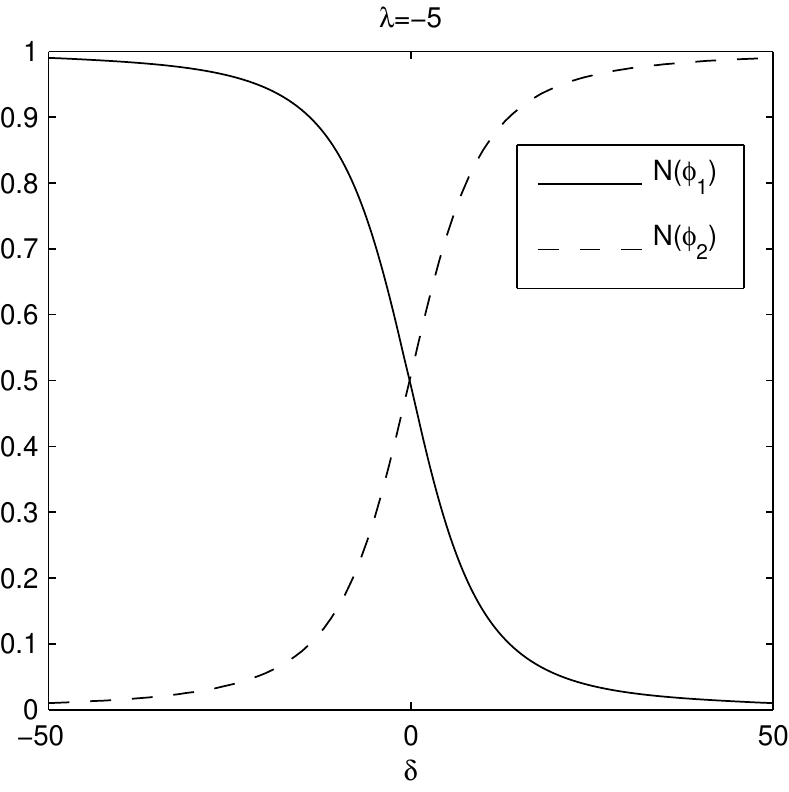,height=5cm,width=5cm,angle=0} \quad
\psfig{figure=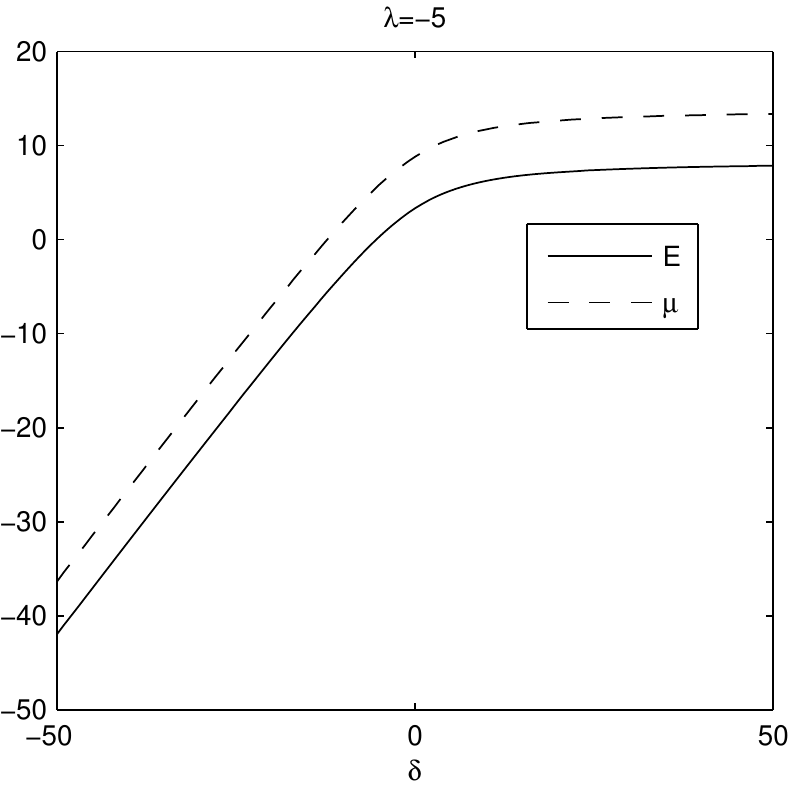,height=5cm,width=5cm,angle=0}}
\caption{Mass of each component $N(\phi_j)=\|\phi_j\|^2$ ($j=1,2$),
energy $E:=E(\Phi_g)$ and chemical potential
 $\mu:=\mu(\Phi_g)$ of the ground states in Example~\ref{exm:1:sec9} when  $\beta=100$ and
 $\ld=0,-5$  for different $\dt$.} \label{fig:fig5:sec9}
\end{figure}
\begin{example}\label{exm:1:sec9} Ground states of a two-component BEC with an external
driving field when $B$ is positive definite, i.e. we take $d=1$,
$V(x)=\frac 12 x^2$ and
$\beta_{11}:\beta_{12}:\beta_{22}=(1:0.94:0.97)\beta$ in
(\ref{eq:minimize:sec9}) \cite{Bao,BaoCai0}. In this case, since
$\ld\le 0$ and $B$ is positive definite when $ \bt>0$, thus we know
that the positive ground state $\Phi_g=(\phi_1, \phi_2)^T$ is
unique. In our computations, we take the computational domain
$U=[-16,16]$ with mesh size $h=\frac{1}{32}$ and time step
$\tau=0.1$. The initial data in (\ref{eq:init1:sec9}) is chosen as \be
\phi_1^0(x)=\phi_2^0(x)=\frac{1}{\pi^{1/4}\sqrt{2}}e^{-x^2/2},
\qquad x\in{\mathbb R}. \ee

Fig.~\ref{fig:fig1:sec9} plots the ground states $\Phi_g$ when $\dt=0$ and
$\ld=-1$ for different $\beta$. Fig.~\ref{fig:fig4:sec9}  shows  mass of each component $N(\phi_j)=\|\phi_j\|^2$ ($j=1,2$),
  energy $E:=E(\Phi_g)$ and  chemical potential
 $\mu:=\mu(\Phi_g)$ of the ground states
 when $\beta=100$ and $\dt=0,1$ for different $\ld$, and
Fig.~\ref{fig:fig5:sec9} depicts  similar results when $\beta=100$ and $\ld=0,-5$ for
different $\dt$.
\end{example}


\section{Perspectives and challenges}
\label{sec:chall}
\setcounter{equation}{0}\setcounter{figure}{0}\setcounter{table}{0}
So far, we have introduced mathematical results  and numerical methods for ground states and dynamics of a single/two component rotating/nonrotating BEC with/without dipole-dipole interactions described by mean field GPE. Despite these BEC systems, much progress has been made towards realizing other kinds of gaseous BEC, such as spinor condensates, condensates at finite temperature, Bose-Fermi mixtures, etc. These  achievements have brought great challenges  to atomic physics community and scientific computing community for  modeling, simulating and understanding various interesting phenomenons.
\subsection{Spin-1 BEC}
In earlier BEC experiments, the atoms were confined in magnetic trap
\cite{Anderson, Davis, Bradley},
in which the spin degrees of freedom is frozen.  In recent years, experimental
achievement of spin-1 and spin-2
condensates \cite{Barrett, Ueda}
offers new
regimes to study various quantum phenomena that are generally absent in
a single component condensate.
The spinor condensate is achieved experimentally when an optical trap,
instead of a magnetic trap,
is used to provide equal confinement for all hyperfine states.

The theoretical studies of spinor condensate have been carried out
in several papers since the achievement of it in experiments
\cite{Ho, Law}. In contrast to single component
condensate, a spin-$F$ ($F\in {\Bbb N}$) condensate is described
by a generalized coupled GPEs which consists of $2F+1$ equations,
each governing one of the $2F+1$ hyperfine states $(m_F = -F,
-F+1,..., F-1, F)$ within the mean-field approximation. For a
spin-1 condensate, at temperature much lower than the critical
temperature $T_c$, the three-components wave function $\Psi:=\Psi(\bx,t)
= (\psi_1(\bx,t), \psi_0(\bx,t), \psi_{-1}(\bx,t))^T$ are well
described by the following coupled GPEs \cite{Ueda,BaoLim},
\begin{align}
\label{eq:CGPE_1:sec11}
i\hbar\,\p_t\psi_{1}(\bx,t)=&\left[-\fl{\hbar^2}{2m}\nabla^2
+V(\bx)+
(c_0+c_2)\left(|\psi_1|^2+|\psi_0|^2\right)+(c_0-c_2)|\psi_{-1}|^2\right]
\psi_1\nn\\
&+c_2\,\bar{\psi}_{-1}\,\psi_0^2, \\
 \label{eq:CGPE_2:sec11}
i\hbar\,\p_t\psi_{0}(\bx,t)=&\left[-\fl{\hbar^2}{2m}\nabla^2
+V(\bx)+
(c_0+c_2)\left(|\psi_1|^2+|\psi_{-1}|^2\right)+c_0|\psi_{0}|^2\right]
\psi_0 \nn \\
&+2c_2\,\psi_{-1}\,\bar{\psi}_{0}\,\psi_1, \\
\label{eq:CGPE_3:sec11}
i\hbar\,\p_t\psi_{-1}(\bx,t)=&\left[-\fl{\hbar^2}{2m}\nabla^2
+V(\bx)+
(c_0+c_2)\left(|\psi_{-1}|^2+|\psi_0|^2\right)+(c_0-c_2)|\psi_{1}|^2\right]
\psi_{-1} \nn\\
&+c_2\,\psi_0^2\,\bar{\psi}_{1},\qquad \bx=(x,y,z)^T\in {\Bbb R}^3.
\end{align}
Here $V(\bx)$ is an external trapping potential.  There are two atomic collision terms,
$c_0 = \frac{4\pi\hbar^2}{3m}(a_0+2a_2)$ and $c_2 =
\frac{4\pi\hbar^2}{3m}(a_2-a_0)$, expressed in terms of the
$s$-wave scattering lengths, $a_0$ and $a_2$, for scattering
channel of total hyperfine spin 0 (anti-parallel spin collision)
and spin 2 (parallel spin collision), respectively. The usual
mean-field interaction, $c_0$, is positive for repulsive
interaction and negative for attractive interaction. The
spin-exchange interaction, $c_2$, is positive for
antiferromagnetic interaction and negative for ferromagnetic
interaction. The wave function is normalized according to \be
\label{eq:Ns:sec11} \|\Psi\|_2^2:=\int_{{\Bbb R}^3} |\Psi(\bx,t)|^2\;d\bx=
\int_{{\Bbb R}^3} \sum_{l=-1}^1 |\psi_l(\bx,t)|^2\;d\bx
:=\sum_{l=-1}^1 \|\psi_l\|_2^2 =N, \ee where $N$ is the total number
of particles in the condensate. This normalization is conserved by coupled GPEs (\ref{eq:CGPE_1:sec11})-(\ref{eq:CGPE_3:sec11}), and so are the   magnetization
\be \label{eq:magn:sec11} M(\Psi(\cdot,t)):=\int_{{\Bbb R}^3}
\left[|\psi_1(\bx,t)|^2- |\psi_{-1}(\bx,t)|^2\right]d\bx\equiv
M(\Psi(\cdot,0))= M \ee and the energy per particle
\begin{eqnarray}
\label{eq:energy:sec11}
E(\Psi(\cdot,t))&=&\int_{{\Bbb
R}^3}\biggl\{\sum_{l=-1}^{1} \left(\frac{\hbar^2}{2m}|\nabla
\psi_l|^2+V(\bx)|\psi_l|^2\right)
+(c_0-c_2)|\psi_1|^2 |\psi_{-1}|^2 \nn\\
&&+\frac{c_0}{2}|\psi_0|^4
+\fl{c_0+c_2}{2}\Bigl[|\psi_1|^4+|\psi_{-1}|^4
+2|\psi_0|^2
\left(|\psi_1|^2+|\psi_{-1}|^2\right)\Bigr]\nn\\
&&+c_2\left(\bar{\psi}_{-1}\psi_0^2\bar{\psi}_{1}+
\psi_{-1}\bar{\psi}_0^2\psi_{1}\right)\biggr\}\; d\bx \equiv
E(\Psi(\cdot,0)), \qquad t\ge0.
\end{eqnarray}
Then  ground states of spin-1 BEC can be defined as the minimizer of energy $E$ under the normalization and magnetization constraints \cite{BaoLim,BaoWang2,BaoChernZhang}. In particular, when the external traps for all the components are the same,
ground states for ferromagnetic and antiferromagnetic spin-1 BECs can be simplified \cite{BaoChernZhang}. Generally speaking, for spin-$F$ BEC, the complicated nonlinear terms in (\ref{eq:CGPE_1:sec11})-(\ref{eq:CGPE_3:sec11}) lead to new difficulties for mathematical analysis and numerical simulation \cite{Wang-spin}. Much work needs to be done in future, especially when rotational frame and dipole-dipole interactions are taken into account in spin-$F$ BECs \cite{Ueda}.
\subsection{Bogoliubov excitation}
The theory of interacting Bose gases, developed by Bogoliubov in 1947, is very useful and important to understand BEC in dilute atomic gases. One of the key issue is the Bogoliubov excitation.

 To describe the condensate, we have the lowest order approximation, i.e., the Gross-Pitaevskii energy by assuming that all particles are in the ground state. However, due to the interactions between the atoms, there is a small portion  occupying the excited states. Thus, if a higher order approximation of the ground state energy  is considered, excitations have to be included. Using a perturbation technique, Bogoliubov has investigated this problem and shown that the excited states of a system of interacting Bose particles can be described by a system of noninteracting quasi-particles satisfying the Bogoliubov dispersion relation.

 To determine the Bogoliubov excitation spectrum, we consider small perturbations
around the ground state of Eq. (\ref{eq:ngpe}).  For simplicity we
assume a vanishing  harmonic potential $V(\bx) = 0$ and
homogeneous density $\nu$.    A stationary state of the
GPE (\ref{eq:ngpe}) is given by $\psi(\bx, t) =
\psi(t) = e^{-i\mu t} \sqrt\nu$ with the chemical potential
\begin{equation}
  \mu = \beta \nu.
\end{equation}
Now we add a local perturbation $\xi(\bx, t)$ to the
stationary state $\psi(t)$, that is, $\psi(\bx, t) =
e^{-i\mu t} [\sqrt\nu + \xi(\bx, t)]$.  We expand the
perturbation in a plane wave basis as $\xi(\bx, t) =
\int_{\Bbb R^3}  \bigl( u_{{\bf q}}e^{i ({\bf q}\cdot\bx -
  \omega_{{\bf q}} t)} + \overline{v_{{\bf q}}} e^{-i ({\bf q}\cdot\bx -
  \omega_{{\bf q}} t)} \bigr)d{\bf q}$ and insert $\psi(\bx, t)$ into
Eq. (\ref{eq:ngpe}).  Here, $\omega_{{\bf q}}$ are the excitation
frequencies of quasimomentum ${\bf q}$ and $u_{{\bf q}}$, $v_{{\bf
  q}}$ are the mode functions.  Keeping terms
linear in the excitations $u_{{\bf q}}$ and $v_{{\bf q}}$ we find
the Bogoliubov-de Gennes equations
\begin{equation}\label{eq:u_q:sec11}
  \begin{split}
    \omega_{{\bf q}} u_{{\bf q}} &= \frac{{\bf q}^2}{2} u_{{\bf q}} +
    \nu \beta(v_{{\bf q}} + u_{{\bf q}}),\\
    - \omega_{{\bf q}} v_{{\bf q}} &= \frac{{\bf q}^2}{2} v_{{\bf q}} +
    \nu  \beta (v_{{\bf q}} + u_{{\bf q}}).
  \end{split}
\end{equation}
 Then we can find the eigenenergies of
 Eq. (\ref{eq:u_q:sec11}) by solving the eigenvalue problem.  The resulting Bogoliubov
energy $E_B({\bf q}) = \omega_{{\bf q}}$ is determined by
\begin{equation}
    E_B^2({\bf q}) = \frac{{\bf q}^2}{2} \left(\frac{{\bf q}^2}{2} + 2\beta\nu\right).
\end{equation}
When an external potential is considered, the Bogoliubov energy would be more complicated.
In 1999, the Bogoliubov excitation spectrum was observed for the first time in atomic BEC \cite{Stamper}, using light scattering. Later in 2008, observation of Bogoliubov excitations was announced in
exciton-polariton condensates \cite{Utsun}. Such elementary excitations are crucial in understanding various phenomenon in BEC.
\subsection{BEC at finite temperature}
The process of creating a BEC in a trap by means of evaporative cooling starts in a
regime covered by the quantum Boltzmann equation (QBE) and finishes in a regime where the GPE is expected to be valid. The GPE is capable to describe the
main properties of the condensate at very low temperatures, it treats the condensate
as a classical field and neglects quantum and thermal fluctuations. As
a consequence, the theory breaks down at higher temperatures where the non-condensed fraction of the gas cloud is significant. An approach which allows
the treatment of both condensate and noncondensate parts simultaneously
was developed in \cite{Zaremba,BaoPM45}.

The resulting equations of motion reduce to a generalized GPE for the
condensate wave function coupled with a semiclassical QBE for the thermal
cloud:
\begin{equation}\label{eq:qbe:sec11}
\begin{split}
&i\hbar\p_t\psi(\bx,t)=\left[-\frac{\hbar^2}{2m}\nabla^2+(n_c(\bx,t)+2n(\bx,t))g-iR(\bx,t)\right]\psi,\\
&\frac{\p F}{\p t}+\frac{\bp}{m}\cdot\nabla_{\bx} F-\nabla_{\bx}U\cdot\nabla_{\bp}F=Q(F)+Q_c(F),
\end{split}
\end{equation}
where $n_c(\bx,t)=|\psi(\bx,t)|^2$ is the condensate density, $F:=F(\bx,\bp,t)$ describes the distribution of thermal atoms in the phase space and it gives the particle number with momentum $\bp$ at position $\bx$ and time $t$ in the thermal cloud. $n(\bx,t)=\int_{\Bbb R^3}F(\bx,\bp,t)/(2\pi\hbar)^3d\bp$, $V(\bx)$ is the confining potential and $g=4\pi \hbar^2a_s/m$. The collision integral $Q(F)$  is given by
\begin{align*}
Q(F)=&\frac{2g^2}{(2\pi)^5\hbar^7}\int_{\Bbb R^3\times\Bbb R^3\times\Bbb R^3}\delta(\bp+\bp_\ast-\bp^\prime-\bp_\ast^\prime)\times
\delta(\epsilon+\epsilon_\ast-\epsilon^\prime-\epsilon_\ast^\prime)\\
&\times[(1+F)(1+F_\ast)F^\prime F^\prime_\ast-FF_\ast(1+F^\prime)(1+F^\prime_\ast)]\,d\,\bp_\ast\,d\bp^\prime\,d\bp_\ast^\prime,
\end{align*}
where $\epsilon=U(\bx,t)+\bp^2/2m$, $U(\bx,t)=V(\bx)+2gn_c(\bx,t)+2gn(\bx,t)$ and $\delta(\cdot)$ is the Dirac distribution.
$Q_c(F)$ which describes collisions between condensate and non condensate particles
 is given by
 \begin{align*}
Q_c(F)=&\frac{2g^2n_c}{(2\pi)^2\hbar^4}\int_{\Bbb R^3\times\Bbb R^3\times\Bbb R^3}\delta(m{\bf v}_c+\bp_\ast-\bp^\prime-\bp_\ast^\prime)\\
&\times
\delta(\epsilon_c+\epsilon_\ast-\epsilon^\prime-\epsilon_\ast^\prime)[\delta(\bp-\bp_\ast)
-\delta(\bp-\bp^\prime)-\delta(\bp-\bp_\ast^\prime)]\\
&\times[(1+F_\ast)F^\prime F^\prime_\ast-F_\ast(1+F^\prime)(1+F^\prime_\ast)]\,d\,\bp_\ast\,d\bp^\prime\,d\bp_\ast^\prime,
\end{align*}
where
\be
\epsilon_c(\bx,t)=\frac12m {\bf v}_c(\bx,t)^2+\mu_c(\bx,t),
\ee
and ${\bf v}_c$  is the quantum hydrodynamic velocity, $\mu_c$ is the effective  potential acting on the condensate \cite{Zaremba,Griffin}. $R(\bx,t)$ is then written as
\be
R(\bx,t)=\frac{\hbar}{2n_c}\int_{\Bbb R^3}\frac{Q_c(F)}{(2\pi\hbar)^3}\,d\bp.
\ee
Note that for low temperatures $T\to 0$  we have $n,R\to 0$ and we recover the
conventional GPE. The system (\ref{eq:qbe:sec11}) is normalized as $N_c(0)=N_c^0$ and $N_t(0)=N_t^0$ with
\be
N_c(t)=\int_{\Bbb R^3}|\psi(\bx,t)|^2\,d\bx,\quad N_t(t)=\int_{\Bbb R^3}|n(\bx,t)|^2d\,\bx, \quad t\ge0,
\ee
where $N_c^0$ and $N_t^0$ are the number of particles in the condensate and thermal
cloud at time $t=0$, respectively. It is easy to see from the equations (\ref{eq:qbe:sec11}) that the total number of particles defined as $N_{\rm total}(t)=N_c(t)+
N_t(t)= N^0_{\rm total}=N^0_c+N^0_t$ is conserved. For this set of equations, the GPE part can be solved efficiently, and the main trouble comes from the Boltzmann equation part. Alternatively,  projected GPE model is  also used for simulating BEC at finite temperature \cite{Davis-ft1,Davis-ft2}.

\section*{Acknowledgments}
We would like to thank our collaborators: Naoufel Ben Abdallah, Francois Castella, I-Liang Chern, Qiang Du, Jiangbin Gong, Dieter Jaksch,  Shi Jin, Baowen Li, Hai-Liang Li, Fong-Yin Lim, Peter A. Markowich,  Florian M\'ehats, Lorenzo Pareschi, Han Pu, Matthias Rosenkranz, Christian Schmeiser, Jie Shen, Weijun Tang, Hanquan Wang, Rada M. Weish\"aupl, Yanzhi Zhang, etc.
for their significant contributions and fruitful collaborations on the topic over the last decade.
We have learned a lot from them during the fruitful collaboration and interaction.
This work was supported by the Academic Research
Fund of Ministry of Education of Singapore grant
R-146-000-120-112.


\medskip
Received xxxx 20xx; revised xxxx 20xx.
\medskip

\end{document}